\documentclass[12pt, lof]{puthesis}
\usepackage{hyperref}

\usepackage{graphicx}

\usepackage{amsmath,amsthm}
\usepackage{amsfonts}
\usepackage{amssymb}
\usepackage{mathrsfs}
\usepackage{color}
\usepackage{xspace}
\usepackage[T1]{fontenc}
\usepackage[capitalize]{cleveref}
\usepackage{isomath}
\usepackage{epigraph}
\usepackage{natbib}
\usepackage[title,titletoc]{appendix}
\setcitestyle{authoryear,open={(},close={)}}
\allowdisplaybreaks[1] 

\newcommand{\gke}{\textsc{Gkeyll}\xspace}
\newcommand{\dx}[1]{\textnormal{d}#1}
\newcommand{\pderiv}[2]{
\frac{\partial #1}{\partial #2}
}
\newcommand{\deriv}[2]{
\frac{\dx{#1}}{\dx{#2}}
}
\newcommand{\pderivInline}[2]{
\partial #1/\partial #2
}
\renewcommand{\vec}[1]{\ensuremath{\mbox{\boldmath$ {#1} $}}} 
\newcommand{\uv}[1]{\ensuremath{\mathbf{\hat{#1}}}} 
\newcommand{\uvg}[1]{\boldsymbol{\hat{\boldsymbol{#1}}}}
\newcommand{\jac}{\mathcal{J}}
\let\ccdot=\cdot
\renewcommand{\cdot}{\boldsymbol{\ccdot}}

\makeatletter
\newcommand{\dashover}[2][\mathop]{#1{\mathpalette\df@over{{\dashfill}{#2}}}}
\newcommand{\fillover}[2][\mathop]{#1{\mathpalette\df@over{{\solidfill}{#2}}}}
\newcommand{\df@over}[2]{\df@@over#1#2}
\newcommand\df@@over[3]{%
  \vbox{
    \offinterlineskip
    \ialign{##\cr
      #2{#1}\cr
      \noalign{\kern1pt}
      $\m@th#1#3$\cr
    }
  }%
}
\newcommand{\dashfill}[1]{%
  \kern-.5pt
  \xleaders\hbox{\kern.5pt\vrule height.4pt width \dash@width{#1}\kern.5pt}\hfill
  \kern-.5pt
}
\newcommand{\dash@width}[1]{%
  \ifx#1\displaystyle
    2pt
  \else
    \ifx#1\textstyle
      1.5pt
    \else
      \ifx#1\scriptstyle
        1.25pt
      \else
        \ifx#1\scriptscriptstyle
          1pt
        \fi
      \fi
    \fi
  \fi
}
\newcommand{\solidfill}[1]{\leaders\hrule\hfill}
\makeatother

\theoremstyle{plain}
\newtheorem{proposition}{Proposition}

\newtheorem{lemma}{Lemma}
\theoremstyle{remark}
\newtheorem*{remark}{Remark}

\author{Noah Roth Mandell}
\adviser{Gregory W. Hammett}
\title{Magnetic Fluctuations in Gyrokinetic Simulations of Tokamak Scrape-Off Layer Turbulence}
\abstract{
Understanding turbulent transport physics in the tokamak edge and scrape-off layer (SOL) is critical to developing a successful fusion reactor. The dynamics in these regions plays a key role in achieving high fusion performance by determining the edge pedestal that suppresses turbulence in the high-confinement mode (H-mode). Additionally, the survivability of a reactor is set by the heat load to the vessel walls, making it important to understand turbulent spreading of heat as it flows along open magnetic field lines in the SOL. Large-amplitude fluctuations, magnetic X-point geometry, and plasma interactions with material walls make simulating turbulence in the edge/SOL more challenging than in the core region, necessitating specialized gyrokinetic codes. Further, the inclusion of electromagnetic effects in gyrokinetic simulations that can handle the unique challenges of the boundary plasma is critical to the understanding of phenomena such as the pedestal and  edge-localized modes, for which electromagnetic dynamics are expected to be important.

In this thesis, we develop the first capability to simulate electromagnetic gyrokinetic turbulence on open magnetic field lines. This is an important step towards comprehensive electromagnetic gyrokinetic simulations of the coupled edge/SOL system. By using a continuum full-$f$ approach via an energy-conserving discontinuous Galerkin (DG) discretization scheme that avoids the Amp\`ere cancellation problem, we show that electromagnetic fluctuations can be handled in a robust, stable, and efficient manner in the gyrokinetic module of the \gke code. We then present results which roughly model the scrape-off layer of the National Spherical Torus Experiment (NSTX), and show that electromagnetic effects can affect blob dynamics and transport. We also formulate the gyrokinetic system in field-aligned coordinates for modeling realistic edge and scrape-off layer geometries in experiments. A novel DG algorithm for maintaining positivity of the distribution function while preserving conservation laws is also presented. 
}
\acknowledgements{
I have been extremely fortunate to have been supported by an incredible group of family, friends, mentors, teachers, and collaborators over my academic career. 

First I would like to thank my thesis adviser, Greg Hammett. His physical intuition and insight has been a tremendous resource and inspiration throughout my time at Princeton. ``Physics is about thinking slowly,'' Greg told me on numerous occasions as we puzzled over a problem in his office.\footnote{Greg attributes this quote to Fred Skiff, a very deep thinker.} He is masterful at slowly and carefully thinking through every part of a calculation, model, or result. This attention to detail, taking the little steps that are needed to realize big ideas, is among the most important things that I have learned from Greg. I look forward to many more years of collaboration and friendship.

I would also not be the scientist I am today without the mentorship and friendship of Bill Dorland. Bill introduced me to the world of fusion, computational plasma physics, and turbulence when I was a young undergraduate at the University of Maryland. I was immediately drawn to his passion and his big, exciting ideas. Thank you for patiently teaching me so much, believing in me enough to give me such an ambitious project as an undergraduate, and for the continued mentorship and collaborations while I was at Princeton. (Remember when he almost stole me for good in year 3, Greg? Haha.) More importantly, thank you for all the support, guidance, advice, and inspiration over the years. 

Thank you to Ammar Hakim for fearlessly leading the development of the \gke code. Ammar's computational insights have been invaluable throughout my thesis work, and there is no question that I have become a better software developer and physicist because of Ammar's help. 

Working on the \gke code has been incredibly stimulating and rewarding, mostly thanks to the amazing people in the \gke group. All of this would not have been possible without all of your hard work. Thank you to Jimmy Juno, Mana Francisquez, Petr Cagas, Tess Bernard, Rupak Mukherjee, Liang Wang, and many others for your constant support, often involving late-night Slack chats and code commits. I also thank Eric Shi, whose thesis paved the path for much of my work.

Thank you to my thesis readers, Matt Kunz and Walter Guttenfelder. Their careful reading and insightful feedback have greatly improved the text. Greg Hammett also provided valuable comments. Also, thank you to my thesis committee members: Matt, Walter, Greg, Ammar, and Sir Steve Cowley.

I thank Stewart Zweben for mentoring me for my first-year project on NSTX gas-puff imaging, taking a student who had already made up his mind that he wanted to be a theorist and having the patience to show him the nuances of tokamak experiments.

To all my fellow grad students and post-docs at the lab, thank you for making my time at Princeton so enjoyable. I especially enjoyed APS gallivanting with Denis St. Onge and Brian Kraus and others; Denis' brisket parties; watching Packers football games with Peter Bolgert, Daniel Ruiz, Jeff Lestz and others; playing softball with the Tokabats; and of course, playing ping pong. Some of my fondest memories are from that little room tucked away behind Science Ed, playing for hours with Peter B. (original commissioner of the PPPL Ping Pong League, or PPPLPPL), Daniel R., Jonathan Ng, Jeff L., Vasily Geyko, Brian K., Lee Ellison, Jack Matteucci, Hongxuan Zhu, Joaquim Loizu, Vinicius Duarte, David Pfefferle, Jacob Schwartz, Charles Swanson, Ian Ochs, Alex Glasser, Elijah Kolmes, Andy Alt, Nick McGreivy, and many others/left-handed-alter-egos. The Bolgert Open tournament was an annual summer tradition (could never get past Hongxuan in the finals...), and Brian even developed a computerized ratings system. We probably spent way too much time playing ping pong when we should've been working, but we all got pretty good and it was a lot of fun. I am also honored to have been a member of the party office, which has a rich history of housing outstanding Princeton graduates. 

Thank you to Dara Lewis and Beth Leman for all their help with navigating the administrative details of being a graduate student at PPPL.

There are many more teachers and professors to thank for helping me along the way, including the outstanding professors and lecturers in the Princeton Program in Plasma Physics, but I would like to say a special thank you to my high school physics teacher, Pieter Kreunen. My love of physics started back in Mr. Kreunen's classroom, getting shocked by Van der Graff generators and building robots. Thank you for inspiring me and encouraging me to pursue physics those many years ago.

I have been blessed with an amazingly loving and supportive family. Mom and Dad, I am forever grateful for every opportunity you gave me growing up, and the sacrifices you made to give them to me. Madison, you mean more to me than you know, and I am so proud to have you as my sister. Thank you for supporting me unconditionally every step of the way. I am also so grateful for the invaluable months I have been able to spend over the past year\footnote{I do not want to say this without solemnly acknowledging the COVID-19 pandemic that has devastated the country and the world for over a year. We have been fortunate to keep our health, but millions have not been so lucky. In these strange times, the Stoppard quote at the beginning, originally included as a quip about mathematical consistency, has truly taken on a new meaning...} back home in Gurnee with Mom, Dad, and Madison, and in Maine with Calla, Arthur, Mariah, Stephen and Jack.

Lastly, most importantly:
Haley, we made it! You've been with me every step of this journey, and I wouldn't have gotten here without you. This is yours as much as it is mine. I am inspired by you every day. Thank you for always lovingly supporting me and believing in me. I wouldn't trade a single day with you for the world.
\noindent
**************************************************************************
\\\\ \noindent
I was very fortunate to be funded for four years by the Department of Energy Computational Science Graduate Fellowship, provided under DOE grant DE-FG02-97ER25308. Thank you for the generous support, and for all the friends I made as part of the fellowship program. Additional funding came from DOE contract DE-AC02-09CH11466.
Simulations were performed on the Perseus and Eddy clusters at Princeton University/PPPL, and the Cori system at the National Energy Research Scientific Computing Center. 

}
\dedication{to Haley, with love}

\begin{document}

\makefrontmatter

\chapter{Introduction} \label{ch:intro}

\section{Motivation: the promise of fusion energy}

After Einstein first discovered the relationship between energy and mass, governed by the iconic equation $E= mc^2$, it was soon realized that this relationship was the key to the process that produces the energy of the Sun and stars: nuclear fusion. Fifteen years after Einstein's discovery, British astrophysicist Arthur Eddington was the first to describe how the Sun and similarly-sized stars create their energy by fusing hydrogen atoms into helium. Eddington realized that the tiny difference in mass between a helium atom and its constituent hydrogen parts, as had been recently shown by Aston, meant that the `missing' mass is converted into energy via Einstein's equation. ``The store is well nigh inexhaustible, if only it could be tapped,'' Eddington said in a lecture at the annual meeting of the British Association for the Advancement of Science in Cardiff \citep{eddington1920}. Thus began the promise of man-made fusion as a terrestrial energy source. 

One of the main allures of fusion power is the abundance of the fuel. Unlike fossils fuels, which at current energy-consumption rates would be burned through in less than 1,000 years (causing catastrophic global warming in the process), the fuel for fusion is virtually limitless because it can be extracted from seawater. The most promising fusion reaction for use on Earth is not the proton-proton reaction that powers the Sun, but an easier-to-initiate reaction between deuterium ($^2$H) and tritium ($^3$H):
\begin{equation}
    ^2\mathrm{H} + {}^3\mathrm{H} \rightarrow {}^4\mathrm{He}\ (3.5\ \mathrm{MeV}) + \mathrm{n}\ (14.1\ \mathrm{MeV}).
\end{equation}
Deuterium is a naturally abundant isotope of hydrogen that can be readily extracted from seawater at minimum cost, with each liter of seawater containing $\sim 0.02$ g of deuterium. Tritium is not naturally abundant due to its relatively short half-life of 12.3 years. However, fusion reactors can use the energetic neutron from the D-T reaction to breed their own tritium via lithium ($^6$Li) blankets via the reaction
\begin{equation}
    \mathrm{n}  + {}^6\mathrm{Li} \rightarrow {}^4\mathrm{He}\ (2.1\ \mathrm{MeV}) + {}^3\mathrm{H}\ (2.7\ \mathrm{MeV}).
\end{equation}
Current world lithium supplies are approximately 13.5 million tons, but lithium is also contained in seawater at a concentration of 0.2 mg per liter. Thus there is enough fusion fuel readily available in the oceans to power the Earth for millions of years, several orders of magnitude longer than other terrestrial fuel sources other than solar energy \citep{cowley2016}. 

Other major benefits of fusion are its minimal environmental impact and operational safety. Fusion would be clean and virtually carbon-neutral, emitting no greenhouse gases and not contributing to climate change. While this benefit is also shared by nuclear fission (where energy is produced by splitting the nuclei of heavy elements like uranium), fusion has the additional advantage that it has no long-lived radioactive byproducts. Helium is an inert gas, and while the energetic neutron from the D-T reaction can transmute the materials in the walls of a reactor and make them radioactive over time, the use of low-activation wall materials would make the waste substantially safer than fission waste. Further, a fusion power plant would be safer to operate than a fission power plant, as there is no runaway meltdown scenario. Unlike fission reactions, fusion reactions immediately shut down when the fuel is removed or cooled.

Unfortunately, a fusion reaction is very difficult to get started. To produce fusion, the positively-charged fuel nuclei must have enough energy to overcome the repulsive Coulomb force between them; only then can they get close enough to fuse together via the nuclear strong force and release energy. Unlike in the Sun, where immense gravitational pressure creates conditions necessary for fusion, terrestrial fusion must achieve fusion conditions via other methods. The most promising approach is to heat a gas of deuterium and tritium to very high temperature, so that particles have enough energy that random collisions can overcome the Coulomb repulsion. The energies required, of order $10$ keV, are well above the electron binding energy, resulting in the fuel gases ionizing fully and becoming a plasma.

At these extreme temperatures, the fuel cannot simply be contained by material walls. Instead, we can take advantage of the fact that a plasma is composed of charged particles. In the presence of magnetic fields, charged particles spiral helically around the field lines, providing a way to control the particle motion. Particles are still free to move parallel to the magnetic field lines, so one way to confine them is to wrap the field lines into a torus shape, creating a `magnetic bottle'. However, a simple ring configuration leads to vertical particle drifts, making the configuration intrinsically unstable. This problem can be overcome by twisting the field lines into a helical shape wrapping around the torus. The twisting magnetic field guides the particles up or down, counteracting the drift motion and enhancing confinement. This configuration is the basis of both the tokamak and stellarator concepts. In a tokamak, the twist in the field is produced by a current driven toroidally through the plasma, whereas in a stellarator the twist is produced by shaped helical field coils. 
These two configurations are shown schematically in \cref{fig:tokamak-stellarator}. 
We will focus on tokamaks in this work.

\begin{figure}
    \centering
    \includegraphics[width=\textwidth]{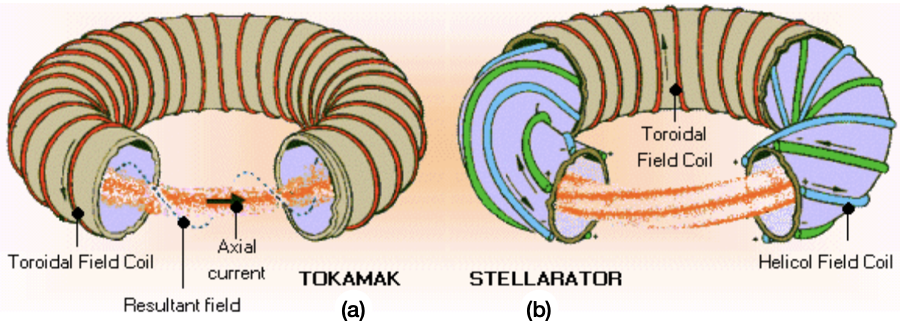}
    \caption[Schematic diagrams for a tokamak and a stellarator.]{Schematic diagrams for a tokamak (left) and a stellarator (right). In a tokamak, the poloidal component of the magnetic field that gives the helical twist is produced by current driven toroidally through the plasma. In a stellarator, specially-shaped helical field coils produce the poloidal component of the magnetic field. (Source: CEA)}
    \label{fig:tokamak-stellarator}
\end{figure}

\section{Turbulent transport in fusion plasmas}

With the plasma confined (to lowest order), the plasma can be heated without direct contact with the vessel walls. The goal is then to keep the plasma hot enough and dense enough for long enough for fusion to occur. This is the idea behind the fusion triple product, $n T \tau_E$, where $n$ is the plasma density, $T$ is the mean temperature, and $\tau_E$ is the energy confinement time. Lawson's criterion gives the condition for `break-even', at which point the plasma's self-heating from fusion exceeds its losses \citep{wesson2005},
\begin{equation}
    n T \tau_E \geq 10^{21}\ \mathrm{keV\cdot s/m^3}.
\end{equation}
In practice, the energy confinement time has proved to be the most challenging component to maximize. It is defined as
\begin{equation}
    \tau_E = \frac{W}{P_\mathrm{loss}},
\end{equation}
where $W$ is the energy content of the plasma and $P_\mathrm{loss}$ is the energy loss rate. 
While the particles are well-confined along the magnetic field lines so that the parallel (with respect to the field lines) confinement time is very large, particles can also diffuse radially outward, perpendicular to the magnetic field.
As a result, the perpendicular diffusion rate is the limiting factor on the energy confinement time. The transport was originally thought to be dominated by collisional processes (yielding ``classical'' and geometry-modified ``neoclassical'' transport), but these processes were found to greatly under-predict the transport seen in tokamaks, with most of the measured transport denoted ``anomalous''. It is now recognized that plasma turbulence is responsible for this anomalous component, so that tokamak plasma confinement is dominated by turbulent transport. 

In the tokamak core, turbulence is driven by small-scale, low-frequency ``micro-instabilities''. These instabilities feed off the density and temperature gradients that inherently result from the requirement that the temperature must be low ($\sim 10^3$ K) near the walls of the device but very hot in the core ($\sim 10^8$ K). Despite fluctuation levels of only order $1\%$, core turbulence leads to significant transport of particles, momentum, and heat. The fluctuations typically have length scales perpendicular to the background magnetic field on the order of the ion gyroradius $\rho_i=v_{ti}/\Omega_i$ and frequencies (and growth rates) on the order of the diamagnetic drift frequency, $\omega_\ast = k_\theta \rho_i v_{ti}/L_n$, where $k_\theta$ is a typical poloidal wavenumber, $v_{ti} = \sqrt{T_i/m_i}$ is the ion thermal speed, $\Omega_i= Ze B/m_i$ is the ion cyclotron frequency, and $L_n = -(\dx{\ln n}/\dx{r})^{-1}$ is the density scale length. Given these length and time scales, we can make a simple mixing-length estimate of the diffusivity,
\begin{equation}
    D \sim \frac{(\Delta x)^2}{\Delta t} \sim \rho_i^2 \omega_\ast \sim \rho_i^2 k_\theta \rho_i \frac{v_{ti}}{L_n}.
\end{equation}
Taking $k_\theta \rho_i \sim 1$ yields the so-called gyro-Bohm diffusivity, $D_\mathrm{gB} \sim \rho_i^2 v_{ti}/L_n$. While this gives a rough scaling of the transport, additional theory and numerical simulation are required for meaningful understanding and quantitative prediction of turbulent transport in tokamaks.


\section{The boundary plasma}

\begin{figure}
    \centering
    \includegraphics[width=.5\textwidth]{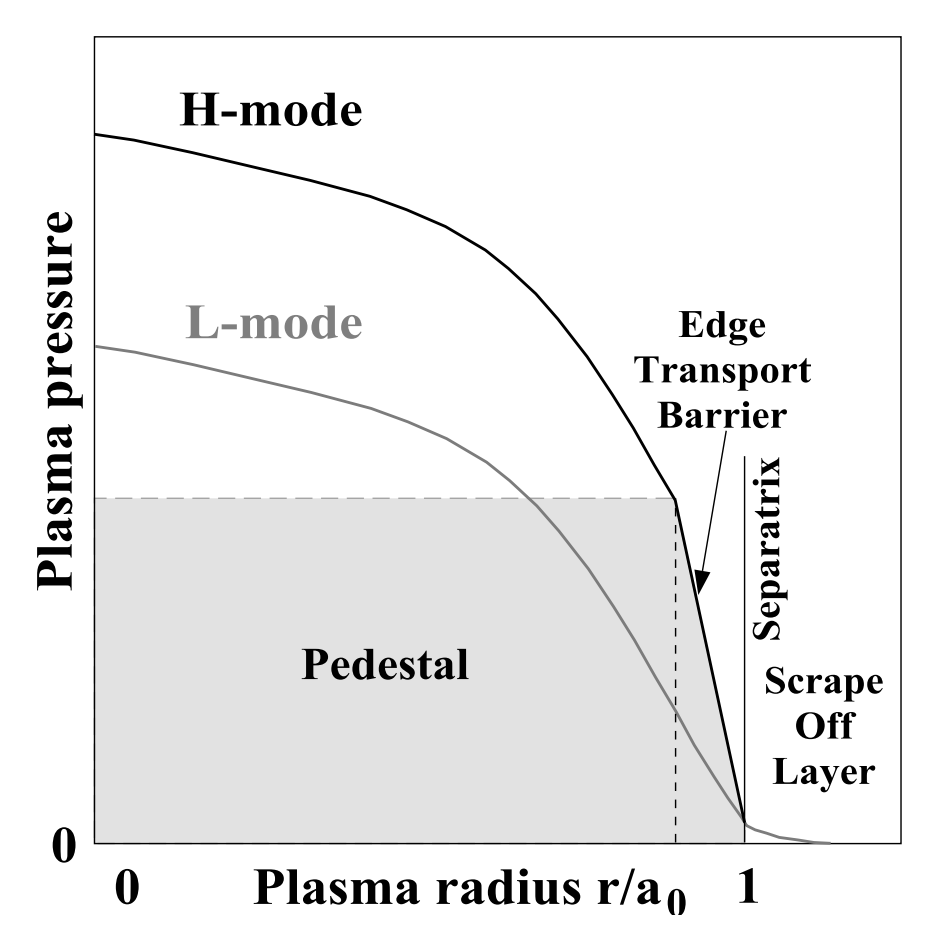}
    \caption[Diagram of plasma profiles in L- and H-mode regimes.]{Diagram of plasma profiles in L- and H-mode regimes. In H-mode profiles are elevated by the pedestal, which is the result of a transport barrier that forms near the last closed flux surface. Figure from \citep{lang2013}.}
    \label{fig:pedestal}
\end{figure}

While the core attracted much of the focus in the early days of fusion research, it was soon realized that the edge and scrape-off layer (SOL), which we together refer to as the boundary plasma, greatly affect the device performance and dynamics. Performance is strongly determined by the edge profiles because core profiles of density and temperature are relatively stiff \citep{doyle2007,kinsey2011}. A primary example of this is the high-confinement mode (H-mode), first discovered by \citet{wagner1982}, where a steep-gradient transport barrier region called the pedestal forms in the edge and raises the core profiles (as if they were standing on a pedestal), as shown in \cref{fig:pedestal}. Strong sheared poloidal flows are observed in this region, correlated with a reduction in turbulent fluctuation levels and fluxes. Understanding pedestal formation and predicting the pedestal height are of great current interest \citep{snyder2011}, and a major motivator for first-principles modeling of the boundary plasma.

The scrape-off layer (SOL) is the region outside the last closed flux surface (LCFS) where the field lines are open and terminate on material walls. Charged particles move freely along the field lines and are lost when they strike the walls (until they recombine and reenter the plasma as cold neutrals, a process called recycling). The dynamics in the SOL is primarily set by the interplay between particles and heat crossing the LCFS from the edge, parallel losses to the walls, cross-field turbulent transport, and plasma surface interactions (PSIs), including recycling and impurity fluxes. As a result of these processes, the SOL plasma is rather cold, with $T_e \sim 10-100$ eV. 

\begin{figure}
    \centering
    \includegraphics[width=\textwidth]{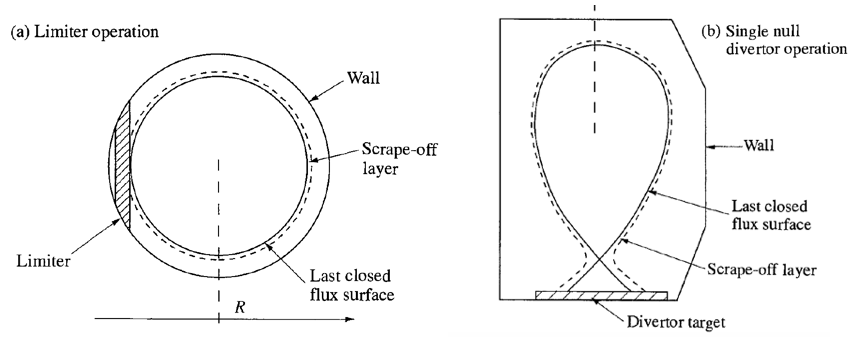}
    \caption[Diagram of limiter and divertor configurations of tokamak operation.]{Diagram of limiter $(a)$ and divertor $(b)$ configurations of tokamak operation. Figure from \citet{wesson2005}.}
    \label{fig:limiter-divertor}
\end{figure}

The termination points of the open field lines in the SOL are determined by whether the tokamak is operated in a limiter or divertor configuration, as shown in the diagram in \cref{fig:limiter-divertor}. In the former, material limiters are placed at various locations on the first wall. The field lines that intersect the limiters then define the SOL. While the limiter configuration is operational,
the divertor configuration is generally preferred in high-performance devices. In the divertor configuration, an external current in the direction of the plasma current is applied at the top and/or bottom of the device, resulting in the formation of X-point nulls. This moves the plasma-wall interactions onto the divertor targets, which are much further away from the main core plasma than limiter plates. This is beneficial since neutrals and impurities released from the divertor plates cannot directly enter the core plasma. Divertor configurations are also preferable for handling the heat exhaust requirements of the SOL and removing impurities and fusion ash via pumping \citep{wesson2005}.

\subsection{Intermittent SOL transport and blob dynamics} \label{sec:blob-dynamics}

The cross-field transport in the SOL is highly intermittent. Unlike in the core, where the transport is dominated by small fluctuations, fluctuations in the SOL can be comparable to the equilibrium quantities. This is primarily due to the convective transport of coherent structures of enhanced density and temperature called blobs or filaments. These structures propagate quasi-ballistically, moving radially outwards and resulting in significant particle and heat transport.
Blobs are highly extended along the field line with parallel lengths $\sim 1-10$ m and much smaller scales $\sim 1-10$ cm perpendicular to the field \citep{zweben2017}. The intermittent nature of blob transport suggests that a simple picture of diffusive transport is inadequate \citep{naulin2007}. Instead, the transport is avalanche-like, suggesting that the system gets pushed up against some critical gradient threshold and then intermittently releases bursts of transport when the threshold is exceeded \citep{labombard2005,labombard2008}. An expansive review of experimental evidence and theoretical understanding of intermittent edge turbulence and blobs is given by \citet{dippolito2011}.

The basic mechanism of blob transport is plasma polarization due to magnetic drifts. On the outboard side of the tokamak, the curvature and $\nabla B$ drifts are vertical, with ions drifting in one direction and electrons drifting in the other. The resulting charge polarization produces a vertical electric field across the blob, giving a radially outward $E\times B$ drift \citep{krasheninnikov2001}. This is shown schematically in \cref{fig:blob-diagram}. 

\begin{figure}[t]
    \centering
    \includegraphics[width=.55\textwidth]{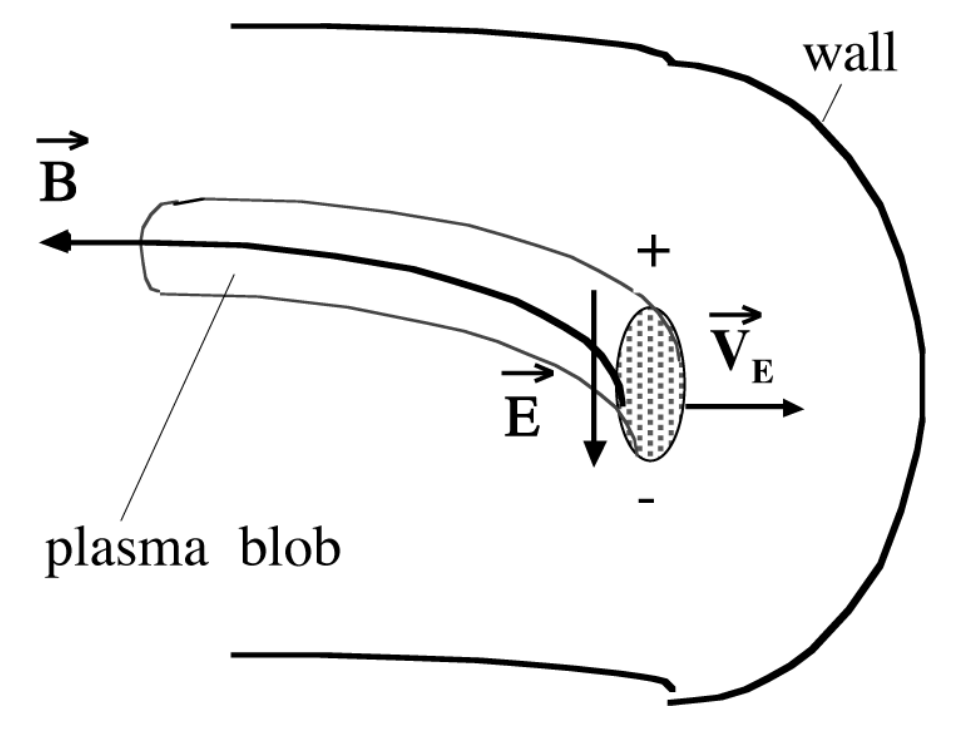}
    \caption[Diagram of basic blob transport mechanism from charge polarization.]{Diagram of basic blob transport mechanism from charge polarization. Figure from \citep{krasheninnikov2008}.}
    \label{fig:blob-diagram}
\end{figure}

\begin{figure}[t]
    \centering
    \includegraphics[width=.7\textwidth]{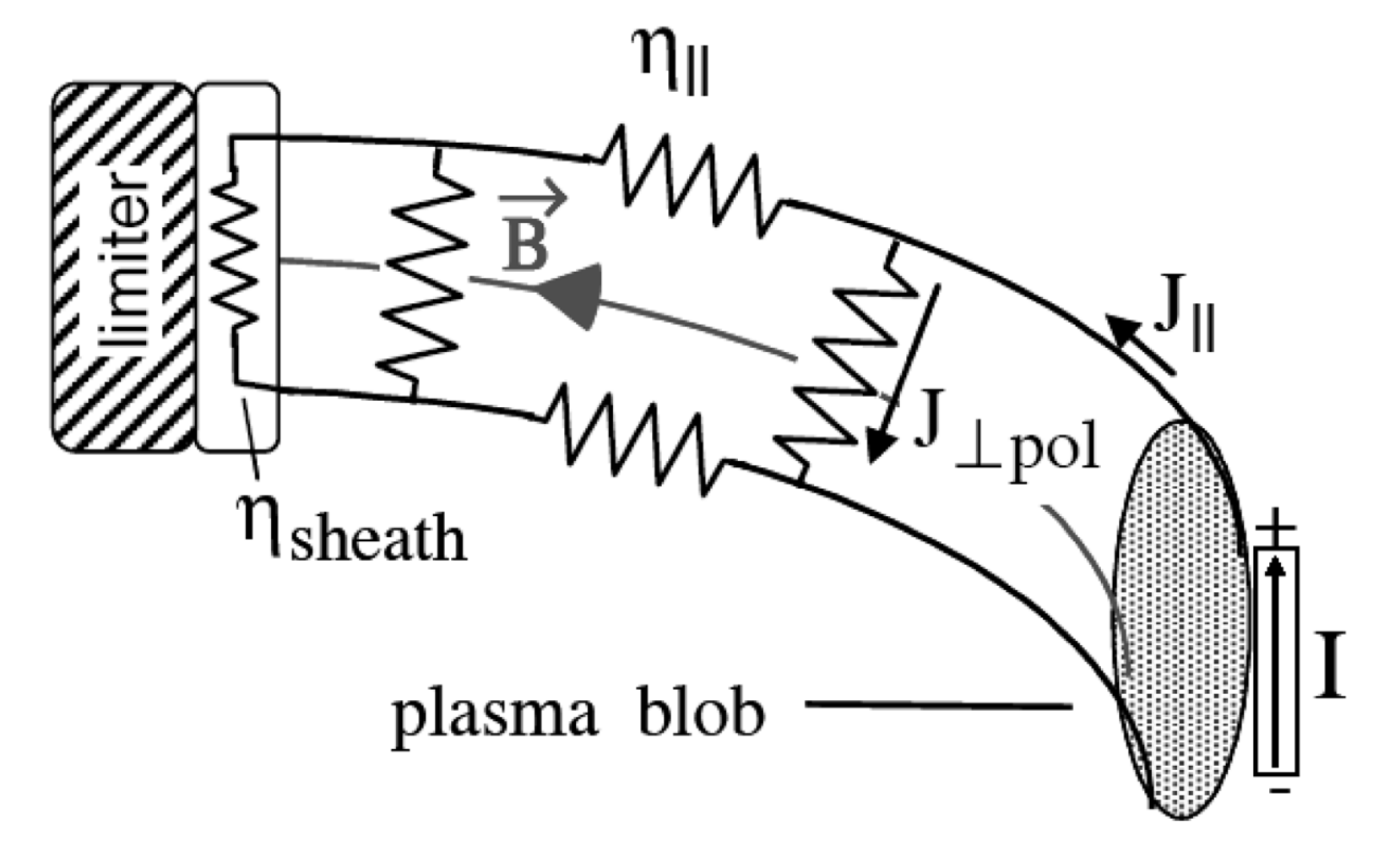}
    \caption[Diagram of equivalent blob circuit.]{Diagram of equivalent blob circuit. Charge polarization in the blob provides a constant current source. The current flows along the field lines in a dipolar structure. The current can be closed through the sheath or through the plasma, depending on the resistivities of the respective current paths. Figure from \citep{krasheninnikov2008}. Note that the circuit element through which the polarization current flows may be more appropriately characterized as a capacitor \citep{xu2010}.
    }
    \label{fig:blob-circuit}
\end{figure}

The magnitude of the blob electric field, and thereby the blob speed, is affected by the balance between the polarization current and parallel currents. To explain this, it is useful to visualize the currents in the blob via a blob equivalent circuit \citep{myra2005,krasheninnikov2008, xu2010}, as shown in the circuit diagram in \cref{fig:blob-circuit} from \citet{krasheninnikov2008}. (Note that the circuit element through which the polarization current flows may be more appropriately characterized as a capacitor due to plasma inertia \citep{xu2010}.) The magnetic drifts act as a local current source. At constant current, the potential drop across the blob is determined by the resistance in the circuit. 
If the plasma has low resistivity ($\eta_\parallel$), the current flows freely along the field lines to the sheath, and the effective sheath resistance ($\eta_\mathrm{sheath}$) will determine the blob potential and thereby the blob velocity. This is known as the sheath-limited regime, and it can lead to reduced blob speed and transport as the blob polarization current can be effectively shorted out by the current closure through the sheath. Conversely, if the plasma resistivity $\eta_\parallel$ is larger due to increased collisionality, the effective resistivity in the circuit will increase and lead to larger blob velocity. At large enough resistivity, parallel currents are hindered enough that cross-field current closure happens away from the sheath via ion polarization currents or collisional currents, with complete disconnection giving the inertial or resistive-ballooning regime. Magnetic shear (especially near the X-point) can have the opposite effect on the blob velocity, as it can lead to a thin, elongated region of the blob where magnetic shear is strong. This makes it easier for cross-field currents to close the circuit through the thin sheared part of the blob, reducing the resistivity of the current loop and thereby slowing the blob. Current closure through regions of high magnetic shear can thus also effectively disconnect the blob from the sheath. However, notice that sheath disconnection due to increased collisionality gives the opposite effect on blob velocity than sheath disconnection via magnetic shear; the former results in increased effective blob circuit resistivity and larger velocities, while the latter decreases resistivities and slows the blobs \citep{dippolito2011,krasheninnikov2008}.

\subsection{SOL heat exhaust problem}

Particles and heat from the core are transported across the LCFS and exhausted in the SOL. The heat flows quickly along the open field lines to the walls, with the parallel heat flux in the SOL reaching above 500 MW/m$^2$ in some present devices and expected to be $\sim 1$ GW/m$^2$ in ITER \citep{loarte2007}. The maximum heat load for present materials with active cooling is typically $10$ MW/m$^2$ normal to the surface in steady state and $20$ MW/m$^2$ for transients \citep{loarte2007}. Thus the heat load must be reduced below these material limits in order to avoid damage to the wall plates and the introduction of impurities that degrade fusion performance. 
The heat load can be reduced in part by making the incidence angle of the magnetic field lines on the walls very shallow $(\sim 2-5^\circ)$ to reduce the component of the flux normal to the walls, but this still leaves a significant portion of the heat to be dissipated via other means. The width of the heat flux channel becomes an important parameter, since spreading the heat over a larger area reduces the peak heat load. Here, cross-field turbulent transport is beneficial as it can widen the heat flux width. An empirical scaling of the heat flux width, $\lambda_q$, computed from a multi-machine database has shown that the heat-flux width, mapped to the outboard midplane, varies strongest with the inverse of the plasma current (or equivalently, the inverse of the poloidal magnetic field strength) \citep{eich2013}. Simply extrapolating the scaling from present-day experiments to the upcoming ITER experiment suggests that the heat flux width for the ITER $Q=10$ baseline could be $\approx 1$ mm \citep{eich2013}, much smaller than the $3-3.5$ mm result from the ITER physics basis based mostly on JET ELM-averaged data \citep{loarte2007}. SOLPS transport modeling has suggested $\lambda_q = 3.6$ mm \citep{kukushkin2013}. The validity of these empirical scalings for the ITER heat flux width is an important issue that must be addressed by first-principles modeling. A recent XGC1 electrostatic gyrokinetic simulation predicted $\lambda_q \approx 5.9$ mm \citep{chang2017}, with the width widened due to electron turbulence. Additional analysis of XGC1 data has suggested that trapped electron mode (TEM) turbulence in particular is responsible for increased SOL heat transport. While $E \times B$ shear suppresses TEM in the SOL of present devices, $E\times B$ shear is predicted to be weaker in ITER, allowing TEM to drive transport. These results have suggested a new scaling of $\lambda_q \sim 1/B_\mathrm{pol}(a/\rho_{i,\mathrm{pol}})$, with the new parameter $a/\rho_{i,\mathrm{pol}}$ related to the neoclassical $E\times B$ shearing rate \citep{chang2020}.

%

\section{Electromagnetic effects in the boundary plasma}

In this thesis we will focus in particular on electromagnetic effects in the plasma boundary.
The edge/SOL region features steep pressure gradients, especially in the H-mode transport barrier and SOL regions, which contribute to the importance of electromagnetic effects.
Experimental evidence has indicated that the edge plasma state is controlled by electromagnetic drift wave dynamics \citep{labombard2005,labombard2008}. 
In this regime, the parallel electron dynamics is no longer fast relative to the drift turbulence, so electrons can no longer be treated adiabatically \citep{scott1997}. This leads to coupling of the perpendicular vortex motions and kinetic shear Alfv\'en waves, which results in field-line bending \citep{xu2010}. 
The slowing of parallel electron dynamics can also add impedance along the field line, leading to blobs becoming electrically disconnected from the sheath and resulting in enhanced blob velocities. While in the electrostatic case the sheath potential is communicated to the upstream plasma rapidly on the order of the electron transit time ($\tau_e = L_\parallel / v_{te}$, with $v_{te} = \sqrt{T_e/m_e}$), in the electromagnetic case Alfv\'en waves communicate the potential on the order of the Alfv\'en time, $\tau_A = L_\parallel / v_A = L_\parallel / (\sqrt{2/\beta} c_\mathrm{s}) = L_\parallel\sqrt{\mu_0 n m_i}/B$, with $c_\mathrm{s} = \sqrt{T_e/m_i}$ the sound speed. Thus a basic condition for electromagnetic effects to alter sheath connection is $v_A < v_{te}$, or $\hat{\beta} \equiv (\beta/2)m_i/m_e > 1$. If in the time $\tau_A$ the blob is able to move more than its width across the field, the information about the sheath will never reach it. Thus the blob will move as if the sheath did not exist if $\tau_A \gtrsim L_\perp / v_\perp$, or ${\beta} \gtrsim (L_\perp/L_\parallel)^2(c_\mathrm{s}/v_\perp)^2$, where $L_\perp$ is the typical length scale of the potential of the blob, and $v_\perp$ is the blob radial velocity at the midplane  \citep{lee2015a,hoare2019}. Given these conditions, electromagnetic effects could especially be important for the high beta filaments found in edge localized modes (ELMs), which are large-scale magnetohydrodynamic (MHD) modes that result in large, high pressure filaments originating from the pedestal. ELM filaments also carry a large unidirectional current, distinguishing them from standard blobs and further enhancing the electromagnetic effects of ELMs by inducing magnetic field perturbations \citep{myra2007,kirk2005, kirk2006, migliucci2010,vianello2011}. Additionally, experiments have found correlations between (non-ELM) large blobs and MHD modes \citep{zweben2020}. 

The following subsections briefly illustrate the role of magnetic induction in determining the parallel electron dynamics and producing field-line bending, mostly following \citet{xu2010} and \citet{scott1997}.

\subsection{Parallel electron dynamics and the role of magnetic induction}

The strong mobility of electrons along the field line makes it important to understand how the parallel current responds to forces due to parallel gradients in the density $n$, electron temperature $T_e$, and the electrostatic potential $\Phi$. The linear response determines the propagation of wave-like disturbances along the field line. The dynamics is governed by the electron parallel force balance equation, also known as the parallel component of the generalized Ohm's law. In the electrostatic limit, we have
\begin{equation}
     \frac{m_e}{n e^2}\deriv{J_\parallel}{t} + \eta_\parallel J_\parallel= \frac{1}{n e} {\nabla}_\parallel p_e - {\nabla}_\parallel \Phi .
\end{equation}
On the right-hand side, we have the balance between the parallel pressure and electric forces, where $p_e=n T_e$ is the electron pressure, $n=n_i=n_e$ is the plasma density (assuming a quasi-neutral plasma with singly charged ions), and $E_\parallel = -\nabla_\parallel \Phi$ is the parallel electric field in the electrostatic limit. Here, $\nabla_\parallel = \uv{b}\cdot\nabla$ denotes a derivative in the direction of the background magnetic field, $\uv{b} = \vec{B}_0/B$. On the left-hand side, the first term is electron inertia, which gives finite-electron-mass ($m_e$), collisionless effects. Here, $\dx{}/\dx{t} = \pderivInline{}{t}+ \vec{v}_E\cdot\nabla$ is the total time derivative, with $\vec{v}_E = (1/B)\uv{b}\times\nabla \Phi$ the $E\times B$ velocity. The second term on the left-hand side is resistive friction, with $J_\parallel \approx - e n u_{\parallel e}$ the parallel current  (dominated by electron parallel flow $u_{\parallel e}$) and $\eta_\parallel = 0.51 m_e \nu_{ei}/(n e^2)$ the parallel resistivity, which is proportional to the electron collision frequency. The electrons are said to be ``adiabatic'' when the forces on the right-hand side balance. After linearizing and assuming the electrons are sufficiently fast to isothermalize along the field line so that $\nabla_\parallel T_e = 0$, we have
\begin{equation}
    T_{e0} \nabla_\parallel n = n_0 e \nabla_\parallel \Phi,
\end{equation}
which results in the adiabatic electron density response, given by the Boltzmann distribution $n= n_0e \Phi/T_{e0}$, with subscript $0$ denoting background quantities. 

Now we will introduce electromagnetic (finite $\beta$) effects. We will consider only perpendicular magnetic fluctuations of the form $\vec{B}_1 = \nabla \times (A_\parallel \uv{b}) \approx - \uv{b}\times \nabla A_\parallel$, where $A_\parallel$ is the parallel component of the magnetic vector potential.  This is related to the parallel current via the parallel component of Amp\`ere's law, 
\begin{equation}
    -\nabla_\perp^2 A_\parallel = \mu_0 J_\parallel.
\end{equation}

Electromagnetic effects enter into parallel force balance in two ways. First, the parallel gradient must be taken along perturbed field lines, resulting in an additional ``magnetic flutter'' component due to $\vec{B}_1$:
\begin{equation}
    \tilde{\nabla}_\parallel \equiv \frac{1}{B}(\vec{B}_0 + \vec{B}_1)\cdot \nabla = \nabla_\parallel - \uv{b}\times \nabla A_\parallel \cdot \nabla.
\end{equation}
Second, magnetic induction adds to the parallel electric field, which is now given by
\begin{equation}
    E_\parallel = -\tilde{\nabla}_\parallel \Phi - \pderiv{A_\parallel}{t}.
\end{equation}
As a result, the parallel force balance equation becomes
\begin{equation}
    \pderiv{A_\parallel}{t} + \frac{m_e}{n e^2}\deriv{J_\parallel}{t} + \eta_\parallel J_\parallel= \frac{1}{n e} \tilde{\nabla}_\parallel p_\parallel - \tilde{\nabla}_\parallel \Phi .
\end{equation}
Balancing the first two terms on the left-hand side, we can see that induction is dominant over inertia at perpendicular scales larger than the collisionless skin depth $d_e = \sqrt{m_e/(n_0 e^2 \mu_0)}$, so that $k_\perp d_e < 1$. Balancing the first and third terms on the left-hand side gives that induction is dominant over resistivity at perpendicular scales larger than the collisional skin depth, so that $k_\perp d_e \sqrt{\nu_{ei}/\omega}<1$, with $\omega$ some characteristic frequency so that $\partial/\partial {t}\sim\omega$. 
Any imbalance of the forces on the right-hand side will result in non-adiabatic electrons, providing a channel to exchange the internal particle energy with the magnetic energy of field-line bending (via induction), or producing irreversible dissipation of magnetic energy (via resistivity). Defining the parallel electromotive force (emf) as \citep{hinton2003,xu2010}
\begin{equation}
    \psi \equiv \int (\eta_\parallel J_\parallel - E_\parallel)\dx{\ell},
\end{equation}
with $\ell$ the length along the perturbed field line, we can rewrite parallel force balance compactly as
\begin{equation}
    \pderiv{A_\parallel}{t} + \eta_\parallel J_\parallel  = \tilde{\nabla}_\parallel(\psi - \Phi). \label{force-bal}
\end{equation}

\subsection{Field-line bending}

To compute the evolution (bending) of magnetic field lines, we use Faraday's law,
\begin{equation}
\pderiv{\vec{B}_1}{t}= - \nabla\times \vec{E}.
\end{equation}
The electric field is given by 
\begin{equation}
    \vec{E} = -\nabla \Phi - \uv{b}\pderiv{A_\parallel}{t} = -\nabla \phi - \tilde{\nabla}_\parallel(\psi -\Phi)\uv{b} + \eta \vec{J} = -\nabla_\perp \Phi - \tilde{\nabla}_\parallel \psi \uv{b}+ \eta \vec{J},
\end{equation}
where we have used parallel force balance and dropped the parallel subscripts on the resistivity term ($\eta_\parallel J_\parallel \gg \eta_\perp J_\perp$). Note that $\nabla_\perp \Phi =\vec{v}_E\times (\vec{B}_0 + \vec{B}_1)$. Substituting into Faraday's law, we have
\begin{equation}
    \pderiv{\vec{B}_1}{t} = - \nabla\times \left[\vec{v}_E\times(\vec{B}_0 + \vec{B}_1)\right] + \nabla\times(\tilde{\nabla}_\parallel \psi \uv{b})+ D_m \nabla^2 \vec{B}_1 - \nabla \eta \times \vec{J}.
\end{equation}
From left to right, on the right-hand side we have a frozen-in term, a drift term, a magnetic diffusion term (with $D_m = \eta/\mu_0$) and a resistivity gradient term. In the limit of small resistivity, we can write
\begin{equation}
    \pderiv{\vec{B}_1}{t} \approx \nabla\times\left[ \tilde{\nabla}_\parallel(\psi -\Phi)\uv{b}\right].
\end{equation}
This shows that the net parallel gradient force (the right-hand side of \cref{force-bal}) drives line bending via non-adiabatic electrons exchanging energy with the magnetic field  \citep{xu2010}.

Note that we can also write the electric field as 
\begin{equation}
    \vec{E} = (\vec{B}_0+\vec{B}_1)\times \vec{v}_F - \nabla \psi +\eta \vec{J},
\end{equation}
where $\vec{v}_F \equiv (1/B)\uv{b}\times\nabla(\Phi-\psi)$ is the velocity of field lines (neglecting resistive magnetic diffusion). From this, Faraday's law becomes
\begin{equation}
    \pderiv{\vec{B}_1}{t} = - \nabla\times \left[\vec{v}_F\times(\vec{B}_0+\vec{B}_1)\right] + D_m \nabla^2 \vec{B}_1 - \nabla \eta \times \vec{J}.
\end{equation}
From this it follows that magnetic flux is conserved in the limit of no resistivity or diffusion, with field lines advected with velocity $\vec{v}_F$. The difference between the $E\times B$ velocity and the velocity of the field lines is a function of the parallel emf: $\Delta \vec{v} = \vec{v}_E - \vec{v}_F = (1/B)\uv{b}\times \nabla \psi$ \citep{xu2010}.

\section{Modeling the boundary plasma}

The boundary of a tokamak is a complicated nonlinear system. As such, numerical modeling is a critical tool for helping to understand the physics of the boundary plasma. As detailed below, several approaches have produced valuable results and insights at varying levels of complexity and computational expense. Overviews of some of the numerical modeling approaches and associated simulation codes for the boundary plasma are given by \citet{ricci2015, loarte2007, shi-thesis}, and briefly detailed below.

\subsection{Empirical modeling}
Empirical extrapolation of data obtained from present devices serves as the basis for much of the modeling and design of tokamak divertors and wall systems. Codes solve a simplified set of transport equations based on the Braginskii fluid equations in two dimensions, assuming axisymmetry. Since plasma turbulence is not captured directly in these models, \emph{ad hoc} cross-field anomalous diffusion coefficients are used, with the parameters adjusted to fit existing experimental data. This system is often coupled to a Monte-Carlo neutral particle model so that pumping, fueling, and plasma-wall interactions can be modeled. Codes using this approach include SOLPS (formerly B2-EIRENE) \citep{reiter1991,schneider1992}, UEDGE \citep{rognlien1994}, EDGE2D \citep{simonini1994}, and SOLDOR \citep{shimizu2003}. SOLPS has been used extensively as the SOL simulation code for the ITER divertor design \citep{pitts2009,kukushkin2011}. 

\subsection{Fluid modeling}
Given the relatively low temperatures and high collisionalities of the scrape-off layer, a fluid approach is reasonable to reduce the computational cost of global turbulence simulations. As such, models based on the drift-reduced Braginskii equations \citep{braginskii1965, zeiler1997} have provided valuable results and insights on boundary plasma phenomenon. Since these models only evolve the first three moments of the distribution function, they rely on high collisionality to provide fluid closure. This implicitly assumes that the distribution function is close to thermal equilibrium. Codes employing the drift-reduced Braginskii approach include BOUT++ \citep{xu2008boundary}, GBS \citep{ricci2012simulation,halpern2016a}, TOKAM3X \citep{tamain2010tokam}, GDB \citep{zhu2018}, and GRILLIX \citep{stegmeir2018}.

There have also been efforts to extend the validity of the moment-based approach to more kinetic regimes by using gyrofluid models \citep{ribeiro2008, madsen2013, held2016}, based on earlier work on gyrofluid models for core turbulence \citep{hammett1990, dorland1993, beer1996, snyder2001}. Another recent approach uses an Hermite-Laguerre formulation to allow the use of an arbitrary number of moments, although this approach has not yet produced numerical results \citep{jorge2017, frei2020}.

\subsection{Gyrokinetic modeling}
Despite the high collisionality of the plasma boundary, kinetic treatments will inevitably be necessary for reliable quantitative predictions in some cases \citep{jenko2001nonlinear,cohen2008}. Significant deviations from thermal equilibrium can occur due to transient events such as ELMs \citep{batishcheva1996}. Kinetic treatments are also required if one wishes to cross the LCFS and model the coupled dynamics of the pedestal and the SOL within a single framework, since the fluid approximations break down in the hot pedestal.

While the most general approach would involve solving the full six-dimensional Vlasov-Maxwell or Fokker-Planck-Maxwell system to model the plasma, this is impractical due to the high dimensionality and wide range of timescales involved, including the fast cyclotron motion. Instead, we can take advantage of the fact that the characteristic turbulent modes have frequencies much lower than the cyclotron frequency, allowing us to average over the cyclotron motion and eliminate one of the velocity dimensions (the gyrophase angle). The result is the gyrokinetic model, which describes the evolution of particle guiding centers in a reduced five-dimensional phase space.
Gyrokinetic theory and direct numerical simulation have become important tools for studying turbulence and transport in fusion plasmas, especially in the core region \citep{dimits2000comparisons}. This includes the simulation codes GEM \citep{parker1993gyrokinetic}, GS2 \citep{kotschenreuther1995comparison, dorland2000electron}, GTC \citep{lin2000gyrokinetic}, GENE \citep{jenko2000}, EUTERPE \citep{jost2001}, 
GYRO \citep{candy2003},  
GT3D \citep{idomura2003}, GKV \citep{watanabe2005velocity},  GTS \citep{wang2006}, ORB5 \citep{jolliet2007, lanti2019orb5}, GT5D \citep{idomura2008}, GKW \citep{peeters2009nonlinear}, CGYRO \citep{candy2016}, and GX \citep{mandell2018}. 
In the edge and SOL, gyrokinetic simulations are particularly challenging because the large, intermittent fluctuations in the SOL make assumptions of scale separation between equilibrium and fluctuations not strongly valid. This necessitates a full-$f$ approach that self-consistently evolves the full distribution function, $f$ (as opposed to the $\delta f$ approach commonly used in the core, where one assumes $f=F_0+\delta f$ with a fixed background $F_0$ so that only $\delta f$ perturbations must be evolved, and the parallel electric field nonlinearity is frequently neglected). Additional complications of the edge/SOL region include: open field line regions requiring sheath boundary conditions and models of plasma-wall interactions; X-point geometry in diverted configurations, which makes the use of efficient field-aligned coordinate systems challenging;
a wide range of collisionality regimes, from the hot pedestal top to the cold SOL; and atomic physics and neutral interactions. Major extensions to existing core gyrokinetic codes or altogether new efforts are required to meet these challenges. 
To this end, steady progress in gyrokinetic boundary plasma modeling has been made with both particle-in-cell (PIC) and continuum methods. Codes employing the PIC method in the plasma boundary include XGC1 \citep{ku2009full,ku2016new} and ELMFIRE \citep{korpilo2016gyrokinetic}. Continuum methods are used by the codes COGENT \citep{dorf2016continuum}, Gkeyll \citep{shi2017,shi2019,mandell2020} and a modified version of GENE \citep{pan2018}. Both PIC and continuum methods have their own advantages and disadvantages, as we detail briefly below. XGC1 is currently the most sophisticated code for the plasma boundary, capable of simulating electrostatic gyrokinetic turbulence in realistic diverted geometries, including neutral and atomic physics. It is critical to have at least a few successful codes that can cross-check against each other on the difficult problems in the edge/SOL, so as to give more confidence to the predictions.

\subsubsection{Particle-in-cell (PIC) approach}
The first gyrokinetic simulation algorithms used  particle-in-cell (PIC) methods \citep{lee1983, dimits1993, parker1993fully, denton1995, dimits1996}. 
In the PIC approach, the 5D phase space is sampled with an ensemble of $N_p$ markers or `superparticles', representing some clump of physical particles with given position and velocity. The markers are advanced through the domain according to the characteristics of the gyrokinetic equation, while the electromagnetic fields are evaluated and solved on a fixed three-dimensional grid. Communication between the markers and fields requires interpolation: in order to solve the field equations, the markers must be interpolated onto the grid positions so that charges and currents can be computed; likewise, the effects of the fields must be interpolated onto the marker positions to advance the particles. Since the PIC method is essentially a Monte Carlo sampling technique, sampling noise (which scales as $1/\sqrt{N_p}$) arises in moment calculation and can be problematic in some cases \citep{nevins2005,krommes2007,wilkie2016}. The sampling noise can be reduced but not completely eliminated by $\delta f$ methods \citep{denton1995}. Various other techniques have been used to reduce noise \citep{chen2007, garbet2010}. Noise-related issues have contributed to the challenges of handling electromagnetic fluctuations in PIC codes due to the Amp\`ere cancellation problem, as we discuss below. 
In general, PIC methods benefit from being rather intuitive, fairly efficient, and easily parallelizable, with straight-forward generalization to higher dimensionality. The lack of a need for a velocity-space grid is attractive. Further, PIC methods automatically guarantee positivity of the distribution function. PIC methods also have a longer history to draw on than continuum methods. 

\subsubsection{Continuum (grid-based) approach}
The first continuum gyrokinetic codes were developed some years later \citep{kotschenreuther1995comparison, dorland2000electron,jenko2000, jenko2001nonlinear, candy2003}. In the continuum method, the full five-dimensional gyrokinetic distribution function is discretized on a 5D phase-space grid. Conventional numerical methods for solving partial differential equations are then used to advance the distribution function according to the gyrokinetic equation, including finite-difference, finite-volume, (pseudo)spectral, finite-element, and discontinuous Galerkin (DG) methods. Since the electromagnetic fields are discretized on the configuration-space subset of the grid, no interpolation is required to solve the field equations, only moment calculations. Continuum methods do not suffer from statistical noise issues, which has contributed to the success of continuum codes in including electromagnetic effects where some PIC codes have failed.  
Discretization on a five-dimensional phase-space grid presents some additional challenges for parallelization and memory handling, but these issues can still be handled efficiently with well-designed schemes. In particular, continuum schemes can make use of high-order methods that perform more calculations per grid point and potentially enable faster convergence.
One key disadvantage of continuum methods is the strict Courant-Friedrichs-Lewy (CFL) stability limit placed on the time step for explicit time-advance schemes, which can be especially restrictive for electrostatic simulations \citep{lee1987gyrokinetic} and highly-collisional regimes. Another disadvantage of the continuum approach is that the typical numerical methods used do not guarantee positivity of the distribution function, which can cause numerical stability issues.

\subsubsection{Including electromagnetic effects}

Including electromagnetic effects in gyrokinetic simulations has proved numerically and computationally challenging, both in the core and in the edge. The so-called Amp\`ere cancellation problem is one of the main numerical issues that has troubled primarily PIC codes
\citep{reynders1993gyrokinetic,cummings1994gyrokinetic}. Various $\delta f$ PIC schemes to address the cancellation problem have been developed and there are interesting recent advances in this area \citep{chen2003deltaf,mishchenko2004,hatzky2007,mishchenko2014pullback,startsev2014,bao2018conservative}.
Meanwhile, some continuum $\delta f$ core codes avoided the cancellation problem completely \citep{rewoldt1987,kotschenreuther1995comparison}, while others had to address somewhat minor issues resulting from it \citep{jenko2000,candy2003}. 
With respect to the cancellation problem, one possible reason for the differences might be that in continuum codes the fields and particles are discretized on the same grid, whereas in PIC codes the particle positions do not coincide with the field grid. 
Because particle positions are randomly located relative to the field grid, one might need to be more careful in some way when treating the interaction of the particles and electromagnetic fields.

Prior to the work described in this thesis, all published nonlinear electromagnetic gyrokinetic results had focused on the core region, mostly within the $\delta f$ formulation neglecting the $E_\parallel$ nonlinearity (although the ORB5 PIC code includes the $E_\parallel$ nonlinearity and is effectively full-$f$ \citep{lanti2019orb5}). 
The XGC1 code is also full-$f$ and is focused on both the core and the edge/SOL; it has an option for a gyrokinetic ion/drift-fluid massless electron hybrid model \citep{hager2017verification}, with a fully kinetic implicit electromagnetic scheme based on \citet{chen2015multi} recently implemented and under further development \citep{ku2018fully}. The GENE-X code is a recent extension of the core gyrokinetic code GENE \citep{jenko2000} to a full-$f$ electromagnetic formulation similar to the one presented in this work. GENE-X has now produced preliminary (but not yet published) global electromagnetic gyrokinetic simulations including the SOL and X-point. 
Other gyrokinetic codes working on the SOL are not yet electromagnetic. To our knowledge, the results presented here were the first nonlinear electromagnetic full-$f$ gyrokinetic turbulence simulations on open field lines. The demonstration of full-$f$ electromagnetic capabilities, handled in stable and efficient manner that does not significantly increase the computational cost, is a major contribution of this thesis.

\subsubsection{Handling diverted geometries with X-point}
Another challenge is the magnetic geometry of the edge/SOL region, which requires treatment of open and closed magnetic field-line regions and the resulting plasma interactions with material walls on open field lines. The X-point in a diverted geometry is an additional complication which makes the use of field-aligned coordinates challenging.

Core gyrokinetic codes typically use such a field-aligned coordinate system, which allows one to take advantage of the elongated nature of the turbulence along the field line. This reduces the computational demands by allowing a coarse grid along the field line. Unfortunately, field-aligned coordinate systems are singular at the separatrix in diverted geometries due to the presence of the X-point \citep{stegmeir2016}. This has lead some codes to abandon field-aligned coordinates altogether, opting instead for simpler cylindrical coordinates. XGC1 uses a cylindrical coordinate system for the particle motion and an unstructured field-following triangular mesh for the field solver \citep{ku2016new}. BOUT++ uses multiple blocks, each with separate field-aligned coordinates systems, that conform to the X-point but still avoid it \citep{leddy2017}. Recent interest has focused on ideas like the flux-coordinate independent (FCI) approach, which abandons field- and flux-aligned coordinates in the poloidal plane but retains a field-line-following discretization of the parallel gradient operator to regain some of the advantages of field-aligned domains \citep{hariri2013,hariri2014,stegmeir2016}. This approach has been pioneered by the GRILLIX fluid code \citep{stegmeir2016,stegmeir2018}, and recently adopted by several codes, including GDB, GBS, and GENE-X. 
Another recent approach by the COGENT code uses a flux-aligned poloidal grid with controlled dealignment near the X-point \citep{McCorquodale2015,dorf2016continuum,dorf2020}. After breaking the toroidal direction into several blocks (wedges), a local field-aligned coordinate system is used in each block. Interpolation (similar to what is done in the FCI approach) is required to compute the parallel derivatives between blocks. 

Among gyrokinetic codes, currently only XGC1 \citep{ku2016new} has published results simulating turbulence in a three-dimensional diverted geometry with an X-point. As mentioned above, the recently-developed GENE-X code is also capable of including the X-point in global gyrokinetic turbulence simulations.

\section{Thesis overview}

First-principles modeling is crucial for understanding the dynamics in the boundary plasma. In particular, there is a need for comprehensive gyrokinetic simulations including electromagnetic effects. To this end, our efforts in this thesis are focused on demonstrating and advancing the capabilities of the gyrokinetic modules of the \gke plasma simulation framework (which also includes solver modules for the Vlasov--Maxwell system \citep{cagas2017continuum,juno2018} and multi-moment fluid equations \citep{wang2015comparison}). \gke was the first successful continuum gyrokinetic code on open field lines due to the pioneering work of \citet{shi-thesis, shi2017,shi2019}. In this work we make the critical step of including electromagnetic fluctuations in \gke and demonstrating that this additional physics can be handled in a stable and efficient manner. A primary goal is then to investigate how electromagnetic effects can influence SOL turbulence and transport dynamics.

In \cref{ch:emgk} we derive the 5D full-$f$ electromagnetic gyrokinetic system in Hamiltonian form using the symplectic ($v_\parallel$) formulation. In \cref{ch:dg} we describe an energy-conserving high-order discontinuous Galerkin discretization scheme for the EMGK system that has been implemented in \gke, building on the electrostatic scheme of \citet{shi-thesis}. In \cref{ch:nstx-results} we leverage \gke's new electromagnetic capabilities to produce the first published electromagnetic gyrokinetic results on open field lines. These simulations use a simple helical scrape-off layer as a model of the SOL of the National Spherical Torus Experiment (NSTX) experiment at PPPL, extending the electrostatic results of \citet{shi2019}. \cref{ch:geometry} moves towards more realistic geometry by describing and formulating field-aligned coordinate systems for use in SOL geometries with magnetic shear and shaping. \cref{ch:positivity} develops a novel positivity-preserving DG scheme which can improve robustness and accuracy of the simulations while maintaining critical conservation laws. Finally, we conclude in \cref{ch:conclusion} by reviewing the main results and describing important areas for future work.

\chapter{Theoretical background: the full-$f$ electromagnetic gyrokinetic system} \label{ch:emgk}

Turbulence in strongly magnetized plasma is characterized by frequencies much smaller than the ion cyclotron frequency $(\omega \ll \Omega_i)$ and strong anisotropy, with correlation lengths along the background field much longer than perpendicular to it $(k_\parallel \ll k_\perp)$. These two properties are the basis for gyrokinetic theory, which reduces the full six-dimensional (three position dimensions and three velocity dimensions) kinetic phase space to five dimensions by averaging over the cyclotron motion. This eliminates a velocity coordinate (the gyrophase angle) and results in a kinetic description of the dynamics of charged gyro-rings. The first derivations of gyrokinetics used a recursive procedure to generate an order-by-order asymptotic expansion, yielding the local $\delta f$ gyrokinetic equation with the distribution function separated into equilibrium ($F_0$) and perturbed ($\delta f$) parts \citep{catto1978,antonsen1980,frieman1982,abel2013}. 

Alternative approaches were later presented which derived (global, full-$f$) gyrokinetics via Lagrangian and Hamiltonian Lie-transform perturbation methods \citep{dubin1983,hahm1988,brizard2007foundations,sugama2000}. We will take this latter approach in this chapter, using phase-space-Lagrangian Lie perturbation methods \citep{littlejohn1983} to systematically derive self-consistent, energy-conserving, global gyrokinetic equations. We primarily follow \citet{brizard2007foundations} and references therein, but we have also found a series of Ph.D. dissertations from the GENE group \citep{dannert-thesis,pueschel-thesis,gorler-thesis,lapillonne-thesis,told-thesis}\footnote{\citep{dannert-thesis} is written in German but Google Translate does an admirable job of parsing it.} to be helpful for understanding parts of the derivation at a more introductory level.

\section{Gyrokinetic single-particle dynamics}
The goal of this first section is to obtain gyrokinetic equations of motion for single particles. 
We start from the description of a non-relativistic charged particle with charge $q$, mass $m$, and velocity $\vec{v}$ at position $\vec{x}$ in the presence of an electrostatic potential $\Phi(\vec{x})$ and magnetic potential $\vec{A}(\vec{x})$. The single-particle phase-space Lagrangian is
\begin{equation}
    L(\vec{x},\vec{v},\dot{\vec{x}},\dot{\vec{v}},t) = \left(m\vec{v} + q \vec{A}\right)\cdot\dot{\vec{x}} -\left(\frac{1}{2}m|\vec{v}|^2 + q\Phi\right) \equiv \vec{p}\cdot\dot{\vec{x}}-H, \label{lagrangian}
\end{equation}
where $\vec{p}=m\vec{v} + q \vec{A}$ is the canonical momentum of the particle 
(in SI units),
$\vec{v}=\dot{\vec{x}}\equiv \dx{\vec{x}}/\dx{t}$, and $H$ is the Hamiltonian \citep[for an introduction to the phase-space Lagrangian formulation of mechanics, see section II of][]{cary2009}. This Lagrangian will now be subjected to a series of coordinate transformations that will separate the fast gyromotion from the guiding-center and gyrocenter dynamics.

\subsection{Ordering assumptions}
The fundamental ordering requirement that allows us to effectively ignore the fast gyromotion of a charged particle in a magnetic field is
\begin{equation}
    \frac{\omega}{\Omega} \sim \epsilon \ll 1, \label{freq-ordering}
\end{equation}
where $\omega$ is a typical frequency of interest, and $\Omega=qB/m$ is the gyrofrequency for some species of interest. We will focus on drift wave turbulence, which has characteristic frequencies  $\omega \sim \omega_* \sim v_t/L_p$, where $v_t$ is the thermal speed and $L_p$ is a characteristic macroscopic scale length over which profile quantities vary. This implies that
\begin{equation}
    \frac{\rho}{L_p} \sim \epsilon \ll 1 \label{length-ordering}
\end{equation}
is another small parameter, where $\rho = v_t/\Omega$ is the thermal gyroradius. This ordering is valid in many tokamaks over a wide range of experimental conditions, including the edge and scrape-off layer. The frequency and length-scale orderings from \cref{freq-ordering,length-ordering} comprise the primary ordering in $\epsilon$. We must then decide how to deal with fluctuations, flows, and magnetic geometry within the model, resulting in additional parameters ordered with $\epsilon$ and $\epsilon^2$. 

The standard nonlinear gyrokinetic ordering \citep{frieman1982} also assumes small fluctuations,
\begin{equation*}
    \frac{\delta f}{F_0} \sim \frac{q \delta\Phi}{T} \sim \epsilon,
\end{equation*}
where $\delta f$ and $\delta \Phi$ are perturbations of the distribution function and potential, respectively, and $F_0$ is the equilibrium distribution function. Typical wavenumbers (of the  perturbations) are then ordered as
\begin{equation*}
    k_\parallel \rho \sim \epsilon, \qquad k_\perp \rho \sim 1.
\end{equation*}
This is the ``$\delta f$'' ordering, which is usually well-satisfied in the core of tokamak plasmas, and has been used successfully to study core microturbulence for many years \citep{parker1993gyrokinetic,kotschenreuther1995comparison,lin2000gyrokinetic,dimits2000comparisons,dorland2000electron,jenko2000,jost2001,candy2003,idomura2003,watanabe2005velocity,jolliet2007,idomura2008,peeters2009nonlinear,lanti2019orb5}. In the edge region, however, the $\delta f$ ordering is not strongly valid due to the presence of large fluctuations, even though the fundamental frequency ordering is still satisfied. 

The ordering can be generalized to allow larger perturbations by instead taking a drift ordering \citep{dimits1992,parra2008,dimits2012} 
\begin{equation}
    \epsilon_V \equiv \frac{v_E}{v_t} \simeq k_\perp \rho \frac{q\Phi}{T} \sim \epsilon \ll 1, \label{weak-flow}
\end{equation}
where $v_E$ is the $E\times B$ drift velocity from $\Phi$. Here, we have defined a new ordering parameter $\epsilon_V$, which we take to be $\mathcal{O}(\epsilon)$. This is commonly referred to as the ``weak-flow'' ordering since it takes $E\times B$ flows to be small compared to the thermal speed; this is generally satisfied in the edge and scrape-off layer \citep{brower1987,gohil1994,zweben2015}. By constraining gradients of $\Phi$ instead of $\Phi$ itself, the weak-flow ordering simultaneously allows large perturbations $q\Phi/T\sim1$ at long wavelengths ($k_\perp \rho\sim\epsilon_V\ll1$) and small perturbations $q\Phi/T\sim\epsilon_V$ at short wavelengths ($k_\perp \rho\sim 1$), along with perturbations at intermediate scales. Another way to think about this is that one can use $\Phi(\vec{R})$ at the center of a gyro-orbit (denoted $\vec{R}$) as a reference point, and then one can require that the variation of the potential energy around a gyro-orbit be small compared to the kinetic energy,
\begin{equation}
    q\Phi(\vec{R}+\vec{\rho}) - q\Phi(\vec{R}) \approx q\vec{\rho}\cdot\nabla \Phi \ll T,
\end{equation}
which leads to the same criterion as above \citep{hammettprivate2016}. Here $\vec{\rho}$ is the gyroradius vector which points from the center of the gyro-orbit $\vec{R}$ to the particle location $\vec{x}$ (it will be defined more precisely below). 
The ordering has also been extended further to allow strong $E\times B$ flows of order the thermal speed ($\epsilon_V\sim1$) \citep{artun1994,brizard1995,hahm1996,qin2007,hahm2009,dimits2010,sharma2015,mcmillan2016,sharma2020}; we will not consider this here.

We also need an ordering parameter pertaining to the magnetic geometry. For this, we introduce the equilibrium magnetic field scale length $L_B\sim|\nabla_\perp \ln B|^{-1}$ $\sim R$, where $B$ is the background magnetic field and $R$ is the major radius. In the core we expect $L_B\sim L_p$, but the edge features much stronger gradients so that $L_p/L_B\sim L_p/R \lesssim \rho/L_p\sim\epsilon$ is small \citep{gohil1994,burrell1994,zweben2007}. This leads us to a strong-gradient ordering, in which we define an additional ordering parameter
\begin{equation}
    \epsilon_B \equiv \frac{\rho}{L_B} \sim \epsilon^2.
\end{equation}
This ordering has been employed in several edge gyrokinetic models \citep{hahm2009,dimits2012,frei2020}, as it allows the gyrokinetic derivation to proceed fully consistently (up to second order in $\epsilon$) in a two-step process, first by using the $\epsilon_B$ ordering to derive the guiding-center motion, and then subsequently using the $\epsilon_V$ ordering to introduce electromagnetic perturbations. Without the strong-gradient ordering, $\epsilon_B \sim \epsilon_V \sim \epsilon$ are of the same order, which is usually the case in the core. Even though many derivations still use the two-step procedure in this case \citep[see \emph{e.g.}][]{brizard2007foundations}, \citet{parra2011} have shown that the two-step procedure does not yield fully consistent results at second order (and higher), as it misses terms of order $\mathcal{O}(\epsilon_V\epsilon_B)$ that involve both geometric effects and field perturbations. The strong-gradient ordering $\epsilon_B\sim\epsilon^2$ eliminates these concerns, since the geometry only enters at even order in $\epsilon$.

With these ordering assumptions, we take the  background magnetic field to be $\vec{B}_0=\nabla\times\vec{A}_0$, with $\vec{A}_0\sim\mathcal{O}(1/\epsilon_B)$ the background vector potential. We do not include an $\mathcal{O}(1/\epsilon_B)$ background electrostatic potential because this would violate the weak-flow ordering. We then consider electromagnetic perturbations $\vec{A}_1$ and $\Phi_1$ of the form \citep{dimits1992}
\begin{gather}
    \vec{A}_1(\vec{x},t) = \vec{A}_1(\vec{R},t) + \epsilon_V\,\delta\vec{A}_1(\vec{R},\vec{\rho},t) \label{a-pert} \\
    \Phi_1(\vec{x},t) = \Phi_1(\vec{R},t) + \epsilon_V\,\delta\Phi_1(\vec{R},\vec{\rho},t), \label{phi-pert}
\end{gather}
where the guiding-center component of the perturbation is $\mathcal{O}(1)$, and the $\mathcal{O}(\epsilon_V)$ part of the perturbation is the deviation of the potential around the gyro-orbit, effectively giving the finite-Larmor-radius (FLR) correction to the potential, 
\begin{align}
    \delta \vec{A}_1 &\equiv \vec{A}_1(\vec{x},t)-\vec{A}_1(\vec{R},t) \approx \vec{\rho}\cdot\nabla\vec{A}_1 \sim \frac{T}{q}\frac{ B_{1\perp}}{B}\sim\frac{T}{q}\frac{v_{f}}{v_t}\sim \mathcal{O}(\epsilon_V) \\
    \delta \Phi_1 &\equiv \Phi_1(\vec{x},t)-\Phi_1(\vec{R},t) \approx \vec{\rho}\cdot\nabla\Phi_1 \sim \frac{T}{q}\frac{v_E}{v_t}\sim \mathcal{O}(\epsilon_V),
\end{align}
with $v_f = v_\parallel B_{1\perp}/B$ the magnetic flutter velocity. Thus the total 
electromagnetic potentials can be written as
\begin{align}
    \vec{A}(\vec{x}) &= \frac{1}{\epsilon_B}\vec{A}_0(\vec{x}) + \vec{A}_1(\vec{x}) \notag \\&= \left[\frac{1}{\epsilon_B}\vec{A}_0(\vec{R})+ \vec{\rho}\cdot\nabla \vec{A}_0(\vec{R}) + \mathcal{O}(\epsilon_B)\right] + \vec{A}_1(\vec{R}) + \epsilon_V \delta\vec{A}_1(\vec{R},\vec{\rho}) \label{a-tot} \\
    \Phi(\vec{x}) &= \Phi_1(\vec{x}) = \Phi_1(\vec{R})+\epsilon_V \delta\Phi_1(\vec{R},\vec{\rho}), \label{phi-tot}
\end{align}
where we have Taylor expanded the background magnetic potential $\vec{A}_0(\vec{x})=\vec{A}_0(\vec{R}+\vec{\rho})$ to first order in $\epsilon_B$ around the guiding-center position $\vec{R}$. With these definitions, the Lagrangian from \cref{lagrangian} can be written as
\begin{align}
    {L} = {L}_0 + \epsilon_V \delta{L} &= \left[\frac{q}{\epsilon_B}\vec{A}_0(\vec{R}) + q\vec{\rho}\cdot\nabla \vec{A}_0(\vec{R}) +q\vec{A}_1(\vec{R})+  m \vec{v}+\mathcal{O}(\epsilon_B)\right]\cdot\dot{\vec{x}} \notag \\
    &\quad- \left[\frac{1}{2}mv^2 + q\Phi_1(\vec{R}) \right] + \epsilon_V\left[q\,\delta\vec{A}_1(\vec{R},\vec{\rho}) \cdot \dot{\vec{x}} - q\,\delta\Phi_1(\vec{R},\vec{\rho})\right]. \label{lagrangiansplit}
\end{align}

Note that we will hereafter drop the subscript $0$ in the background magnetic field $\vec{B}_0$, unless it is needed for clarity, and simply write $\vec{B}$. For magnetic perturbations we will retain the subscript in $\vec{B}_1$, but these will frequently be expressed instead in terms of the perturbed vector potential. The total magnetic field, including background and perturbation, will be written as $\vec{B}_0 + \vec{B}_1$.

Writing the potentials in the form of \cref{a-tot,phi-tot} has the advantage that we can clearly see that FLR corrections are higher order. Thus, we could self-consistently neglect FLR corrections by taking only the zeroth-order terms. While we will proceed with the general derivation up to order $\epsilon\sim\epsilon_V$ in this chapter, for simplicity we have elected to neglect FLR corrections in the current implementation of the system in the \gke code. Extension to the more general system including FLR corrections is left as important future work. The system that we solve in the current version of \gke is summarized in \cref{sec:gk-summary}. 

\subsection{Transformation to guiding-center coordinates}

Following \citet{littlejohn1983} and \citet{cary2009}, we first transform the zeroth order Lagrangian $L_0$ to guiding-center coordinates, $\vec{Z}=(t,\vec{R},v_\parallel,\mu,\vartheta)$, where $\vec{R}$ is the guiding-center position, $v_\parallel$ is the guiding-center velocity along the background magnetic field, $\mu=mv_\perp^2/(2B)$ is the  lowest-order magnetic moment, and $\vartheta$ is the gyrophase angle.
In terms of the guiding-center coordinates, the particle coordinates $\vec{z}=(t,\vec{x},\vec{v})$ coordinates can be expressed (with the time coordinate $t$ staying the same) as
\begin{gather}
    \vec{x} = \vec{R} + \epsilon_B\rho(\vec{R})\uv{a}(\vec{R},\vartheta) = \vec{R} + \epsilon_B \vec{\rho} \label{gcx}\\
    \vec{v} = v_\parallel \uv{b}(\vec{R}) + v_\perp(\vec{R})\uv{c}(\vec{R},\vartheta) + \epsilon_V\vec{u}_\perp(\vec{R},v_\parallel)= v_\parallel \uv{b} + \sqrt{\frac{2\mu B}{m}} \uv{c} + \epsilon_V \vec{u}_\perp \label{gcv},
\end{gather}
where $\rho(\vec{R}) = v_\perp/\Omega(\vec{R}) = \sqrt{2m \mu/(q^2 B(\vec{R}))}$ is the gyroradius,  $\uv{b}=\vec{B}/B$ is the unit vector along the background field, and $\vec{u}_\perp(\vec{R},v_\parallel)$ is a to-be-defined $\mathcal{O}(\epsilon_V)$ velocity perpendicular to $\uv{b}$, taken to be the velocity of the reference frame; note that $\vec{u}_\perp$ is evaluated at the guiding-center position $\vec{R}$ and is assumed to be gyrophase-independent. From standard guiding-center motion, we might expect this reference frame velocity to be something like the $E\times B$ drift velocity.\footnote{A moving reference frame is typically used in strong-flow derivations of gyrokinetics, where $\epsilon_V\sim1$ so that all terms in \cref{gcv} are the same order. In most weak-flow derivations \citep[][is an exception]{frei2020}, the reference frame is assumed to be stationary ($\vec{u}_\perp=0$), but here we allow for a slowly moving frame.} We also define
\begin{gather}
    \uv{a}(\vec{R},\vartheta)=\cos\vartheta\,\vec{e}_1(\vec{R})-\sin\vartheta\,\vec{e}_2(\vec{R}) \\ \uv{c}(\vec{R},\vartheta)=\pderiv{\uv{a}(\vec{R},\vartheta)}{\vartheta}=-\sin\vartheta\, \vec{e}_1(\vec{R}) - \cos\vartheta\, \vec{e}_2(\vec{R})
\end{gather}
to be unit vectors in the radial and tangential directions to the gyro-orbit that rotate with $\vartheta$, where $\vec{e}_1$ and $\vec{e}_2$ are some arbitrary pair of perpendicular unit vectors in the plane perpendicular to the background field such that $\vec{e}_1\times\vec{e}_2=\uv{b}$. Here and in the following, we will keep track of the order of various terms in $\epsilon_B$ and $\epsilon_V$, but formally these parameters are equal to unity so that the expressions retain the same dimensional form after taking $\epsilon_B=\epsilon_V=1$.

Inserting the guiding-center coordinate transformations into \cref{lagrangiansplit}, the zeroth order Lagrangian in guiding-center coordinates is
\begin{equation}
    {L}_0 = \left[\frac{q}{\epsilon_B} \vec{A}_0+q\vec{A}_1 + q(\vec{\rho}\cdot\nabla) \vec{A}_0 +  mv_\parallel \uv{b} + m v_\perp \uv{c}\right]\cdot \left(\dot{\vec{R}}+\dot{\vec{\rho}}\right) -H_0,
\end{equation}
where  $H_0 = mv_\parallel^2/2 + \mu B + q\Phi_1$ is the zeroth order Hamiltonian, and spatially-varying quantities are evaluated at the guiding-center position $\vec{R}$ unless otherwise noted. Note that although the gyroradius vector $\vec{\rho}$ is $\mathcal{O}(\epsilon_B)$, its time derivative $\dot{\vec{\rho}}$ is $\mathcal{O}(1)$, as it is given by
\begin{align}
    \dot{\vec{\rho}} &= \epsilon_B\left[(\dot{\vec{R}}\cdot\nabla)\vec{\rho} + \dot{\mu}\pderiv{\vec{\rho}}{\mu}\right] + \dot{\vartheta}\pderiv{\vec{\rho}}{\vartheta}= \epsilon_B\left[\frac{\vec{\rho}}{2B}\dot{\vec{R}}\cdot\nabla B + \frac{1}{qv_\perp}\uv{a}\dot{\mu}\right] + \frac{v_\perp}{\Omega}\uv{c}\dot{\vartheta} \notag \\
    &= \frac{v_\perp}{\Omega}\uv{c}\dot{\vartheta} + \mathcal{O}(\epsilon_B),
\end{align}
since $\dot{\vartheta}=\Omega\sim \epsilon_B^{-1}$ \citep{cary2009}.

We then make a series of gauge transformations to eliminate the dependence on the gyrophase to lowest order in $\epsilon_B$, following \citet{littlejohn1983}. These gauge transformations take the form of adding a total time derivative to the Lagrangian, ${L}\rightarrow{L}+\dot{S}$, which does not affect the equations of motion. Taking $S=-q \vec{\rho}\cdot(\vec{A}_0+\epsilon_B\vec{A}_1)$, so that
\begin{align}
    \dot{S} &
    = -\left(\frac{q}{\epsilon_B}\vec{A}_0+q\vec{A}_1\right) \cdot\dot{\vec{\rho}} - q\dot{\vec{R}}\cdot\nabla\vec{A}_0\cdot\vec{\rho} + \mathcal{O}(\epsilon_B),
\end{align}
the Lagrangian can be transformed as
\begin{align}
    {L}_0 \rightarrow {L}_0 + \dot{S} 
    &=\left(\frac{q}{\epsilon_B}\vec{A}_0^*+q\vec{A}_1\right)\cdot\dot{\vec{R}} + \left[q\vec{\rho}\cdot\nabla\vec{A}_0+mv_\parallel\uv{b}+mv_\perp\uv{c}\right]\cdot\dot{\vec{\rho}}-H_0,
\end{align}
where cancellations resulted from noting that $q\left[\vec{\rho}\cdot\nabla\vec{A}_0- 
    (\nabla \vec{A}_0)\cdot\vec{\rho}\right]=-q\vec{\rho}\times(\nabla\times\vec{A}_0)=-q\vec{\rho}\times\vec{B}=-mv_\perp\uv{c}$, and we have defined the modified vector potential
\begin{equation}
    \vec{A}_0^* \equiv \vec{A}_0  +\epsilon_B \frac{mv_\parallel}{q}\uv{b}. \label{a0star}
\end{equation} 
Recognizing $\vec{A}_0^*$ as the gyroaveraged canonical momentum from the background field,
we can see that this gauge transformation has effectively gyroaveraged the first term in \cref{lagrangiansplit}.

We then make an additional gauge transformation with $S=-\epsilon_B (q/2)(\vec{\rho}\cdot\nabla)\vec{A}_0\cdot\vec{\rho}$, which gives
\begin{align}
    \dot{S} = -\frac{q}{2}\left[(\nabla\vec{A}_0)\cdot\vec{\rho}+\left(\vec{\rho}\cdot\nabla\right)\vec{A}_0\right]\cdot\dot{\vec{\rho}} - \epsilon_B\frac{q}{2}\vec{\rho}\cdot\frac{\dx{(\nabla \vec{A}_0)}}{\dx{t}}\cdot\vec{\rho},
\end{align}
where we will drop the last term because it is higher order. The Lagrangian is then transformed as
\begin{align}
    {L}_0\rightarrow{L}_0+\dot{S} 
    &=\left(\frac{q}{\epsilon_B}\vec{A}_0^*+q\vec{A}_1\right)\cdot\dot{\vec{R}} + \frac{m\mu }{q}\dot{\vartheta} - \left[\frac{1}{2}mv_\parallel^2 + \mu B + q\Phi_1 \right] + \mathcal{O}(\epsilon_B). \label{gam0}
\end{align}
This Lagrangian describes the motion of charged particles in a strong background magnetic field and slowly varying electromagnetic potentials (with no variation at the gyroradius scale), and it could be used to derive the drift-kinetic Vlasov equation (up to zeroth order in $\epsilon_B$). Since \cref{gam0} is independent of gyrophase, Noether's theorem gives that the quantity $\partial {L}_0/\partial \dot{\vartheta}=m\mu/q$ is a constant, which is confirmation that $\mu$ is an adiabatic invariant in the absence of electromagnetic perturbations on the scale of the gyroradius (to lowest order). For an alternative derivation of the guiding-center Lagrangian, which eliminates the gyrophase dependence via gyroaveraging instead of gauge transformations, see \citet{helandersigmar}.

\subsection{Transformation to gyrocenter coordinates}

Now we must account for the variations of the electromagnetic fields on the scale of the gyroradius, which are contained in $\delta L$ from \cref{lagrangiansplit},
\begin{align}
    \delta{L} &= \left[q\,\delta\vec{A}_1(\vec{R},\vec{\rho})+ m \vec{u}_\perp \right]\cdot(\dot{\vec{R}}+\dot{\vec{\rho}}) -\left[q\,\delta \Phi_1(\vec{R},\vec{\rho}) + m\vec{v}_\perp\cdot\vec{u}_\perp + \frac{\epsilon_V}{2}m u_\perp^2 \right]\notag\\
    &= q\,\delta\vec{A}_1^*\cdot\dot{\vec{R}} + \frac{m}{B}\vec{v}_\perp\cdot\delta\vec{A}_1^*\,\dot{\vartheta} 
    -\left[q\,\delta \Phi_1 + m\vec{v}_\perp\cdot\vec{u}_\perp + \frac{\epsilon_V}{2}m u_\perp^2 \right] + \mathcal{O}(\epsilon_B),
\end{align}
where we have defined 
\begin{equation}
    \delta \vec{A}_1^*\equiv \delta \vec{A}_1(\vec{R},\vec{\rho}) + (m/q)\vec{u}_\perp(\vec{R},v_\parallel).
\end{equation}
Since the perturbations $\delta \vec{A}_1$ and $\delta \Phi_1$ depend on $\vec{\rho}$, we have reintroduced gyrophase dependence in the Lagrangian and broken $\mu$ conservation. Unlike in the lowest-order case above, we cannot simply use gauge transformations to eliminate the gyrophase dependence at this order because the perturbations depend non-trivially on $\vec{\rho}$. Instead, we use another kind of coordinate transformation known as a Lie transform \citep[for details, see ][]{cary1981,littlejohn1982,cary1983}. The Lie transform offers a systematic method for making perturbative coordinate transformations and computing the resulting changes in functions of those coordinates. 

For these transformations, it will be convenient to adopt the Poincar\'e-Cartan one-form formalism \citep[see \emph{e.g.}][]{cary1983}, where the one-form $\gamma(\vec{Z})$ is defined via the action integral
\begin{equation}
   \mathcal{I} = \int L(\vec{Z})\, \dx{t} = \int \gamma(\vec{Z}).
\end{equation}
Here,
\begin{align}
    \gamma &= \gamma_{Z^\alpha} \dx{Z^\alpha} = L\, \dx{t}\label{gamma}
\end{align}
for $Z^\alpha$ (with $\alpha=0,1,\dots,6$) the components of the extended phase-space coordinates that include time as the zeroth element so that $\gamma_t = - H$, the Hamiltonian. The remaining components, $\gamma_{Z^i}$ with $i=1,\dots,6$, are together called the \emph{symplectic} component of the one-form. 

We will define $\gamma = \gamma_0 + \epsilon_V \gamma_1$, with
\begin{align}
    \gamma_0 &\equiv L_0 \,\dx{t} = q\left(\vec{A}_0^*+ \vec{A}_1\right)\cdot\dx{\vec{R}} + \frac{m\mu }{q}\dx{\vartheta} - \left[\frac{1}{2}mv_\parallel^2 + \mu B + q\Phi_1 \right]\dx{t} \label{Gam0}\\
    \gamma_1 &\equiv \delta L\, \dx{t} = q\,\delta\vec{A}_1^*\cdot\dx{\vec{R}} + \frac{m}{B}\vec{v}_\perp\cdot\delta\vec{A}_1^*\,\dx{\vartheta}-\left[q\,\delta \Phi_1 + m\vec{v}_\perp\cdot\vec{u}_\perp + \frac{\epsilon_V}{2}m u_\perp^2 \right]\dx{t},
    \end{align}
where all quantities are evaluated at $\vec{R}$ (although $\delta \vec{A}_1$ and $\delta\Phi_1$ still also depend on $\vec{\rho}$), and we have dropped the $1/\epsilon_B$ ordering parameter on $\vec{A}_0^*$.
%
The Lie transformation that yields the gyrocenter one-form $\Gamma = \Gamma_0 + \epsilon_V \Gamma_1 + \epsilon_V^2 \Gamma_2 + \dots$ is given (up to second order) by
\begin{align}
    \Gamma_0 &= \gamma_0 + \dx{S_0} \\
    \Gamma_1 &= \gamma_1 - \mathscr{L}_1 \gamma_0 + \dx{S_1} \label{lie1a}\\
    \Gamma_2 &= \gamma_2 - \mathscr{L}_2 \gamma_0 -\frac{1}{2}\mathscr{L}_1(\gamma_1 + \Gamma_1) + \dx{S_2}, \label{lie2}
\end{align}
where $\mathscr{L}_n$ denotes the $n$th order Lie derivative and $S_n$ is an arbitrary $n$th order scalar gauge function. The Lie derivative acting on a one-form $\gamma$ is given by\footnote{The expression in \cref{liedef} is only part of the formal Lie derivative. There is another part of the form $\dx{(G_n^\beta \gamma_{Z^\beta})}$, but this part can be absorbed into $\dx{S}$ in \cref{lie1a,lie2} because $S$ can be chosen arbitrarily.}
\begin{equation}
    \mathscr{L}_n \gamma =G_n^\beta \omega_{\alpha \beta}\dx{Z^\alpha} =  G_n^\beta\left(\pderiv{\gamma_{Z^\alpha}}{Z^\beta} - \pderiv{\gamma_{Z^\beta}}{Z^\alpha}\right)\dx{Z^\alpha}, \label{liedef}
\end{equation}
where the functions $G_n^\beta$ are the components of the $n$th order generating vector field of the Lie transform, and $\omega_{\alpha\beta}$ are the elements of the Lagrange tensor. With this, the first order gyrocenter one-form can be rewritten as
\begin{equation}
    {\Gamma}_{1, Z^\alpha} = \gamma_{1,Z^\alpha} - G_1^\beta\left(\pderiv{\gamma_{0,Z^\alpha}}{Z^\beta}-\pderiv{\gamma_{0,Z^\beta}}{Z^\alpha}\right) + \pderiv{S_1}{Z^\alpha}. \label{lie1}
\end{equation}

The goal now is to find generating functions $\vec{G}_n$ and gauge functions $S_n$ such that the gyrocenter one-form no longer depends on the gyrophase at each order. At zeroth order, we will simply take $S_0=0$, so that
\begin{equation}
    {\Gamma}_0 = \gamma_0,
\end{equation}
since we have already removed gyrophase dependence from the zeroth order one-form, \cref{Gam0}. 

At first order, we can use \cref{lie1} to compute
\begin{align}
    \Gamma_{1,\mathbold{R}} &= q\,\delta\vec{A}_1^* + qG_1^{\mathbold{R}}\times \vec{B}^* - m G_1^{v_\parallel}\uv{b} + \nabla S_1 \label{G1X}\\
    \Gamma_{1,v_\parallel} &= m\uv{b}\cdot G_1^{\mathbold{R}} + \pderiv{S_1}{v_\parallel} \\
    \Gamma_{1,\mu} &= \frac{m}{q}G_1^\vartheta + \pderiv{S_1}{\mu} \\
    \Gamma_{1,\vartheta} &= \frac{m}{B}\vec{v}_\perp\cdot\delta\vec{A}_1^*-\frac{m}{q}G_1^\mu + \pderiv{S_1}{\vartheta} \\
    \Gamma_{1,t} &= -q\,\delta\Phi_1 -\epsilon_V q  G_1^{\mathbold{R}}\cdot \vec{E}_1^*  + G_1^{v_\parallel}m v_\parallel + G_1^\mu B + \pderiv{S_1}{t},
\end{align}
where $\vec{E}_1^* = -\epsilon_V(\nabla \Phi_1 + \partial \vec{A}_1/\partial t) -\epsilon_B (\mu/q) \nabla B= \epsilon_V\vec{E}_1 - \epsilon_B(\mu/q) \nabla B \sim\mathcal{O}(\epsilon_V)$ and  $\vec{B}^* = \nabla \times \vec{A}_0^*+\epsilon_V \nabla\times\vec{A}_1 = \vec{B}_0^* + \epsilon_V \vec{B}_1$. We have also taken $G_1^t=0$ since we do not need to make a coordinate transformation in time. 

We now have some freedom to choose the $\vec{G}_n$ and $S_n$ to simplify the form of the gyrocenter Lagrangian. To this end, we choose to enforce $\Gamma_{1,v_\parallel}=\Gamma_{1,\mu}=\Gamma_{1,\vartheta}=0$, which gives
\begin{align}
    \uv{b}\cdot G_1^{\mathbold{R}} = -\frac{1}{m} \pderiv{S_1}{v_\parallel},\qquad 
    G_1^\vartheta = -\frac{q}{m}\pderiv{S_1}{\mu}, \qquad
    G_1^\mu = \frac{q}{B}\vec{v}_\perp\cdot\delta\vec{A}_{1}^* +\frac{q}{m}  \pderiv{S_1}{\vartheta} .
\end{align}
Dotting \cref{G1X} with $\vec{B}^*$ gives
\begin{gather}
    \vec{B}^*\cdot \Gamma_{1,\mathbold{R}} = q\vec{B}^*\cdot\delta\vec{A}_1^* - m G_1^{v_\parallel} B_\parallel^* + \vec{B}^*\cdot\nabla S_1,
\end{gather}
while crossing \cref{G1X} with $\uv{b}$ gives
\begin{align}
    \uv{b}\times\Gamma_{1,\mathbold{R}} 
    &= q\uv{b}\times \delta\vec{A}_1^* + qB_\parallel^* G_1^{\mathbold{R}} + \frac{q}{m}\vec{B}^*\pderiv{S_1}{v_\parallel} + \uv{b}\times \nabla S_1,
\end{align}
so that we have
\begin{align}
    G_1^{v_\parallel} &= \frac{\vec{B}^*}{m B_\parallel^*} \cdot\left( - \Gamma_{1,\mathbold{R}} + \nabla S_1 + q\delta \vec{A}_1^* \right) \\
    G_1^{\mathbold{R}} &= -\frac{\vec{B}^*}{m B_\parallel^*}\pderiv{S_1}{v_\parallel} - \frac{\uv{b}}{q B_\parallel^*} \times \left( - \Gamma_{1,\mathbold{R}} + \nabla S_1 + q\delta \vec{A}_1^* \right),
\end{align}
where $B_\parallel^*\equiv \uv{b}\cdot\vec{B}^*\approx B + \mathcal{O}(\epsilon_B)$.
The first order gyrocenter Hamiltonian is then
\begin{align}
    H_1 &= - \Gamma_{1,t} =q\,\delta \Phi_1 + m\vec{v}_\perp\cdot\vec{u}_\perp + \epsilon_V\frac{1}{2}m u_\perp^2 - \epsilon_V\frac{q}{m B_\parallel^*}\pderiv{S_1}{v_\parallel}\vec{B}^*\cdot\vec{E}_1^* \notag \\
    &\qquad - \frac{\epsilon_V\vec{E}_1^*\times\uv{b} + v_\parallel \vec{B}^*}{B_\parallel^*}\cdot\left(-\Gamma_{1,\mathbold{R}} + \nabla S_1 + q\,\delta \vec{A}_1^*\right)  - q \vec{v}_\perp\cdot \delta\vec{A}_1^* - \Omega \pderiv{S_1}{\vartheta} - \pderiv{S_1}{t} \notag\\
    &= q\left( \delta\Phi_1 -  \vec{v}_\perp\cdot \delta\vec{A}_1\right)+\epsilon_V\frac{1}{2}mu_\perp^2- \frac{\epsilon_V\vec{E}_1^*\times\uv{b} + v_\parallel \vec{B}^*}{B_\parallel^*}\cdot\left(-\Gamma_{1,\mathbold{R}} + q\,\delta \vec{A}_1^*\right) - \frac{\dx{S_1}}{\dx{t}} \label{ham1}.
\end{align}
We now choose the gauge function $S_1$ to cancel the gyrophase dependence in \cref{ham1}. We will leave $\Gamma_{1,\mathbold{R}}$ unspecified for now;
by construction, the gyrocenter one-form will be gyrophase-independent, so the choice of $\Gamma_{1,\mathbold{R}}$ will not affect the choice of $S_1$. We can define $S_1$ as the solution to
\begin{gather}
    \frac{\dx{S_1}}{\dx{t}} \approx \Omega \pderiv{S_1}{\vartheta} + \mathcal{O}(\epsilon_V) = q\left( \widetilde{\delta\Phi_1} - \widetilde{ \vec{v}_\perp\cdot \delta\vec{A}_1}\right) - \frac{\epsilon_V\vec{E}_1^*\times\uv{b} + v_\parallel \vec{B}^*}{B_\parallel^*}\cdot q\,\widetilde{\delta \vec{A}_1} ,
\end{gather}
where
\begin{equation}
    \widetilde{A} = A(\vec{R}+\vec{\rho}) - \langle A \rangle
\end{equation}
is the gyrophase-dependent part of a quantity $A(\vec{R}+\vec{\rho})$, and $\langle \cdot \rangle$ denotes a gyroaverage, defined by
\begin{equation}
    \langle A(\vec{R}+\vec{\rho}) \rangle \equiv  \frac{1}{2\pi} \int_0^{2\pi} A(\vec{R}+\vec{\rho})\,\dx{\vartheta}. \label{gyavg}
\end{equation}
Noting that $\widetilde{\delta\Phi_1}=\widetilde{\Phi_1}$, $\widetilde{\delta \vec{A}_1} = \widetilde{ \vec{A}_1}$, and $\widetilde{\vec{v}_\perp\cdot\delta \vec{A}_1}=\widetilde{\vec{v}_\perp\cdot \vec{A}_1}-\vec{v}_\perp\cdot\vec{A}_1$, the gauge function $S_1$ becomes
\begin{align}
    S_1 \approx \frac{q}{\Omega}\int^\vartheta \left[ \widetilde{\Phi_1} - \widetilde{ \vec{v}_\perp\cdot \vec{A}_1} + \vec{v}_\perp\cdot\vec{A}_1 - \frac{\epsilon_V\vec{E}_1^*\times\uv{b} + v_\parallel \vec{B}^*}{B_\parallel^*}\cdot \widetilde{\vec{A}_1}\right]\dx{\vartheta'},
\end{align}
where we will take the solution with $\langle S_1 \rangle =0$ (this is required to prevent $S_1$ from becoming unbounded; see \citet{cary1983}). The Hamiltonian then becomes
\begin{align}
    H_1 &= q\left[\langle\Phi_1 - \vec{v}_\perp \cdot\vec{A}_1 \rangle - \Phi_1\right]+\epsilon_V\frac{1}{2}mu_\perp^2 \notag \\
    &\quad- \frac{\epsilon_V\vec{E}_1\times\uv{b} + \epsilon_B(\mu/q)\uv{b}\times\nabla B + v_\parallel \vec{B}^*}{B_\parallel^*}\cdot\left(-\Gamma_{1,\mathbold{R}} + q\langle \vec{A}_1\rangle - q\vec{A}_1 + m \vec{u}_\perp \right),
\end{align}
so that the gyrocenter one-form is
\begin{equation}
    \Gamma=\Gamma_0 + \epsilon_V\Gamma_1 = q \left(\vec{A}_0^* + \vec{A}_1 + \epsilon_V \Gamma_{1,\mathbold{R}}\right)\cdot\dx{\vec{R}} + \frac{m\mu}{q}\dx{\vartheta} - \left[\frac{1}{2}mv_\parallel^2 + \mu B + q\Phi_1 + \epsilon_V H_1\right]\dx{t}.
\end{equation}

Now we must choose $\Gamma_{1,\mathbold{R}}$. One option is to use $\Gamma_{1,\mathbold{R}}=-q\vec{A}_1$. This would eliminate $\vec{A}_1$ from the symplectic part of the one-form, opting instead to move all dependence on field perturbations to the Hamiltonian. This is known as the ``Hamiltonian'' formulation of electromagnetic gyrokinetics \citep{brizard2007foundations}, so-named because all field perturbations reside in the Hamiltonian. This approach has the advantage that the equations of motion do not contain explicit time derivatives of the magnetic potential, which can be advantageous in some discretization schemes. As a result, however, the parallel momentum coordinate becomes the canonical momentum, which depends on the perturbed magnetic potential. In other discretization schemes (namely, in the one that we pursue in \cref{ch:dg}) having the perturbed magnetic potential in the Hamiltonian can be disadvantageous.

Thus we will take another approach, known as the ``symplectic'' formulation, so-named because the symplectic part of the gyrocenter one-form is allowed to retain gyrophase-independent parts of the perturbed fields. In this approach, we eliminate the dependence on $\vec{A}_1$ from the Hamiltonian at first order. This results in the parallel momentum coordinate remaining the kinetic momentum. Thus we take
\begin{gather}
    \Gamma_{1,\mathbold{R}} = q\langle \delta \vec{A}_1 \rangle = q\langle \vec{A}_1(\vec{R}+\vec{\rho}) - \vec{A}_1(\vec{R}) \rangle = q\langle \vec{A}_1 \rangle - q\vec{A}_1 \label{Gam1X},
\end{gather}
where note that $\Gamma_{1,\mathbold{R}}\neq - \widetilde{\vec{A}}_1$, since the non-averaged term in \cref{Gam1X} is evaluated at $\vec{R}$, not $\vec{R}+\vec{\rho}$. With this choice, the gyrocenter one-form is given by
\begin{equation}
    \Gamma =  q\left(\vec{A}_0^*+ \langle\vec{A}_1\rangle\right)\cdot\dx{\vec{R}} + \frac{m\mu }{q}\dx{\vartheta} - \mathcal{H}\dx{t},
\end{equation}
with the total gyrocenter Hamiltonian
\begin{align}
    \mathcal{H} &= \frac{1}{2}mv_\parallel^2 + \mu B +q\langle\Phi_1-\vec{v}_\perp\cdot\vec{A}_1\rangle + \epsilon_V^2\frac{1}{2}mu_\perp^2 \notag \\
    &\qquad- \left[\frac{v_\parallel \vec{B}_0^*+\epsilon_V\left(\vec{E}_1\times\uv{b} +v_\parallel \vec{B}_{1\perp}\right)+ \epsilon_B(\mu/q)\uv{b}\times\nabla B }{B_\parallel^*}\right]\cdot \epsilon_V m \vec{u}_\perp,
\end{align}
and $\vec{B}_{1\perp}=(\nabla\times \vec{A}_1)_\perp$.
By taking the velocity $\vec{u}_\perp$ to be the $\mathcal{O}(\epsilon_V)$ part of the term in square brackets above, which is equivalent to the sum of the guiding-center $E\times B$ velocity, $\vec{v}_E$, and the ``magnetic flutter'' component of the parallel velocity perpendicular to the background field, $\vec{v}_f$, so that
\begin{equation}
    \vec{u}_\perp(\vec{R},v_\parallel) \equiv \frac{\vec{E}_1\times\uv{b}}{B_\parallel^*} + v_\parallel \frac{ \vec{B}_{1\perp}}{B_\parallel^*}= \vec{v}_E + \vec{v}_{f},
\end{equation}
we can reduce the Hamiltonian to
\begin{equation}
    \mathcal{H} = \frac{1}{2}mv_\parallel^2 + \mu B +q\langle\Phi_1-\vec{v}_\perp\cdot\vec{A}_1\rangle - \epsilon_V^2\frac{m}{2}\left|\vec{v}_E + \vec{v}_f \right|^2 - \epsilon_V\epsilon_B\,
    m\vec{v}_d \cdot (\vec{v}_E+\vec{v}_f). \label{ham2}
\end{equation}
Here,
\begin{equation}
    \vec{v}_d \equiv \frac{mv_\parallel^2}{qB_\parallel^*}\uv{b}\times(\uv{b}\cdot\nabla\uv{b}) + \frac{\mu}{qB_\parallel^*}\uv{b}\times\nabla B \sim \mathcal{O}(\epsilon_B)
\end{equation}
is the combined curvature and $\nabla B$ drifts, with $(\nabla\times\uv{b})_\perp=\uv{b}\times(\uv{b}\cdot\nabla\uv{b})$. Thus the final form of the gyrocenter one-form is
\begin{align}
    \Gamma &=  q\left(\vec{A}_0^*+ \langle\vec{A}_1\rangle\right)\cdot\dx{\vec{R}} + \frac{m\mu }{q}\dx{\vartheta} - \left[\frac{1}{2}mv_\parallel^2 + \mu B +q\langle\Phi_1-\vec{v}_\perp\cdot\vec{A}_1\rangle \right. \notag \\
    &\qquad- \left. \epsilon_V^2\frac{1}{2}m\left|\vec{v}_E + \vec{v}_f \right|^2 - \epsilon_V\epsilon_B\,
    m\vec{v}_d \cdot (\vec{v}_E+\vec{v}_f) \right]\dx{t}. \label{gam}
\end{align}

We can now see that the first-order correction to the gyrocenter one-form, $\Gamma_1$, has effectively  replaced the guiding-center potentials $\vec{A}_1(\vec{R})$ and $\Phi_1(\vec{R})$ that appeared in $\Gamma_0$ with gyroaveraged versions. This also resulted in additional higher-order terms in the Hamiltonian. While we could (and should) develop the Lie transform to second order using \cref{lie2} to obtain FLR corrections to these second-order terms, and possibly other second-order terms, the guiding-center (long-wavelength) versions of these terms as they appear in \cref{gam} will be sufficient for our current purposes, since we will only use the first-order Hamiltonian to compute the equations of motion. 

Proper treatment of the second-order $E\times B$ energy term $-m/2|\vec{v}_E|^2$ is necessary for deriving an energetically-consistent gyrokinetic Poisson equation, as we will see in \cref{sec:fieldtheory}. We have obtained this term without needing to compute the next-order Lie transform by making a convenient choice for the reference frame velocity $\vec{u}_\perp$ in \cref{gcv}, effectively guessing an $\mathcal{O}(\epsilon_V)$ correction to the velocity $\vec{v}$ that one could also find from continuing the Lie transform. We can compare the second-order terms in \cref{gam} to Eq. (54) from \citet{brizard2007foundations}, which gives the Hamiltonian that results from computing the Lie transform to second order, with the second-order terms given in the long-wavelength limit. We see that indeed we have recovered some of the second-order terms, but missed a term of the form $\mu|\vec{B}_{\perp1}|^2/(2B)$.

\subsection{Gyrocenter equations of motion} \label{sec:eom}
Now that we have the gyrocenter one-form given by \cref{gam}, we can derive the gyrocenter Poisson bracket and the gyrocenter equations of motion. At this point, we will simplify the system by assuming
$\vec{A}_1=A_\parallel\uv{b}$, so that
\begin{equation}
\vec{B}_{1} = \nabla \times (A_\parallel \uv{b})=-\uv{b}\times\nabla A_\parallel + A_\parallel \nabla \times \uv{b}. \label{b1}
\end{equation}
This results in the neglect of most compressional fluctuations of the magnetic field, although even in this form there remains a small compressional component, $\uv{b}\cdot\vec{B}_1=A_\parallel \uv{b}\cdot\nabla\times\uv{b}$, which may vanish or be finite depending on the particular magnetic geometry. Note that the second term in \cref{b1} is frequently dropped since it is smaller than the first term by $\mathcal{O}(\epsilon_B)$, but we will choose to keep it, in part so that $\nabla \cdot \vec{B}_1 = 0$ exactly. Future work will include the full compressional fluctuations $\delta B_\parallel$, which can influence microinstabilities not only at large $\beta \sim 1$ but also when gradients of $\beta$ are large, particularly in spherical torus machines like NSTX \citep{bourdelle2003,joiner2010,belli2010,zocco2015}.

We will also drop some second-order terms in the Hamiltonian, but for now we will leave the exact form of the Hamiltonian unspecified, since the Hamiltonian does not affect the form of the Poisson bracket. Thus we will write the gyrocenter Lagrangian as
\begin{equation}
    \mathcal{L} = q\left(\vec{A}_0^*+ \langle{A}_\parallel\rangle\uv{b}\right)\cdot\dot{\vec{R}} + \frac{m\mu }{q}\dot{\vartheta} - H = \vec{\Lambda} \cdot \dot{\vec{Z}} - H, \label{Ls}
\end{equation}
where $\vec{\Lambda}$ denotes the symplectic part of the Lagrangian.

The phase-space Euler-Lagrange equations are then given by
\begin{equation}
    \frac{\dx{}}{\dx{t}}\left(\pderiv{\mathcal{L}}{\dot{Z}^i} \right) = \pderiv{\mathcal{L}}{Z^i},
\end{equation}
which yields
\begin{equation}
    \dot{\Lambda}_i = \pderiv{\Lambda_j}{Z^i}\dot{Z}^j - \pderiv{H}{Z^i},
\end{equation}
where $Z^i\ (i=1,\dots,6)$ are the phase-space coordinates $\vec{Z}$ not including time.
We can expand the total time derivative $\dot{\Lambda}_i$ and rearrange terms to obtain
\begin{equation}
    \omega_{ij}\dot{ Z}^j = \pderiv{H}{Z^i} + \pderiv{\Lambda_i}{t}, \label{eom-omega}
\end{equation}
where the Lagrange tensor $\boldsymbol{\omega}$ is defined by
\begin{equation}
    \omega_{ij} \equiv \pderiv{\Lambda_j}{Z^i} - \pderiv{\Lambda_i}{Z^j}. \label{lagrange-tensor}
\end{equation}
Assuming $\det \boldsymbol{\omega}\neq 0$,
we can define the Poisson tensor $\mathbf{\Pi}$ to be the inverse of the Lagrange tensor, (i.e., $\Pi^{ik}\omega_{kj}=\delta^i_j$), so that the Euler-Lagrange equations from \cref{eom-omega} can be inverted to give the equations of motion as
\begin{equation}
    \dot{Z}^i=\Pi^{ij}\left(\pderiv{H}{Z^j} + \pderiv{\Lambda_j}{t}\right). \label{eom}
\end{equation}
Defining the (non-canonical) Poisson bracket as
\begin{equation}
    \{f,g\}\equiv \pderiv{f}{Z^i}\Pi^{ij}\pderiv{g}{Z^j}, \label{pbpi}
\end{equation}
and recognizing that $\Pi^{ij}=\{Z^i,Z^j\}$, we can also write the equations of motion as
\begin{equation}
   \dot{Z}^i=\{Z^i,H\} + \{Z^i,Z^j\}\pderiv{\Lambda_j}{t}. \label{eom2}
\end{equation}
Inserting the gyrocenter phase-space Lagrangian from \cref{gam} into \cref{lagrange-tensor}, the non-zero tensor elements are \citep{cary2009}
\begin{gather}
    \omega_{R_i v_\parallel} = -\omega_{v_\parallel R_i} = -m\, b_i \\
    \omega_{R_i R_j} = -\omega_{R_j R_i} = q\, \epsilon_{ijk} \bar{B}^{*k} \\
    \omega_{\mu \vartheta} = - \omega_{\vartheta \mu} = -\frac{m}{q},
\end{gather}
where $\epsilon_{ijk}$ is the Levi-Civita tensor and $\bar{B}^{*k}$ are the components of $\vec{\bar{B}}^*\equiv\langle \vec{B}^* \rangle = \vec{B}_0^*+ \nabla\times(\langle A_\parallel \rangle \uv{b})$, so that the full tensor takes the form
\begin{equation}
    \vec{\omega} = \begin{pmatrix}
    0 & q \bar{B}^{*3} & -q \bar{B}^{*2} & -m b_1 & 0 & 0 \\
    -q \bar{B}^{*3} & 0 & q \bar{B}^{*1} & -m b_2 & 0 & 0 \\
    q \bar{B}^{*2} & -q \bar{B}^{*1} & 0 & -m b_3 & 0 & 0 \\
    m b_1 & m b_2 & m b_3 & 0 & 0 & 0 \\
    0 & 0 & 0 & 0 & 0 & -\frac{m}{q} \\
    0 & 0 & 0 & 0 & \frac{m}{q} & 0 
    \end{pmatrix}.
\end{equation}
We can then invert this to obtain the Poisson tensor,
\begin{equation}
    \mathbf{\Pi} = \vec{\omega}^{-1} = \frac{1}{\bar{B}_\parallel^*} \begin{pmatrix}
    0 & -\frac{1}{q} b_3 & \frac{1}{q} b_2 & \frac{1}{m} \bar{B}^{*1} & 0 & 0 \\
    \frac{1}{q} b_3 & 0 & -\frac{1}{q} b_1 & \frac{1}{m} \bar{B}^{*2} & 0 & 0 \\
    -\frac{1}{q} b_2 & \frac{1}{q} b_1 & 0 & \frac{1}{m} \bar{B}^{*3} & 0 & 0 \\
    -\frac{1}{m} \bar{B}^{*1} & -\frac{1}{m} \bar{B}^{*2} & -\frac{1}{m} \bar{B}^{*3} & 0 & 0 & 0 \\
    0 & 0 & 0 & 0 & 0 & \frac{q}{m} \\
    0 & 0 & 0 & 0 & - \frac{q}{m} & 0 
    \end{pmatrix},
\end{equation}
with $\bar{B}_\parallel^*=\uv{b}\cdot\vec{\bar{B}}^*$.
The Poisson bracket is then given by \cref{pbpi},
\begin{align}
    \{f,g\} &= \frac{\vec{\bar{B}}^*}{m\bar{B}_\parallel^*}\cdot\left(\nabla f \pderiv{g}{v_\parallel} - \pderiv{f}{v_\parallel}\nabla g\right) - \frac{\uv{b}}{q \bar{B}_\parallel^*}\cdot \nabla f\times \nabla g + \frac{q}{m}\left(\pderiv{f}{\vartheta}\pderiv{g}{\mu} - \pderiv{f}{\mu}\pderiv{g}{\vartheta}\right).
\end{align}
We can also compute the Jacobian of the transformation from particle coordinates to gyrocenter coordinates $\vec{z}=(\vec{x},\vec{v})\rightarrow\vec{Z}=(\vec{R},v_\parallel,\mu,\vartheta)$,
\begin{equation}
    \mathcal{J} = \frac{1}{m^3}\sqrt{\det \boldsymbol{\omega}} = \frac{\bar{B}_\parallel^*}{m} = \frac{1}{m}\left[B + \left(\frac{m}{q}v_\parallel + \langle A_\parallel \rangle\right) \uv{b}\cdot\nabla\times \uv{b} \right] = \frac{B}{m} + \mathcal{O}(\epsilon_B),
    \label{jacob}
\end{equation}
where we will neglect $\uv{b}\cdot\nabla\times\uv{b} \lesssim \mathcal{O}(\epsilon_B)$ in the Jacobian. This approximation breaks the exact equivalence of $\vec{\omega}^{-1}$ and $\mathbf{\Pi}$, but otherwise does not affect conservation properties.
Note that the factor of $1/m^3$ comes from the Jacobian of the transformation from canonical to non-canonical particle coordinates $(\vec{x},\vec{p})\rightarrow(\vec{x},\vec{v})$.\footnote{In some texts this $1/m^3$ factor does not appear and the Jacobian is given as $m^2 B_\parallel^*$, which is the Jacobian of the transformation $(\vec{x},\vec{p})\rightarrow \vec{Z}$; an additional factor of $1/m$ can also appear when the parallel momentum $p_\parallel$ is used as a gyrocenter coordinate instead of $v_\parallel$.}

Now we can use \cref{eom2} to obtain the gyrocenter equations of motion, 
\begin{align}
    \dot{\vec{R}} &= \{\vec{R},H\} + \frac{\uv{b}}{q \bar{B}_\parallel^*}\times \pderiv{\vec{\Lambda}}{t} = \frac{\vec{\bar{B}}^*}{m\bar{B}_\parallel^*}\pderiv{H}{v_\parallel} + \frac{\uv{b}}{q \bar{B}_\parallel^*}\times\nabla H 
    \label{eomR}\\
   \dot{v}_\parallel &= \{v_\parallel,H\} - \frac{\vec{\bar{B}}^*}{m\bar{B}_\parallel^*}\cdot \pderiv{\vec{\Lambda}}{t}= -\frac{\vec{\bar{B}}^*}{m\bar{B}_\parallel^*}\cdot\nabla H - \frac{q}{m}\pderiv{\langle A_\parallel\rangle}{t} 
   \label{eomV} .
\end{align}
Finally, note that the zeroth-order (guiding-center) equations of motion are the same, except with all gyroaverages replaced by evaluation of the quantity at the guiding-center position.

\section{Gyrokinetic field theory} \label{sec:fieldtheory}
In the previous section, we derived the phase-space Lagrangian and equations of motion for a single charged particle in the presence of electromagnetic fields. Now we describe the collective behavior of a system of many such particles and the interactions between the particles and the fields. 

The system Lagrangian is given by integrating the single-particle Lagrangian over phase space, weighted by the distribution function $f_s$ and summed over all species $s$, plus an additional field term:
\begin{equation}
    \mathcal{L} = \sum_s \int\mathcal{J} f_s \mathcal{L}_s\, \dx{^6\vec{Z}} + \int   \mathcal{L}_f\, \dx{^3\vec{x}}.
\end{equation}
Here, $\dx{^6\vec{Z}}\equiv\dx{^3\vec{R}}\,\dx{v_\parallel}\, \dx{\mu} \,\dx{\vartheta}$ is the gyrocenter phase-space volume element with $\mathcal{J}=B/m$ the Jacobian (dropping the $\mathcal{O}(\epsilon_B)$ terms in \cref{jacob}), $\mathcal{L}_s$ is the single-particle Lagrangian for species $s$, and $\mathcal{L}_f$ is the field Lagrangian. Note that formally the Jacobian might be included in the definition of $\dx{^6\vec{Z}}$, but we instead opt to have the Jacobian appear explicitly in the expressions.

\subsection{The gyrokinetic Vlasov equation}
In the absence of sources and collisions (which we address later), the evolution of the distribution function $f$ is governed by the gyrokinetic Vlasov equation. This takes the form of Liouville's equation, which states that the distribution function is conserved along the nonlinear phase-space characteristics. This is expressed by
\begin{equation}
    \frac{\dx{f(\vec{Z},t)}}{\dx{t}}= \pderiv{f}{t} + \dot{\vec{Z}}\cdot\pderiv{f}{\vec{Z}} = \pderiv{f}{t} + \dot{\vec{R}}\cdot\nabla f + \dot{v}_\parallel \pderiv{f}{v_\parallel} = 0,
\end{equation}
with the phase-space characteristics given by \cref{eomR,eomV}. From \cref{eom2}, this can also be written in terms of the Poisson bracket as
\begin{equation}
    \pderiv{f}{t} + \{f,H\} + \{f,\vec{Z}\}\cdot\pderiv{\vec{\Lambda}}{t} = \pderiv{f}{t} + \{f,H\} - \frac{q}{m}\pderiv{\langle A_\parallel \rangle}{t}\pderiv{f}{v_\parallel} = 0.
\end{equation}

Together with Liouville's theorem, which states that phase-space volume is conserved, as expressed by
\begin{equation}
    \pderiv{\mathcal{J}}{t} + \pderiv{}{\vec{Z}}\cdot\left(\mathcal{J} \dot{\vec{Z}}\right) = \pderiv{\mathcal{J}}{t} +\nabla\cdot\left(\mathcal{J}\dot{\vec{R}}\right) + \pderiv{}{v_\parallel}\left(\mathcal{J}\dot{v}_\parallel\right) = 0, \label{liouvillethm}
\end{equation}
the gyrokinetic Vlasov equation can also be written in conservative form as
\begin{equation}
    \pderiv{(\mathcal{J}f)}{t} + \pderiv{}{\vec{Z}}\cdot\left(\mathcal{J}f \dot{\vec{Z}}\right) = \pderiv{(\mathcal{J}f)}{t} +\nabla\cdot\left(\mathcal{J}f\dot{\vec{R}}\right) + \pderiv{}{v_\parallel}\left(\mathcal{J}f\dot{v}_\parallel\right) = 0.
\end{equation}

\subsection{Variational derivation of the gyrokinetic field equations} \label{variational}
We follow \citet{sugama2000,scott2010}, to derive the gyrokinetic field equations. The field equations are  derived directly from the Lagrangian by requiring variations of the action, $\mathcal{I} = \int \mathcal{L}\, \dx{t}$, to vanish with respect to the fields $\Phi_1$ and $\vec{A}_1$. In this way, approximations and simplifications can be made at the level of the Lagrangian, and then the resulting field equations will be  consistent with those approximations, so that momentum and energy conservation are preserved.

Thus we must first specify the form of the Lagrangian. We consider three different cases: (1) keeping second-order terms in the single-particle Hamiltonian; (2) dropping second-order terms in the Hamiltonian; (3) dropping first- and second-order terms in the single-particle Lagrangian, resulting in the guiding-center Lagrangian. 

\subsubsection{Case 1: Single-particle Hamiltonian with second-order terms}
In the first case, we take the single-particle Lagrangian to be the gyrocenter Lagrangian from \cref{Ls}. For the Hamiltonian, we will keep second-order terms, but we will neglect all second-order terms involving magnetic fluctuations in \cref{ham2}, so that we are left with
\begin{equation}
    H = \frac{1}{2}mv_\parallel^2 + \mu B + q\langle \Phi\rangle - \frac{1}{2}m v_E^2 = \frac{1}{2}mv_\parallel^2 + \mu B + q\langle \Phi\rangle - \frac{m}{2B^2}|\nabla_\perp \Phi|^2.
\end{equation}
Note that here and after we will drop the subscript on $\Phi_1$, since there is no $\Phi_0$ to confuse it with.

The field Lagrangian $\mathcal{L}_f$ comes from the standard electrodynamic field term $(E^2-B^2)/(2\mu_0)$, but we assume quasineutrality, which eliminates the electric field term. After neglecting parallel fluctuations of the magnetic field, we have
\begin{equation}
    \mathcal{L}_f = -\frac{B_{1\perp}^2}{2\mu_0} \approx -\frac{|\nabla_\perp A_\parallel|^2}{2\mu_0}.
\end{equation}
The system Lagrangian for this case is now
\begin{align}
    \mathcal{L} &= \sum_s \int \mathcal{J} f_s \left[q_s\left(\vec{A}_0^* + \langle A_\parallel \rangle \uv{b}\right) \cdot \dot{\vec{R}} + \frac{m_s\mu}{q_s}\dot{\vartheta}\right. \notag \\ &\qquad\left.- \left(\frac{1}{2}m_s v_\parallel^2 + \mu B + q\langle \Phi\rangle - \frac{m_s}{2B^2}|\nabla_\perp \Phi|^2\right) \right]\dx{^6\vec{Z}} -\int \frac{|\nabla_\perp A_\parallel|^2}{2\mu_0}\dx{^3\vec{x}}.
\end{align}

The field equation for the electrostatic potential $\Phi$ is found from the requirement that variations of the action with respect to $\Phi$ vanish. This gives the condition $\delta \mathcal{I}/\delta \Phi(\vec{x})=0$, where the functional derivative is given by
\begin{align}
    \frac{\delta \mathcal{I}}{\delta \Phi(\vec{x})} &= \sum_s \int \left[ -\mathcal{J}f_s \pderiv{H_s}{\Phi(\vec{x})} + \nabla\cdot\left( \mathcal{J}f_s \pderiv{H_s}{\nabla \Phi(\vec{x})}\right) \right] \dx{^6\vec{Z}} \notag \\
    &= \sum_s \int \left[-q_s\mathcal{J} f_s \delta(\vec{x} - \vec{R} - \vec{\rho})  -\nabla\cdot\left( \mathcal{J}f_s\delta(\vec{x}-\vec{R})\frac{m_s}{B^2}\nabla_\perp \Phi\right) \right]\dx{^6\vec{Z}}.
\end{align}
Requirement that this quantity vanish yields an equation for $\Phi(\vec{x})$ that takes the form of the quasineutrality condition,
\begin{equation}
    \sigma_{gy} + \sigma_{pol} = \sigma_{gy} - \nabla \cdot \vec{P} = 0, \label{qneut}
\end{equation}
where the gyrocenter charge density is
\begin{equation}
    \sigma_{gy} = \sum_s q_s \int \langle\mathcal{J} f_s \rangle^\dagger \dx{^3\vec{v}} \equiv \sum_s q_s \bar{n}_s, \label{gycharge}
\end{equation}
with $\dx{^3\vec{v}}\equiv 2\pi \dx{v_\parallel}\dx{\mu}$. We will continue writing $\dx{^3\vec{v}}$ in integrals defined this way throughout, even though there are only two evolved velocity dimensions; the factor of $2\pi$ comes from a trivial integration over the gyroangle, the third velocity dimension. This factor cancels the factors of $(2\pi)^{-1}$ included in the defintions of the gyroaveraging operations, where integration over the gyroangle does appear explicitly. Again, we also choose to leave the phase-space Jacobian out of our definition of $\dx{^3\vec{v}}$ so that it appears explicitly in the expressions. The notation
\begin{equation}
    \langle f \rangle^\dagger \equiv \frac{1}{2\pi}\int_0^{2\pi} f(\vec{x}-\vec{\rho})\, \dx{\vartheta}
\end{equation}
denotes a gyroaverage taken at constant $\vec{x}$ (as opposed to constant $\vec{R}$). Note that this operator is the adjoint of the gyroaverage taken at constant $\vec{R}$ defined in \cref{gyavg}, \emph{i.e.} it satisfies the property
\begin{equation}
    \int \langle f \rangle\, g\, \dx{^3\vec{R}} = \int  f\langle g \rangle^\dagger\, \dx{^3\vec{x}}.
\end{equation}
The polarization charge density $\sigma_{pol}$ is given as the divergence of the polarization vector
\begin{equation}
    \vec{P} = - \sum_s \int  \mathcal{J}f_s \frac{m_s}{B^2}\nabla_\perp \Phi\,\dx{^3\vec{v}} = - \sum_s \frac{m_s n_s}{B^2}\nabla_\perp \Phi. \label{pol}
\end{equation}
The quasineutrality condition can then be written as the gyrokinetic Poisson equation, 
\begin{equation}
    -\nabla \cdot \sum_s \frac{m_s n_s}{B^2}\nabla_\perp \Phi = \sum_s q_s \bar{n}_s.
\end{equation}
Finally, note that the second order $E\times B$ energy term in the Hamiltonian was required to obtain the polarization charge density. Without it, the quasineutrality condition would not give us an equation for the potential.

The field equation for the parallel vector potential $A_\parallel$ is found from the requirement that variations of the action with respect to $A_\parallel$ vanish. This gives the condition $\delta \mathcal{I}/\delta A_\parallel(\vec{x})=0$, where the functional derivative is given by
\begin{align}
    \frac{\delta \mathcal{I}}{\delta A_\parallel(\vec{x})} &= -\nabla \cdot \pderiv{\mathcal{L}_f}{\nabla A_\parallel(\vec{x})} + \sum_s \int \mathcal{J} f_s \pderiv{\mathcal{L}_s}{A_\parallel(\vec{x})}\,\dx{^6\vec{Z}} \notag \\
    &= \frac{1}{\mu_0}\nabla_\perp^2 A_\parallel + \sum_s q_s\int \uv{b}\cdot\dot{\vec{R}}\, \mathcal{J}f_s\, \delta(\vec{x}-\vec{R}-\vec{\rho})\, \dx{^6 \vec{Z}} \notag \\
    &= \frac{1}{\mu_0}\nabla_\perp^2 A_\parallel + \sum_s q_s\int \frac{1}{m_s}\pderiv{H_s}{v_\parallel}\mathcal{J} f_s\, \delta(\vec{x}-\vec{R}-\vec{\rho})\,  \dx{^6 \vec{Z}}\notag \\
     &= \frac{1}{\mu_0}\nabla_\perp^2 A_\parallel + \sum_s q_s\int v_\parallel \mathcal{J}f_s\, \delta(\vec{x}-\vec{R}-\vec{\rho})\, \dx{^6 \vec{Z}}, \label{ampderiv}
\end{align}
where we have used \cref{eomR} to substitute for $\uv{b}\cdot\dot{\vec{R}}=v_\parallel$.
The requirement that this quantity vanish results in the parallel component of Amp\`ere's law,
\begin{equation}
    -\nabla_\perp^2 A_\parallel = \mu_0 \bar{J}_{\parallel}, \label{ampere}
\end{equation}
with
\begin{equation}
    \bar{J}_{\parallel} \equiv \sum_s q_s \int  v_\parallel \langle \mathcal{J}f_s \rangle^\dagger\,\dx{^3\vec{v}} = \sum_s q_s \overline{n_s u}_{\parallel s}
\end{equation}
the gyrocenter parallel current density.

\subsubsection{Case 2: Single-particle Hamiltonian without second-order terms}
In this case, we will drop all second-order terms in the Hamiltonian, so that we are left with
\begin{equation}
    H = \frac{1}{2}mv_\parallel^2 + \mu B + q\langle \Phi\rangle.
\end{equation}
As we noted above, the second order $E \times B$ energy term (which we have now dropped) was needed to obtain the polarization charge density in the quasineutrality equation. Without the second-order term, we need another way to obtain the polarization term. 

Following \citet{sugama2000,scott2010}, we can instead cast the higher-order $E\times B$ energy term into a field term, \emph{i.e.} as part of the field Lagrangian $\mathcal{L}_f$. To do this, we replace the distribution function multiplying this term in the system Lagrangian by a time-independent background distribution function $f_0$, giving
\begin{align}
    \mathcal{L} &=  \sum_s \int \mathcal{J}f_s\left(\mathcal{L}_{s0}+\mathcal{L}_{s1}\right)\dx{^6\vec{Z}}  + \sum_s \mathcal{J}f_{0s}\mathcal{L}_{s2}\,\dx{^6\vec{Z}}  -\int \frac{|\nabla_\perp A_\parallel|^2}{2\mu_0} \dx{^3\vec{x}} \\
    &=\sum_s \int\mathcal{J} f_s \left[q_s\left(\vec{A}_0^* + \langle A_\parallel \rangle \uv{b}\right) \cdot \dot{\vec{R}} + \frac{m_s\mu}{q_s}\dot{\vartheta} - \left(\frac{1}{2}m_s v_\parallel^2 + \mu B + q_s\langle \Phi\rangle \right) \right] \dx{^6\vec{Z}} \notag \\&\qquad + \sum_s \int\mathcal{J}f_{0s}\frac{m_s}{2B^2}|\nabla_\perp \Phi|^2\,\dx{^6\vec{Z}}-\int \frac{|\nabla_\perp A_\parallel|^2}{2\mu_0}\dx{^3\vec{x}}, \label{lagrangian2}
\end{align}
with the entire second line of \cref{lagrangian2} now comprising $\mathcal{L}_f$.

After following the same steps as in Case 1 to derive the quasineutrality condition from $\delta \mathcal{I}/\delta \Phi(\vec{x})=0$, we obtain the same expression for the gyrocenter charge density in \cref{gycharge}, but the polarization vector is modified as
\begin{equation}
    \vec{P} = - \sum_s \int \mathcal{J}f_{0s} \frac{m_s}{B^2}\nabla_\perp \Phi\, \dx{^3\vec{v}} = - \sum_s \frac{m_s n_{0s}}{B^2}\nabla_\perp \Phi,
\end{equation}
where the density in \cref{pol} has been replaced by some time-independent background density $n_0$. This result is sometimes referred to as \emph{linearized polarization}. From a computational perspective, it is helpful that in this case the kernel that must be inverted in the gyrokinetic Poisson equation does not change in time, allowing for parts of the inversion (\emph{e.g.} matrix factorization) to be done only at the beginning of the calculation. Thus the linearized polarization is commonly used for computational efficiency \citep{idomura2008,ku2018fast,shi2019}, even in the edge/SOL where large density fluctuations could lead to questions of the validity of replacing the full density with a background density.

It is important to note that using the linearized polarization approximation requires neglecting the second order $E\times B$ energy term in the Hamiltonian, and vice versa, in order for the resulting system to be energetically consistent. This will be shown explicitly in \cref{conservation}, where we show conservation properties of the system. 

Finally, the parallel Amp\`ere equation for $A_\parallel$ remains unchanged from \cref{ampere}; since the Hamiltonian does not depend on $A_\parallel$ (in the symplectic formulation), dropping second-order terms in the Hamiltonian has no effect on variations of the action with respect to $A_\parallel$.

\subsubsection{Case 3: Guiding-center single-particle Lagrangian}
We finally consider the case where we drop first- and second-order terms in the single-particle Lagrangian, resulting in the guiding-center Lagrangian (with no gyroaverages). Similar to Case 2, we can cast the first- and second-order terms into field terms multiplying a background distribution function:
\begin{align}
    \mathcal{L} &=  \sum_s \int \mathcal{J}f_s\mathcal{L}_{s0}\,\dx{^6\vec{Z}} + \sum_s \int\mathcal{J}f_{0s}\left(\mathcal{L}_{s1}+\mathcal{L}_{s2}\right)\dx{^6\vec{Z}} -\int \frac{|\nabla_\perp A_\parallel|^2}{2\mu_0}\dx{^3\vec{x}} \label{L0gen} \\
    &=\sum_s \int\mathcal{J} f_s \left[q_s\left(\vec{A}_0^* + A_\parallel \uv{b}\right) \cdot \dot{\vec{R}} + \frac{m_s\mu}{q_s}\dot{\vartheta} - \left(\frac{1}{2}m_s v_\parallel^2 + \mu B + q_s \Phi \right) \right]  \dx{^6\vec{Z}}\notag \\&\qquad + \sum_s \int\mathcal{J}f_{0s}\left(q_s[\langle A_\parallel \rangle -A_\parallel]\uv{b}\cdot\dot{\vec{R}} -q_s[\langle \Phi \rangle -\Phi]+ \frac{m_s}{2B^2}|\nabla_\perp \Phi|^2\right) \dx{^6\vec{Z}}\notag \\
    &\qquad-\int \frac{|\nabla_\perp A_\parallel|^2}{2\mu_0}\dx{^3\vec{x}}. \label{lagrangian3}
\end{align}

The functional derivative of the action with respect to variations of $\Phi(\vec{x})$ gives
\begin{align}
    \frac{\delta \mathcal{I}}{\delta \Phi(\vec{x})} &= \sum_s \int\bigg[ - q_s \mathcal{J}f_s \delta(\vec{x}-\vec{R}) - q_s\mathcal{J} f_{0s}[\delta(\vec{x}-\vec{R}-\vec{\rho})-\delta(\vec{x}-\vec{R})]\Bigg.\notag \\
    &\qquad \left.+ \nabla \cdot \left(\mathcal{J}f_{0s} \delta(\vec{x}-\vec{R})\frac{m_s}{B^2}\nabla_\perp \Phi\right)\right] \dx{^6\vec{Z}} \notag\\
    &= -\sum_s q_s\int  \mathcal{J}\left(f_s + [\langle f_{0s} \rangle^\dagger - f_{0s}]\right)\dx{^3\vec{v}} -\nabla\cdot \left(\sum_s \frac{m_s}{B^2}\nabla_\perp \Phi\int  \mathcal{J}f_{0s}\,\dx{^3\vec{v}}\right) \notag \\
    &= -\sum_s q_s (n_s + \bar{n}_{0s}-n_{0s}) -\nabla\cdot \sum_s \frac{m_s n_{0s}}{B^2}\nabla_\perp \Phi.
\end{align}
If we assume that the background density $n_0$ varies slowly on the gyroradius scale (consistent with the ordering $\rho/L_p\ll 1$), we can approximate $\bar{n}_{0s}-n_{0s}\approx 0$ (and to be consistent, we should also drop the $f_0 \mathcal{L}_1$ field term in \cref{L0gen}), so that the gyrokinetic Poisson equation becomes
\begin{equation}
    -\nabla \cdot \sum_s \frac{m_s n_{0s}}{B^2}\nabla_\perp \Phi = \sum_s q_s{n}_s. \label{poisson0}
\end{equation}
Consistent with dropping gyroaverages in the single-particle Lagrangian, the charge density is no longer gyroaveraged here compared to the other cases. Note that even after dropping all gyroaverage operations in the single-particle Lagrangian and the Poisson equation, the polarization density on the left-hand side of \cref{poisson0} still incorporates some lowest-order finite-Larmor-radius (FLR) effects.

Similarly, the Amp\`ere equation becomes
\begin{equation}
    -\nabla_\perp^2 A_\parallel = \mu_0 {J}_{\parallel} = \mu_0 \sum_s q_s n_s u_{\parallel s}, \label{ampere0}
\end{equation}
with the gyroaverage of the current density dropped as well.

\subsection{Conservation properties of the gyrokinetic Vlasov-Poisson-Amp\`ere system} \label{conservation}

The Hamiltonian structure of the gyrokinetic Vlasov-Poisson-Amp\`ere system guarantees conservation of arbitrary functions of $f$ along the characteristics,
\begin{equation}
    \pderiv{G(f)}{t} + \dot{\vec{Z}}\cdot\pderiv{}{\vec{Z}}G(f) = 0,
\end{equation}
along with corresponding Casimir invariants $\int \mathcal{J}G(f)\,\dx{^6\vec{Z}}$. Thus, the system has an infinite number of conserved quantities, including the total particle number (or $L_1$ norm) $N=\int \mathcal{J} f\,\dx{^6\vec{Z}}$, the $L_2$ norm $M=\int\mathcal{J} f^2\,\dx{^6\vec{Z}}$, and the kinetic entropy $S=-\int  \mathcal{J}f\ln f\,\dx{^6\vec{Z}}$ \citep{idomura2008}.

Conservation laws of energy and momentum can be derived by applying Noether's theorem to the action integral \citep{sugama2000}. The Noether energy $\mathcal{E}$ is given by varying the action with respect to time variations, which results in
\begin{align}
    \mathcal{E} &= \sum_s \int\mathcal{J} f_s \,\vec{\Lambda}_s\cdot\dot{\vec{Z}}\, \dx{^6\vec{Z}} -\mathcal{L} = \sum_s \int\mathcal{J} f_s H_s\, \dx{^6\vec{Z}} - \mathcal{L}_f,
\end{align}
where recall $\vec{\Lambda}_s=\partial \mathcal{L}_s/\partial \dot{\vec{Z}}$ is the symplectic part of the single-particle Lagrangian. We can verify that this is indeed a conserved quantity for each of the cases discussed in the previous section, by inserting the corresponding definitions of the Hamiltonian and field Lagrangian. The proof relies on the fact that the field equations have been derived consistently from the system Lagrangian, with all approximations made at the level of the Lagrangian. Considering Case 1 from the previous section, which includes the second order $E\times B$ term in the Hamiltonian and the full polarization density (as opposed to the linearized polarization density in the other cases), we can explicitly compute the time derivative of $\mathcal{E}$ as
\begin{align}
    \pderiv{\mathcal{E}}{t} &= \sum_s \int \left(\mathcal{J} f_s \pderiv{H_s}{t} + H_s \pderiv{(\mathcal{J}f_s)}{t}\right) \dx{^6\vec{Z}} - \pderiv{\mathcal{L}_f}{t} \notag \\
    &= \sum_s \int  \left(\mathcal{J}f_s\left[q_s\pderiv{\langle\Phi\rangle}{t} - \frac{m_s}{B^2}\nabla_\perp \Phi \cdot \nabla_\perp \pderiv{\Phi}{t} \right] \right. \notag \\
    &\qquad\left.-\ H_s\pderiv{}{\vec{Z}}\cdot\left(\mathcal{J}f_s\dot{\vec{Z}}\right) \right)\dx{^6\vec{Z}}+ \int \frac{1}{\mu_0}\nabla_\perp A_\parallel\cdot\nabla_\perp \pderiv{A_\parallel}{t}\, \dx{^3\vec{x}}
\end{align}
We can integrate by parts in several terms, and after assuming that boundary contributions vanish (boundary contributions are allowed, they just must be properly accounted for), this results in
\begin{align}
    \pderiv{\mathcal{E}}{t} &= \sum_s \int \left(\left[q_s \mathcal{J}f_s \pderiv{\langle\Phi\rangle}{t} +\nabla\cdot \left(\mathcal{J}f_s\frac{m_s}{B^2}\nabla_\perp\Phi\right)\pderiv{\Phi}{t}\right]\right. \notag \\
    &\qquad \left.+\ \mathcal{J}f_s \dot{\vec{Z}}\cdot\pderiv{H_s}{\vec{Z}}\right)\dx{^6\vec{Z}} - \int \frac{1}{\mu_0}\nabla_\perp^2 A_\parallel \pderiv{A_\parallel}{t}\,\dx{^3\vec{x}} \notag \\
    &= \sum_s \int \left(\left[q_s \langle \mathcal{J}f_s \rangle^\dagger +\nabla\cdot \left(\mathcal{J}f_s\frac{m_s}{B^2}\nabla_\perp\Phi\right)\right]\pderiv{\Phi}{t}\right. \notag \\
    &\qquad \left.-\ qv_\parallel\langle\mathcal{J}f_s\rangle^\dagger \pderiv{ A_\parallel}{t}\right)\dx{^3\vec{x}}\,\dx{^3\vec{v}} - \int\frac{1}{\mu_0}\nabla_\perp^2 A_\parallel \pderiv{A_\parallel}{t}\, \dx{^3\vec{x}} \notag \\
    &=\int \left( \left[\sigma_{gy} - \nabla \cdot \vec{P} \right] \pderiv{\Phi}{t} - \frac{1}{\mu_0}\left[\mu_0\bar{J}_\parallel + \nabla_\perp^2 A_\parallel \right]\pderiv{A_\parallel}{t}\right)\dx{^3\vec{x}} \notag \\
    &=0,
\end{align}
where we used \cref{qneut} and \cref{ampere}, and
\begin{align}
    \dot{\vec{Z}}\cdot\pderiv{H}{\vec{Z}}&= \{H,H\} + \{H,\vec{Z}\}\cdot\pderiv{\vec{\Lambda}}{t} =  \{H,\vec{Z}\}\cdot\pderiv{\vec{\Lambda}}{t} \notag \\&=-\uv{b}\cdot \dot{\vec{R}}\pderiv{ \langle A_\parallel\rangle}{ t}= -\frac{q}{m}\pderiv{H}{v_\parallel}\pderiv{\langle A_\parallel\rangle}{t} = - q v_\parallel \pderiv{\langle A_\parallel\rangle}{t},
\end{align}
with $\{H,H\}=0$ from antisymmetry of the Poisson bracket.

Similarly, the Noether toroidal momentum is given by varying the action with respect to spatial variations, which results in
\begin{equation}
    \mathcal{P} = \sum_s \int \mathcal{J} f_s P_\varphi\, \dx{^6\vec{Z}},
\end{equation}
where 
\begin{equation}
    P_\varphi \equiv \pderiv{\mathcal{L}}{\dot{\varphi}} = q{A}_\varphi + (q A_\parallel + mv_\parallel) b_\varphi,
\end{equation}
with $A_\varphi = \vec{A}_0\cdot(\partial \vec{R}/\partial \varphi)$ and $b_\varphi = \uv{b}\cdot(\partial \vec{R}/\partial \varphi)$.

\section{Summary of gyrokinetic system, in limit of current interest} \label{sec:gk-summary}
Here we summarize the gyrokinetic system, in the limit that we will use for the remainder of this thesis. As a first step towards full-$f$ electromagnetic gyrokinetic simulations of the plasma boundary region, we have implemented the lowest-order (guiding-center, or drift-kinetic) limit of the system in the \gke code, neglecting all gyroaveraging operations. This is a matter of simplicity, and implementing gyroaveraging effects given by the next order terms we have derived is important future work. We emphasize that this ``long-wavelength'' limit is a valid limit of our full-$f$ {gyrokinetic} derivation since we took care to include the guiding-center components of the field perturbations at $\mathcal{O}(1)$ in \cref{a-tot,phi-tot}. Further, although one may think of this as a drift-kinetic limit, the presence of the ion polarization term in the quasineutrality equation distinguishes the long-wavelength gyrokinetic model from versions of drift-kinetics that include the polarization drift in the equations of motion or that determine the potential from some other equation.

In this limit, the gyrokinetic Poisson bracket is given by
\begin{equation}
    \{F,G\} = \frac{\vec{B}^*}{m B_\parallel^*} \cdot \left(\nabla F \frac{\partial G}{\partial v_\parallel} - \frac{\partial F}{\partial v_\parallel}\nabla G\right) - \frac{ \uv{b}}{q B_\parallel^*}\times \nabla F \cdot \nabla G, \label{gkpb}
\end{equation}
with $\vec{B}^*=\vec{B} + (mv_\parallel/q)\nabla\times\uv{b}+\vec{B}_{1}$, $\vec{B}_{1}= \nabla\times(A_\parallel\uv{b})$, and $B_\parallel^*=\uv{b}\cdot\vec{B}^*\approx B$. The Hamiltonian is
\begin{equation}
    H = \frac{1}{2}m v_\parallel^2 + \mu B + q\Phi. \label{ham}
\end{equation}
Inserting this into \cref{eomR,eomV}, this results in the (guiding-center) equations of motion,
\begin{gather}
    \dot{\vec{R}} = \{\vec{R},H\} = \frac{\vec{B^*}}{B_\parallel^*}v_\parallel + \frac{\uv{b}}{q B_\parallel^*}\times\left(\mu\nabla B + q \nabla \Phi\right), \label{eomR0}\\
    \dot{v}_\parallel = \dot{v}^H_\parallel -\frac{q}{m}\pderiv{A_\parallel}{t} = \{v_\parallel,H\}-\frac{q}{m}\pderiv{A_\parallel}{t} = -\frac{\vec{B^*}}{m B_\parallel^*}{\cdot}\left(\mu\nabla B + q \nabla \Phi\right)-\frac{q}{m}\pderiv{A_\parallel}{t}. \label{eomV0}
\end{gather}
In \cref{eomV0} we have separated $\dot{v}_\parallel$ into a term that comes from the Hamiltonian, $\dot{v}^H_\parallel = \{v_\parallel,H\}$, and the term that comes from the symplectic part of the Lagrangian that is proportional to the inductive component of the parallel electric field, $(q/m)\partial {A_\parallel}/\partial{t}$. We use this notation for convenience, and so that the time derivative of $A_\parallel$ appears explicitly.

The gyrokinetic equation for species $s$ is then given by
\begin{align}
    \pderiv{f_s}{t} + \dot{\vec{R}} \cdot\nabla f_s + \dot{v}^H_\parallel \pderiv{f_s}{v_\parallel}- \frac{q_s}{m_s}\pderiv{A_\parallel}{t}\pderiv{f_s}{v_\parallel} = C[f_s] + S_s,
\end{align}
or equivalently,
\begin{equation}
    \pderiv{f_s}{t} + \{f_s,H_s\} - \frac{q_s}{m_s}\pderiv{ A_\parallel }{t}\pderiv{f_s}{v_\parallel} = C[f_s]+S_s,
\end{equation}
or in conservative form as
\begin{equation}
\pderiv{(\mathcal{J}f_s)}{t} + \nabla{\cdot}( \mathcal{J} \dot{\vec{R}} f_s) + \pderiv{}{v_\parallel}\left(\mathcal{J}\dot{v}^H_\parallel f_s\right)- \pderiv{}{v_\parallel}\left(\mathcal{J}\frac{q_s}{m_s}\pderiv{A_\parallel}{t} f_s \right)= \mathcal{J} C[f_s] + \mathcal{J} S_s. \label{emgk} 
\end{equation}
Here we have included collisions, $C[f_s]$, and sources, $S_s$, which we did not derive in this chapter. Details about the model collision operator are included briefly below.

The field equations are the ones derived in \cref{variational}, Case 3, consistent with neglecting gyroaveraging operations in the equations of motion. The gyrokinetic Poisson equation is 
\begin{equation}
    -\nabla \cdot \left(\epsilon_\perp \nabla_\perp \Phi\right) = \sum_s q_s \int  \mathcal{J} f_s\, \dx{^3\vec{v}}, \label{poisson1}
\end{equation}
with 
\begin{equation}
    \epsilon_\perp = \sum_s \frac{m_s n_{0s}}{B^2},
\end{equation}
and the parallel Amp\`ere equation is
\begin{equation}
    -\nabla_\perp^2 A_\parallel = \mu_0 \sum_s q_s \int  v_\parallel \mathcal{J} f_s\,\dx{^3\vec{v}}. \label{ampere1}
\end{equation}
Note that we can also take the time derivative of this equation to get a generalized Ohm's law which can be solved directly for $\pderivInline{A_\parallel}{t}$, the inductive component of the parallel electric field $E_\parallel$ \citep{reynders1993gyrokinetic, cummings1994gyrokinetic, chen2001gyrokinetic}i:
\begin{equation}
    -\nabla_\perp^2 \pderiv{A_\parallel}{t} = \mu_0 \sum_s q_s \int v_\parallel \pderiv{(\mathcal{J} f_s)}{t}\, \dx{^3\vec{v}}.
\end{equation}
Writing the gyrokinetic equation as
\begin{equation}
    \pderiv{(\mathcal{J}f_s)}{t} = 
    \pderiv{(\mathcal{J}f_s)}{t}^\star + \pderiv{}{v_\parallel}\left(\mathcal{J} \frac{q_s}{m_s}\pderiv{A_\parallel}{t} f_s\right), \label{fstar}
\end{equation}
where $\partial{(\mathcal{J}f_s)^\star}/\partial{t}$ denotes all the terms in the gyrokinetic equation (including sources and collisions) except the $\pderivInline{A_\parallel}{t}$ term, Ohm's law can be rewritten (after an integration by parts) as 
\begin{equation}
    \left(-\nabla_\perp^2 + \sum_s \frac{\mu_0 q_s^2}{m_s} \int\mathcal{J} f_s\, \dx{^3\vec{v}}\right) \pderiv{A_\parallel}{t} = \mu_0 \sum_s q_s \int v_\parallel \pderiv{(\mathcal{J}f_s)}{t}^\star\,\dx{^3\vec{v}}. \label{ohmstar}
\end{equation}

Finally, the conserved energy in this system is
\begin{align}
    \mathcal{E} &= \mathcal{E}_H - \mathcal{E}_E + \mathcal{E}_B \notag \\
    &= \sum_s \int \mathcal{J} f_s H_s \,\dx{^6\vec{Z}} - \int \frac{\epsilon_\perp}{2}|\nabla_\perp \Phi|^2\,\dx{^3\vec{R}} + \int \frac{1}{2\mu_0}|\nabla_\perp A_\parallel|^2\,\dx{^3\vec{R}}.
\end{align}

\section{Model collision operator}
To model the effect of collisions we use a conservative Lenard--Bernstein (or Dougherty) collision operator \citep{Lenard1958,Dougherty1964},
\begin{align} \label{eq:GkLBOEq}
\mathcal{J}C[f] &= \nu\left\lbrace\pderiv{}{v_\parallel}\left[\left(v_\parallel - u_\parallel\right)\mathcal{J} f+v_{t}^2\pderiv{(\mathcal{J} f)}{v_\parallel}\right]+\pderiv{}{\mu}\left[2\mu \mathcal{J} f+2\mu\frac{m}{B}v_{t}^2\pderiv{(\mathcal{J} f)}{\mu}\right]\right\rbrace,
\end{align}
where
\begin{align}
    n u_\parallel = \int  \mathcal{J} v_\parallel f \,\dx{^3\vec{v}}, \qquad\qquad n u_\parallel^2 + 3 n v_{t}^2 = \int  \mathcal{J}\left(v_\parallel^2 + 2\mu B/m\right)f\,\dx{^3\vec{v}},
\end{align}
with $n = \int \mathcal{J} f\,\dx{^3\vec{v}}$.
This collision operator contains the effects of drag and pitch-angle scattering, and it conserves number, momentum and energy density. Consistent with our present long-wavelength treatment of the gyrokinetic system, finite-Larmor-radius effects are ignored. For simplicity we restrict ourselves to the case in which the collision frequency $\nu$ is velocity independent, i.e. $\nu\neq\nu(v)$.  
Further details about this collision operator, including its conservation properties and its numerical discretization, are shown in  \citet{francisquez2020}.

\chapter{Numerical methods: an electromagnetic full-$f$ gyrokinetic scheme} \label{ch:dg}

The electromagnetic gyrokinetic system described in the previous chapter requires robust numerical methods that can preserve the underlying conservation laws. For this, we have chosen a numerical scheme based on the discontinuous Galerkin (DG) finite-element method. In this chapter, we develop a DG scheme that is explicitly constructed to conserve energy in Hamiltonian systems, and then we apply the general Hamiltonian scheme to electromagnetic gyrokinetics.

\section{The discontinuous Galerkin method}
The discontinuous Galerkin method comprises a class of Galerkin methods for numerically solving partial differential equations that combines attractive features of finite-element and finite-volume methods. The result is a method with flexibility in choice of local, arbitrarily high-order basis functions (as in finite-element methods), along with the ability to locally enforce conservation laws (as in finite-volume methods). DG methods first appeared in the study of neutron transport \citep{reed1973}. The work of \citet{Cockburn1998,Cockburn2001} introduced the Runge-Kutta discontinuous Galerkin (RKDG) method for the solution of nonlinear, time-dependent hyperbolic systems, leading to the use and study of DG methods for a wide variety of problems in computational fluid dynamics and other areas. For a more detailed introduction to DG methods, see the textbooks of \citet{hesthaven2007} and \citet{durran2010numerical} and the review by \citet{Cockburn2001}.

\subsection{DG for hyperbolic conservation laws} \label{sec:hyperbolic-cons}
To introduce the DG scheme, we will first focus on a scalar hyperbolic conservation law in one dimension of the generic form
\begin{equation}
    \pderiv{f}{t} + \pderiv{F(f)}{x} = 0, \label{pde}
\end{equation}
with ${F}(f)$ some arbitrary (possibly nonlinear) flux, and the system defined on some region $x\in \Omega$ and subject to some boundary conditions and initial conditions. 

We begin by dividing the region $\Omega$ into a mesh $\mathcal{T}$ of $N$ non-overlapping cells $\mathcal{K}_i\in \mathcal{T}$, with cell $i$ defined by $\mathcal{K}_i = [x_{i-1/2},x_{i+1/2}]$, where $x_{i+1/2}=(x_i+x_{i+1})/2$ and $x_i$ is the center of cell $i$. 
We next define a piecewise-polynomial approximation space for the solution,
\begin{equation}
    \mathcal{V}_h^p = \{\psi : \psi|_{\mathcal{K}_i}\in \vec{P}^p, \forall\ \mathcal{K}_i\in \mathcal{T}\},
\end{equation}
where $\vec{P}^p$ is some space of polynomials with maximum degree $p$ (by some measure) containing polynomial functions $\psi=\psi(x)$ local to each cell. The approximate solution is then defined in each cell as a finite sum of expansion functions $\psi_k(x)$,
\begin{equation}
    f_h^i(x,t) = \sum_{k=0}^{p} f_k^i(t) \psi_k(x). \label{fexp}
\end{equation}
The global approximate solution $f_h(x,t)$, composed as a direct sum of the $N$ local solutions as
\begin{equation}
    f_h(x,t) = \bigoplus_{i=1}^N f_h^i(x,t), \label{fexpand}
\end{equation}
is assumed to approximate the exact solution $f(x,t) \simeq f_h(x,t)$. In the form defined above, the global solution $f_h$ can be discontinuous at the interface between two cells; there are no restrictions on the local coefficients $f_k^i$ in neighboring cells, so continuity is not enforced in general. The discontinuous piecewise-polynomial form of the global solution is a key part of the discontinuous Galerkin method. 

Inserting the approximate solution $f_h$ into \cref{pde}, we obtain the residual
\begin{equation}
    \mathcal{R}_h(x, t) = \pderiv{f_h}{t} + \pderiv{F(f_h)}{x}.
\end{equation}
Various schemes can be given by particular choices for how to minimize the residual. The DG scheme is given by minimizing the residual via the Galerkin condition, which can be stated as the requirement that the residual in each cell be orthogonal to all test functions $\psi$ in the solution space,
\begin{equation}
     \int_{\mathcal{K}_i} \psi\,\mathcal{R}_h\,\dx{x} = \int_{\mathcal{K}_i} \psi \left(\pderiv{f_h}{t} + \pderiv{F(f_h)}{x}\right)\dx{x} = 0, \qquad \forall\ \psi\in \mathcal{V}_h^p. \label{galerkin}
\end{equation}
Since $F(f_h)$ can be discontinuous at cell boundaries, the spatial derivative of $F(f_h)$ that appears in \cref{galerkin} introduces delta functions at the boundaries. To avoid this, we integrate by parts in space to move the derivative off of $F$. This results in the DG \emph{weak form} of the system,
\begin{equation}
    \int_{\mathcal{K}_i} \psi \pderiv{f_h}{t}\,\dx{x} - \int_{\mathcal{K}_i} \pderiv{\psi}{x}F_h\,\dx{x} + \bigg[ \psi \hat{F} \bigg]^{x_{i+1/2}}_{x_{i-1/2}} = 0.
\end{equation}
The weak form now contains a volume integral term (the second term on the left-hand side above) and a surface integral term (the third term on the left-hand side, where the `surface' of cell $\mathcal{K}_i$ is just the endpoints of the cell in this simple 1D example). We have introduced the numerical flux $\hat{F}=\hat{F}(F^-,F^+)$ in the surface term. This arises because the flux $F_h=F(f_h)$ is not unique at the cell boundaries since $f_h$ can be discontinuous at the cell boundary, resulting in different values of the flux when evaluated just inside ($F^-$) or just outside ($F^+$) the boundary. The choice of the form of the numerical flux depends on the system of interest. For advection, \textit{i.e.} when $F(f) = u f$ with $u$ the advection velocity, a common choice is the upwind flux,
\begin{equation}
    \hat{F}(F^-,F^+) = \frac{1}{2}\left[F^+ + F^-\right] - \frac{1}{2}\text{sgn}(u)\left[F^+-F^-\right] =  \begin{cases} 
        F^- \quad \text{if } u>0 \\
        F^+ \quad \text{if } u<0.
    \end{cases}
\end{equation}  
The numerical flux is common to both sides of the cell boundary, so that the flux of particles out of one cell is identical to the the flux into the adjacent cell through the shared boundary. This ensures that the $L_1$ norm $N=\int f\, \dx{\vec{x}}$ is conserved. In general, the numerical flux must be consistent, so that $\hat{F}(F,F)=F$. Finally, drawing on the success of monotone finite-volume methods, the flux should be monotone by requiring it to be non-decreasing in the first argument ($F^-$) and non-increasing in the second argument ($F^+$) \citep{leveque2002}.

We can now obtain a system of coupled differential equations for the time evolution of the weights $f_k^i(t)$ by inserting the expansion from \cref{fexp} into the weak form:
\begin{equation}
    \sum_{k=0}^p\int_{\mathcal{K}_i} \psi_j \psi_k \pderiv{f_k^i}{t}\,\dx{x} - \int_{\mathcal{K}_i} \pderiv{\psi_j}{x}F_h\,\dx{x} + \bigg[ \psi_j \hat{F} \bigg]^{x_{i+1/2}}_{x_{i-1/2}} = 0, \qquad j=0,\dots,p. \label{dgweakgen1d}
\end{equation}
Note that in the special case where the polynomials $\psi\in\mathcal{P}^p$ are orthonormal, this reduces to
\begin{equation}
    \pderiv{f_j^i}{t} = \int_{\mathcal{K}_i} \pderiv{\psi_j}{x}F_h\,\dx{x} - \bigg[ \psi_j \hat{F} \bigg]^{x_{i+1/2}}_{x_{i-1/2}} = 0, \qquad j=0,\dots,p.
\end{equation}
This system of equations can now be discretized in time using an explicit scheme such as a high-order Runge--Kutta (RK) time discretization scheme, resulting in the RKDG method \citep{Cockburn1998}.

The scheme can be easily generalized to a multi-dimensional hyperbolic conservation law,
\begin{equation}
    \pderiv{f}{t} + \nabla\cdot \vec{F}(f) = 0. \label{pde2}
\end{equation}
As above, the DG weak form in cell $i$ is given by multiplying \cref{pde2} by a test function $\psi$, and integrating (by parts) over the cell. This gives
\begin{equation}
    \int_{\mathcal{K}_i}  \psi \pderiv{f_h}{t}\dx{\vec{x}}- \int_{\mathcal{K}_i} \vec{F}_h\cdot\nabla \psi\, \dx{\vec{x}}  + \oint_{\partial \mathcal{K}_i}\psi^- \vec{\hat{F}}\cdot\dx{\vec{s}} = 0, \label{DGweakgen}
\end{equation}
where now the surface term takes the form of a surface integral over the cell boundary $\partial \mathcal{K}_i$, with $\dx{\vec{s}}$ the differential element pointing outward normal to the surface.

\subsection{Choice of basis functions}

There is significant freedom for the choice of basis functions $\psi \in \vec{P}^p$. The class of possible basis functions is typically grouped into nodal and modal families. In the nodal approach, the basis functions are usually taken to be the Lagrange interpolating polynomials, which in one-dimension are given by
\begin{equation}
    \ell_k(x) = \prod_{\begin{smallmatrix}0\le m\le p\\ m\neq k\end{smallmatrix}} \frac{x-x_m}{x_k-x_m}.
\end{equation}
with $x_k$ the set of nodes chosen to represent the solution. The nodes are typical chosen to be quadrature points so that the integrals in the DG weak form can be computed efficiently. The coefficients in the expansion of the solution are then just the values of the solution at the nodes, so that $f_k^i(t) = f_h^i(x_k,t)$ in \cref{fexpand}.

We instead take the modal approach, where we obtain the expansion coefficients by projecting the solution onto some set of `modes' $\psi_k$, so that
\begin{equation}
    f_k^i(t) = \int_{\mathcal{K}_i} \psi_k(x) f_h^i(x,t) \, \dx{x}.
\end{equation}
It is convenient to map each cell $\mathcal{K}_i$ to the interval $[-1,1]$ in each dimension. For the one-dimensional weak form from \cref{dgweakgen1d}, this can be accomplished using the transformation
\begin{equation}
    \xi = \frac{2(x-x_i)}{\Delta x_i},
\end{equation}
where cell $\mathcal{K}_i=[x_i - \Delta x_i/2, x_i+\Delta x_i/2]$ has cell center $x_i$ and width $\Delta x_i$. This gives
\begin{equation}
    \dx{x} = \frac{\Delta x_j}{2}\dx{\xi}, \qquad \pderiv{}{x} = \frac{2}{\Delta x_j}\pderiv{}{\xi}, 
\end{equation}
so that \cref{dgweakgen1d} becomes
\begin{equation}
    \frac{\Delta x_i}{2}\sum_{k=0}^p \int_{-1}^1 \psi_j\psi_k \pderiv{f_k}{t}\, \dx{\xi} - \int_{-1}^1 \pderiv{\psi_j}{\xi} F_h\, \dx{\xi} + \bigg[ \psi_j \hat{F} \bigg]^{\,1}_{-1} = 0.
\end{equation}
The first term on the left-hand side contains the mass matrix, 
\begin{equation}
    M_{jk} \equiv \int_{-1}^1 \psi_j \psi_k\, \dx{\xi}.
\end{equation}
It is then efficient to choose an orthogonal basis so that the mass matrix is diagonal, or even better, an orthonormal basis so that the mass matrix is the identity matrix. This can be done by starting with the simple monomial basis $\psi_k(x) = x^k$ and using a Gram-Schmidt orthogonalization procedure to generate an orthogonal basis on the interval $[-1,1]$, which can then be appropriately normalized so that the basis is orthonormal. As a result, the 1D DG weak form in cell $i$ becomes
\begin{equation}
    \pderiv{f_k}{t} = \frac{2}{\Delta x_i}\int_{-1}^1 \pderiv{\psi_j}{\xi} F_h\, \dx{\xi} - \frac{2}{\Delta x_i} \bigg[ \psi_j \hat{F} \bigg]^{\,1}_{-1}.
\end{equation}

These approaches can be generalized to higher dimensions by taking Lagrange tensor products of the one-dimensional basis sets. This results in the number of basis functions within a cell scaling like $(p+1)^d$ for $d$ dimensions, which gives an exponential cost scaling with dimensionality. Since this can be prohibitive for a five-dimensional system like gyrokinetics, we instead reduce the tensor product basis by employing the serendipity basis set \citep{arnold2011serendipity}. Starting from a tensor product basis of the monomials with degree $p$, elements are dropped if they have super-linear degree greater than $p$, defined to be the total degree of the polynomial with respect to variables which enter super-linearly (so for example, the super-linear degree of $x^2 y z^3$ is 5). For a two-dimensional $p=2$ serendipity basis, this means that the $x^2 y^2$ element is dropped because its super-linear degree is four, while $x y^2$ and $x^2 y$ are kept because they have super-linear degree of two, equal to $p$. The resulting reduced set of monomials can then be orthogonalized and orthonormalized with a Gram-Schmidt procedure as in one dimension. The serendipity basis set has the advantage of using fewer basis functions while giving the same formal convergence order (although it is less accurate) as the Lagrange tensor
basis. Note however that for $p=1$ the serendipity basis is equivalent to the Lagrange tensor basis. 

A more complete treatment of the advantages of various choices for DG basis sets is given by \citet{juno-thesis}.

\section{An energy-conserving discontinuous Galerkin \\scheme for general Hamiltonian systems} \label{sec:dgenergy}
A broad class of problems in fluid mechanics and plasma physics are described by Hamiltonian systems. As we saw in \cref{ch:emgk}, this includes the electromagnetic gyrokinetic system. In this section, we discuss a discontinuous Galerkin scheme for general Hamiltonian systems that is explicitly constructed to be energy-conserving. We will apply this scheme to the electromagnetic gyrokinetic system in \cref{sec:gkscheme}.
The scheme presented here (and in more detail in \citet{hakim2019}) is a generalization of the DG scheme presented by \citet{liu2000high} for the 2D incompressible Euler equations, which can be expressed in Hamiltonian form as we will see below. 

\subsection{Evolution of general Hamiltonian systems} \label{sec:hamil}

The phase-space evolution of a Hamiltonian system is in general given by the Liouville equation, which describes the conservation of the phase-space distribution function $f(t,\vec{Z})$ along trajectories in phase space,
\begin{equation}
    \pderiv{f}{t} + \dot{\vec{Z}}\cdot\pderiv{f}{\vec{Z}}=0. \label{liouville}
\end{equation}
In this section we will assume that the coordinates $\vec{Z}=(Z^1,\ldots,Z^{N_d})$, which label the $N_d$-dimensional
phase space in which the distribution function evolves, are canonical or that they resulted from a time-independent non-canonical coordinate transformation\footnote{The symplectic formulation of electromagnetic gyrokinetics is derived via a time-\emph{dependent} non-canonical coordinate transformation. In \cref{sec:gkscheme} we will show that the scheme in this section can be generalized to account for time dependence in the symplectic structure.}. This means that the equations of motion can be written as
\begin{equation}
    \dot{\vec{Z}} = \{\vec{Z},H\},
\end{equation}
where $H$ is the Hamiltonian and $\{g,h\}$ is the Poisson bracket operator. Equivalently, \cref{liouville} can be written in terms of the Poisson bracket as
\begin{equation}
    \pderiv{f}{t} + \{f,H\} = 0. \label{eq:liouville}
\end{equation}
Liouville's theorem also provides that phase-space volume is conserved under (possibly non-canonical) coordinate transformations. Given a coordinate transformation $\bar{\vec{Z}}\rightarrow \vec{Z}$ with Jacobian $\jac$ such that $\dx{\bar{\vec{Z}}} = \jac \dx{\vec{Z}}$, this can be stated as
\begin{equation}
    \pderiv{}{\vec{Z}}\cdot\left(\jac \dot{\vec{Z}}\right) = 0,
\end{equation}
where again in this section we assume the Jacobian of the transformation is time-independent.
This allows us to write the Liouville equation in conservative form as
\begin{equation}
\pderiv{}{t}\left(\jac f\right) + \pderiv{}{\vec{Z}}\cdot\left(\jac f \dot{\vec{Z}} \right) = 0, \label{eqn:liouville-cons}
\end{equation}
which now has the same form of a hyperbolic conservation law as in \cref{sec:hyperbolic-cons}, with $\vec{F} = \jac f\dot{\vec{Z}}$ the flux. Finally, the form of the Hamiltonian and the equation(s) governing its evolution depend on the system of interest. 

Hamiltonian systems conserve the total energy of the system, given by
\begin{equation}
    \mathcal{E} = \int H \mathcal{J}f\, \dx{\vec{Z}} - \mathcal{L}_f, \label{hamil-energy}
\end{equation}
where $\mathcal{L}_f$ accounts for possible field terms,
such that
\begin{equation}
    \pderiv{\mathcal{E}}{t} = \int \left(H \pderiv{(\mathcal{J} f)}{t} + \mathcal{J} f\pderiv{H}{t}\right)\dx{\vec{Z}} - \pderiv{\mathcal{L}_f}{t} = 0.
\end{equation}
The first term vanishes upon integration by parts, assuming no boundary contributions, since 
$\dot{\vec{Z}}\cdot\partial H/\partial\vec{Z} = \{H, H\} = 0$; physically, this is because the flow $\dot{\vec{Z}}$ is along contours of constant energy in phase space. For systems with time-dependent Hamiltonians, the field equation governing the Hamiltonian is required to show that the second and third terms above cancel exactly.

%
\subsection{Discontinuous Galerkin discretization scheme} \label{sec:dg-hamil}

Now we can follow the ideas from \cref{sec:hyperbolic-cons} to make a DG discretization of \cref{eqn:liouville-cons}. We start by decomposing the global phase-space domain $\Omega$ into a {structured} phase-space mesh $\mathcal{T}$ with cells $\mathcal{K}_i \in \mathcal{T},\ i=1,...,N$. As above, we then introduce a piecewise-polynomial approximation space for the distribution function $f(t,\vec{Z})$,
\begin{equation}
    \mathcal{V}_h^p = \{\psi : \psi|_{\mathcal{K}_i}\in \vec{P}^p, \forall\ \mathcal{K}_i\in \mathcal{T}\},
\end{equation}
The DG weak form in cell $i$ is then obtained by multiplying \cref{eqn:liouville-cons} by a test function $\psi\in\mathcal{V}^p_h$ and integrating (by parts) in the cell, yielding
\begin{equation}
    \int_{\mathcal{K}_i}  \psi \pderiv{(\mathcal{J}f_h)}{t}\dx{\vec{Z}}- \int_{\mathcal{K}_i} \mathcal{J}f_h \dot{\vec{Z}}_h\cdot\pderiv{\psi}{\vec{Z}}\, \dx{\vec{Z}}  + \oint_{\partial \mathcal{K}_i}\psi^- \widehat{\jac f_h \dot{\vec{Z}}_h}\cdot\dx{\vec{s}} = 0, \label{eq:dis-weak-form}
\end{equation}
where $\dot{\vec{Z}}_h=\{\vec{Z},H_h\}$, and $\hat{\vec{F}} = \widehat{\mathcal{J} f_h \dot{\vec{Z}}_h}$ is a numerical flux function. Solving \cref{eq:dis-weak-form} for all test functions $\psi\in\mathcal{V}_h^p$  in all cells $\mathcal{K}_i\in\mathcal{T}$ yields the discretized distribution function $f_h\in\mathcal{V}^p_h$. However, noting that the quantity $\mathcal{J}f_h$ always appears together, it is convenient to instead discretize the Jacobian-weighted distribution function, $\mathcal{J} f_h \in \mathcal{V}^p_h$. 

We have not yet addressed the approximation space for the discrete Hamiltonian, $H_h$. For this, we will introduce a subset of $\mathcal{V}^p_h$ where the piecewise polynomials are continuous across cell interfaces, denoted by $\fillover{\mathcal{V}}^p_h=\mathcal{V}^p_h\cap C_0(\vec{Z})$, with $C_0(\vec{Z})$ the set of continuous functions.
As we will show later, in order to maintain energy conservation in our discrete scheme, we will require that the discrete Hamiltonian be continuous across cell interfaces, \textit{i.e.} $H_h\in\fillover{\mathcal{V}}^p_h$  \citep{hakim2019,liu2000high,shi2017,shi-thesis}.
This leads to the following Lemma:
\begin{lemma}\label{lem:norm-alpha}
  The component of the phase-space characteristic velocity normal to a face of a cell is continuous across the cell boundary, as long as both the Hamiltonian and the Poisson tensor are continuous across the boundary.
\end{lemma} 
\begin{proof}
  Recall from \cref{sec:eom} that for a general Hamiltonian
  system, the Poisson bracket operator is defined as
  \begin{align}
    \{f,g\} = \pderiv{f}{Z^\alpha}\Pi^{\alpha\beta}\pderiv{g}{Z^\beta},
  \end{align}
  where $\Pi^{\alpha\beta}$ is the anti-symmetric {Poisson tensor}. The
  characteristic velocity, $\dot{Z}^\alpha=\{Z^\alpha,H\}$, can then be written
  as $\dot{Z}^\alpha=\Pi^{\alpha\beta}\partial H/\partial Z^\beta$. 
  Let $n_\alpha$ be a unit
  vector normal to a cell surface in dimension $\alpha$. We have
  \begin{align}
    n_\alpha \dot{Z}^\alpha = n_\alpha \Pi^{\alpha\beta}\pderiv{H}{Z^\beta} = \tau^\beta \pderiv{H}{Z^\beta} =
    \vec{\tau}\cdot\pderiv{H}{\vec{Z}},
  \end{align}
  where $\tau^\beta \equiv n_\alpha \Pi^{\alpha\beta}$. Hence, $\vec{\tau}\cdot\vec{n}
  = n_\alpha\Pi^{\alpha\beta} n_\beta = 0$, as $\Pi^{\alpha\beta}$ is anti-symmetric, showing
  that the vector $\vec{\tau}$ is orthogonal to $\vec{n}$, and thus
  tangent to the cell surface. Hence, as the Hamiltonian is
  continuous within the cell (including the cell surface), the tangential component of its gradient (the normal
  component of the characteristic velocity) is also continuous on the cell surface. However, this also requires that the tangent vector $\vec{\tau}$ is continuous across the boundary, which means the Poisson tensor itself must be continuous across cell boundaries.
\end{proof}
\begin{remark}
When the conditions of \cref{lem:norm-alpha} are met so that the phase-space characteristic velocities are indeed continuous across cell interfaces, we can pull the characteristic velocity out of the numerical flux function, since we will have $\dot{Z}_h^{\alpha\,-} = \dot{Z}_h^{\alpha\,+}$ in each dimension $\alpha$. Thus the numerical flux becomes $\hat{\vec{F}}=\dot{\vec{Z}}_h \widehat{ \jac f_h}$.
\end{remark}
 
\subsection{Discrete energy conservation} \label{sec:dg-hamil-energy}
We will now show that the scheme given by \cref{eq:dis-weak-form} conserves the total energy of the system exactly in the continuous-time limit. The energy is given by
\begin{equation}
    \mathcal{E}_h = \int_\Omega H_h \mathcal{J}f_h \dx{\vec{Z}} - \mathcal{L}_f = \sum_i \int_{\mathcal{K}_i} H_h \mathcal{J}f_h \dx{\vec{Z}} - \mathcal{L}_f,
\end{equation}
where $\mathcal{L}_f$ is a possible field term. This quantity evolves as
\begin{equation}
    \pderiv{\mathcal{E}_h}{t} = \sum_i \int_{\mathcal{K}_i} \left(H_h \pderiv{(\mathcal{J}f_h)}{t} +\mathcal{J} f_h \pderiv{H_h}{t} \right) \dx{\vec{Z}} - \pderiv{\mathcal{L}_f}{t}. \label{DGenergyevolve}
\end{equation}
To show that the first term vanishes, we can take $\psi=H_h$ in \cref{eq:dis-weak-form} since $\psi\in \mathcal{V}_h^p$ and $H_h \in \fillover{\mathcal{V}}_h^p \subset \mathcal{V}_h^p$. This gives
\begin{align}
    \int_{\mathcal{K}_i}  H_h \pderiv{(\mathcal{J}f_h)}{t}\dx{\vec{Z}}& = \int_{\mathcal{K}_i} \mathcal{J}f_h \dot{\vec{Z}}_h\cdot\pderiv{H_h}{\vec{Z}}\, \dx{\vec{Z}}  - \oint_{\partial \mathcal{K}_i}H_h^- \widehat{\jac f_h} \dot{\vec{Z}}_h\cdot\dx{\vec{s}}.
\end{align}
Now note that the volume term vanishes exactly because $\dot{\vec{Z}_h}\cdot\partial H_h/\partial\vec{Z} = \{H_h, H_h\} = 0$; physically, this is because the (discrete) flow $\dot{\vec{Z}}_h$ is along contours of constant (discrete) energy in phase space. Summing over all cells, the surface term also vanishes; the requirement that the Hamiltonian is continuous across cell boundaries, along with Lemma 1, means that the integrand only differs by a sign across cell boundaries, resulting in exact cancellations between the surface contributions from either side of each cell face. Thus we have
\begin{equation}
    \sum_i  \int_{\mathcal{K}_i}  H_h \pderiv{(\mathcal{J}f_h )}{t}\dx{\vec{Z}} = 0.
\end{equation}
This gives energy conservation in systems in which the Hamiltonian is time-independent. In systems where the Hamiltonian is an evolving quantity, we must use the field equation governing the Hamiltonian to show that the remaining terms in \cref{DGenergyevolve} cancel. We will see an example of this in the next section, when we apply our energy-conserving DG scheme to the electromagnetic gyrokinetic system.

\subsection{Example: the 2D incompressible Euler system} \label{sec:incomp-euler}
A well-known example of a Hamiltonian system is the two-dimensional incompressible Euler equations, expressed in the vorticity stream-function formulation \citep{christiansen1973,olver1982}. In this formulation, the (incompressible) fluid flow $\vec{u}=\nabla\times(\vec{e}_z \phi)$ is expressed in terms of the stream function $\phi$, and the vorticity is $\varpi = \vec{e}_z\cdot\nabla\times\vec{u}$, with $\vec{e}_z$ the direction perpendicular to the plane of motion. The evolution of the vorticity is then given by the Liouville equation, 
\begin{equation}
    \pderiv{\varpi}{t} + \vec{u}\cdot\nabla \varpi = 0,
\end{equation}
or equivalently in terms of the canonical Poisson bracket $\{f,g\} = \partial_x f \partial_y g - \partial_y f \partial_x g$ as
\begin{equation}
    \pderiv{\varpi}{t} + \{\varpi,\phi\} = 0.  
\end{equation}
Comparing this equation to \cref{eq:liouville} above, we see that in the 2D incompressible Euler system the ``phase space'' is composed of the configuration space dimensions $(x,y)$, the vorticity $\varpi$ is the ``distribution function'', and the stream function $\phi$ plays the role of the Hamiltonian. The stream function is determined by the Poisson equation,
\begin{equation}
    -\nabla_\perp^2\phi = \varpi, \label{euler-poisson}
\end{equation}
with $\nabla_\perp = \vec{e}_x \partial_x + \vec{e}_y \partial_y $.

From \cref{hamil-energy}, the conserved energy in this system is
\begin{equation}
    \mathcal{E} = \int \phi\varpi \,\dx{\vec{Z}} - \mathcal{L}_f = \int \phi \varpi\,\dx{\vec{Z}} - \int \frac{1}{2}|\nabla_\perp \phi|^2\dx{\vec{Z}}= \int \frac{1}{2}|\nabla_\perp \phi|^2 \,\dx{\vec{Z}},
\end{equation}
where the field energy is
\begin{equation}
    \mathcal{L}_f = \int \frac{1}{2}|\nabla_\perp \phi|^2\dx{\vec{Z}}.
\end{equation}
To prove that energy is indeed conserved, we can compute $\pderivInline{\mathcal{E}}{t}$ by first taking
\begin{equation}
    \int \phi \pderiv{\varpi}{t}\dx{\vec{Z}} = -\int \phi\, \vec{u}\cdot \nabla\varpi\,\dx{\vec{Z}} = \int \varpi\, \vec{u}\cdot \nabla\phi\,\dx{\vec{Z}} = 0,
\end{equation}
where $\vec{u}\cdot\nabla \phi = \{\phi,\phi\} = 0$, and we have neglected boundary terms after integrating by parts. When the Hamiltonian $\phi$ is time-dependent, we also have
\begin{equation}
    \int \varpi \pderiv{\phi}{t}\dx{\vec{Z}} = -\int \nabla_\perp^2 \phi \pderiv{\phi}{t}\dx{\vec{Z}} = \int \nabla_\perp \phi \cdot \pderiv{}{t}\nabla_\perp \phi \,\dx{\vec{Z}} = \pderiv{}{t}\int \frac{1}{2}|\nabla_\perp \phi|^2\dx{\vec{Z}},
\end{equation}
which exactly cancels the evolution of the field energy term,
\begin{equation}
    \pderiv{\mathcal{L}_f}{t} = \pderiv{}{t}\int \frac{1}{2}|\nabla_\perp\phi|^2\dx{\vec{Z}},
\end{equation}
so that we are left with energy conservation,  $\pderivInline{\mathcal{E}}{t}=0$.

\subsubsection{Discontinuous Galerkin discretization (Liu-Shu scheme)}

In the scheme of \citet{liu2000high}, the discrete energy is conserved exactly by the spatial scheme for 2D incompressible flow if the basis functions for the stream function $\phi_h$ are in the continuous subset of the basis functions for the vorticity $\varpi_h$, irrespective of the numerical fluxes selected for
the vorticity equation. In our notation, this means $\varpi_h \in \mathcal{V}_h^p$ and $\phi_h \in \fillover{\mathcal{V}}^p_h$.  Identifying the vorticity as the ``distribution function'' and the stream function as the Hamiltonian, we can see that this is a special case of the general scheme prescribed in \cref{sec:dg-hamil}.

The DG weak form of the vorticity evolution equation in cell $i$ follows from \cref{eq:dis-weak-form}, and is given by
\begin{equation}
    \int_{\mathcal{K}_i} \psi \pderiv{\varpi_h}{t}\dx{\vec{Z}}- \int_{\mathcal{K}_i} \varpi_h \dot{\vec{Z}}_h\cdot\pderiv{\psi}{\vec{Z}}\, \dx{\vec{Z}}  + \oint_{\partial \mathcal{K}_i}\psi^- \widehat{\varpi_h \dot{\vec{Z}}_h}\cdot\dx{\vec{s}} = 0, \label{euler-dg}
\end{equation}
with $\dot{\vec{Z}}_h=\{\vec{Z},\phi_h\}$. In order to impose the continuity requirement on $\phi_h$, we can use the (continuous) finite-element method (FEM) to solve the Poisson equation. The discrete local weak form of the Poisson equation is obtained by multiplying \cref{euler-poisson} by a test function $\xi \in \fillover{\mathcal{V}}^p_h$ and integrating (by parts) in each cell $\mathcal{K}_i$,
\begin{equation}
    \int_{\mathcal{K}_i} \nabla_\perp \phi_h \cdot \nabla_\perp \xi^{(i)} \dx{\vec{Z}} - \oint_{\partial \mathcal{K}_i} \xi^{(i)} \nabla_\perp\phi_h \cdot \dx{\vec{s}} = \int_{\mathcal{K}_i} \xi^{(i)} \varpi_h \dx{\vec{Z}}, \label{euler-FEMpoisson}
\end{equation}
where  $\xi^{(i)}$ denotes the restriction of $\xi$ to cell $i$. The global weak form is then obtained by summing \cref{euler-FEMpoisson} over all cells, which results in cancellation of the surface terms at cell interfaces and leaves only a global $\partial \mathcal{T}$ boundary term.

To verify that the discrete energy, $\mathcal{E}_h = \int \phi_h \varpi_h \dx{\vec{Z}}- \int \frac{1}{2}|\nabla_\perp \phi_h|^2 \dx{\vec{Z}}$, is conserved by this discretization scheme, we can first take $\psi = \phi_h$ in \cref{euler-dg} (since $\psi \in \mathcal{V}_h^p$ and $\phi_h \in \fillover{\mathcal{V}}_h^p \subset \mathcal{V}_h^p$) and sum over all cells. This gives
\begin{equation}
    \sum_i \int_{\mathcal{K}_i} \phi_h \pderiv{\varpi_h}{t}\dx{\vec{Z}}=  \sum_i\int_{\mathcal{K}_i} \varpi_h \dot{\vec{Z}}_h\cdot\pderiv{\phi_h}{\vec{Z}}\, \dx{\vec{Z}}  - \sum_i\oint_{\partial \mathcal{K}_i}\phi_h^- \widehat{\varpi_h \dot{\vec{Z}}_h}\cdot\dx{\vec{s}} = 0,
\end{equation}
where as in \cref{sec:dg-hamil-energy}, the volume term vanishes exactly because  $\dot{\vec{Z}_h}\cdot\partial \phi_h/\partial\vec{Z} = \{\phi_h, \phi_h\} = 0$, and the surface terms cancel at cell boundaries because the integrand only differs by a sign on either side due to continuity of $\phi_h$. Thus the evolution of the Hamiltonian part of the discrete energy, $\mathcal{E}_{H\,h} = \int \phi_h \varpi_h \dx{\vec{Z}}$, reduces to
\begin{equation}
    \pderiv{\mathcal{E}_{H\,h}}{t} = \pderiv{}{t}\int \phi_h \varpi_h \dx{\vec{Z}} = \sum_i \int_{\mathcal{K}_i} \varpi_h \pderiv{\phi_h}{t} \dx{\vec{Z}}.
\end{equation}

This remaining term is canceled by the evolution of the field energy term, $\mathcal{L}_{f\,h} = \int \frac{1}{2}|\nabla_\perp\phi_h|^2 \dx{\vec{Z}}$, which is given by
\begin{equation}
    \pderiv{\mathcal{L}_{f\,h}}{t} = \sum_i \int_{\mathcal{K}_i} \nabla_\perp \phi_h \cdot \nabla_\perp \pderiv{\phi_h}{t} \dx{\vec{Z}} = \sum_i\int_{\mathcal{K}_i}\varpi_h  \pderiv{\phi_h}{t}  \dx{\vec{Z}},
\end{equation}
where the second equality is obtained by 
taking $\xi^{(i)} = \partial \phi_h/\partial t$ in \cref{euler-FEMpoisson} and summing over cells. Thus, together we have
\begin{equation}
    \pderiv{\mathcal{E}_h}{t} = \pderiv{\mathcal{E}_{H\,h}}{t} - \pderiv{\mathcal{L}_{f\,h}}{t} = 0,
\end{equation}
and so energy is indeed conserved by the Liu-Shu scheme. Note that this property is independent of the choice of numerical fluxes in the vorticity equation.

\section{Applying the scheme to electromagnetic gyrokinetics}\label{sec:gkscheme}
For the electromagnetic gyrokinetic system, we again start by decomposing the global 5D phase-space domain $\Omega$ into a {structured} phase-space mesh $\mathcal{T}$ with cells $\mathcal{K}_i \in \mathcal{T},\ i=1,...,N$. We then introduce a piecewise-polynomial approximation space for the Jacobian-weighted distribution function $\mathcal{J} f(\vec{R},v_\parallel,\mu)$,
\begin{equation}
    \mathcal{V}_h^p = \{\psi : \psi|_{\mathcal{K}_i}\in \vec{P}^p, \forall \mathcal{K}_i\in \mathcal{T}\},
\end{equation}
where $\vec{P}^p$ is some space of polynomials with maximum degree $p$ (by some measure). 
That is, $\psi(\vec{Z})$ are polynomial functions of $\vec{Z}$ in each cell, and $\vec{P}^p$ is the space of the linear combination of some set of multi-variate polynomials.
In the remainder of this work, we choose $\vec{P}^p$ to be an orthonormalized serendipity polynomial element space \citep{arnold2011serendipity}.
The serendipity basis set
has the advantage of using fewer basis functions while giving the same formal
convergence order (although it is less accurate) as the Lagrange tensor
basis, although note that for $p=1$ the serendipity basis is equivalent to the Lagrange tensor basis. 

We can then obtain the discrete weak form of the gyrokinetic equation by multiplying \cref{emgk} by any test function $\psi\in\mathcal{V}^p_h$
and integrating (by parts) in each cell, giving
\begin{align}
    \int_{\mathcal{K}_i}&\psi \pderiv{(\mathcal{J}f_h)}{t}\dx{^3\vec{R}}\,\dx{^3\vec{v}}  
    \notag \\ &\quad 
    - \int_{\mathcal{K}_i} \mathcal{J}f_h \dot{\vec{R}}_h\cdot\nabla \psi \,\dx{^3\vec{R}}\,\dx{^3\vec{v}} 
    - \int_{\mathcal{K}_i} \mathcal{J}f_h \left(\dot{v}^H_{\parallel h}-\frac{q}{m}\pderiv{A_{\parallel h}}{t}\right) \pderiv{\psi}{v_\parallel}\dx{^3\vec{R}}\,\dx{^3\vec{v}}
    \notag \\ &\quad
    + \oint_{\partial \mathcal{K}_i}\psi^- \widehat{\mathcal{J}f_h}\dot{\vec{R}}_h\cdot \dx{\vec{s}}_R\, \dx{^3\vec{v}}
    + \oint_{\partial \mathcal{K}_i}  \psi^- \widehat{\mathcal{J}f_h} \left(\dot{v}^H_{\parallel h}-\frac{q}{m}\pderiv{A_{\parallel h}}{t}\right)\dx{^3\vec{R}}\,\dx{s}_v 
    \notag \\ &\quad
    = \int_{\mathcal{K}_i} \psi\left(\mathcal{J}C[f_h] + \mathcal{J}S_h\right)\dx{^3\vec{R}}\,\dx{^3\vec{v}}. \label{DGgk}
\end{align}
The discrete phase-space characteristics are defined via the discrete version of the gyrokinetic Poisson bracket, \cref{gkpb}, as
\begin{gather}
    \dot{\vec{R}}_h = \{\vec{R},H_h\}_h = \frac{\vec{B}^*_h}{m B_{\parallel h}^*}\pderiv{H_h}{v_\parallel} + \frac{\uv{b}}{q B_{\parallel h}^*}\times \nabla H_h, \label{eomRh}\\
    \dot{v}_{\parallel h} = \dot{v}^H_{\parallel h} -\frac{q}{m}\pderiv{A_{\parallel h}}{t} = \{v_\parallel,H_h\}_h-\frac{q}{m}\pderiv{A_{\parallel h}}{t} = -\frac{\vec{B}^*_h}{m B_{\parallel h}^*}{\cdot}\nabla H_h-\frac{q}{m}\pderiv{A_{\parallel h}}{t}, \label{eomVh}
\end{gather}
with $\vec{B}^*_h = \vec{B}_{0 h} + (mv_\parallel/q)\nabla\times \uv{b} +\nabla \times(A_{\parallel h}\uv{b})$, and $B^*_{\parallel h} = \uv{b}\cdot\vec{B}^*_h\approx B_{0 h}$. 
Consistent with the energy-conserving DG algorithm formulated in \cref{sec:dgenergy}, we will require the discrete Hamiltonian $H_h$ to be continuous across cell interfaces. We do this by introducing a subset of $\mathcal{V}^p_h$ where the piecewise polynomials are continuous across cell interfaces, denoted by $\fillover{\mathcal{V}}^p_h$, and requiring  $H_h\in\fillover{\mathcal{V}}^p_h$. 
From Lemma 1, this ensures that the discrete phase-space characteristics, $\dot{\vec{R}}_h=\{\vec{R},H_h\}_h$ and $\dot{v}^H_{\parallel h}-({q_s}/{m_s})\partial{A_{\parallel h}}/\partial{t} = \{v_\parallel,H_h\}_h-({q_s}/{m_s})\partial{A_{\parallel h}}/\partial{t}$, are also continuous across cell interfaces in the direction of flow.\footnote{In a general non-orthogonal field-aligned geometry this is not necessarily true. This is because $\vec{B}^*_h\vec{\cdot}\nabla z$ contains $A_{\parallel _h}$, which can be discontinuous in the $z$ direction, making the Poisson tensor itself discontinuous in this direction. This makes the characteristic speed $\dot{\vec{R}}_h\vec{\cdot}\nabla z$ discontinuous across $z$ cell interfaces. We will discuss this issue in \cref{ch:geometry}.}

Solving \cref{DGgk} for all test functions $\psi\in\mathcal{V}_h^p$  in all cells $\mathcal{K}_i\in\mathcal{T}$ yields the discretized Jacobian-weighted distribution function $\mathcal{J}f_h\in\mathcal{V}^p_h$. In the surface terms, $\dx{\vec{s}_R}$ is the differential element on a configuration-space surface (pointing outward normal to the surface), and $\dx{s_v}=2\pi\, \dx{\mu}\, \vec{n}\vec{\cdot} (\partial{\vec{Z}}/\partial{v_\parallel})$ is the differential element on a $v_\parallel$ surface.  We choose to use standard upwind fluxes in our scheme, which depend on the local value of the phase-space characteristic flow normal to the surface evaluated at each Gaussian quadrature point on the surface. Given the phase-space flow $\dot{\vec{Z}}_h$, the upwind flux can be expressed as
\begin{equation}
    \widehat{f_h} = \frac{1}{2}\left(f_h^++f_h^-\right)-\frac{1}{2}\text{sgn}\left(\vec{n}{\cdot}\dot{\vec{Z}}_h\right)\left(f_h^+-f_h^-\right),
\end{equation}
where $\vec{n}=\dx{\vec{s}}/|\dx{\vec{s}}|$ is the unit normal pointing out of the $\partial \mathcal{K}_i$ surface.

We must also discretize the field equations. We introduce the \textit{restriction} of the phase-space mesh to configuration space, $\mathcal{T}^R$, and we  denote the configuration-space cells by $\mathcal{K}_i^R\in\mathcal{T}^R$ for $i=1,...,N_R$, where $N_R$ is the number of configuration-space cells. We also restrict $\mathcal{V}_h^p$ to configuration space as
\begin{equation}
    \mathcal{X}_h^p = \mathcal{V}_h^p \setminus \mathcal{T}^R.
\end{equation}
Further, we introduce the subset of polynomials that are piecewise continuous across configuration-space cell interfaces $\fillover{\mathcal{X}}_h^p \subset \mathcal{X}_h^p$, along with an additional subset $\dashover{\mathcal{X}}_h^p\subset \mathcal{X}_h^p$ where continuity is required in the directions perpendicular to the magnetic field, but not in the direction parallel to the field. Assuming a field-aligned coordinate system \citep[\emph{e.g.}][]{beer1995field}, we will take the perpendicular directions to be $x$ and $y$, and the parallel direction to be $z$.  

Since we require $H_h$ to be continuous across all cell interfaces, this means that we require $\Phi_h$ to be continuous, \textit{i.e.} $\Phi_h\in\fillover{\mathcal{X}}_h^p$. Thus to solve the GK Poisson equation, \cref{poisson1}, we use the (continuous) finite-element method (FEM). While one could ensure $\Phi_h$ is continuous in all directions by using a three-dimensional FEM solve, we instead use a two-dimensional FEM solve in the $x$ and $y$ directions, followed by a one-dimensional smoothing operation in the $z$ direction. That is, we first solve for $\dashover{\Phi}_h\in \dashover{\mathcal{X}}^p_h$ using a two-dimensional FEM solve, and then we use a smoothing/projection operation to ensure continuity in the $z$ direction. We will denote this operation as $\Phi_h = \mathcal{P}_z[\dashover{\Phi}_h]$ and define it below. We can make this splitting because $\nabla_\perp$ only produces coupling in the $x$ and $y$ (perpendicular) directions. 

For the two-dimensional solve, we solve for $\dashover{\Phi}_h\in \dashover{\mathcal{X}}^p_h$ by multiplying \cref{poisson1} by a test function $\xi \in \dashover{\mathcal{X}}^p_h$ and integrating (by parts) in each configuration-space cell $\mathcal{K}^R_i$ to obtain the discrete \emph{local} weak form
\begin{equation}
\int_{\mathcal{K}^R_i} \epsilon_{\perp h} \nabla_\perp \dashover{\Phi}\mskip0.01\thinmuskip_h \vec{\cdot} \nabla_\perp \xi^{(i)}\,\dx{^3\vec{R}} - \oint_{\partial \mathcal{K}^R_i}\xi^{(i)}\ \epsilon_{\perp h} \nabla_\perp\dashover{\Phi}\mskip0.01\thinmuskip_h\cdot \dx{\vec{s}_R} = \int_{\mathcal{K}^R_i} \xi^{(i)}\ \mathcal{P}^*_z[\sigma_{g\,h}]\,\dx{^3\vec{R}}, \label{FEMpoisson}
\end{equation}
where $\xi^{(i)}$ denotes the restriction of $\xi$ to cell $i$, $\epsilon_{\perp h} = \sum_s m_s n_{0s}/B_{0h}^2$, and
\begin{equation}
    \sigma_{g\,h}=\sum_s q_s \int_{\mathcal{T}^v}  \mathcal{J}f_{s\,h}\,\dx{^3\vec{v}},
\end{equation}
with $\mathcal{T}^v$ the restriction of $\mathcal{T}$ to velocity space.
The global weak form is then obtained by summing \cref{FEMpoisson} over cells in $x$ and $y$ (but not in $z$), which results in cancellation of the surface terms at cell interfaces and leaves only a global $\partial \mathcal{T}^R$ boundary term. Note that in order to maintain energetic consistency (as we will see below), the introduction of $\mathcal{P}_z$ necessitates the modification of the right-hand side of \cref{FEMpoisson} with $\mathcal{P}^*_z$, the adjoint of $\mathcal{P}_z$, defined as
\begin{equation}
\int_{\mathcal{T}^R} f \mathcal{P}_z[g]\,\dx{^3\vec{R}} = \int_{\mathcal{T}^R} \mathcal{P}^*_z[f] g\,\dx{^3\vec{R}}.
\end{equation}

For the smoothing operation $\Phi_h=\mathcal{P}_z[\dashover{\Phi}_h]$, we use a one-dimensional FEM solve in the $z$ direction. This can be written as the solution $\Phi_h$ of the global (in $z$) weak equality
\begin{equation}
    \int_{\mathcal{T}^z_j}d\vec{R}\ \chi\ \Phi_h\,\dx{^3\vec{R}} = \int_{\mathcal{T}^z_j}d\vec{R}\ \chi \dashover{\Phi}\mskip0.01\thinmuskip_h \,\dx{^3\vec{R}}, \label{smooth}
\end{equation}
where $\chi\in\widehat{\mathcal{X}}^p_h\subset\mathcal{X}^p_h$, with $\widehat{\mathcal{X}}^p_h$ a subset of the configuration-space basis where continuity is required only in the $z$ direction. Here, $\mathcal{T}^z_j$ denotes a restriction of the domain that is global in $z$ but cell-wise local in $x$ and $y$. 
We remark that using an FEM solve for this operation makes $\mathcal{P}_z$ self-adjoint, so that $\mathcal{P}^*_z=\mathcal{P}_z$. Note, however, that one could instead use a different, local smoothing operation that is not self-adjoint, so we will keep the distinction between $\mathcal{P}_z$ and $\mathcal{P}_z^*$. Also note that $\mathcal{P}_z$ is a projection operator, in that $\mathcal{P}_z[\mathcal{P}_z[\dashover{\Phi}_h]]=\mathcal{P}_z[\dashover{\Phi}_h]$.

The continuous discrete Hamiltonian $H_h\in\fillover{\mathcal{V}}_h^p$ is then given by
\begin{equation}
    H_h = \frac{1}{2}m{v_{\parallel\,h}^2} + \mu B_{0h} + q \mathcal{P}_z[\dashover{\Phi}\mskip0.01\thinmuskip_h],
\end{equation}
where ${v_{\parallel\,h}^2}$ is the projection of $v_\parallel^2$ onto $\fillover{\mathcal{V}}^p_h$. Note that this is only necessary when $v_\parallel^2$ is not in the basis (\emph{i.e.} when $p_v<2$, where $p_v$ is the maximum degree of the $v_\parallel$ monomials in the basis set), resulting in a continuous piecewise-linear approximation to $v_\parallel^2$.

Since $A_{\parallel h}$ does not appear in the Hamiltonian in the symplectic formulation of EMGK, we are free to allow discontinuity in $A_{\parallel h}$. Thus for the parallel Amp\`ere equation we will take $A_{\parallel h}\in\dashover{\mathcal{X}}^p_h$ so that $A_{\parallel h}$ is continuous in $x$ and $y$ but discontinuous in $z$. Multiplying \cref{ampere1} by a test function $\varphi\in\dashover{\mathcal{X}}^p_h$ and integrating, we can obtain the discrete weak form of this equation. The local weak form in cell $i$ is
\begin{equation}
\int_{\mathcal{K}^R_i} \nabla_\perp A_{\parallel h} \vec{\cdot} \nabla_\perp \varphi^{(i)}\,\dx{^3\vec{R}} - \oint_{\partial \mathcal{K}^R_i} \varphi^{(i)} \nabla_\perp A_{\parallel h}\cdot\dx{\vec{s}_R}  = \mu_0 \int_{\mathcal{K}^R_i} \varphi^{(i)}\ J_{\parallel h}\,\dx{^3\vec{R}}, \label{FEMampere-local}
\end{equation}
where again the surface terms will cancel on summing over cells except at the global $\partial \mathcal{T}^R$ boundary, and
\begin{equation}
J_{\parallel h} = \sum_s\frac{q_s}{m_s} \int_{\mathcal{T}^v} \pderiv{H_{s\,h}}{v_\parallel} \mathcal{J}f_{s\,h}\,\dx{^3\vec{v}}. \label{Jparh}
\end{equation}
Here, note that we have replaced the $v_\parallel$ in the $J_\parallel$ definition from \cref{ampere1} with $(1/m)\partial H_h/\partial v_\parallel$; this will be required for energy conservation in the $p_v=1$ case, since $\partial H_h/\partial v_\parallel\neq m v_\parallel$ when $v_\parallel^2$ is not in the basis. Instead, for $p_v=1$, $\partial H_h/\partial v_\parallel=m\bar{v}_\parallel$, the piecewise-constant projection of $mv_\parallel$. Looking back at the variational derivation of Amp\`ere's law in \cref{ampderiv}, we see that indeed using $(1/m)\partial H_h/\partial v_\parallel$ is energetically consistent. 
As before, we solve \cref{FEMampere-local} using a two-dimensional FEM solve in the $x$ and $y$ directions. Note, however, that we do not require the smoothing operation in $z$ here because $A_{\parallel h}$ is allowed to be discontinuous in the $z$ direction.

The discrete weak form of Ohm's law, \cref{ohmstar}, can be obtained by taking the time derivative of the discrete Amp\`ere's law, \cref{FEMampere-local}.  The details of the required manipulations are left to Appendix \ref{app:Ohm}. In the end, the distinction between $p_v=1$ and $p_v>1$ in the definition of $J_{\parallel h}$ leads to two different cases: in the $p_v=1$ case surface terms from the gyrokinetic update appear in the integrals, while volume terms vanish because $\pderivInline{\bar{v}_\parallel}{v_\parallel}=0$; in the $p_v>1$ case we have the opposite, with surface terms cancelling exactly at cell interfaces and volume terms remaining. The local weak form becomes
\begin{align}
    &\int_{\mathcal{K}^R_i} \nabla_\perp \pderiv{A_{\parallel h}}{t} \vec{\cdot} \nabla_\perp \varphi^{(i)}\,\dx{^3\vec{R}} 
    - \oint_{\partial \mathcal{K}^R_i} \varphi^{(i)}\nabla_\perp \pderiv{A_{\parallel h}}{t} \cdot \dx{\vec{s}_R} 
    \notag \\ &\quad
    - \int_{\mathcal{K}_i^R} \varphi^{(i)}\pderiv{A_{\parallel h}}{t} \left[\sum_{s,j} \frac{\mu_0 q_s^2}{m_s}\oint_{\partial\mathcal{K}^v_j}  \bar{v}_\parallel^- \widehat{\mathcal{J} f_{s\,h}}\,\dx{s_v}\right]\dx{^3\vec{R}} 
    \notag\\ &
    = \mu_0\sum_s q_s \int_{\mathcal{K}^R_i} \varphi^{(i)} \Bigg[\int_{\mathcal{T}^v} \bar{v}_\parallel \pderiv{(\mathcal{J}f_{s\,h})}{t}^\star \dx{^3\vec{v}}-\sum_j\oint_{\partial \mathcal{K}_j^v} \bar{v}_\parallel^-{\dot{v}^H_{\parallel h}} \widehat{\mathcal{J}f_{s\,h}}\,\dx{s_v} \Bigg]\dx{^3\vec{R}}, \quad\ (p_v=1) \label{Ohmp1} \\
   &\int_{\mathcal{K}^R_i}\nabla_\perp \pderiv{A_{\parallel h}}{t} \vec{\cdot} \nabla_\perp \varphi^{(i)}\,\dx{^3\vec{R}} - \oint_{\partial \mathcal{K}^R_i} \varphi^{(i)}\nabla_\perp \pderiv{A_{\parallel h}}{t} \cdot \dx{\vec{s}_R} 
   \notag \\ &\quad
   + \int_{\mathcal{K}_i^R} \varphi^{(i)}\pderiv{A_{\parallel h}}{t}\left[ \sum_s \frac{\mu_0 q_s^2}{m_s}\!\int_{\mathcal{T}^v}  \mathcal{J} f_{s\,h}\,\dx{^3\vec{v}}\right]\dx{^3\vec{R}} \notag \\
    &=\mu_0\sum_s q_s \!\!\int_{\mathcal{K}^R_i}  \varphi^{(i)}\left[\int_{\mathcal{T}^v} v_\parallel \pderiv{(\mathcal{J} f_{s\,h})}{t}^\star\dx{^3\vec{v}}\right]\dx{^3\vec{R}}, \qquad\qquad (p_v>1) \label{Ohmp2}
\end{align}
where $\pderivInline{A_{\parallel h}}{t}\in \dashover{\mathcal{X}}^p_h$, and
\begin{align}
    &\int_{\mathcal{K}_i}\psi \pderiv{(\mathcal{J}f_h)}{t}^\star\dx{^3\vec{R}}\,\dx{^3\vec{v}} = 
    \int_{\mathcal{K}_i} \mathcal{J}f_h \dot{\vec{R}}_h\cdot\nabla \psi \,\dx{^3\vec{R}}\,\dx{^3\vec{v}} 
    + \int_{\mathcal{K}_i} \mathcal{J}f_h \dot{v}^H_{\parallel h} \pderiv{\psi}{v_\parallel}\dx{^3\vec{R}}\,\dx{^3\vec{v}}
    \notag \\ &\quad 
    - \oint_{\partial \mathcal{K}_i}\psi^- \widehat{\mathcal{J}f_h}\dot{\vec{R}}_h\cdot \dx{\vec{s}}_R\, \dx{^3\vec{v}}
    + \int_{\mathcal{K}_i} \psi\left(\mathcal{J}C[f_h] + \mathcal{J}S_h\right)\dx{^3\vec{R}}\,\dx{^3\vec{v}}
 \label{partialGK}
\end{align}
so that the gyrokinetic equation can be written as
\begin{align}
    &\int_{\mathcal{K}_i} \psi \pderiv{(\mathcal{J}f_h)}{t}\dx{^3\vec{R}}\,\dx{^3\vec{v}} = 
    \int_{\mathcal{K}_i}\psi \pderiv{(\mathcal{J}f_h)}{t}^\star\dx{^3\vec{R}}\,\dx{^3\vec{v}} 
    \notag\\ &\quad
    - \oint_{\partial \mathcal{K}_i}  \psi^- \widehat{\mathcal{J}f_h} \left(\dot{v}^H_{\parallel h}-\frac{q}{m}\pderiv{A_{\parallel h}}{t}\right)\dx{^3\vec{R}}\,\dx{s}_v 
     - \int_{\mathcal{K}_i} \mathcal{J}f_h \frac{q}{m}\pderiv{A_{\parallel h}}{t} \pderiv{\psi}{v_\parallel}\dx{^3\vec{R}}\,\dx{^3\vec{v}}. \label{DGstar}
\end{align}
Note that some special attention is required to ensure that upwinding of the numerical fluxes is handled consistently in Eqs. (\ref{Ohmp1}) and (\ref{DGstar}) in the $p_v=1$ case. The upwind flow for the $v_\parallel$ surface terms is $\dot{v}^H_{\parallel h}-({q}/{m})\partial{A_{\parallel h}}/\partial{t}$; this is somewhat problematic because we cannot readily solve for $\pderivInline{A_{\parallel h}}{t}$ from \cref{Ohmp1} without first knowing {the upwind direction, which depends on $\pderivInline{A_{\parallel h}}{t}$}. Thus for $p_v=1$ only, we use an approximate $\widetilde{\pderivInline{A_{\parallel h}}{t}}$, calculated using \cref{Ohmp2} (which contains no surface term contributions), to compute the upwind direction for the $v_\parallel$ surface terms in Eqs. (\ref{Ohmp1}) and (\ref{DGstar}). One could extend this algorithm by iterating with a new estimate of the upwind direction based on the previous estimate of $\partial A_{\parallel\, h}/\partial t$, but we leave that for future work. The present algorithm seems to work well for the cases tested so far, {and we expect that $\widetilde{\pderivInline{A_{\parallel h}}{t}}$ results in the correct upwind direction most of the time}.

In our modal DG scheme, integrals in the above weak forms are computed analytically using a quadrature-free scheme that results in exact integrations (of the discrete integrands). This means there are no aliasing errors, and that integration by parts operations that led to these integrals are treated exactly, for the specified discrete representation of $f_h$ and other factors in the integrand.  This is important for ensuring the conservation properties of the scheme, since the conservation laws in the EMGK system are indirect, involving integrals of the gyrokinetic equation \citep{hakim2019}. The fact that integrations are exact also has important implications for the cancellation problem. Since integrals in the discrete Ohm's law are computed exactly, the discretization errors (which are solely embedded in the discrete integrands) cancel exactly, avoiding the cancellation problem. {For more details about the modal scheme, the analytical integrations and the avoidance of the cancellation problem, we have included in \cref{app:disp} a derivation of a semi-discrete Alfv\'en wave dispersion relation that results from our scheme.}

\subsection{Discrete conservation properties} \label{discons}

Now we would like to show that the discrete system (in the continuous-time limit) preserves various conservation laws of the continuous system. As with the continuous system, we will consider the conservation properties in the absence of collisions, sources and sinks, and we will assume that the boundary conditions are either periodic or that the distribution function vanishes at the boundary.

\begin{proposition}
The discrete system conserves total number of particles (the $L_1$ norm).
\end{proposition}
\begin{proof}
Taking $\psi=1$ in the discrete weak form of the gyrokinetic equation, \cref{DGgk}, and summing over all cells, we have
\begin{align}
    \sum_i &\pderiv{}{t}\int_{\mathcal{K}_i}  \mathcal{J}f_h \,\dx{^3\vec{R}}\,\dx{^3\vec{v}} + \sum_i \oint_{\partial \mathcal{K}_i} \widehat{\mathcal{J}f_h}\dot{\vec{R}}_h\cdot\dx{\vec{s}_R}\,\dx{^3\vec{v}}
    \notag \\ &\quad
    +\sum_i \oint_{\partial \mathcal{K}_i}\widehat{\mathcal{J}f_h}\left(\dot{v}^H_{\parallel h}-\frac{q}{m}\pderiv{A_{\parallel h}}{t}\right)\dx{^3\vec{R}}\,\dx{s_v} = 0 \notag \\
    &\qquad\Rightarrow \quad \pderiv{}{t}\int_{\mathcal{T}} \mathcal{J}f_h \,\dx{^3\vec{R}}\,\dx{^3\vec{v}}
    = 0 ,
\end{align}
where the surface terms cancel exactly at cell interfaces because the integrands (both the phase-space characteristics and the numerical fluxes) are continuous across the interfaces. 
\end{proof}
\begin{proposition}
The discrete system conserves a discrete total energy, $\mathcal{E}_h = \mathcal{E}_{H\,h} - \mathcal{E}_{E\,h} + \mathcal{E}_{B\,h}$, where
\begin{gather}
    \mathcal{E}_{H\,h} = \sum_s \int_\mathcal{T}   \mathcal{J}f_{s\,h} H_{s\,h}\,\dx{^3\vec{R}}\,\dx{^3\vec{v}}, \\
    \mathcal{E}_{E\,h} = \int_\mathcal{T}  \frac{\epsilon_{\perp h}}{2}|\nabla_\perp \dashover{\Phi}\mskip0.01\thinmuskip_h|^2\,\dx{^3\vec{R}},
\end{gather}
and
\begin{gather}
    \mathcal{E}_{B\,h} = \int_{\mathcal{T}} \frac{1}{2\mu_0}|\nabla_\perp A_{\parallel h}|^2\,\dx{^3\vec{R}}.
\end{gather}
\end{proposition}
\begin{proof}
We start by calculating 
\begin{equation}
    \pderiv{\mathcal{E}_{H\,h}}{t} = \sum_{s,i}\int_{\mathcal{K}_i}\left( H_{s\,h}\pderiv{(\mathcal{J}f_{s\,h})}{t} + \mathcal{J}f_{s\,h}\pderiv{H_{s\,h}}{t}\right)\,\dx{^3\vec{R}}\,\,\dx{^3\vec{v}}.\label{dWHdt1}
\end{equation}
The first term can be calculated by taking $\psi=H_h$ in \cref{DGgk} and summing over cells and species, since $\psi\in\mathcal{V}^p_h$ and $H_h\in\fillover{\mathcal{V}}_h^p\subset \mathcal{V}_h^p$:
\begin{align}
    \sum_{s,i}&\int_{\mathcal{K}_i} H_{s\,h}\pderiv{(\mathcal{J}f_{s\,h})}{t}\,\dx{^3\vec{R}}\,\dx{^3\vec{v}}
     - \sum_{s,i}\int_{\mathcal{K}_i}\mathcal{J}f_{s\,h} \left(\dot{\vec{R}}_h\vec{\cdot} \nabla H_{s\,h}+\dot{v}^H_{\parallel h}\pderiv{H_{s\,h}}{v_\parallel}\right)\dx{^3\vec{R}}\,\dx{^3\vec{v}} \notag\\ 
     &+ \sum_{s,i}\int_{\mathcal{K}_i} \mathcal{J}f_{s\,h} \frac{q_s}{m_s}\pderiv{A_{\parallel h}}{t} \pderiv{H_{s\,h}}{v_\parallel}\,\dx{^3\vec{R}}\,\dx{^3\vec{v}} 
     + \sum_{s,i}\oint_{\partial \mathcal{K}_i}H_{s\,h}^- \widehat{\mathcal{J}f_{s\,h}}\dot{\vec{R}}_h\cdot \dx{\vec{s}}_R\, \dx{^3\vec{v}}
     \notag \\ &\quad 
     +  \sum_{s,i}\oint_{\partial \mathcal{K}_i}  H_{s\,h}^- \widehat{\mathcal{J}f_{s\,h}} \left(\dot{v}^H_{\parallel h}-\frac{q}{m}\pderiv{A_{\parallel h}}{t}\right)\dx{^3\vec{R}}\,\dx{s}_v 
     = 0. 
\end{align}
Here, we see why we must require $H_h$ to be continuous; we want the surface terms to vanish, which means the integrands must be continuous across cell interfaces so that the contributions from either side of the interface cancel exactly when we sum over cells. The numerical flux $\widehat{\mathcal{J}f_h}$ is by definition continuous across the interface, and we have already noted above that the phase-space characteristics $\dot{\vec{R}}_h$ and $\dot{v}^H_{\parallel h}-({q}/{m})\partial{A_{\parallel h}}/\partial{t}$ are also continuous across cell interfaces. This leaves the Hamiltonian, which we require to be continuous so that the surface terms do indeed vanish. Further, the first volume term vanishes exactly because $\dot{\vec{R}}_h\vec{\cdot} \nabla H_{h}+\dot{v}^H_{\parallel h}\pderivInline{H_{h}}{v_\parallel}=\{H_{h},H_{h}\}_h=0$ by definition of the Poisson bracket. However, since the symplectic formulation of EMGK is derived via a time-dependent coordinate transformation (which we did not consider in \cref{sec:hamil}), we still have a leftover term involving $\partial A_{\parallel h} /\partial t$, so that we have
\begin{align}
    \sum_{s,i}\int_{\mathcal{K}_i}  H_{s\,h}\pderiv{(\mathcal{J}f_{s\,h})}{t}\,\dx{^3\vec{R}}\,\dx{^3\vec{v}}
    &= 
    -\sum_{s,i}\int_{\mathcal{K}_i} \mathcal{J}f_{s\,h} \frac{q_s}{m_s}\pderiv{A_{\parallel h}}{t} \pderiv{H_{s\,h}}{v_\parallel} \,\dx{^3\vec{R}}\,\dx{^3\vec{v}} \notag \\
    &= 
    - \int_{\mathcal{T}^R}  \pderiv{A_{\parallel h}}{t} J_{\parallel h}\,\,\dx{^3\vec{R}}. \label{dWHdt2}
\end{align}
Here, we see why we have defined $J_{\parallel h}$ using the derivative of $H_h$ instead of $v_\parallel$, as noted after \cref{Jparh}. 
For the second term in \cref{dWHdt1}, we now take into account time dependence in the Hamiltonian, which gives
\begin{align}
    \sum_{s,i}\int_{\mathcal{K}_i}   \mathcal{J}f_{s\,h}\pderiv{H_{s\,h}}{t}\,\dx{^3\vec{R}}\,\dx{^3\vec{v}} &= \sum_{s,i}\int_{\mathcal{K}_i} \mathcal{J}f_{s\,h} q_s\mathcal{P}_z[\pderiv{\dashover{\Phi}\mskip0.01\thinmuskip_h}{t}]\,\dx{^3\vec{R}}\,\dx{^3\vec{v}} \notag \\
    &= \int_{\mathcal{T}^R} \sigma_{g\,h}\mathcal{P}_z[\pderiv{\dashover{\Phi}\mskip0.01\thinmuskip_h}{t}]\,\dx{^3\vec{R}}. \label{dWHdt2_h}
\end{align}
Thus we have
\begin{equation}
    \pderiv{\mathcal{E}_{H\,h}}{t} =- \int_{\mathcal{T}^R} \pderiv{A_{\parallel h}}{t} J_{\parallel h}\,\dx{^3\vec{R}}+ \int_{\mathcal{T}^R} \sigma_{g\,h}\mathcal{P}_z[\pderiv{\dashover{\Phi}\mskip0.01\thinmuskip_h}{t}]\,\dx{^3\vec{R}}. \label{dWHdt_h}
\end{equation}

Next, we calculate 
\begin{align}
    \pderiv{\mathcal{E}_{E\,h}}{t} &= \sum_{i} \int_{\mathcal{K}^R_i} \epsilon_{\perp h}\nabla_\perp\dashover{\Phi}\mskip0.01\thinmuskip_h \vec{\cdot} \nabla_\perp \pderiv{\dashover{\Phi}\mskip0.01\thinmuskip_h}{t}\,\dx{^3\vec{R}} = \int_{\mathcal{T}^R} \mathcal{P}^*_z[\sigma_{g\,h}]\pderiv{\dashover{\Phi}\mskip0.01\thinmuskip_h}{t}\,\dx{^3\vec{R}} \notag \\
    &=\int_{\mathcal{T}^R} \sigma_{g\,h}\mathcal{P}_z[\pderiv{\dashover{\Phi}\mskip0.01\thinmuskip_h}{t}]\,\dx{^3\vec{R}}, \label{dWEdt_h}
\end{align}
where we have used $\xi^{(i)}=\pderivInline{\dashover{\Phi}\mskip0.01\thinmuskip_h}{t}$ in \cref{FEMpoisson} to make the second equality, noting that the surface term vanishes upon summing over cells because $\dashover{\Phi}_h\in\dashover{\mathcal{X}}^p_h$ is continuous in the perpendicular directions. Here, we see why we modified the right-hand side of \cref{FEMpoisson} with $\mathcal{P}_z^*$, so that the resulting term in \cref{dWEdt_h} matches the one in \cref{dWHdt2_h}.

Finally, we calculate
\begin{align}
    \pderiv{\mathcal{E}_{B\,h}}{t} &= \sum_{i} \int_{\mathcal{K}^R_i} \frac{1}{\mu_0}\nabla_\perp A_{\parallel h} \vec{\cdot} \nabla_\perp \pderiv{A_{\parallel h}}{t} \,\dx{^3\vec{R}} = \int_{\mathcal{T}^R} \pderiv{A_{\parallel h}}{t} J_{\parallel h}\,\dx{^3\vec{R}},
\end{align}
where we have used $\varphi^{(i)}=({1}/{\mu_0})\partial{A_{\parallel h}}/\partial{t}$ in \cref{FEMampere-local} to make the second equality, again noting that the 
surface term vanishes upon summing over cells because $\pderivInline{A_{\parallel h}}{t}\in\dashover{\mathcal{X}}^p_h$ is continuous in the perpendicular directions.

We now have conservation of discrete total energy:
\begin{equation}
     \pderiv{\mathcal{E}_{h}}{t} =  \pderiv{\mathcal{E}_{H\,h}}{t}- \pderiv{\mathcal{E}_{E\,h}}{t} +  \pderiv{\mathcal{E}_{B\,h}}{t} = 0.
\end{equation}
We note that this proof did not rely on the particular choice of numerical flux function.
\end{proof}

\subsection{Time-discretization scheme}\label{sec:timedisc}
So far we have considered only the discretization of the phase space for the system, and we have considered the conservation properties of the scheme in the continuous-time limit. Indeed, in the discrete-time system the conservation properties are no longer exact due to truncation error in the non-reversible time-stepping methods that we consider. However the errors will be \emph{independent} of the phase-space discretization, and errors can be reduced by taking a smaller time step or by using a high-order time-stepping scheme to improve convergence. Following the approach of the Runge-Kutta discontinuous Galerkin method \citep{Cockburn1998,Cockburn2001,shu2009discontinuous}, we have implemented several explicit multi-stage strong stability-preserving Runge-Kutta high-order schemes \citep{gottlieb2001strong,shu2002survey}; most of the results in this thesis use a three-stage, third-order scheme (SSP-RK3), which is sufficiently accurate for our calculations; it is  also unconditionally stable if the CFL condition is satisfied, unlike SSP-RK2. These schemes have the property that a high-order scheme can be composed of several first-order forward-Euler stages. For example, for SSP-RK3, the time advance is given by
\begin{align}
    f^{(1)} &= \mathcal{F}[f^n, t^n] \\
    f^{(2)} &= \frac{3}{4}f^n + \frac{1}{4}\mathcal{F}[f^{(1)}, t^n+\Delta t] \\
    f^{n+1} &= \frac{1}{3}f^n + \frac{2}{3}\mathcal{F}[f^{(2)}, t^n+\Delta t/2],
\end{align}
where
\begin{equation}
    \mathcal{F}[f, t] = f + \Delta t\, \mathbb{L}[f]
\end{equation}
denotes a first-order forward-Euler step,
with $\mathbb{L}[f]$ denoting the right-hand side operator resulting from the DG spatial discretization scheme. 
Thus we will detail our time-stepping scheme for a single forward-Euler stage, which can then be combined into a multi-stage high-order scheme. 

Given $f_h^n=f_h(t=t^n)$ and $A_{\parallel h}^n=A_{\parallel h}(t=t^n)$ at time $t^n$, the steps of the forward-Euler scheme to advance to time $t^{n+1}=t^n+\Delta t$ are as follows:
\begin{enumerate}
    \item Calculate $\dashover{\Phi}_h^n$ using \cref{FEMpoisson}, and then $\Phi_h^n = \mathcal{P}_z[\dashover{\Phi}_h^n]$ using \cref{smooth}.
    \begin{gather}
\int_{\mathcal{K}^R_i} \epsilon_{\perp h} \nabla_\perp \dashover{\Phi}\mskip0.01\thinmuskip_h^n {\cdot} \nabla_\perp \xi^{(j)}\,\dx{^3\vec{R}} - \oint_{\partial \mathcal{K}^R_i} \xi^{(j)} \epsilon_{\perp h}\nabla_\perp\dashover{\Phi}\mskip0.01\thinmuskip_h^n\cdot\dx{\vec{s}}_R = \int_{\mathcal{K}^R_j} \xi^{(j)} \mathcal{P}^*_z[\sigma_{g\,h}^n]\,\dx{^3\vec{R}} \label{FEMpoisson-local-dt} \\
    \int_{\mathcal{T}^z_j}\chi \Phi_h^n\,\dx{^3\vec{R}}  = \int_{\mathcal{T}^z_j} \chi \dashover{\Phi}\mskip0.01\thinmuskip_h^n \,\dx{^3\vec{R}}
\end{gather}
    \item Calculate the partial EMGK update $\left({\partial(\mathcal{J}f_h)}^\star/{\partial t}\right)^n$ using \cref{partialGK}.
\begin{align}
    \int_{\mathcal{K}_i}\psi &\left(\pderiv{(\mathcal{J}f_h)}{t}^\star\right)^n\dx{^3\vec{R}}\,\dx{^3\vec{v}} = 
    \int_{\mathcal{K}_i} \mathcal{J}f_h^n \dot{\vec{R}}_h^n\cdot\nabla \psi \,\dx{^3\vec{R}}\,\dx{^3\vec{v}} 
    + \int_{\mathcal{K}_i} \mathcal{J}f_h^n \dot{v}^{H\,n}_{\parallel h} \pderiv{\psi}{v_\parallel}\dx{^3\vec{R}}\,\dx{^3\vec{v}}
    \notag \\ &\qquad \qquad
    - \oint_{\partial \mathcal{K}_i}\psi^- \widehat{\mathcal{J}f_h}^n\dot{\vec{R}}_h^n\cdot \dx{\vec{s}}_R\, \dx{^3\vec{v}}
    + \int_{\mathcal{K}_i} \psi\left(\mathcal{J}C[f_h^n] + \mathcal{J}S_h^n\right)\dx{^3\vec{R}}\,\dx{^3\vec{v}}  \label{partialGK-dt}
\end{align}
    \item Calculate $\left({\partial A_{\parallel h}}/{\partial t}\right)^n$ from \cref{Ohmp2} [\textit{for $p_v=1$, this is only a provisional value, which we will denote as $(\widetilde{\partial{A_{\parallel h}}/\partial{t}})^n$}].
\begin{align}
   \int_{\mathcal{K}^R_i}\nabla_\perp \left(\pderiv{A_{\parallel h}}{t}\right)^n \vec{\cdot}& \nabla_\perp \varphi^{(i)}\,\dx{^3\vec{R}} - \oint_{\partial \mathcal{K}^R_i} \varphi^{(i)}\nabla_\perp \left(\pderiv{A_{\parallel h}}{t}\right)^n \cdot \dx{\vec{s}_R} 
   \notag \\ &\quad
   + \int_{\mathcal{K}_i^R} \varphi^{(i)}\left(\pderiv{A_{\parallel h}}{t}\right)^n \left[ \sum_s \frac{\mu_0 q_s^2}{m_s}\!\int_{\mathcal{T}^v}  \mathcal{J} f_{s\,h}^n\,\dx{^3\vec{v}}\right]\dx{^3\vec{R}} \notag \\
    &=\mu_0\sum_s q_s \!\!\int_{\mathcal{K}^R_i}  \varphi^{(i)}\left[\int_{\mathcal{T}^v} v_\parallel \left(\pderiv{(\mathcal{J} f_{s\,h})}{t}^\star\right)^n\dx{^3\vec{v}}\right]\dx{^3\vec{R}} 
\end{align}
    \item ($p_v=1\ only$) Use the provisional $(\widetilde{\partial{A_{\parallel h}}/\partial{t}})^n$ from step 3 to calculate the upwinding direction in the surface terms in \cref{Ohmp1}, and then calculate $({\partial{A_{\parallel h}}/\partial{t}})^n$.
\begin{align}
    &\int_{\mathcal{K}^R_i} \nabla_\perp \left(\pderiv{A_{\parallel h}}{t}\right)^n \vec{\cdot} \nabla_\perp \varphi^{(i)}\,\dx{^3\vec{R}} 
    - \oint_{\partial \mathcal{K}^R_i} \varphi^{(i)}\nabla_\perp \left(\pderiv{A_{\parallel h}}{t}\right)^n \cdot \dx{\vec{s}_R} 
    \notag \\ &\quad
    - \int_{\mathcal{K}_i^R} \varphi^{(i)}\left(\pderiv{A_{\parallel h}}{t}\right)^n \left[\sum_{s,j} \frac{\mu_0 q_s^2}{m_s}\oint_{\partial\mathcal{K}^v_j}  \bar{v}_\parallel^- \widehat{\mathcal{J} f_{s\,h}}^n\,\dx{s_v}\right]\dx{^3\vec{R}} 
    \notag\\ &
    = \mu_0\sum_s q_s \int_{\mathcal{K}^R_i} \varphi^{(i)} \Bigg[\int_{\mathcal{T}^v} \bar{v}_\parallel \left(\pderiv{(\mathcal{J}f_{s\,h})}{t}^\star\right)^n \dx{^3\vec{v}}-\sum_j\oint_{\partial \mathcal{K}_j^v} \bar{v}_\parallel^-{\dot{v}^{H\,n}_{\parallel h}} \widehat{\mathcal{J}f_{s\,h}}^n\,\dx{s_v} \Bigg]\dx{^3\vec{R}}
\end{align}
    \item Calculate the full EMGK update, $\left(\pderivInline{(\mathcal{J}f_h)}{t}\right)^{n}$, using \cref{DGstar}. For $p_v=1$, the provisional $\left(\widetilde{\pderivInline{A_{\parallel h}}{t}}\right)^n$ from step 3 should again be used to calculate the upwinding direction in the surface terms for consistency.
\begin{align}
    &\int_{\mathcal{K}_i} \psi \left(\pderiv{(\mathcal{J}f_h)}{t}\right)^n\dx{^3\vec{R}}\,\dx{^3\vec{v}} = 
    \int_{\mathcal{K}_i}\psi \left(\pderiv{(\mathcal{J}f_h)}{t}^\star\right)^n\dx{^3\vec{R}}\,\dx{^3\vec{v}} 
    \notag\\ &\quad
    - \oint_{\partial \mathcal{K}_i}  \psi^- \widehat{\mathcal{J}f_h}^n \left[\dot{v}^{H\,n}_{\parallel h}-\frac{q}{m}\left(\pderiv{A_{\parallel h}}{t}\right)^n\right]\dx{^3\vec{R}}\,\dx{s}_v 
    \notag \\ &\quad
     - \int_{\mathcal{K}_i} \mathcal{J}f_h^n \frac{q}{m}\left(\pderiv{A_{\parallel h}}{t}\right)^n \pderiv{\psi}{v_\parallel}\dx{^3\vec{R}}\,\dx{^3\vec{v}}. \label{DGstar-dt}
\end{align}
    \item Advance $f_h$ and $A_{\parallel h}$ to time $t_{n+1}$.
\begin{gather}
   \mathcal{J} f_h^{n+1} = \mathcal{J} f_h^n + \Delta t \left(\pderiv{(\mathcal{J} f_h)}{t}\right)^n \label{f-dt} \\
   A_{\parallel h}^{n+1} = A_{\parallel h}^n + \Delta t \left(\pderiv{A_{\parallel h}}{t}\right)^n \label{Apar-dt}
\end{gather}
\end{enumerate}
Note that the parallel Amp\`ere equation, \cref{FEMampere-local}, is only used to solve for the initial condition of $A_{\parallel h}(t=0)$. For all other times, \cref{Apar-dt} is used to advance $A_{\parallel h}$. This prevents the system from being over-determined and ensures consistency between $A_{\parallel h}$ and $\pderivInline{A_{\parallel h}}{t}$.

\section{Linear benchmarks}
\label{sec:linear}
The scheme presented above has been implemented into the \gke plasma simulation framework. In this section we present some linear benchmarks that verify the implementation.

\subsection{Kinetic Alfv\'en wave}\label{res:kaw}
As a first benchmark of our electromagnetic scheme, we consider the kinetic Alfv\'en wave.
In a slab (straight background magnetic field) geometry, with stationary singly-charged ions (assuming $\omega\gg k_\parallel v_{ti}$), the gyrokinetic equation for electrons reduces to
\begin{equation}
    \pderiv{f_e}{t} = \{H_e,f_e\}-\frac{e}{m}\pderiv{f_e}{v_\parallel}\pderiv{A_\parallel}{t} = - v_\parallel \pderiv{f_e}{z} {-} \frac{e}{m}\pderiv{f_e}{v_\parallel}\left(\pderiv{\Phi}{z}+\pderiv{A_\parallel}{t} \right).
\end{equation}
Taking a single Fourier mode with perpendicular wavenumber $k_\perp$ and parallel wavenumber $k_\parallel$, the field equations become
\begin{gather}
    k_\perp^2 \frac{m_i n_{0}}{B^2}\Phi = e n_0-e\int f_e\,\dx{v_\parallel} \label{alf-poisson} \\
    k_\perp^2 A_\parallel = -\mu_0 e \int v_\parallel f_e \,\dx{v_\parallel} \label{alf-ampere} \\
    \left(k_\perp^2 + \frac{\mu_0 e^2}{m_e}\int f_e\,\dx{v_\parallel}\right)\pderiv{A_{\parallel}}{t} = -\mu_0 e\int v_\parallel \{H_e, f_e\}\,\dx{v_\parallel}
\end{gather}
After linearizing the gyrokinetic equation by assuming a uniform Maxwellian background with density $n_0$ and temperature $T_e$, so that $f_e = F_{Me} + \delta f_e$, the dispersion relation becomes
\begin{equation}
    \omega^2\left[1+\frac{\omega}{\sqrt{2}k_\parallel v_{te}}Z\left(\frac{\omega}{\sqrt{2}k_\parallel v_{te}}\right)\right] = \frac{k_\parallel^2 v_{te}^2}{\hat{\beta}}\left[1+k_\perp^2\rho_\mathrm{s}^2+\frac{\omega}{\sqrt{2}k_\parallel v_{te}}Z\left(\frac{\omega}{\sqrt{2}k_\parallel v_{te}}\right)\right], \label{kawdisp}
\end{equation}
where $\hat{\beta}=(\beta_e/2) m_i/m_e$, with $\beta_e = 2\mu_0 n_0 T_e/B^2$, $v_{te}=\sqrt{T_e/m_e}$ is the electron thermal speed, $\rho_\mathrm{s} = c_\mathrm{s}/\Omega_i$ is the ion sound gyroradius with $c_\mathrm{s}=\sqrt{T_e/m_i}$ the sound speed and $\Omega_i = e B/m_i$ the gyrofrequency, and $Z(x)$ is the plasma dispersion function \citep{fried1961plasma}. Note that $\rho_\mathrm{s}$ can also be defined in terms of the electron skin depth,  $d_e = (n_0 e^2 \mu_0/m_e)^{-1/2}$, so that $\rho_\mathrm{s}=d_e \hat{\beta}^{1/2}$. In the limit $k_\perp\rho_\mathrm{s}\ll1$ the wave becomes the standard shear Alfv\'en wave from magnetohydrodynamics (MHD), which is an undamped wave with frequency $\omega=k_\parallel v_A$, where $v_A=v_{te}/\hat{\beta}^{1/2}$ is the Alfv\'en velocity. For larger values of $k_\perp \rho_\mathrm{s}$, the mode is damped by kinetic effects.

\begin{figure}
    \hspace{-.5cm}
    \centering
    \includegraphics[width=.48\textwidth]{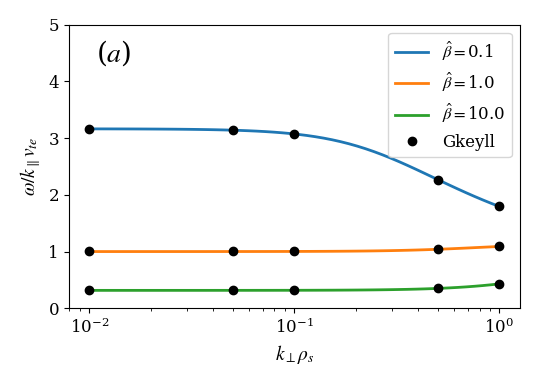}
    \includegraphics[width=.48\textwidth]{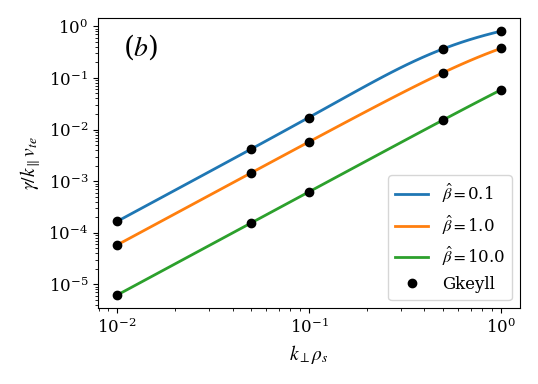}
    \caption[Real frequencies and damping rates for the kinetic Alfv\'en wave from \gke.]{Real frequencies ($a$) and damping rates ($b$) for the kinetic Alfv\'en wave vs $k_\perp \rho_\mathrm{s}$. Solid lines are the exact values from \cref{kawdisp} for three different values of $\hat{\beta}=(\beta_e/2) m_i/m_e$, and black dots are the numerical results from \gke.}
    \label{fig:kaw}
\end{figure}

In \cref{fig:kaw}, we show the real frequencies ($a$) and damping rates ($b$) obtained by solving \cref{kawdisp} for a few values of $\hat{\beta}$. We also show numerical results from \gke, which match the analytic results very well. These results are a good indication that our scheme avoids the Amp\`ere cancellation problem, which can cause large errors for modes with length-scales large compared to the electron skin depth, $k_\perp^2 d_e^2 \ll 1$, or equivalently, $\hat{\beta}/k_\perp^2\rho_\mathrm{s}^2\gg1$ (see \cref{sec:cancel}); we see no such errors, even for the case with $\hat{\beta}/k_\perp^2\rho_\mathrm{s}^2=10^5$. Each \gke simulation was run using piecewise-linear basis functions ($p=1$) in a reduced dimensionality mode with one configuration space dimension and one velocity space dimension, with $(N_z,N_{v_\parallel})=(32,64)$ the number of cells in each dimension. The perpendicular dimensions ($x$ and $y$), which appear only in the field equations in this simple system, were handled by replacing $\nabla_\perp^2 \rightarrow - k_\perp^2$, as in Eqs. (\ref{alf-poisson}) and (\ref{alf-ampere}). We use periodic boundary conditions in $z$ and zero-flux boundary conditions in $v_\parallel$.

We also show in \cref{fig:kaw-fields} the fields $\Phi_h$ and $\pderivInline{A_{\parallel h}}{t}$ for the case with $\hat{\beta}=10$ and $k_\perp \rho_\mathrm{s}=0.01$, which gives $\hat{\beta}/k_\perp^2\rho_\mathrm{s}^2=10^5$. For these parameters the system is near the MHD limit, which means we should expect $E_\parallel=-\pderivInline{\Phi}{z}-\pderivInline{A_\parallel}{t}\approx0$. While this condition is never enforced, getting the physics correct requires the scheme to allow $\pderivInline{\Phi_h}{z}\approx-\pderivInline{A_{\parallel h}}{t}$. The fact that our scheme allows discontinuities in $A_\parallel$ in the parallel direction is an advantage in this case. Because $\Phi_h$ is piecewise-linear here, $\pderivInline{\Phi_h}{z}$ is piecewise-constant; this is necessarily discontinuous for non-trivial solutions. Thus the scheme produces a piecewise-constant $\pderivInline{A_{\parallel h}}{t}$ in this MHD-limit case, as shown in \cref{fig:kaw-fields}, resulting in $E_{\parallel h}\approx 0$.
If our scheme did not allow discontinuities in $A_{\parallel h}$, a continuous $\pderivInline{A_{\parallel h}}{t}$ would never be able to exactly cancel a discontinuous $\pderivInline{\Phi_h}{z}$, and the resulting $E_{\parallel h}\neq0$ would make the solution inaccurate. Notably, this would be the case had we chosen the Hamiltonian ($p_\parallel$) formulation of the gyrokinetic system, which uses $p_\parallel=mv_\parallel + q A_\parallel$ as the parallel velocity coordinate. This is because $A_\parallel$ is included in the Hamiltonian in the $p_\parallel$ formulation, which would require continuity of $A_{\parallel h}$ (and thereby $\pderivInline{A_{\parallel h}}{t}$) to conserve energy in our discretization scheme. 

\begin{figure}
    \centering
    \includegraphics[width=.6\textwidth]{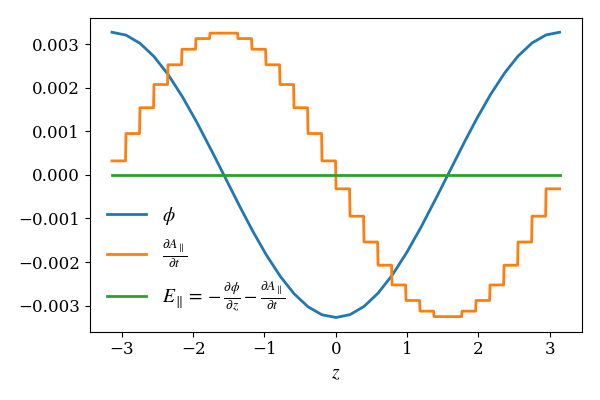}
    \caption[Electrostatic potential and magnetic vector potential for Alfv\'en wave with $\hat{\beta}=10$ and $k_\perp \rho_\mathrm{s}=0.01$.] {$\Phi_h$ (blue) and $\partial{A_{\parallel h}}/\partial {t}$ (yellow) for the case with $\hat{\beta}=10$ and $k_\perp \rho_\mathrm{s}=0.01$. The amplitude of $E_{\parallel h}$ (green) is $\sim 10^{-9}$.}
    \label{fig:kaw-fields}
\end{figure}

\subsection{Kinetic ballooning mode (KBM)} \label{sec:kbm}
We use the kinetic ballooning mode (KBM) instability in the local limit as a second linear benchmark of our electromagnetic scheme. \citet{kim1993electromagnetic} obtain the dispersion relation by solving 
\begin{gather}
    \omega\left[1 + \tau  - P_0\right]\Phi = \left[\tau(\omega-\omega_{*e}) - k_\parallel P_1\right]\frac{\omega}{k_\parallel}A_\parallel \label{kbmdisp1kim}\\
    \frac{2 k_\parallel^2 k_\perp^2}{\beta_i}A_\parallel = k_\parallel\left[k_\parallel P_1 - \tau(\omega-\omega_{*e})\right]\Phi - \left[k_\parallel^2P_2 - \tau\left(\omega(\omega-\omega_{*e})-2\omega_{de}(\omega-\omega_{*e}(1+\eta_e))\right)\right]A_\parallel \label{kbmdisp2}
\end{gather}
where
\begin{gather}
    P_m = \int_0^{\infty}dv_\perp \ v_\perp\int_{-\infty}^{\infty} dv_\parallel\  \frac{1}{\sqrt{2\pi}}e^{-(v_\parallel^2+v_\perp^2)/2}(v_\parallel)^m\frac{\omega-\omega_{*i}\left[1+\eta_i(v^2/2-3/2)\right]}{\omega-k_\parallel v_\parallel - \omega_{di}(v_\parallel^2+v_\perp^2/2)}J_0^2(v_\perp\sqrt{b}),
\end{gather}
with $\tau=T_i/T_e$, $\omega_{*e}=k_y$, $\omega_{*i}=-k_y$, $\eta_s=L_{n}/L_{Ts}$, $b=k_\perp^2$, and $\Gamma_0(b)=I_0(b)e^{-b}$ with $I_0(b)=J_0(i b)$ the modified Bessel function. Here, the wavenumbers $k_y$ and $k_\parallel$ are normalized to $\rho_i$ and $L_n$, respectively, and the frequencies $\omega$ and $\omega_*$ are normalized to $v_{ti}/L_n$. In the local limit, $\omega_{ds}=\omega_{*s}L_{n}/R$ and $k_\perp=k_y$ do not vary along the field line. The above equations include FLR effects even beyond the order of the general system that we derived in \cref{ch:emgk}, with the ion polarization density given by 
\begin{equation}
    n_\mathrm{pol}^\mathrm{Kim} = n_0\left[\Gamma_0(b) - 1\right]\frac{e \Phi}{T_i}.
\end{equation}
The term in square brackets on the left-hand side of \cref{kbmdisp1kim} implicitly contains  a term proportional to $-n_\mathrm{pol}^\mathrm{Kim}$. In our system we only keep the first-order part of the polarization density (because higher-order corrections require even higher-order terms in the Hamiltonian), leaving
\begin{equation}
    n_\mathrm{pol} = -n_0 b \frac{e \Phi}{T_i}.
\end{equation}
Thus we must modify the left-hand side of \cref{kbmdisp1kim} by taking $-n_\mathrm{pol}^\mathrm{Kim} \rightarrow -n_\mathrm{pol}$; we can do this by adding a term proportional to $-(n_\mathrm{pol} - n_\mathrm{pol}^\mathrm{Kim})$ in the square brackets, resulting in
\begin{gather}
    \omega\left[\Gamma_0(b) + b + \tau  - P_0\right]\Phi = \left[\tau(\omega-\omega_{*e}) - k_\parallel P_1\right]\frac{\omega}{k_\parallel}A_\parallel. \label{kbmdisp1mod}
\end{gather}
Finally, we additionally modify the FLR terms by taking $b\rightarrow0$ while keeping the first-order polarization density (which we now write in terms of $k_\perp^2$ instead of $b$) and $k_y=k_\perp\neq0$ in the non-FLR terms, which gives
\begin{gather}
    \omega\left[1 + k_\perp^2 + \tau  - P_0\right]\Phi = \left[\tau(\omega-\omega_{*e}) - k_\parallel P_1\right]\frac{\omega}{k_\parallel}A_\parallel, \label{kbmdisp1}
\end{gather}
where now we will also assume $b=0$ in all $P_m$ expressions.

The local limit can be achieved by simulating a helical flux tube with no magnetic shear, which gives a system with constant magnetic curvature that corresponds to $\omega_d=\text{const}$. This geometry has been previously used for SOL turbulence studies with \gke \citep{shi2019,bernard2019}, except in this section we take the boundary condition along the field lines to be periodic. We will provide further details about the helical geometry and the coordinates in \cref{ch:nstx-results}.

\begin{figure}
    \centering
    \includegraphics[width=.7\textwidth]{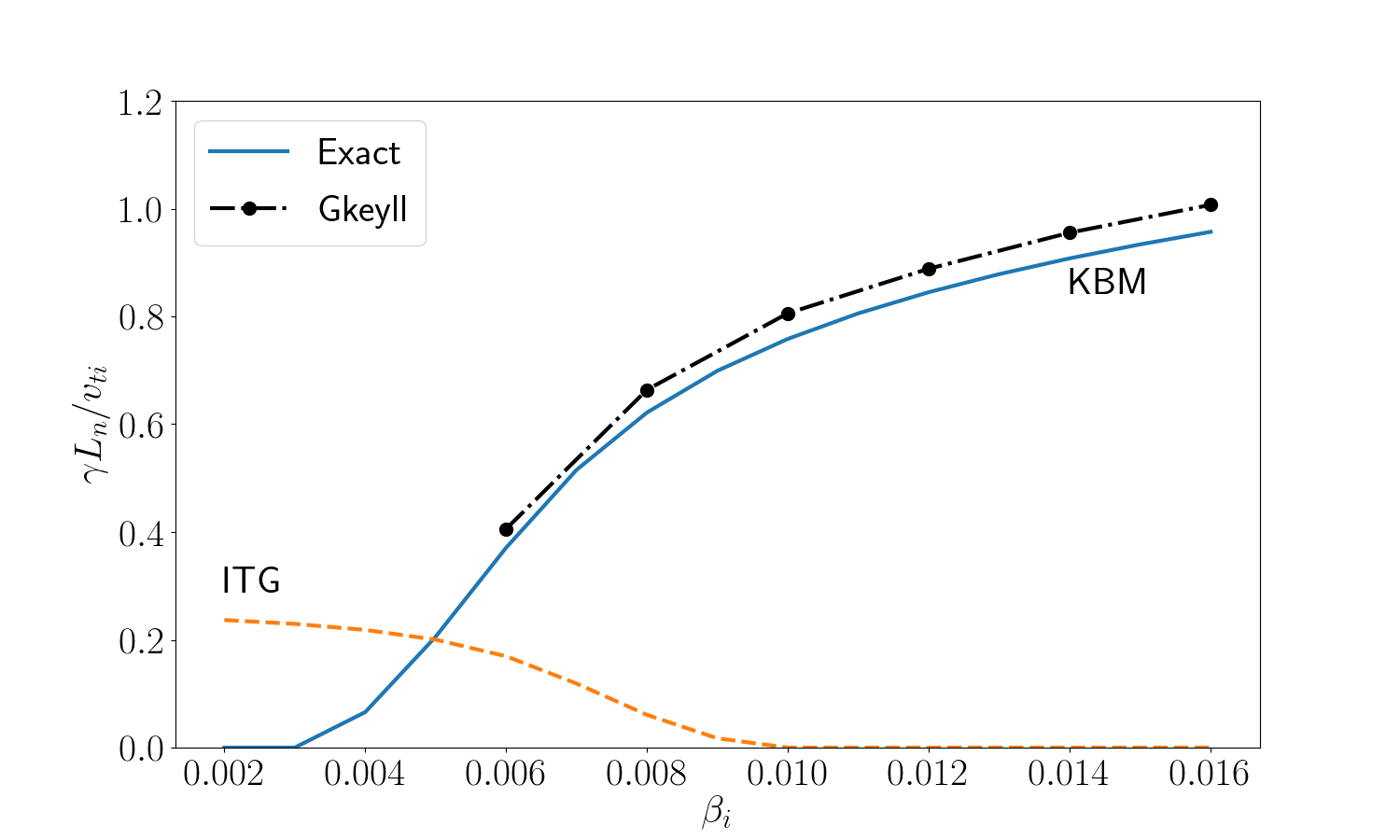}
    \caption[Growth rates for the KBM instability in the local limit.]{Growth rates for the KBM instability in the local limit, as a function of $\beta_i$, with $k_\perp \rho_i=0.5,\ k_\parallel L_n = 0.1,\ R/L_n=5,\ R/L_{Ti}=12.5,\  R/L_{Te}=10,$ and $\tau=1$. The black dots are numerical results from \gke, and the colored lines are the result of numerically solving the analytic dispersion relation given by Eqs. (\ref{kbmdisp1}-\ref{kbmdisp2}).}
    \label{fig:kbm}
\end{figure}

We show the results of \gke simulations of the KBM instability in the local-limit helical geometry for several values of $\beta_i$ in \cref{fig:kbm}. The results agree well with the analytic result obtained by numerically solving Eqs. (\ref{kbmdisp2}) and (\ref{kbmdisp1}). The parameters $k_\perp \rho_i=0.5,\ k_\parallel L_n = 0.1,\ R/L_n=5,\ R/L_{Ti}=12.5,\  R/L_{Te}=10,\ \tau=1$ are chosen to match those used in figure 1 of \citet{kim1993electromagnetic}, although the differences in FLR terms ($b=0$) cause our growth rates to be larger than those in \citet{kim1993electromagnetic}.
{Finally, we note that since \gke is designed primarily for nonlinear calculations, the fact that Fourier modes are not eigenfunctions of the DG discretization of the system makes these linear tests somewhat difficult for \gke. This may play a role in the small deviation of the results from the analytical theory. Because of this, Fourier modes other than the one initialized can grow and pollute the results. In particular, we have not included results from the ion temperature gradient (ITG) branch because we find that a mode with $k_\parallel=0$ grows and overcomes the  finite $k_\parallel$ mode before its growth rate has converged.}

\section{Avoiding the Amp\`ere cancellation problem} \label{sec:cancel}
Historically, including electromagnetic effects in gyrokinetic simulations has proved numerically and computationally challenging, both in the core and in the edge. The so-called Amp\`ere cancellation problem is one of the main numerical issues that has troubled primarily PIC codes
\citep{reynders1993gyrokinetic,cummings1994gyrokinetic}. 

To understand where the cancellation problem comes from, let us reexamine the simple Alfv\'en wave case from \cref{res:kaw}. The cancellation problem is usually discussed in the context of the $p_\parallel$ (Hamiltonian) formulation of electromagnetic gyrokinetics, which is the formulation used by most PIC codes (in order to avoid the appearance of the explicit time derivative of $A_\parallel$ in the gyrokinetic equation). In the $p_\parallel$ formulation, the simple gyrokinetic system that we looked at in \cref{res:kaw} becomes
\begin{gather}
    \pderiv{f_e}{t} = \{H_e,f_e\}= - \frac{1}{m_e}p_\parallel \pderiv{f_e}{z} {-} e\pderiv{f_e}{p_\parallel}\pderiv{\Phi}{z} \\
    k_\perp^2 \frac{m_i n_{0}}{B^2}\Phi = e n_0-e\int f_e\,\dx{p_\parallel} \label{app:alf-poissonH} \\
    \left(k_\perp^2 + C_N\ \frac{\mu_0 e^2}{m_e^2}\int  f_e\,\dx{p_\parallel}\right)A_{\parallel} = -C_{J}\ \frac{\mu_0 e}{m_e^2}\int  p_\parallel f_e\,\dx{p_\parallel} \label{app:alf-ohmH}.
\end{gather}
In \cref{app:alf-ohmH}, we have introduced two constants, $C_N$ and $C_J$. We will use these constants to represent small errors that could arise in the numerical calculation of these integrals. As in Section \ref{res:kaw}, we can calculate the dispersion relation for this system, but now we will take the limit $\omega\gg k_\parallel v_{te}$, so that the dispersion relation reduces to
\begin{equation}
    \omega^2 = \frac{k_\parallel^2 v_A^2}{C_N + k_\perp^2 \rho_\mathrm{s}^2/\hat{\beta}}\left[1+(C_N - C_{J})\frac{\hat{\beta}}{k_\perp^2\rho_\mathrm{s}^2}\right], \label{cancel-dispH}
\end{equation}
where recall that $\hat{\beta}=(\beta_e/2) m_i/m_e$. This reduces to the correct result if $C_N=C_{J}=1$. However, if $C_N\neq C_{J}$, there will be a spurious numerical term from the second term in the brackets, leading to large errors for modes with $\hat{\beta}/(k_\perp^2\rho_\mathrm{s}^2)\gg 1$. 
This means that one must be very careful in how the integrals in \cref{app:alf-ohmH} are computed. The integrals need not be computed exactly, but one must ensure that they are computed \emph{consistently} so that any numerical error is identical in both integrals (\emph{i.e.} $C_N=C_J$), resulting in the errors cancelling exactly. This can be challenging in PIC codes, in part because the moments of the distribution function involve some finite sampling noise. Another complication is that the particle positions do not coincide with the field grid, necessitating interpolations. Various $\delta f$ PIC schemes to address the cancellation problem have been developed and there are interesting recent advances in this area \citep{chen2003deltaf,mishchenko2004,hatzky2007,mishchenko2014pullback,startsev2014,bao2018conservative}.

Meanwhile, some continuum $\delta f$ core codes avoided the cancellation problem completely \citep{rewoldt1987,kotschenreuther1995comparison}, while others had to address somewhat minor issues resulting from it \citep{jenko2000,candy2003}. In particular, the use of the $v_\parallel$ (symplectic) formulation of EMGK in \citep{kotschenreuther1995comparison} results in an Amp\`ere's law that contains only one integral, as in \cref{ampere}, so one does not need to worry about two large integral terms cancelling appropriately. 

However, in our scheme based on the $v_\parallel$ formulation, we solve Ohm's law for $\pderivInline{A_\parallel}{t}$. Recall from \cref{res:kaw} that the simple gyrokinetic system is given by
\begin{gather}
    \pderiv{f_e}{t} = \{H_e,f_e\}-\frac{e}{m_e}\pderiv{f_e}{v_\parallel}\pderiv{A_\parallel}{t} = - v_\parallel \pderiv{f_e}{z} {-} \frac{e}{m}\pderiv{f_e}{v_\parallel}\left(\pderiv{\Phi}{z}+\pderiv{A_\parallel}{t} \right) \\
    k_\perp^2 \frac{m_i n_{0}}{B^2}\Phi = e n_0-e\int f_e\,\dx{v_\parallel} \label{app:alf-poisson} \\
    \left(k_\perp^2 + C_N\ \frac{\mu_0 e^2}{m_e}\int  f_e\,\dx{v_\parallel}\right)\pderiv{A_{\parallel}}{t} = -C_{J}\ \mu_0 e\int v_\parallel \{H_e, f_e\}\,\dx{v_\parallel} \label{app:alf-ohm}.
\end{gather}
Ohm's law \emph{does} have two integrals on either side of the equation. As above, we have inserted constants, $C_N$ and $C_{J}$ to represent small numerical errors in the calculation of the integrals. Again taking the limit $\omega\gg k_\parallel v_{te}$, the dispersion relation reduces to
\begin{equation}
    \omega^2 = \frac{k_\parallel^2 v_A^2}{C_N + k_\perp^2 \rho_\mathrm{s}^2/\hat{\beta}}\left[1+(C_N - C_{J})\frac{\hat{\beta}}{k_\perp^2\rho_\mathrm{s}^2}\right]. \label{cancel-disp}
\end{equation}
This is the same dispersion relation as we obtained in \cref{cancel-dispH} from the $p_\parallel$ formulation. Thus even though we are using the $v_\parallel$ formulation, we still have to worry about the cancellation problem when we use Ohm's law to solve for $\pderivInline{A_\parallel}{t}$. We must be careful to compute the integrals in Ohm's law consistently so that numerical errors cancel exactly. In the following section, we derive a semi-discrete Alfv\'en wave dispersion relation that results from our DG discretization scheme to show that our scheme does indeed avoid the cancellation problem.

\subsection{Semi-discrete dispersion relation for Alfv\'en wave} \label{app:disp}
Here we will derive a semi-discrete Alfv\'en wave dispersion relation by using a piecewise-linear DG discretization for only the $v_\parallel$ coordinate, with the remaining coordinates not discretized (for simplicity).
The main purpose is to show how our discrete scheme avoids the Amp\`ere cancellation problem. We will also show how the integrals in the DG weak form are computed analytically in our modal scheme.

The semi-discrete gyrokinetic weak form for this system is
\begin{gather}
    \int_{\mathcal{K}_j} \psi \pderiv{f_h}{t}\,\dx{v_\parallel}+ \int_{\mathcal{K}_j}\psi \frac{1}{m_e}\pderiv{H_h}{v_\parallel}\pderiv{f_h}{z}\,\dx{v_\parallel} - \frac{e}{m_e}\left(\pderiv{\Phi}{z}+\pderiv{A_\parallel}{t}\right)\int_{\mathcal{K}_j}\pderiv{\psi}{v_\parallel}f_h\,\dx{v_\parallel} \qquad\notag \\
    + \frac{e}{m_e}\left(\pderiv{\Phi}{z}+\pderiv{A_\parallel}{t}\right)(\psi^- \widehat{f}_h)\bigg\rvert_{\partial\mathcal{K}_j} = 0.
\end{gather}
We begin by mapping each cell $\mathcal{K}_j$ to $\xi\in [-1,1]$ via the transformation $\xi = 2(v_\parallel - \bar{v}_\parallel^j)/\Delta v_\parallel$, where $\bar{v}_\parallel^j$ is the cell center of cell $j$, resulting in
\begin{gather}
    \int_{-1}^1 \psi \pderiv{f_h}{t}\,\dx{\xi} + \int_{-1}^1 \psi \frac{2}{\Delta v_\parallel}\frac{1}{m_e}\pderiv{H_h}{\xi}\pderiv{f_h}{z}\,\dx{\xi} - \frac{e}{m_e}\left(\pderiv{\Phi}{z}+\pderiv{A_\parallel}{t}\right)\int_{-1}^1 \frac{2}{\Delta v_\parallel}\pderiv{\psi}{\xi}f_h\,\dx{\xi} \qquad \notag \\
    + \frac{e}{m_e}\left(\pderiv{\Phi}{z}+\pderiv{A_\parallel}{t}\right)\frac{2}{\Delta v_\parallel}(\psi^- \widehat{f}_h)\bigg\rvert_{-1}^{\ 1} = 0.
\end{gather}
Taking an orthonormal piecewise-linear basis in $\xi$, $\psi = [\frac{1}{\sqrt{2}}, \frac{\sqrt{3}}{\sqrt{2}}\xi]$, we expand $f_h$ on the basis in cell $j$ as
\begin{equation}
    f_h^j(z,v_\parallel,t) = \sum_k \psi_k(\xi) f_k^j(z,t) = \frac{1}{\sqrt{2}}f_0^j + \frac{\sqrt{3}}{\sqrt{2}} f_1^j \xi.
\end{equation}
(Note that in the fully discretized case all coordinate dependence would be contained in multi-variate basis functions.) We can then analytically integrate the weak form for each $\psi_k$ to obtain the modal evolution equation for each DG `mode' $f_k$:
\begin{gather}
    \pderiv{f_0^j}{t} + \bar{v}_\parallel^j\pderiv{f_0^j}{z} + \frac{e}{m_e}\left(\pderiv{\Phi}{z} + \pderiv{A_\parallel}{t}\right)\frac{\sqrt{2}}{\Delta v_\parallel}\widehat{f}_{h}^{\,j}\bigg\rvert_{-1}^{\ 1} = 0 \\
    \pderiv{f_1^j}{t} + \bar{v}_\parallel^j\pderiv{f_1^j}{z} + \frac{e}{m_e}\left(\pderiv{\Phi}{z} + \pderiv{A_\parallel}{t}\right)\frac{\sqrt{6}}{\Delta v_\parallel}\xi\widehat{f}_{h}^{\,j}\bigg\rvert_{-1}^{\ 1} - \frac{e}{m_e}\left(\pderiv{\Phi}{z} + \pderiv{A_\parallel}{t}\right)\frac{2\sqrt{3}}{\Delta v_\parallel}f_0^j = 0.
\end{gather}
Finally, we will make the ansatz $f_k = F_{Mk} + f_k e^{i k_\parallel z- i \omega t}$ and linearize:
\begin{gather}
    -i(\omega-k_\parallel \bar{v}_\parallel^i){f_0^j} + \frac{e}{m_e}\left(ik_\parallel{\Phi} + \pderiv{A_\parallel}{t}\right)\frac{\sqrt{2}}{\Delta v_\parallel}\widehat{F}_{Mh}^{j}\bigg\rvert_{-1}^{\ 1} = 0 \label{f0w} \\
    -i(\omega-k_\parallel \bar{v}_\parallel^i){f_1^j} + \frac{e}{m_e}\left(ik_\parallel{\Phi} + \pderiv{A_\parallel}{t}\right)\frac{\sqrt{6}}{\Delta v_\parallel}\xi\widehat{F}_{Mh}^{j}\bigg\rvert_{-1}^{\ 1} - \frac{e}{m_e}\left(ik_\parallel{\Phi} + \pderiv{A_\parallel}{t}\right)\frac{2\sqrt{3}}{\Delta v_\parallel}F_{M0}^j = 0.
\end{gather}

We now turn to the field equations. The Poisson equation is
\begin{equation}
    k_\perp^2\frac{m_i n_0}{B^2}\Phi = e n_0 - e\sum_j \int_{\mathcal{K}_j} f_h\,\dx{v_\parallel}.
\end{equation}
Expanding $f_h$ and using the ansatz, this becomes
\begin{equation}
    k_\perp^2\frac{m_i n_0}{B^2}\Phi = e n_0 - e\sum_j \frac{\Delta v_\parallel}{\sqrt{2}}F_{M0}^j - e \sum_j \frac{\Delta v_\parallel}{\sqrt{2}} f_0^j = - e \sum_j \frac{\Delta v_\parallel}{\sqrt{2}} f_0^j,
\end{equation}
where we will define $F_{Mh}$ so that $\sum_j \frac{\Delta v_\parallel}{\sqrt{2}}F_{M0}^j = n_0$ by definition. For Ohm's law, we must use the $p_v=1$ form from Eq. (\ref{Ohmp1}), which gives
\begin{equation}
    k_\perp^2\pderiv{A_\parallel}{t} - \pderiv{A_\parallel}{t} \frac{\mu_0 e^2}{m_e}\sum_j \bar{v}_\parallel^j \widehat{f}_{h}^{\,j}\bigg\rvert_{\partial \mathcal{K}_j} = - \mu_0 e \sum_j \int_{\mathcal{K}_j} \bar{v}_\parallel^j \pderiv{f_h}{t}^\star\,\dx{v_\parallel} + ik_\parallel \Phi \frac{\mu_0 e^2}{m_e}\sum_j \bar{v}_\parallel^j \widehat{f}_{h}^{\,j}\bigg\rvert_{\partial \mathcal{K}_j},
\end{equation}
where
\begin{equation}
    \int_{\mathcal{K}_j}\psi \pderiv{f_h}{t}^\star \,\dx{v_\parallel}= -i k_\parallel \int_{\mathcal{K}_j} \psi \frac{1}{m_e}\pderiv{H_h}{v_\parallel} f_h\,\dx{v_\parallel} + \frac{e}{m_e}ik_\parallel\Phi\int_{\mathcal{K}_j}dv_\parallel\ \pderiv{\psi}{v_\parallel}f_h\,\dx{v_\parallel}.
\end{equation}
Again expanding and using the ansatz, Ohm's law becomes
\begin{equation}
k_\perp^2\pderiv{A_\parallel}{t} - \pderiv{A_\parallel}{t} \frac{\mu_0 e^2}{m_e}\sum_j \bar{v}_\parallel^j \widehat{F}_{Mh}^{j}\bigg\rvert_{-1}^{\ 1} = - \mu_0 e(i k_\parallel) \sum_j \frac{\Delta v_\parallel}{\sqrt{2}}\bar{v}_\parallel^{j\,2} f_0^j + ik_\parallel \Phi \frac{\mu_0 e^2}{m_e}\sum_j \bar{v}_\parallel^j \widehat{F}_{Mh}^{j}\bigg\rvert_{-1}^{\ 1}. \label{ohmw} 
\end{equation}
Analogously to \cref{app:alf-ohm}, we can rewrite this equation as
\begin{equation}
k_\perp^2\pderiv{A_\parallel}{t} + \pderiv{A_\parallel}{t} \frac{\mu_0 e^2 n_0}{m_e} C_N = - \mu_0 e(i k_\parallel) \sum_j \frac{\Delta v_\parallel}{\sqrt{2}}\bar{v}_\parallel^{j\,2} f_0^j - ik_\parallel \Phi \frac{\mu_0 e^2 n_0}{m_e}C_J, \label{ohmw2} 
\end{equation}
where we have defined
\begin{gather}
    C_N = -\sum_j\frac{1}{n_0}\bar{v}_\parallel^j \widehat{F}_{Mh}^{j}\bigg\rvert_{-1}^{\ 1}, \\
    C_J = -\sum_j\frac{1}{n_0}\bar{v}_\parallel^j \widehat{F}_{Mh}^{j}\bigg\rvert_{-1}^{\ 1}.
\end{gather}
Clearly $C_N = C_J$, which allows us to move the $C_N$ term to the right-hand side, giving a term proportional to the total parallel electric field, $E_\parallel = -ik_\parallel\Phi -\pderiv{A_\parallel}{t}$:
\begin{equation}
    k_\perp^2\pderiv{A_\parallel}{t} = - \mu_0 e(i k_\parallel) \sum_j \frac{\Delta v_\parallel}{\sqrt{2}}\bar{v}_\parallel^{j\,2} f_0^j - \frac{\mu_0 e^2 n_0}{m_e}\left(ik_\parallel \Phi +\pderiv{A_\parallel}{t}\right)C_N.
\end{equation}
This is essential for avoiding the cancellation problem because if we instead had $C_N\neq C_J$, we would have had a leftover term proportional to $(C_N-C_J)\pderiv{A_\parallel}{t}$ on the left-hand side. This leftover term would then lead to the spurious term  proportional to $\hat{\beta}/(k_\perp^2\rho_\mathrm{s}^2)$ in Eq. (\ref{cancel-disp}).

In order to compute the integral quantities in the field equations, we use Eq. (\ref{f0w}) to compute
\begin{align}
    f_0^j &= -\frac{e}{m_e}\left( ik_\parallel \Phi + \pderiv{A_\parallel}{t}\right)\frac{i}{\omega - k_\parallel \bar{v}_\parallel^j}\frac{\sqrt{2}}{\Delta v_\parallel} \widehat{F}_{Mh}^{j}\bigg\rvert_{-1}^{\ 1} \notag\\
    &\approx -\frac{e}{m_e}\left( ik_\parallel \Phi + \pderiv{A_\parallel}{t}\right)\frac{i}{\omega}\left(1+\frac{k_\parallel \bar{v}_\parallel^j}{\omega} + \frac{k_\parallel^2 \bar{v}_\parallel^{j\,2}}{\omega^2} + \frac{k_\parallel^3 \bar{v}_\parallel^{j\,3}}{\omega^3}\right) \frac{\sqrt{2}}{\Delta v_\parallel} \widehat{F}_{Mh}^{j}\bigg\rvert_{-1}^{\ 1} \ \quad (\omega \gg k_\parallel v_{te}),
\end{align}
where we have expanded in the limit $\omega \gg k_\parallel v_{te}$.
Now we can calculate
\begin{gather}
    \sum_j \frac{\Delta v_\parallel}{\sqrt{2}} f_0^j = -\frac{e}{m_e}\left( ik_\parallel \Phi + \pderiv{A_\parallel}{t}\right)\frac{i}{\omega}\sum_j\left(1+\frac{k_\parallel \bar{v}_\parallel^j}{\omega} + \frac{k_\parallel^2 \bar{v}_\parallel^{j\,2}}{\omega^2} + \frac{k_\parallel^3 \bar{v}_\parallel^{j\,3}}{\omega^3}\right)  \widehat{F}_{Mh}^{j}\bigg\rvert_{-1}^{\ 1} \\
    \sum_j \frac{\Delta v_\parallel}{\sqrt{2}}\bar{v}_\parallel^{j\,2}f_0^j = -\frac{e}{m_e}\left( ik_\parallel \Phi + \pderiv{A_\parallel}{t}\right)\frac{i}{\omega}\sum_j\left(1+\frac{k_\parallel \bar{v}_\parallel^j}{\omega}\right) \bar{v}_\parallel^{j\,2} \widehat{F}_{Mh}^{j}\bigg\rvert_{-1}^{\ 1}
\end{gather}
Substituting these integral quantities into the field equations, the Poisson equation becomes
\begin{align}
    k_\perp^2 \frac{m_i n_0}{B^2}\Phi &= \frac{e^2}{m_e}\left(ik_\parallel \Phi + \pderiv{A_\parallel}{t}\right)\frac{i k_\parallel}{\omega^2}\sum_j \left[\bar{v}_\parallel^j \widehat{F}_{Mh}^{j}\bigg\rvert_{-1}^{\ 1} +\frac{k_\parallel}{\omega} \left(1+\frac{k_\parallel \bar{v}_\parallel^j}{\omega} \right)\bar{v}_\parallel^{j\,2} \widehat{F}_{Mh}^{j}\bigg\rvert_{-1}^{\ 1}\right] \notag \\
    &= \frac{e^2}{m_e}\left(ik_\parallel \Phi + \pderiv{A_\parallel}{t}\right)\frac{i k_\parallel}{\omega^2} \left[-n_0 C_N +\sum_j\frac{k_\parallel}{\omega} \left(1+\frac{k_\parallel \bar{v}_\parallel^j}{\omega} \right)\bar{v}_\parallel^{j\,2} \widehat{F}_{Mh}^{j}\bigg\rvert_{-1}^{\ 1}\right] \label{poissonw}
\end{align}
and Ohm's law becomes
\begin{align}
    k_\perp^2 \pderiv{A_\parallel}{t} &= \frac{\mu_0 e^2}{m_e} \left(ik_\parallel \Phi + \pderiv{A_\parallel}{t}\right)\sum_j \left[\bar{v}_\parallel^j \widehat{F}_{Mh}^{j}\bigg\rvert_{-1}^{\ 1} +\frac{k_\parallel}{\omega} \left(1+\frac{k_\parallel \bar{v}_\parallel^j}{\omega} \right)\bar{v}_\parallel^{j\,2} \widehat{F}_{Mh}^{j}\bigg\rvert_{-1}^{\ 1}\right] \notag \\
    &= \frac{\mu_0 e^2}{m_e} \left(ik_\parallel \Phi + \pderiv{A_\parallel}{t}\right) \left[-n_0 C_N +\sum_j\frac{k_\parallel}{\omega} \left(1+\frac{k_\parallel \bar{v}_\parallel^j}{\omega} \right)\bar{v}_\parallel^{j\,2} \widehat{F}_{Mh}^{j}\bigg\rvert_{-1}^{\ 1}\right], \label{ohmw3}
\end{align}
where we have substituted the definition of $C_N$.
We can now combine Eqs. (\ref{poissonw}) and (\ref{ohmw3}) by multiplying Eq. (\ref{poissonw}) by $i k_\parallel T_e/n_0$, multiplying Eq. (\ref{ohmw3}) by $\rho_\mathrm{s}^2 = m_i T_e/(e^2 B^2)$ and summing the two equations to get
\begin{align}
&k_\perp^2\rho_\mathrm{s}^2\left(ik_\parallel \Phi + \pderiv{A_\parallel}{t}\right)= \notag \\
&\quad \left(ik_\parallel \Phi + \pderiv{A_\parallel}{t}\right)\left(\hat{\beta}-\frac{k_\parallel^2 v_{te}^2}{\omega^2}\right) \left[-C_N +\frac{1}{n_0}\sum_j\frac{k_\parallel}{\omega} \left(1+\frac{k_\parallel \bar{v}_\parallel^j}{\omega} \right)\bar{v}_\parallel^{j\,2} \widehat{F}_{Mh}^{j}\bigg\rvert_{-1}^{\ 1}\right],
\end{align}
with $\hat{\beta}=(\beta_e/2)m_i/m_e$. This then yields the dispersion relation
\begin{align}
k_\perp^2\rho_\mathrm{s}^2 &=\left(\hat{\beta}-\frac{k_\parallel^2 v_{te}^2}{\omega^2}\right)\left[-C_N +\frac{1}{n_0}\sum_j\frac{k_\parallel}{\omega} \left(1+\frac{k_\parallel \bar{v}_\parallel^j}{\omega} \right)\bar{v}_\parallel^{j\,2} \widehat{F}_{Mh}^{j}\bigg\rvert_{-1}^{\ 1}\right]. \label{dispw}
\end{align}

To evaluate $C_N$ and the other sum, we need to project the background onto the basis in each cell. Taking $F_M = n_{0h} (2\pi v_t^2)^{-1/2}\exp\left(-v_\parallel^2/(2v_t^2)\right)$, we project onto the basis in cell $j$ as
\begin{align}
    F_{M0}^j &= \frac{1}{\sqrt{2}}\int_{-1}^1 \frac{1}{\sqrt{2\pi v_t^2}}e^{\frac{-\left(\bar{v}_\parallel^j+\Delta v_\parallel \xi/2\right)^2}{2 v_t^2}}\,\dx{\xi} \notag \\
    &= \frac{n_{0h}}{\Delta v_\parallel \sqrt{2}}\left[ \text{erf}\left(\frac{(j+1/2)\Delta v_\parallel}{v_t\sqrt{2}}\right) - \text{erf}\left(\frac{(j-1/2)\Delta v_\parallel}{v_t\sqrt{2}}\right)\right] \\
    F_{M1}^j &= \frac{\sqrt{3}}{\sqrt{2}}\int_{-1}^1  \xi\frac{1}{\sqrt{2\pi v_t^2}}e^{\frac{-\left(\bar{v}_\parallel^j+\Delta v_\parallel \xi/2\right)^2}{2 v_t^2}}\,\dx{\xi} \notag \\ &= -\frac{2 n_{0h}v_t}{\Delta v_\parallel^2}\sqrt{\frac{3}{\pi}}\left(e^{\frac{-\left((j+1/2)\Delta v_\parallel\right)^2}{2 v_t^2}} - e^{\frac{-\left((j-1/2)\Delta v_\parallel\right)^2}{2 v_t^2}} \right) \notag \\
    &\qquad-\frac{n_{0h}\sqrt{6}}{\Delta v_\parallel}j\left[ \text{erf}\left(\frac{(j+1/2)\Delta v_\parallel}{v_t\sqrt{2}}\right) - \text{erf}\left(\frac{(j-1/2)\Delta v_\parallel}{v_t\sqrt{2}}\right)\right],
\end{align}
where we have taken the cell center to be $\bar{v}_\parallel^j = j\Delta v_\parallel$. Now we can evaluate integrated quantities such as
\begin{align}
    \sum_{j=-N}^P \frac{\Delta v_\parallel}{\sqrt{2}}F_{M0}^j &= \frac{n_{0h}}{2}\left[\text{erf}\left(\frac{(P+1/2)\Delta v_\parallel}{v_t\sqrt{2}}\right) - \text{erf}\left(\frac{-(N+1/2)\Delta v_\parallel}{v_t\sqrt{2}}\right)\right] \notag \\
    &= \frac{n_{0h}}{2}\left[\text{erf}\left(\frac{v_\text{max}}{v_t\sqrt{2}}\right) - \text{erf}\left(\frac{v_\text{min}}{v_t\sqrt{2}}\right)\right] \notag \\
    &= n_{0h}\text{erf}\left(\frac{v_\text{max}}{v_t\sqrt{2}}\right) \qquad\qquad \text{assuming } v_\text{min} = -v_\text{max}
\end{align}
where now note that we have finite limits on the sum to indicate finite extents of the $v_\parallel\in [-v_\text{max}, v_\text{max}]$ grid.
 As we alluded to before, we will define $n_{0h}$ so that $\sum_j \frac{\Delta v_\parallel}{\sqrt{2}}F_{M0}^j = n_0$ by definition, which means 
\begin{equation}
    n_{0h} = \frac{n_0}{\text{erf}\left(\frac{v_\text{max}}{v_t\sqrt{2}}\right)}.
\end{equation}
Note that $\text{erf}(x)$ quickly approaches 1 with increasing $x$, so that for example when $v_\text{max} = 4 v_t$, $n_{0h} \approx 1.00006 \,n_0$. We can also calculate
\begin{align}
    C_N &= -\frac{1}{n_0}\sum_{j=-N}^P \bar{v}_\parallel^j \widehat{F}_{M h}^j\bigg\rvert_{-1}^{\ 1} \notag \\
    &= -\frac{1}{n_0}\sum_{j=-N}^P \frac{j\Delta v_\parallel}{2}\left[ F_{Mh}^j(1)+F_{Mh}^{j+1}(-1) - F_{Mh}^j(-1)-F_{Mh}^{j-1}(1)\right] \notag \\&\qquad\qquad +\frac{\sigma}{n_0}  \sum_{j=-N}^P\frac{j\Delta v_\parallel}{2}\left[ F_{Mh}^{j+1}(-1)-F_{Mh}^{j}(1) - F_{Mh}^j(-1)+F_{Mh}^{j-1}(1) \right] \notag \\
    &= \frac{1}{n_0}\sum_{j=-N}^P \frac{\Delta v_\parallel}{2}\left[F_{Mh}^j(1)+F_{Mh}^j(-1)\right] \notag \\
    &\qquad\qquad+\frac{\sigma}{n_0} \sum_{j=-N}^P \frac{\Delta v_\parallel}{2}\left[F_{Mh}^j(1)-F_{Mh}^j(-1)\right]+ \text{boundary terms}\notag \\
    &=  \frac{1}{n_0}\sum_{j=-N}^P \frac{\Delta v_\parallel}{\sqrt{2}} F_{M0}^j +\frac{\sigma}{n_0} \sum_{j=-N}^P \frac{\Delta v_\parallel\sqrt{3}}{\sqrt{2}} F_{M1}^j +  \text{boundary terms} \notag \\
    &= 1 + \text{boundary terms} \approx 1,
\end{align}
where $\sigma$ is the sign of the upwind velocity, and the boundary terms that result from the finite limits on the sum are small for $v_\text{max}\gtrsim 4 v_t$. Thus we have $C_N\approx1$ as expected, although it does not need to be exactly equal to unity to eliminate the cancellation problem. Instead, it was sufficient that $C_N=C_J$ on either side of Eq. (\ref{ohmw2}).

One can also show that 
\begin{gather}
    \sum_{j=-N}^P \bar{v}_\parallel^{j\,2} \widehat{F}_{M h}^j\bigg\rvert_{-1}^{\ 1} = \text{boundary terms }\approx 0 \\
    \sum_{j=-N}^P \bar{v}_\parallel^{j\,3} \widehat{F}_{M h}^j\bigg\rvert_{-1}^{\ 1} =-3n_0 v_t^2\left(1 - \frac{\Delta v_\parallel^2}{12 v_t^2}\right) + \text{boundary terms }\approx - 3n_0 v_t^2\left(1 - \frac{\Delta v_\parallel^2}{12 v_t^2}\right).
\end{gather}
Now substituting these results into the dispersion relation from Eq. (\ref{dispw}), we obtain
\begin{align}
k_\perp^2\rho_\mathrm{s}^2 \approx \left(\hat{\beta}-\frac{k_\parallel^2 v_{te}^2}{\omega^2}\right)\left(-C_N - \frac{3 k_\parallel^2 v_{te}^2}{\omega^2}\left(1 - \frac{\Delta v_\parallel^2}{12 v_{te}^2}\right)\right) \approx -\left(\hat{\beta}-\frac{k_\parallel^2 v_{te}^2}{\omega^2}\right)
\end{align}
after again taking the limit $\omega \gg k_\parallel v_{te}$ and assuming $\Delta v_\parallel \sim v_{te}$. This finally gives
\begin{equation}
    \omega^2 \approx \frac{k_\parallel^2 v_{te}^2}{\hat{\beta}+ k_\perp^2 \rho_\mathrm{s}^2},
\end{equation}
which is the expected dispersion relation.

\begin{subappendices}
\section{The discrete weak form of Ohm's law} \label{app:Ohm}
To obtain the discrete weak form of Ohm's law, we start by taking the time derivative of Eq. (\ref{FEMampere-local}):
\begin{align}
\int_{\mathcal{K}^R_i} \nabla_\perp \pderiv{A_{\parallel h}}{t} &\vec{\cdot} \nabla_\perp \varphi^{(i)}\,\dx{^3\vec{R}} - \oint_{\partial \mathcal{K}^R_i} \varphi^{(i)} \nabla_\perp \pderiv{A_{\parallel h}}{t}\cdot\dx{\vec{s}_R} \notag \\
&= \mu_0\sum_s\frac{q_s}{m_s} \int_{\mathcal{K}^R_i} \varphi^{(i)}\left[ \int_{\mathcal{T}^v} \pderiv{H_{s\,h}}{v_\parallel} \pderiv{(\mathcal{J}f_{s\,h})}{t}\,\dx{^3\vec{v}}\right]\dx{^3\vec{R}}. \label{ohm1}
\end{align}
Now, note that, analogously to Eq. (\ref{fstar}), we can write the discrete weak form of the gyrokinetic equation as
\begin{align}
    &\int_{\mathcal{K}_i} \psi \pderiv{(\mathcal{J}f_h)}{t}\dx{^3\vec{R}}\,\dx{^3\vec{v}} = 
    \int_{\mathcal{K}_i}\psi \pderiv{(\mathcal{J}f_h)}{t}^\star\dx{^3\vec{R}}\,\dx{^3\vec{v}} 
    \notag\\ &\quad
    - \oint_{\partial \mathcal{K}_i}  \psi^- \widehat{\mathcal{J}f_h} \left(\dot{v}^H_{\parallel h}-\frac{q}{m}\pderiv{A_{\parallel h}}{t}\right)\dx{^3\vec{R}}\,\dx{s}_v 
     - \int_{\mathcal{K}_i} \mathcal{J}f_h \frac{q}{m}\pderiv{A_{\parallel h}}{t} \pderiv{\psi}{v_\parallel}\dx{^3\vec{R}}\,\dx{^3\vec{v}}, \label{app:DGstar}
\end{align}
where
\begin{align}
    &\int_{\mathcal{K}_i}\psi \pderiv{(\mathcal{J}f_h)}{t}^\star\dx{^3\vec{R}}\,\dx{^3\vec{v}} = 
    \int_{\mathcal{K}_i} \mathcal{J}f_h \dot{\vec{R}}_h\cdot\nabla \psi \,\dx{^3\vec{R}}\,\dx{^3\vec{v}} 
    + \int_{\mathcal{K}_i} \mathcal{J}f_h \dot{v}^H_{\parallel h} \pderiv{\psi}{v_\parallel}\dx{^3\vec{R}}\,\dx{^3\vec{v}}
    \notag \\ &\quad 
    - \oint_{\partial \mathcal{K}_i}\psi^- \widehat{\mathcal{J}f_h}\dot{\vec{R}}_h\cdot \dx{\vec{s}}_R\, \dx{^3\vec{v}}
    + \int_{\mathcal{K}_i} \psi\left(\mathcal{J}C[f_h] + \mathcal{J}S_h\right)\dx{^3\vec{R}}\,\dx{^3\vec{v}}.
 \label{app:partialGK}
\end{align}
Substituting $\psi=\varphi^{(i)} \pderivInline{H_h}{v_\parallel}$ in Eq. (\ref{app:DGstar}) and summing over velocity cells, we obtain
\begin{align}
    &\int_{\mathcal{K}^R_i} \varphi^{(i)} \left[\int_{\mathcal{T}^v} \pderiv{H_{h}}{v_\parallel} \pderiv{(\mathcal{J}f_{h})}{t}\,\dx{^3\vec{v}}\right]\dx{^3\vec{R}} = \int_{\mathcal{K}^R_i}\varphi^{(i)}\left[ \int_{\mathcal{T}^v} \pderiv{H_{h}}{v_\parallel} \pderiv{(\mathcal{J}f_{h})^\star}{t}\,\dx{^3\vec{v}}\right]\dx{^3\vec{R}} \notag\\
    &\qquad- \int_{\mathcal{K}_i^R} \varphi^{(i)}\left[\sum\limits_j\oint_{\partial\mathcal{K}_j^v} \left(\dot{v}^H_{\parallel h}-\frac{q}{m}\pderiv{A_{\parallel h}}{t}\right) \pderiv{H_h}{v_\parallel}^- \widehat{\mathcal{J}f_h}\,\dx{s_v}\right]\dx{^3\vec{R}} \notag \\
    &\qquad- \int_{\mathcal{K}_i^R} \varphi^{(i)} \frac{q}{m}\pderiv{A_{\parallel h}}{t} \left[\int_{\mathcal{T}^v}  \mathcal{J}\pderiv{^2H_h}{v_\parallel^2}f_h\,\dx{^3\vec{v}}\right]\dx{^3\vec{R}}.
\end{align}
Note that, for $p_v>1$, the $v_\parallel$ surface term on the right-hand side vanishes because $\pderivInline{H_h}{v_\parallel}$ is continuous across $v_\parallel$ cell interfaces when $v_\parallel^2$ is included in the basis, resulting in cancellations. However, for $p_v=1$ this term is not continuous, and we must keep this surface term; further, the {last} term on the right-hand side vanishes for $p_v=1$ since $\pderivInline{^2H_h}{v_\parallel^2}=0$.
We can now substitute this result into the right-hand side of Eq. (\ref{ohm1}), giving
\begin{align}
    &\int_{\mathcal{K}^R_i} \nabla_\perp \pderiv{A_{\parallel h}}{t} \vec{\cdot} \nabla_\perp \varphi^{(i)}\,\dx{^3\vec{R}} 
    - \oint_{\partial \mathcal{K}^R_i} \varphi^{(i)}\nabla_\perp \pderiv{A_{\parallel h}}{t} \cdot \dx{\vec{s}_R} 
    \notag \\ &\quad
    - \int_{\mathcal{K}_i^R} \varphi^{(i)}\pderiv{A_{\parallel h}}{t} \left[\sum_{s,j} \frac{\mu_0 q_s^2}{m_s}\oint_{\partial\mathcal{K}^v_j}  \bar{v}_\parallel^- \widehat{\mathcal{J} f_{s\,h}}\,\dx{s_v}\right]\dx{^3\vec{R}} 
    \notag\\ &
    = \mu_0\sum_s q_s \int_{\mathcal{K}^R_i} \varphi^{(i)} \Bigg[\int_{\mathcal{T}^v} \bar{v}_\parallel \pderiv{(\mathcal{J}f_{s\,h})}{t}^\star \dx{^3\vec{v}}-\sum_j\oint_{\partial \mathcal{K}_j^v} \bar{v}_\parallel^-{\dot{v}^H_{\parallel h}} \widehat{\mathcal{J}f_{s\,h}}\,\dx{s_v} \Bigg]\dx{^3\vec{R}}, \quad\ (p_v=1) \label{app:Ohmp1} \\
   &\int_{\mathcal{K}^R_i}\nabla_\perp \pderiv{A_{\parallel h}}{t} \vec{\cdot} \nabla_\perp \varphi^{(i)}\,\dx{^3\vec{R}} - \oint_{\partial \mathcal{K}^R_i} \varphi^{(i)}\nabla_\perp \pderiv{A_{\parallel h}}{t} \cdot \dx{\vec{s}_R} 
   \notag \\ &\quad
   + \int_{\mathcal{K}_i^R} \varphi^{(i)}\pderiv{A_{\parallel h}}{t}\left[ \sum_s \frac{\mu_0 q_s^2}{m_s}\!\int_{\mathcal{T}^v}  \mathcal{J} f_{s\,h}\,\dx{^3\vec{v}}\right]\dx{^3\vec{R}} \notag \\
    &=\mu_0\sum_s q_s \!\!\int_{\mathcal{K}^R_i}  \varphi^{(i)}\left[\int_{\mathcal{T}^v} v_\parallel \pderiv{(\mathcal{J} f_{s\,h})}{t}^\star\dx{^3\vec{v}}\right]\dx{^3\vec{R}}, \qquad\qquad (p_v>1) \label{app:Ohmp2}
\end{align}
In Eq. (\ref{app:Ohmp1}), $\bar{v}_\parallel$ is the piecewise-constant projection of $v_\parallel$. 
\end{subappendices}

\chapter{Simulations of a helical scrape-off layer as a model of the NSTX SOL} \label{ch:nstx-results}
\section{Helical scrape-off layer model}
As a first step towards modeling the tokamak scrape-off layer, we consider a simple  helical scrape-off layer model. In this configuration, the magnetic field is composed of a toroidal component $B_\varphi$ and a vertical component $B_v$, giving helical field lines. All field lines are open, terminating on material walls at the top and bottom of the device. This configuration is also known as a simple magnetized torus (SMT), and has been studied experimentally via devices such as the Helimak \citep{gentle2008} and TORPEX \citep{fasoli2006}. Despite the relative simplicity of the helical SMT configuration, it contains unfavorable magnetic curvature. This gives rise to the interchange instability that drives turbulence and blob dynamics in the SOL. Thus the SMT configuration is a good testbed for investigating SOL blob dynamics. We will use parameters roughly modeling the SOL of the National Spherical Torus Experiment (NSTX) at PPPL.

\subsection{Simplified helical geometry} \label{sec:simple-geo}

We simulate a flux-tube-like domain that wraps helically around the torus and terminates on conducting plates at each end. For this, we use a non-orthogonal, field-aligned coordinate system \citep{beer1995field}, with $x$ the radial coordinate, $z$ the coordinate along the field lines, and $y$ the binormal coordinate that labels field lines at constant $x$ and $z$. One can think of these coordinates roughly mapping to physical cylindrical coordinates ($R,\varphi,Z)$ via $R=x$, $\varphi=(y\sin\chi+z\cos\chi)/R_c$, $Z=z\sin\chi$ (although this parametrization does not give a truly field-aligned coordinate system; see Appendix \ref{app:alt-heli-geo}). In this chapter, the field-line pitch angle $\chi=\sin^{-1}(B_v/B)$ is taken to be constant, with $B_v$ the vertical component of the magnetic field (analogous to the poloidal field in typical tokamak geometry), and $B$ the total magnitude of the background magnetic field. Further, $R_c=R_0+a$ is the radius of curvature at the center of the simulation domain, with $R_0$ the device major radius and $a$ the minor radius. As in \citet{shi2019}, we neglect all geometrical factors arising from the non-orthogonal coordinate system in this chapter, except for the assumption that perpendicular gradients of $f$ are much stronger than parallel gradients. Thus we can approximate
\begin{equation}
    (\nabla\times\uv{b})\vec{\cdot}\nabla f(x,y,z)\approx \left[(\nabla\times\uv{b})\vec{\cdot} \nabla y\right]\pderiv{f}{y} = \frac{1}{B}\pderiv{B}{x}\pderiv{f}{y}=-\frac{1}{x}\pderiv{f}{y},
\end{equation}
where we have used $\vec{B}\approx B_\text{axis}(R_0/x)\vec{e}_z$, with $B_\text{axis}$ the magnetic field strength at the magnetic axis, and neglected the contribution of the small vertical field $B_v$.\footnote{As a result of these approximations, the actual geometry that we are simulating is purely toroidal (\emph{i.e.} $B_v\rightarrow 0$), as we show in Appendix \ref{app:simple-heli-geo}.} This means that the magnetic (curvature plus $\nabla B$) drift,
\begin{equation}
    \vec{v}_d = \frac{m v_\parallel^2}{q B}\nabla\times\uv{b} + \frac{\mu}{q B}\uv{b}\times \nabla B,
\end{equation}
is purely in the $y$ direction, 
\begin{equation}
    \vec{v}_d \cdot\nabla y = - \left(\frac{m v_\parallel^2 + \mu B}{qB}\right)\frac{1}{x} = -\frac{m v_\parallel^2 + \mu B}{qB_\text{axis} R_0},\qquad \vec{v}_d\cdot\nabla x = \vec{v}_d\cdot\nabla z =0. \label{simplevd}
\end{equation}
Thus this simplified geometry has constant magnetic curvature (the curvature does not vary along the field line, so there is no ballooning structure), and we have neglected magnetic shear in the present setup. Note that while we make several approximations in specifying the geometry in this chapter, we will relax these approximations in \cref{ch:geometry}, where we will account for \emph{all} geometric factors arising from non-orthogonal coordinates in helical geometry and include magnetic shear.

\subsection{Modeling the Debye sheath via boundary conditions}

A distinguishing feature of the SOL is that the magnetic field lines terminate on material surfaces, resulting in the presence of the Debye sheath at the plasma-material interface. Sheath effects play a key role in blob dynamics \citep{krasheninnikov2008}, and can affect particle and heat fluxes to plasma-facing components. 

\begin{figure}
    \centering
    \includegraphics[width=\textwidth]{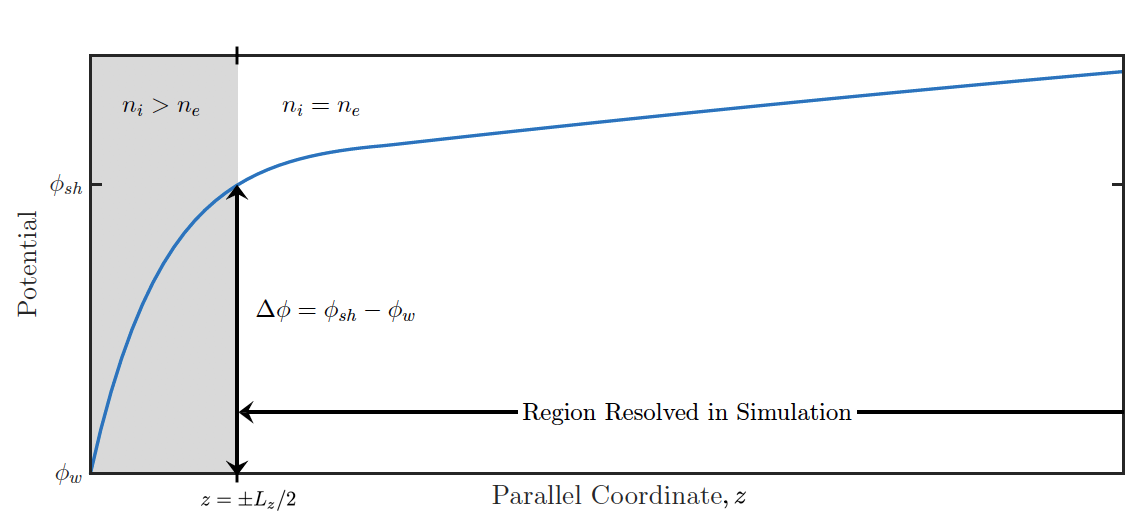}
    \caption[Illustration of the Debye sheath potential.]{Illustration of the Debye sheath potential. Since the non-quasi-neutral region (shaded) violates the gyrokinetic ordering assumptions, we cannot resolve this region directly in our simulations. Instead, we handle the sheath through model boundary conditions. This figure was adapted from \citet[][figure 2.7]{shi-thesis}.}
    \label{fig:sheath}
\end{figure}

The sheath forms because electrons move along field lines much faster than ions, resulting in electrons being initially lost more quickly to the wall. This leads to a layer of excess ions $(n_i > n_e)$ in the immediate vicinity of the wall, which breaks the quasi-neutrality condition. The plasma responds by generating an electric potential that drops near the wall, as shown in \cref{fig:sheath}, which accelerates ions into the wall and reflects low-energy electrons. A quasi-steady state is established, such that the fluxes of ions and electrons into the wall are approximately balanced so that the parallel outflow is roughly ambipolar.

Gyrokinetics, which assumes quasi-neutrality $(n_i = n_e)$, cannot handle the sheath directly. Apart from violating the gyrokinetic quasi-neutrality assumption, the length and time scales are also beyond the ordering regime of gyrokinetics ($\omega \ll \Omega_i$, $k_\perp \rho_i \sim 1$);
the sheath is a few electron Debye lengths wide $(\lambda_{De}\ll \rho_i)$, and it forms on the order of the electron plasma frequency $(\omega_{pe}\gg \Omega_i)$. Thus, we cannot resolve the sheath directly in our gyrokinetic simulations. Instead, we handle the sheath through model boundary conditions. 

\begin{figure}
    \centering
    \includegraphics[width=\textwidth   ]{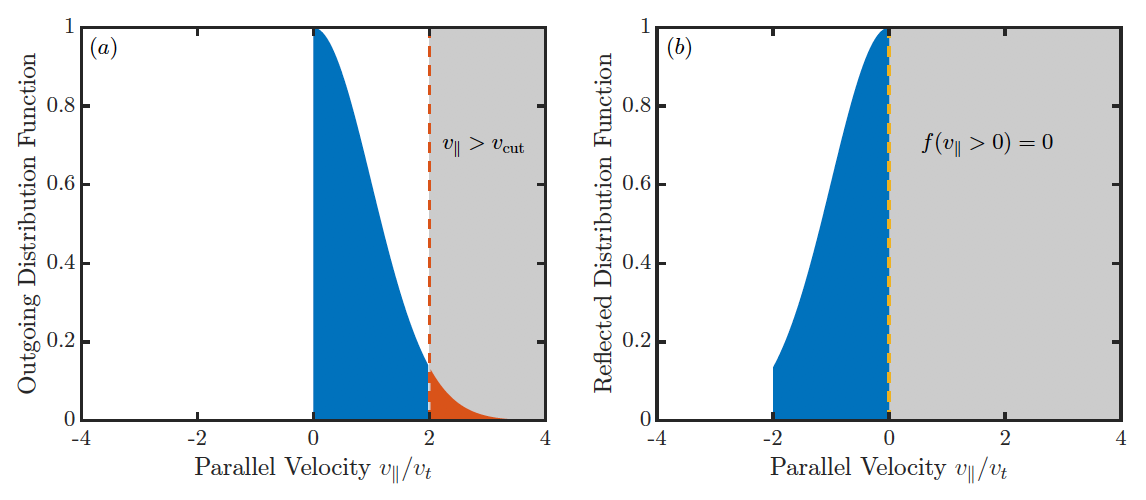}
    \caption[Illustration of the outgoing and reflected electron distribution function due to the conducting sheath boundary condition.]{Illustration of the outgoing ($a$) and reflected ($b$) electron distribution function due to the conducting sheath boundary condition. High-energy electrons with $v_\parallel > v_\mathrm{cut} = \sqrt{2|e|\Delta\Phi/m_e}$ are lost to the wall (orange region) while the lower energy electrons are reflected (blue region). This figure is from \citet[][figure 2.8]{shi-thesis}.}
    \label{fig:sheath2}
\end{figure}

We use a conducting-sheath boundary condition \citep{shi2017,shi2019}, which involves using the potential at the $z$ domain boundaries (obtained by solving the gyrokinetic Poisson equation on the whole domain) as the sheath potential, $\Phi_{sh}(x,y) = \Phi(x,y,z=z_{sh})$, with $z_{sh}=\pm L_z/2$ the $z$ domain boundaries.
By assuming that there is an unresolved non-quasi-neutral region in which the sheath potential drops to some potential at the wall, $\Phi_w$ (which is taken to be zero for a grounded wall), we can use the difference $\Delta \Phi = \Phi_{sh} - \Phi_w$ to reflect particles with $m_s v_\parallel^2/2 < -q_s\Delta \Phi$. For a typical sheath with $\Delta \Phi > 0$, this means that outgoing low-energy electrons ($q_s=-|e|$) will be reflected back into the domain, while high-energy electrons and all ions will be lost to the wall. The resulting reflected electron distribution function is shown in \cref{fig:sheath2}$b$.
Note that unlike in the standard logical sheath boundary condition \citep{parker1993suitable}, we have not directly imposed that the ion and electron currents at the sheath entrance be equal at all times. Instead, the conducting-sheath boundary condition allows local current fluctuations in and out of the sheath. We do not, however, impose the Bohm sheath criterion that ions must be supersonic as they enter the sheath \citep{bohm1949,stangeby2000}. This is one area of potential improvement to our model sheath boundary conditions. Another area of future work is accounting for the shallow incidence angle of the field lines intersecting the wall plates, leading to the development of the Chodura sheath \citep{chodura1982}. Recent work has studied the implications of the Chodura magnetic pre-sheath for gyrokinetic particle dynamics \citep{geraldini2017}. 

\section{Proof of concept: results from the first nonlinear electromagnetic gyrokinetic simulations on open field lines} \label{sec:emgk-res}
We now present preliminary nonlinear electromagnetic results from \gke. As detailed above, we simulate turbulence on helical, open field lines as a rough model of the tokamak scrape-off layer, using a flux-tube-like domain on the outboard side that wraps {helically around the torus} and terminates on conducting plates at each end in $z$. A cartoon diagram of our setup is shown in \cref{fig:heli-approx}.
These simulations are a direct extension of the work of \citet{shi2019} to include electromagnetic fluctuations. This work comprises the first published electromagnetic gyrokinetic results on open field lines, as detailed in \citet{mandell2020,hakim2020a}.

\subsection{Simulation setup}

The simulation box is centered at $(x,y,z)=(R_c,0,0)$ with dimensions $L_x=50\rho_{\mathrm{s}0}\approx 14.6$ cm, $L_y=100\rho_{\mathrm{s}0}\approx 29.1$ cm, and $L_z=L_\mathrm{pol}/\sin\chi=8$ m, where $L_\mathrm{pol}=2.4$ m and $\rho_{\mathrm{s}0}=c_{\mathrm{s}0}/\Omega_i$. {Note that although the domain that we simulate is a flux tube, the simulations are not performed in the local limit; the simulations include radial variation of the magnetic field and the profiles, and are thus effectively global.} The radial boundary conditions model conducting walls at the radial ends of the domain, given by the Dirichlet boundary condition $\Phi=A_\parallel=0$. The condition $\Phi=0$ prevents $E\times B$ flows into walls, while $A_\parallel=0$ makes it so that (perturbed) field lines never intersect the walls. For the latter, one can think of image currents in the conducting wall 
that mirror currents in the domain, resulting in exact cancellation of the perpendicular magnetic fluctuations at the wall. Also note that in this simple magnetic geometry the magnetic drifts do not have a radial component. Thus these radial boundary conditions on the fields are sufficient to ensure that there is no flux of the distribution function to the radial boundaries. 
Periodic boundary conditions are used in the $y$ direction. As discussed in the previous section, conducting-sheath boundary conditions are applied to the distribution function in the $z$ direction, with the end-plates taken to be grounded so that $\Phi_w = 0$. The fields do not require a boundary condition in the $z$ direction since only perpendicular derivatives appear in the field equations. The velocity-space grid has extents $-4v_{ts}\leq v_\parallel \leq 4 v_{ts}$ and $0\leq\mu\leq6T_{s0}/B_0$, where $v_{ts}=\sqrt{T_{s0}/m_s}$ and $B_0=B_\text{axis}R_0/R_c$. We use piecewise-linear ($p=1$) basis functions, with $(N_x,N_y,N_z,N_{v_\parallel},N_\mu)=(16,32,10,10,5)$ the number of cells in each dimension. For $p=1$ DG, one should double each of these numbers to obtain the equivalent number of grid-points for comparison with standard grid-based gyrokinetic codes, or with the number of particles per cell in PIC codes. This level of moderate velocity resolution ($\sim 200$ velocity grid-points per spatial grid-point) has been shown to be quite adequate for these types of problems \citep{candy2006}, where strong turbulence broadens the velocity resonances that might otherwise require high resolution to resolve. Further, since our algorithms conserve energy and particles, we do not need to increase velocity resolution to reduce conservation errors like in other non-conservative codes. Note however that the velocity resolution is far above that of Braginskii fluid codes, which typically keep only several fluid moments ($\sim$ velocity degrees of freedom). This will be more important when simulating the less-collisional pedestal region, where the Braginskii system is not strongly valid. 

\begin{figure}[t]
    \centering
    \includegraphics[width=\textwidth]{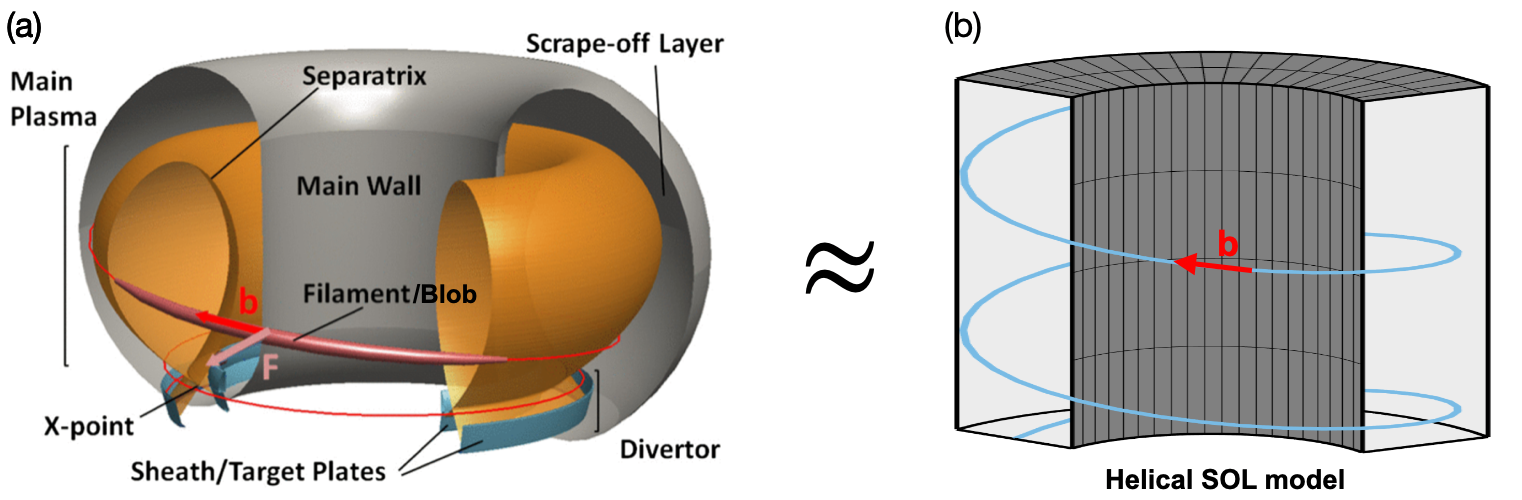}
    \caption[Cartoon diagrams of a tokamak SOL and the helical SOL model.]{$(a)$ Cartoon diagram of tokamak SOL, showing a blob filament elongated along the field line, adapted from \citet{Carralero2015}. $(b)$ Cartoon diagram of the helical SOL model.}
    \label{fig:heli-approx}
\end{figure}

The simulation parameters are similar to those used in \citet{shi2019}, roughly approximating an H-mode deuterium plasma in the NSTX SOL: $B_\text{axis}=0.5$ T, $R_0=0.85$ m, $a=0.5$ m. We use $T_{e0}=T_{i0}=40$ eV to set the velocity grid extents; these values approximate the temperatures that we expect in the simulation, and are used in the initial conditions, but the temperatures are free to evolve during the simulation. For the particle source, we use the same form as in \citet{shi2019} but we increase the source particle rate by a factor of 10 to access a higher $\beta$ regime where electromagnetic effects will be more important. This implies that the total power into the SOL is $P_\mathrm{SOL} = 54$ MW, and the total power into the simulation domain (which is a flux tube that covers a fraction of the SOL) is $P_\mathrm{src} = P_\mathrm{SOL} L_y L_z / (2\pi R_c L_\mathrm{pol}) = 6.2$ MW. The source is localized in the region $x<x_S +3\lambda_S$, with $x_S=R_c-0.05$ m and $\lambda_S=5\times10^{-3}$ m. The location $x=x_S +3\lambda_S$, which separates the source region from the SOL region, can be thought of as the separatrix. A floor of one tenth the peak particle source rate is used near the midplane to prevent regions of $n\ll n_0$ from developing at large $x$.  (In \cref{sec:power-scans}, we drop this floor on the particle source rate, after finding that it seems sufficient to put a floor on the \emph{initial} density.) The source particle rate and temperature are shown in the $x-z$ plane in \cref{fig:source-diagram-jpp}, along with an illustration of the boundary conditions. Unlike in \citet{shi2019} we do not use numerical heating to keep $f>0$ despite the fact that our DG algorithm does not guarantee positivity. While the simulations appear to be robust to negative $f$ in some isolated regions, lowering the source floor in the SOL region can sometimes lead to simulation failures due to positivity issues at large $x$. A more sophisticated algorithm for ensuring positivity is in progress, and detailed in \cref{ch:positivity}.

\begin{figure}[t!]
    \centering
    \includegraphics[width=\textwidth]{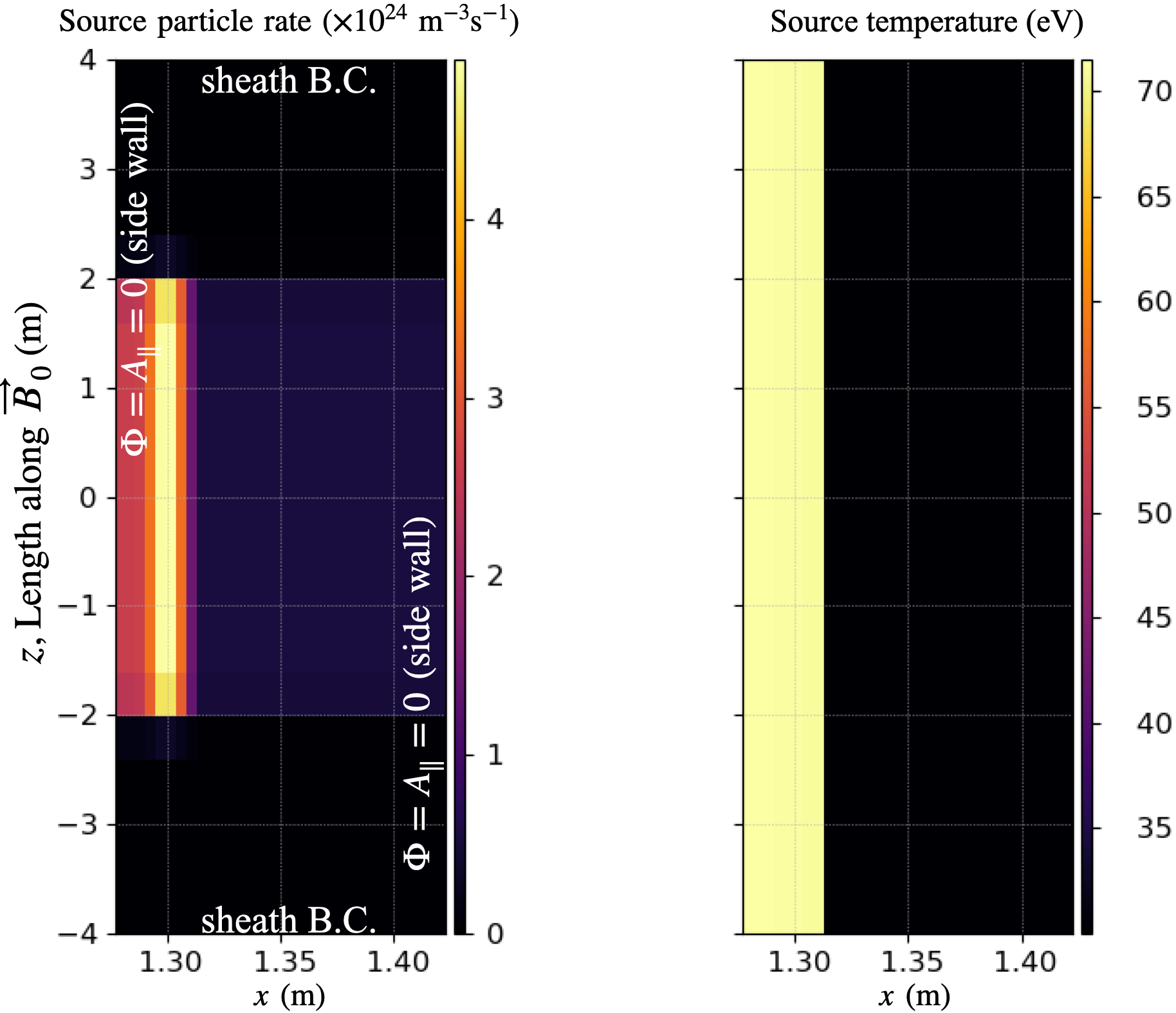}
    \caption[Diagram of the source particle rate and temperature for proof-of-principle simulations.]{Diagram of the source particle rate and temperature. In these cases, a floor of one tenth the peak particle source rate is used in the SOL region to prevent negative density regions from developing. Also shown in the left plot are the boundary conditions: Dirichlet at the radial boundaries with $\Phi=A_\parallel=0$, and sheath boundary conditions along the field line. Not pictured are periodic boundary conditions in the $y$ direction (out-of-plane here).}
    \label{fig:source-diagram-jpp}
\end{figure}

We also artificially lower the collision frequency to one tenth the physical value to offset the increased particle source rate so that the time-step limit from collisions does not become too restrictive. Further, in these initial simulations, we model only ion--ion and electron--electron collisions; cross-species collisions are not included in this section, but they are included in the simulations in \cref{sec:power-scans}. As a result, the typical ion-ion mean free path is $\lambda_{ii} \sim 3$ m, and the typical electron-electron mean free path is $\lambda_{ee} \sim 1$ m. 

The simulations were run in this configuration to  $t=1$ ms, with a quasi-steady state being reached around $t=600\ \mu\text{s}$ when the sources balance losses to the end plates. For reference, the ion transit time is $\tau_i=(L_z/2)/v_{ti}\approx 50\ \mu\text{s}$.
In terms of computational cost, the electromagnetic simulation is less than twice as expensive as the corresponding electrostatic simulation on a per-time-step basis. On 128 cores, the time per time step was 0.41 s for the electrostatic simulation and 0.68 s for the electromagnetic simulation. The increased cost is due to the additional field solves required for Ohm's law, along with additional terms in the gyrokinetic equation. However, due to time-step restrictions on an electrostatic simulation due to the electrostatic shear Alfv\'en mode (also known as the $\omega_H$ mode) \citep{lee1987gyrokinetic}, the electromagnetic simulation makes up some of the additional cost by taking slightly larger time steps. The total wall-clock time (on 128 cores) for the electrostatic simulation was approximately 65 h, and the electromagnetic simulation took about 82 h. Altogether, the cost of these simulations is relatively modest, and the addition of electromagnetic effects only makes the simulations marginally ($\sim 25$\%) more expensive. We also note that the new version of \gke, which uses a quadrature-free modal DG scheme, is approximately 10 times faster than the previous version of \gke used in \citet{shi2019}, which used nodal DG with Gaussian quadrature. For details about the improvements from the quadrature-free modal scheme, see \citet{hakim2020}.

\subsection{Electromagnetic simulation results}

We show snapshots of the density, temperature and $\beta$ of electrons (top row) and ions (bottom row) in \cref{fig:moms}. Note that the ion density is the guiding-center ion density, which does not include the ion polarization density. The snapshots are taken at the midplane ($z=0$) at $t=620\ \mu$s. We can see a blob with a mushroom structure being ejected from the source region. We also show in \cref{fig:fields} snapshots of the electromagnetic fields taken at the same time and location. We show the electrostatic potential $\Phi$, the parallel magnetic vector potential $A_\parallel$ and the normalized magnetic fluctuation amplitude $|\delta B_\perp|/B_0=|\nabla_\perp A_\parallel|/B_0$ (top row), along with the components of the parallel electric field $E_\parallel = -\nabla_\parallel \Phi - \pderivInline{A_\parallel}{t}$ (bottom row). Note that only $\Phi$, $A_\parallel$ and $\pderivInline{A_\parallel}{t}$ are evolved quantities in the simulation, with the other quantities derived. We see that $\pderivInline{A_\parallel}{t}$ is of comparable magnitude to $\nabla_\parallel \Phi$, indicating that the dynamics is in the electromagnetic regime. Significant magnetic fluctuations of over $2.5\%$ can be seen in $|\delta B_\perp|/B_0$ in this snapshot. 

\begin{figure}[t!]
    \centering
    \includegraphics[width=\textwidth]{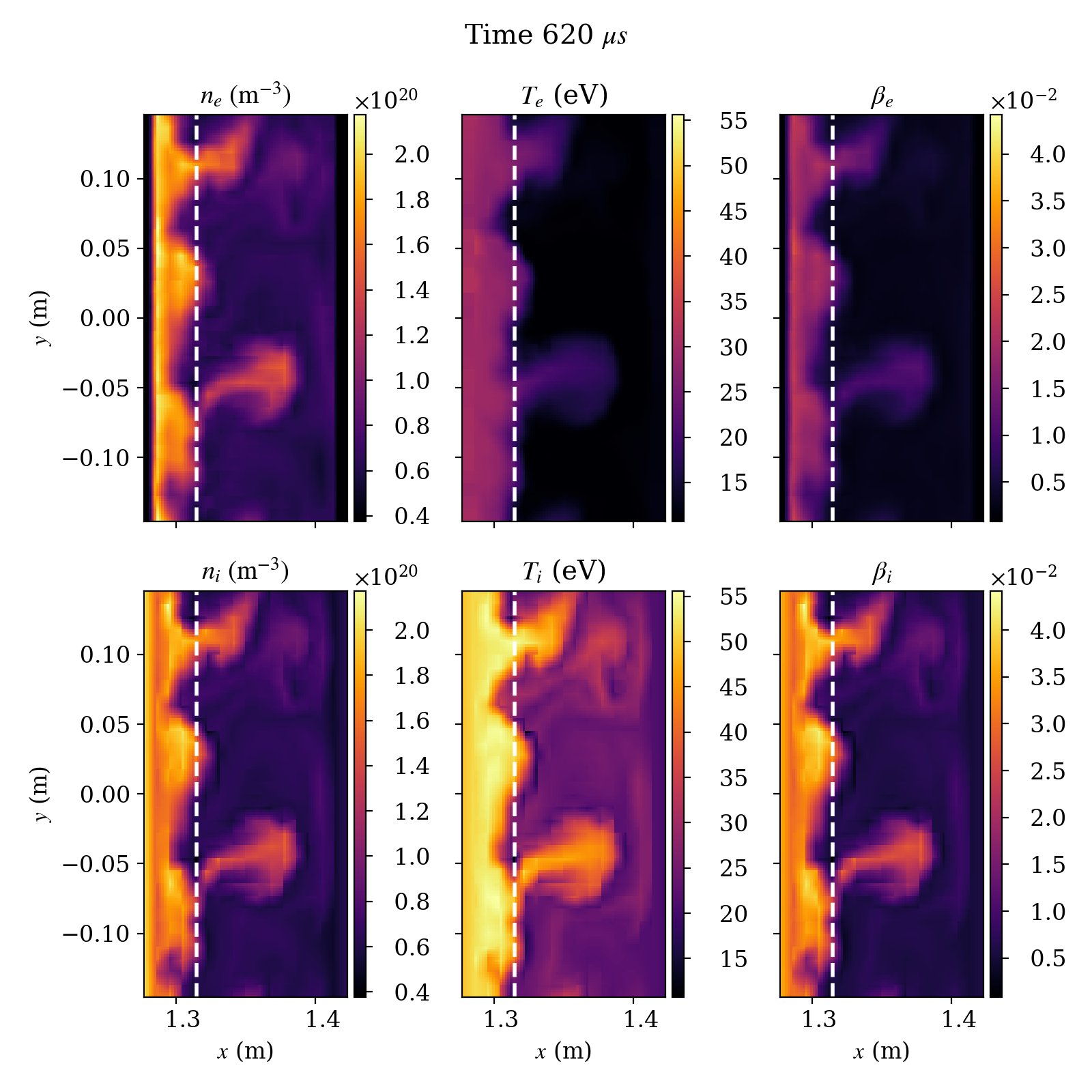}
    \caption[Snapshots of moments from an electromagnetic simulation on open, helical field lines.]{Snapshots from an electromagnetic simulation on open, helical field lines. From left to right, we show the density, temperature, and plasma beta of electrons (top row) and ions (bottom row). Note that the ion density is the guiding-center ion density, which does not include the ion polarization density.  The snapshots are taken at the midplane $(z=0)$ at $t=620\ \mu$s. The dashed line indicates the
boundary between the source and SOL regions. A blob with mushroom structure is being ejected from the source region and propagating radially outward into the SOL region.}
    \label{fig:moms}
\end{figure}

\begin{figure}[t!]
    \centering
    \includegraphics[width=\textwidth]{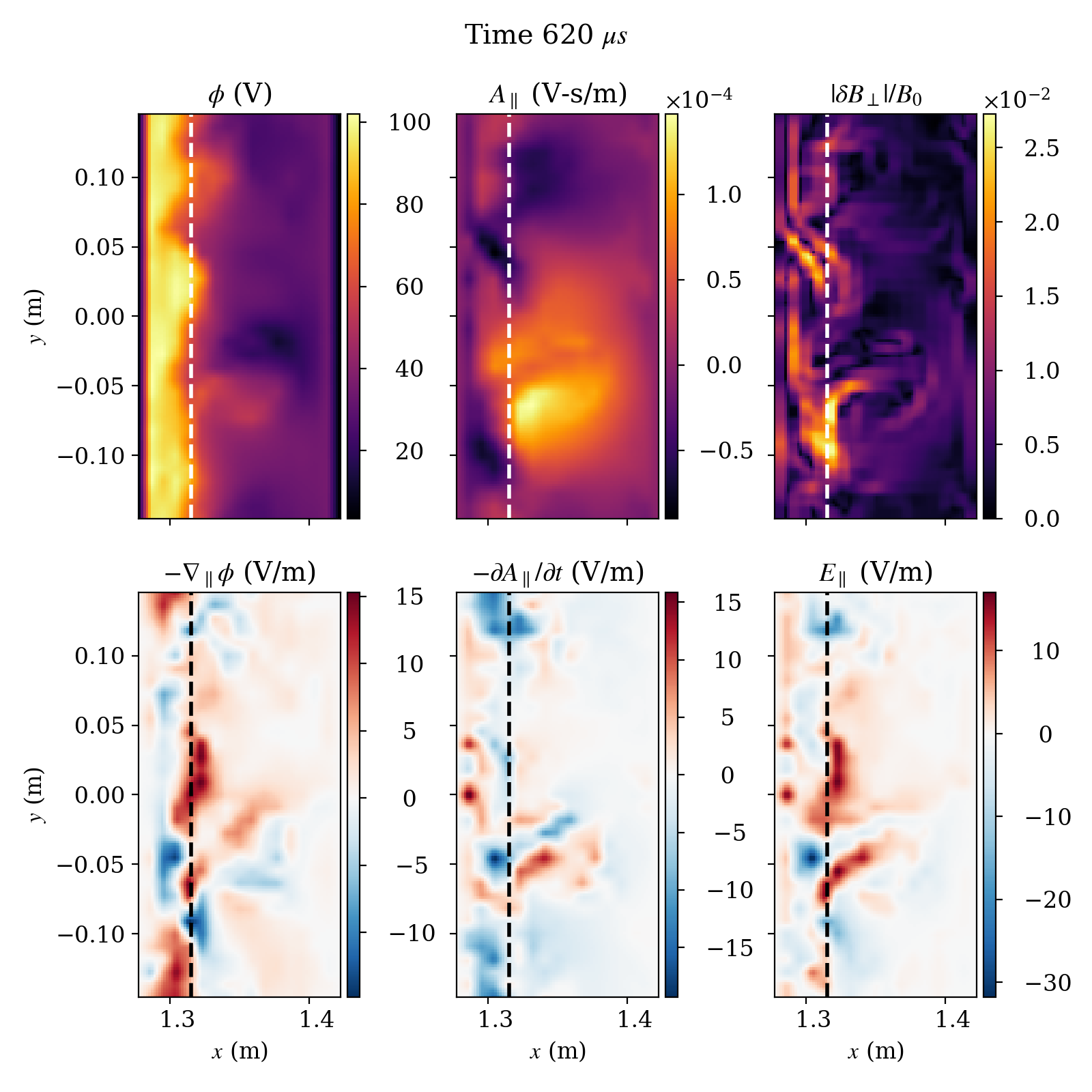}
    \caption[Snapshots of fields from an electromagnetic simulation on open, helical field lines.]{Snapshots (at $z=0$, $t=620\ \mu$s) of the electrostatic potential $\Phi$, parallel magnetic vector potential $A_\parallel$, and normalized magnetic fluctuation amplitude $|\delta B_\perp|/B_0=|\nabla_\perp A_\parallel|/B_0$ (top row), along with the components of the parallel electric field $E_\parallel = -\nabla_\parallel \Phi - \partial{A_\parallel}/\partial{t}$ (bottom row).}
    \label{fig:fields}
\end{figure}

In Figures \ref{fig:bstream-y} and \ref{fig:bstream-x} we show projections of the three-dimensional magnetic field line trajectories. These plots are created by integrating the field line equations for the total (background plus fluctuation) magnetic field. In \cref{fig:bstream-y}, each field line starts at $z=-4$ m and either $x=1.33$ m or $x=1.38$ m for a range of $y$ values and is traced to $z=4$ m. The starting points (at $z=-4$ m ) are marked with circles, while the ending points (at $z=4$ m) are marked with crosses. The trajectories have been projected onto the $x-y$ plane, and we have also plotted the ion density at $z=0$ m in the background. From left to right, we show a short time series of snapshots, with $t=230,\ 240$ and $250\ \mu$s. At $t=230\ \mu$s, a blob is starting to emerge from the source region at $y\approx0.04$ m. The field lines that start at $x=1.33$ m are beginning to be stretched radially outward as the blob emerges. In the $t=240\ \mu$s snapshot, we see that the blob is now propagating radially outward into the SOL region and the $x=1.33$ m field lines have been stretched further. The field lines that start at $x=1.38$ m are now also starting to be stretched near $y\approx0.02$ m, and they are stretched even more in the $t=250\ \mu$s snapshot as the blob continues to propagate. We can also see the remnants of another blob that was ejected near $y=-0.1$ m in previous frames. In the $t=230\ \mu$s snapshot, the field lines have been stretched by this blob, but by $t=250\ \mu$s the field lines in this region have returned closer to their equilibrium position. This behavior of blobs bending and stretching the field lines is an inherently full-$f$ phenomenon. The blobs have a higher density and temperature than the background, so they raise the local plasma $\beta$ as they propagate. This causes the field lines to move with the plasma, allowing the fields lines to be deformed and stretched by the radially propagating blobs and ultimately leading to larger magnetic fluctuations. This behavior has been seen in some electromagnetic Braginskii fluid modeling of SOL blobs \citep{lee2015,lee2015a, hoare2019}, but this is the first time this behavior has been shown with an electromagnetic full-$f$ gyrokinetic model in the SOL. The referenced fluid modeling has also focused on seeded blob simulations, whereas in our simulations the blobs form self-consistently. 

\begin{figure}
    \centering
    \includegraphics[width=\textwidth]{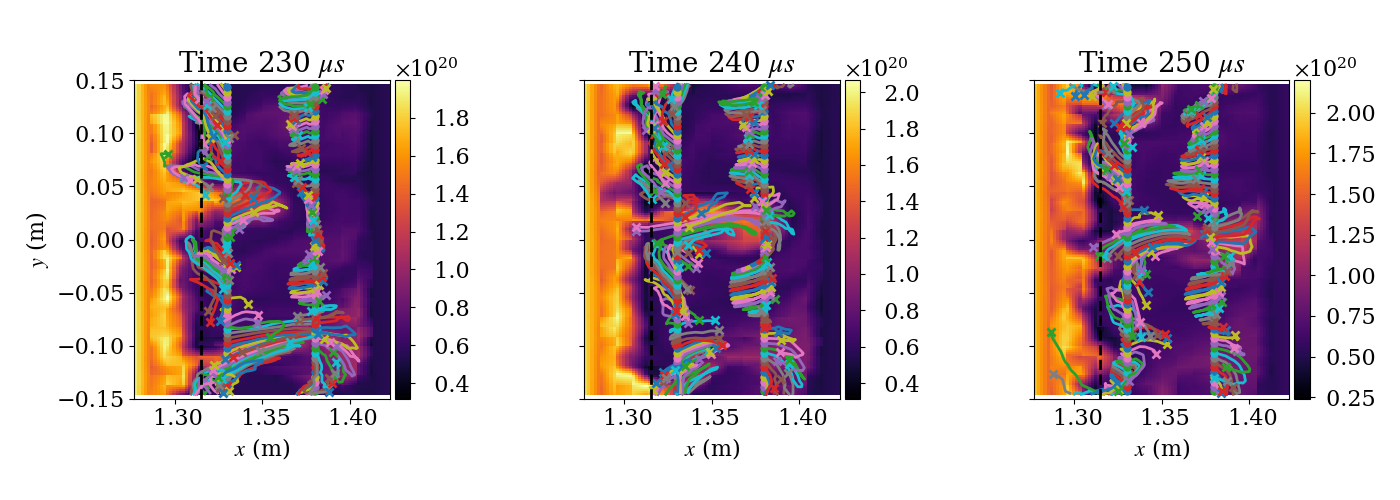}
    \caption[Three-dimensional magnetic field line trajectories for a short timeseries.]{Three-dimensional magnetic field line trajectories at $t=230$, $240$, and $250$~$\mu$s, projected onto the $x-y$ plane.  The ion density at $z=0$~m is plotted in the background.  Each field line starts at $z=-4$~m and either $x=1.33$~m or $x=1$.38~m for a range of $y$ values and is traced to $z=4$~m. The starting points are marked with circles and the ending points are marked with crosses. Focusing on the blob that is being ejected near $y=0$~m, we see that field lines are stretched and bent by the blob as it propagates into the SOL region. In previous frames (not shown) a blob was also ejected near $y=-0.1$~m. At $t=230$~$\mu$s the field lines are still stretched from this event, but they return closer to their equilibrium position by $t=250$~$\mu$s. (A full animation of this timeseries is included in the online supplementary materials of \citet{mandell2020}.)} 
    \label{fig:bstream-y}
    \centering
    \includegraphics[width=.8\textwidth]{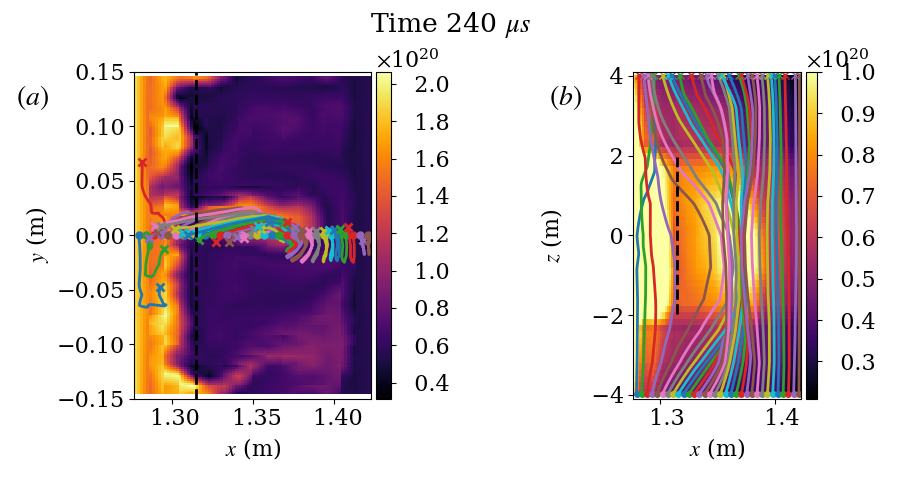}
    \caption[Three-dimensional magnetic field line trajectories in the midplane and along the background field.]{Three-dimensional magnetic field line trajectories at $t=240\ \mu$s, projected onto the $x-y$ plane in $(a)$ and the $x-z$ plane in $(b)$. The ion density is plotted in the background, at $z=0$ m in $(a)$ and averaged over $|y|<0.02$ m in $(b)$. Each field line starts at $y=0$ m and $z=-4$ m for a range of $x$ values and is traced to $z=4$ m. The starting points are marked with circles and the ending points are marked with crosses. Each field line is colored the same in both $(a)$ and $(b)$. The field lines in the near-SOL are stretched radially outward by a blob near $y=0$ m. }
    \label{fig:bstream-x}
\end{figure}
In \cref{fig:bstream-x} we show a slightly different view of the field-line trajectories at $t=240\ \mu$s. Field lines are still traced from the bottom ($z=-4$ m) to the top ($z=4$ m), but now each field line starts at $y=0$ m for a range of $x$. The starting points are again marked with circles and the ending points are marked with crosses. We have projected the three-dimensional trajectories onto the $x-y$ plane in \cref{fig:bstream-x}$(a)$, and onto the $x-z$ plane in \cref{fig:bstream-x}$(b)$. In $(a)$ we again plot the ion density at $z=0$ m in the background; in $(b)$ the ion density has been averaged over $|y|<0.02$ m. As can be seen in \cref{fig:bstream-x}$(b)$, the blob propagating near $y\approx0$ m has stretched several field lines radially outward near the midplane. These bowed-out field lines originate from a range of $x$ values, $1.3\ $m $\lesssim x \lesssim 1.35$ m, and have all been dragged along with the blob as it was ejected from the source region and propagated radially outward. We also see some degree of line-tying in these plots, with many of the field lines ending at a similar point in $x-y$ space to where they began, despite being stretched near the midplane. {The field lines are not perfectly line-tied, however; if they were, the crosses would perfectly align with their corresponding circles in the $x-y$ projections. Because our sheath boundary condition allows current fluctuations at the sheath interface, we can model the finite resistance of the sheath, which makes line-tying only partial \citep{kunkel1966interchange}. This allows the footpoints of the field lines to slip at the sheath interface \citep{ryutov2006}. Examining \cref{fig:bstream-y} and \cref{fig:bstream-x}, we see evidence of this in the simulation, with most of the end points moving slowly and smoothly in the vicinity of their origin, especially at larger $x$. In the source region, however, there are other field lines whose end points suddenly jump further away from their origin. This suggests that we are seeing line breaking 
(reconnection) due to electron inertia effects and numerical diffusion, as field lines are pushed close together by large perturbations in the source region.
}

\begin{figure}
    \centering
    \includegraphics[width=.6\textwidth]{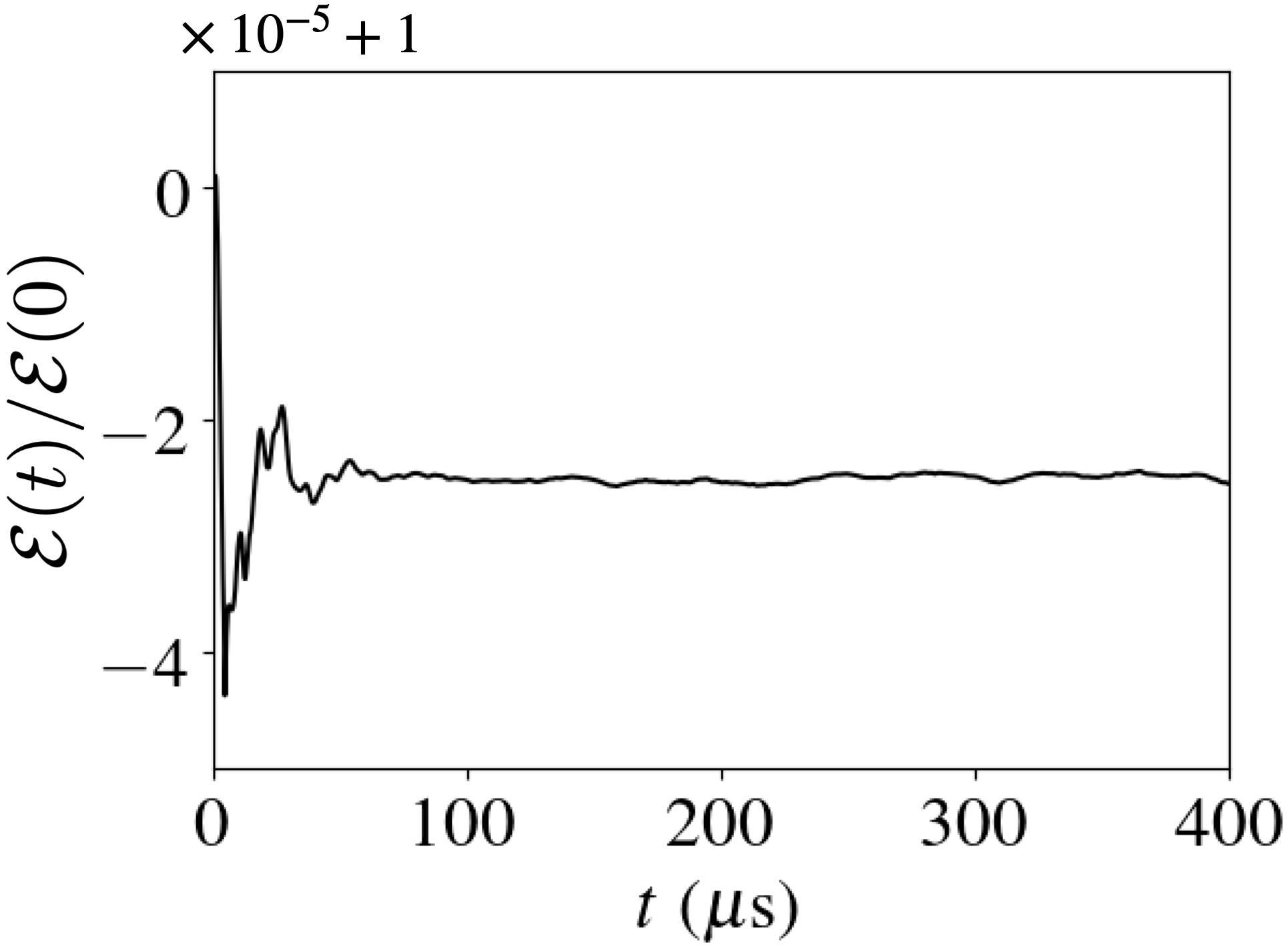}
    \caption{Energy conservation in the electromagnetic NSTX-like simulation, accounting for sources and sinks.}
    \label{fig:em-energy}
\end{figure}


We also show in \cref{fig:em-energy} a time trace of the total energy in the system, accounting for sources and losses to the sheath. 
This is given by
\begin{gather}
    \mathcal{E} = \mathcal{E}_\mathrm{tot} - \mathcal{E}_\mathrm{src} + \mathcal{E}_\mathrm{loss},
\end{gather}
with
\begin{align}
    \mathcal{E}_\mathrm{tot} &= \mathcal{E}_H - \mathcal{E}_E + \mathcal{E}_B \notag \\ &= \sum_s \int_\mathcal{T} \mathcal{J}f_{s\,h} H_{s\,h}\,\dx{^3\vec{R}}\,\dx{^3\vec{v}} -
    \int_\mathcal{T}  \frac{\epsilon_{\perp h}}{2}|\nabla_\perp \dashover{\Phi}\mskip0.01\thinmuskip_h|^2\,\dx{^3\vec{R}} + \int_\mathcal{T} \frac{1}{2\mu_0}|\nabla_\perp A_{\parallel h}|^2 \dx{^3\vec{R}}
\end{align}
and
\begin{gather}
    \mathcal{E}_\mathrm{src} = \sum_s \int \dx{t} \int_\mathcal{T} S_{s\,h} H_{s\,h}\,\dx{^3\vec{R}}\,\dx{^3\vec{v}} \\
    \mathcal{E}_\mathrm{loss} = \sum_s \int \dx{t} \oint_{\partial \mathcal{T}} H_{s\,h}^- \widehat{\mathcal{J} f_{s\,h}} \dot{\vec{R}}_h \cdot \dx{\vec{s}_R} \dx{^3\vec{v}} = \sum_s \int \dx{t} \int H_{s\,h}^- \widehat{\mathcal{J} f_{s\,h}} \uv{b} \cdot \dot{\vec{R}}_h \dx{x}\,\dx{y}\, \dx{^3\vec{v}} \Big|_{z_\mathrm{lower}}^{z_\mathrm{upper}}.
\end{gather}
Energy is preserved to $\sim\mathcal{O}(10^{-5})$, with these finite energy errors likely related to the discrete timestepping scheme \citep{shi-thesis}. 

\subsection{Electrostatic-electromagnetic qualitative comparison}

We have also run a corresponding electrostatic simulation in this configuration for direct comparison. This simulation is identical in configuration to the $L_z=8$ m case from \citet{shi2019} except for the increased particle source rate and lack of cross-species collisions.

An analysis of the blob dynamics in the two cases reveals differences that are supported by theory.  In the electrostatic case, the electron density response is strongly adiabatic. We can see this in \cref{fig:es-adiabatic}, where we break the electron density into adiabatic and non-adiabatic parts. To compute the adiabatic part, we assume that electrons are sufficiently fast to isothermalize along the field line and rapidly communicate the sheath potential upstream, so that parallel force balance becomes
\begin{equation}
    T_e \nabla_\parallel n_e \approx - e n_e E_\parallel = e n_e \nabla_\parallel \Phi.
\end{equation}
The resulting adiabatic density response is given by integrating this equation along the field line subject to the sheath boundary conditions, yielding
\begin{equation}
    n_\mathrm{adiab}(z) = n_\mathrm{sheath} \exp\left[e (\Phi(z) - \Phi_\mathrm{sheath})/T_e\right]. \label{neadiab}
\end{equation}
By subtracting the adiabatic density from the full electron density, we find that non-adiabatic density fluctuations are only of order $1\%$ in the electrostatic case. 
\begin{figure}[t!]
    \centering
    \includegraphics[width=\textwidth]{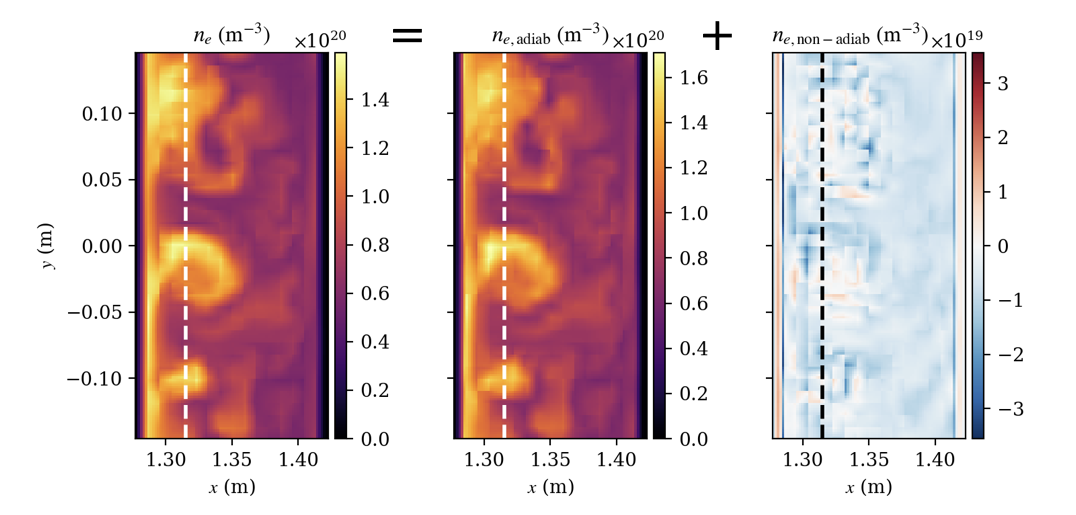}
    \caption[Electrostatic case: adiabatic and non-adiabatic components of electron density.]{Breaking the electron density into adiabatic and non-adiabatic components shows that the electrons are strongly adiabatic in the electrostatic case.}
    \label{fig:es-adiabatic}
    \includegraphics[width=.75\textwidth]{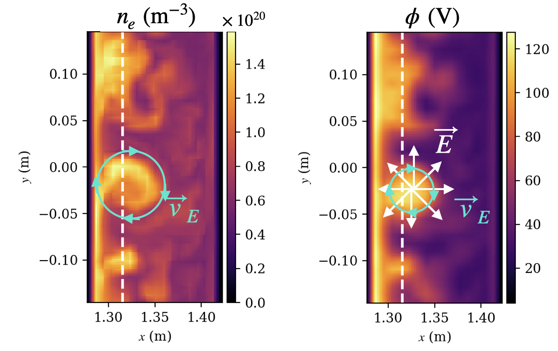}
    \caption[Electrostatic case: Boltzmann spinning of blobs.]{In the electrostatic case, blobs spin due to the Boltzmann spinning effect since the electrons are strongly adiabatic. \citep{angus2012}.}
    \label{fig:es-blob-spin}
\end{figure}
As a result of the strongly adiabatic dynamics, the blobs spin via the Boltzmann spinning effect \citep{angus2012}.  To see the origins of this effect, we rearrange \cref{neadiab} to find the blob potential along the field line, 
\begin{equation}
    \Phi_\mathrm{blob}(z) \approx \Phi_\mathrm{sheath} + (T_e/e) \ln\left( n(z) / n_\mathrm{sheath}\right). \label{boltzmann-blob}
\end{equation}
When the midplane ($z=0$) density is greater than the density at the endplates so that $n(0) > n_\mathrm{sheath}$, a radial (with respect to the blob center) variation in the blob density can give a radial variation in the blob potential via the second term in \cref{boltzmann-blob}. Since the blob density is peaked in the center of the blob, the resulting electric field then points radially outward from the blob center. This produces an $E\times B$ drift that spins the blob about its center, which is what we see in the electrostatic simulation, as shown in \cref{fig:es-blob-spin}. 

\begin{figure}[t!]
    \centering
    \includegraphics[width=\textwidth]{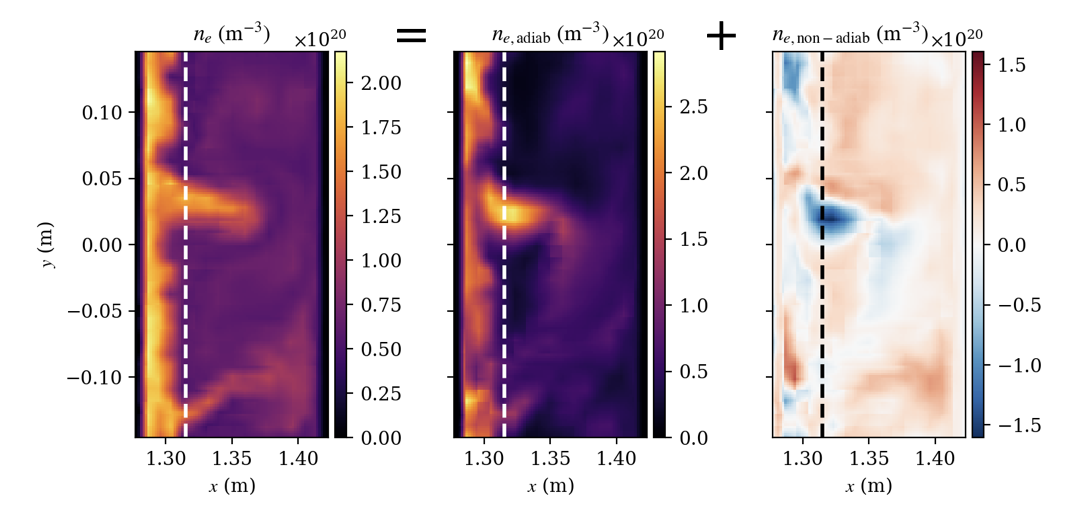}
    \caption[Electromagnetic case: adiabatic and non-adiabatic components of electron density.]{
    The electron density is moderately non-adiabatic in the electromagnetic case (though not so strongly non-adiabatic as to give an MHD-like response with $E_\parallel=0$).}
 \label{fig:em-adiabatic}
    \includegraphics[width=.75\textwidth]{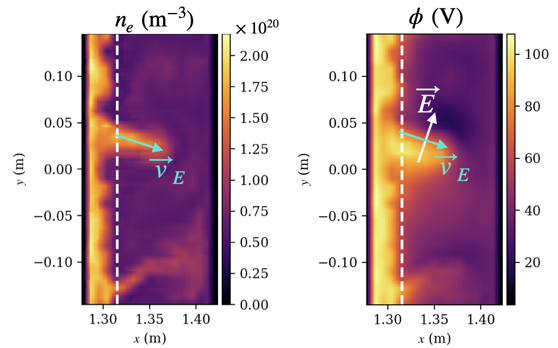}
    \caption[Electromagnetic case: ballistic radial blob propagation.]{In the electromagnetic case, blobs propagate ballistically radially outwards, which suggests electrical disconnection from the sheath \citep{krasheninnikov2008,dippolito2011}.}
    \label{fig:em-blob}
\end{figure}

When we make a similar comparison in the electromagnetic case, we find that the electron density is moderately \emph{non}-adiabatic, as shown in \cref{fig:em-adiabatic}, with adiabatic and non-adiabatic fluctuations on the same order. Note however that the electrons are not so strongly non-adiabatic as to give an MHD-like response with $E_\parallel=0$ (which would require $n_{e,\mathrm{non-adiab}}\gg n_{e,\mathrm{adiab}}$), since $E_\parallel$ is finite as shown in \cref{fig:fields}.

%

Here, the presence of a strong inductive component of $E_\parallel$ indicates that electromagnetic effects are important, so that the propagation speed of waves along the field line (which communicate information about the sheath to the upstream plasma, for example) is limited to the Alv\'en speed, $v_A = v_{te}/ \hat{\beta}^{1/2}$. In this case $\hat{\beta} = (\beta_e/2)m_i/m_e \sim 10$, so the parallel response time, $\tau_A = L_\parallel/v_A$, is about 3 times slower than in the electrostatic case (where the parallel response time is given by the electron transit time, $\tau_e = L_\parallel /v_{te}$).  If the time $\tau_A$ is longer than the time it takes the blob to move more than its width across the field, $\tau_\perp = L_\perp / v_\perp$, the information about the sheath will never reach it, leading to electrical disconnection from the sheath. Thus the blob will move as if the sheath did not exist if $\tau_A \gtrsim \tau_\perp$, or ${\beta} \gtrsim (L_\perp/L_\parallel)^2(c_\mathrm{s}/v_\perp)^2$, where $L_\perp$ is the typical length scale of the potential of the blob, and $v_\perp$ is the blob radial velocity at the midplane  \citep{lee2015a,hoare2019}. This means that the vertical charge polarization in the blob due to the curvature drift cannot be shorted out by the sheath, and the blob moves radially outwards due to the resulting $E\times B$ drift as shown in \cref{fig:em-blob}. The simulation has self-consistently produced the same dipolar potential structure and behavior as shown in \cref{fig:blob-diagram}. Note that there are other effects, including collisional viscosity and magnetic shear, that can cause sheath disconnection apart from electromagnetic effects \citep{myra2006,krasheninnikov2008}, although the resulting blob dynamics is not always the same \citep{dippolito2011}. Even without these other effects, electrostatic blobs could still be sheath-disconnected if the parallel connection length is long enough, so that $\tau_e \gtrsim \tau_\perp$. 

We might expect that we will see significant differences in blob dynamics when comparing electrostatic and electromagnetic simulations (which otherwise have identical parameters) if the blobs are sheath-connected in the electrostatic simulations ($\tau_{ES}\sim\tau_e \gtrsim \tau_\perp$) and sheath-disconnected in the electromagnetic simulations ($\tau_{EM}\sim\tau_A \lesssim \tau_\perp$). Together this gives a condition $\tau_{e} \lesssim \tau_\perp \lesssim \tau_A$ where including electromagnetic effects in the simulation might have the greatest impact on the blob dynamics.


\section{$\beta$ dependence of SOL dynamics}
\label{sec:power-scans}
Motivated by the differences in the electrostatic vs. electromagnetic blob dynamics observed in the previous section, we will now study the effect of $\beta$ on dynamics in our model helical SOL due to electromagnetic effects. In particular, we are interested in varying the Alfv\'en speed, since this can slow the parallel electron dynamics and reduce connectivity with the sheath. Noting that the Alfv\'en speed $v_A = B/\sqrt{\mu_0 n m_i}$ depends on the density $n$ but not the temperature, we will vary $\beta \sim \beta_e = 2\mu_0 n T_e/B^2$ by varying $n$ at constant $T_e$. To do this, we perform a parameter scan of the source particle rate, which roughly controls the density in these flux-driven simulations.

The simulations in this section use the same simplified helical geometry as in \cref{sec:emgk-res}. However, here we no longer use a source floor of one tenth the peak particle source rate. In these simulations, we have found that setting a floor on the \emph{initial} density is sufficient to avoid positivity issues, which seem to be most problematic when an initial burst of blobs propagate into a region of near-zero density. Ensuring that the initial density is finite mitigates this issue. This allows the simulations to run rather robustly without simulation-crashing positivity issues, even without a finite particle source rate in the entire domain.
We also extended the domain 2 cm further radially inward and slightly modified the source profile so that the peak density is more removed from the radial boundary. We again use piecewise-linear $(p=1)$ basis functions; we also slightly increased the resolution in the $z$ direction, so that $(N_x,N_y,N_z,N_{v_\parallel},N_\mu)=(16,32,14,10,5)$. (Recall that these are the number of DG cells in each direction, and that to get the effective number of grid-points for $p=1$ one should double each number.)
Further, we have included electron-ion and ion-electron collisions here, whereas in \cref{sec:emgk-res} only same-species collisions were included.
Note that here, as in the previous section, we have artificially reduced the collision frequency to $10\%$ of its physical value. We do this in part to avoid an expensive timestep restriction from large collisionality (this could be avoided in the future by using an implicit discretization of the collision operator), but also so that we can isolate how electromagnetic effects change with density from collisional effects that also scale with density \citep{myra2006}. In reality, collisional viscosity and magnetic induction compete to slow parallel electron dynamics, with the slowest timescale dominating the behavior \citep{scott1997}. 

The base case for this parameter scan is a case with $L_{z,\mathrm{base}} = 8$ m and $P_{\mathrm{SOL},\mathrm{base}}=5.4$ MW, which is the `nominal' experimental heating power. In the base case, the profiles are mostly unchanged when electromagnetic effects are included, as can be seen in the top row of \cref{fig:power-scan-profiles}. We then scan the source particle rate by taking $P_\mathrm{SOL} = \hat{n} P_\mathrm{SOL,base}$ with $\hat{n} = \{1, 2, 3.5, 5,10\}$ at constant source temperature. \cref{fig:source-diagram} shows the profiles in the $x-z$ plane of the source particle rate and the source temperature for the base case, along with the boundary conditions (which are the same as in \cref{sec:emgk-res}).

\begin{figure}[t!]
    \centering
    \includegraphics[width=.5\textwidth]{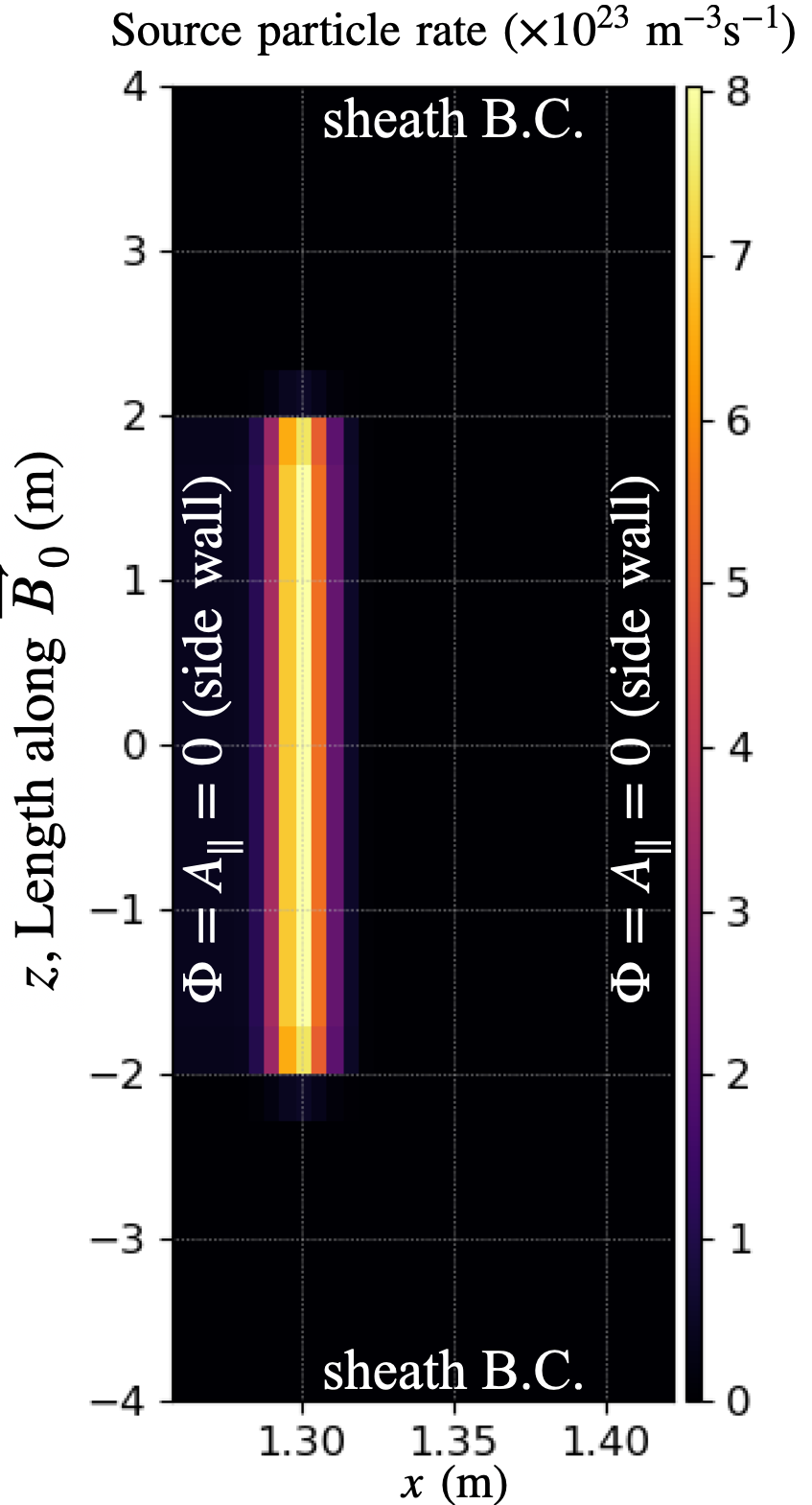}
    \caption[Diagram of the source particle rate and temperature for the power scan simulations.]{Diagram of the source particle rate for the base case ($\hat{n}=1$). The source temperature is  70 eV. For other cases the source profiles are identical, except for a scaling factor $\hat{n}$ on the source particle rate. Also shown in the left plot are the boundary conditions: Dirichlet at the radial boundaries with $\Phi=A_\parallel=0$, and sheath boundary conditions along the field line. Not pictured are periodic boundary conditions in the $y$ direction (out-of-plane here).}
    \label{fig:source-diagram}
\end{figure}

\subsection{Midplane radial profiles and gradients}

In \cref{fig:power-scan-profiles} we see time- and $y$-averaged midplane profiles of density, temperature, and $\beta$ for electrostatic and electromagnetic cases with $\hat{n} = \{1,3.5,10\}$. Electron quantities are shown with solid lines, while dashed quantities are ion guiding-center moments. The midplane density scales with source particle rate scaling factor $\hat{n}$ while the temperature does not, as one would expect. In all cases, we see that $T_i/T_e \sim 2$, which is consistent with experimental results showing $1\lesssim T_i/T_e \lesssim 4$ in the SOL \citep{kocan2011}. As $\beta \sim \hat{n}$ increases we see more differences in the profiles between the electromagnetic and electrostatic cases. At higher $\hat{n}$, electromagnetic effects seem to make the profiles steeper in the source region (shaded) and flatter in the non-source region, which we will denote as the ``SOL region''. Although the profiles in the source region are likely influenced by the form of the sources, the sources are the same in all cases (except for the scaling factor $\hat{n}$). This means that \emph{differences} in the profiles in the source region are still physical, even if the profiles themselves are not due to sensitivity to the source parameters.
\begin{figure}[t!]
    \includegraphics[width=1\textwidth]{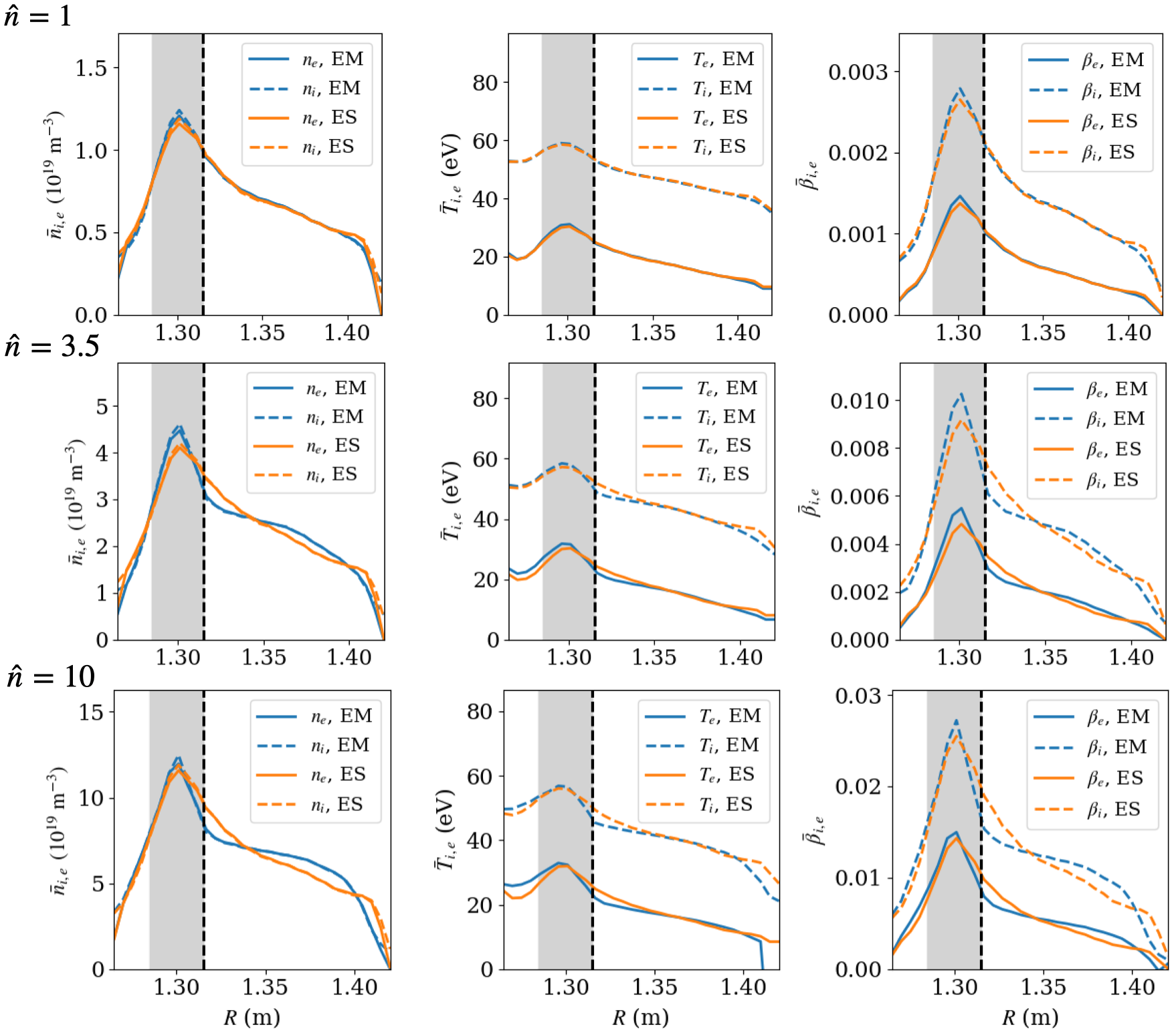}
    \caption[Midplane density, temperature, and $\beta$ profiles for a source power scan.]{Time- and $y-$averaged midplane profiles of density (left), temperature (middle), and $\beta$ (right) for electrostatic (ES) and electromagnetic (EM) cases with source scaling factor $\hat{n} = \{1,3.5,10\}$. Electron quantities are shown with solid lines, while dashed quantities are ion guiding-center moments. As $\hat{n}$ increases, electromagnetic effects cause steepening of gradients in the source region (shaded) and flattening in the SOL region.}
    \label{fig:power-scan-profiles}
\end{figure}

\begin{figure}[t!]
    \includegraphics[width=1\textwidth]{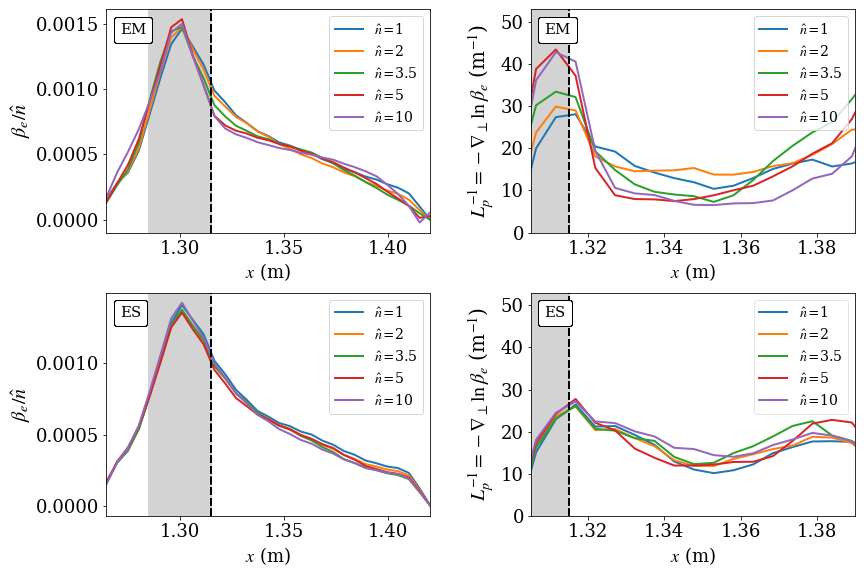}
    \caption[Comparison of electron $\beta$ profiles from electromagnetic and electrostatic cases in power scan.]{On the left, electron $\beta$ profiles normalized to power scaling factor $\hat{n}$ for electromagnetic (EM, top) and electrostatic (ES, bottom) cases. On the right, inverse pressure gradient scale lengths for each case, $L_p^{-1} =- \nabla_\perp \ln \beta_e$. While profiles and gradients change very little with $\hat{n}$ in the electrostatic cases, the electromagnetic cases show steepening of gradients in the source region (shaded) and flattening in the SOL region as $\hat{n}\sim \beta$ increases.}
    \label{fig:power-scan-beta-Lp-profiles}
\end{figure}

\begin{figure}[t!]
    \centering
    \includegraphics[width=.7\textwidth]{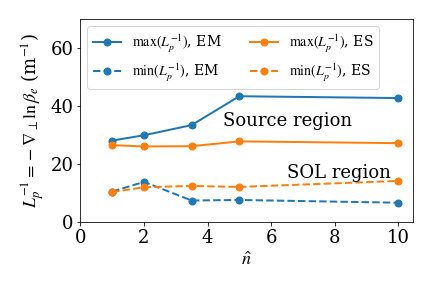}
    \caption[Maximum and minimum pressure gradient values as a function of source power scaling factor for the electromagnetic and electrostatic cases.]{Maximum (solid lines) and minimum (dashed) values of $L_p^{-1}$ as a function of $\hat{n}$ for the electromagnetic (EM) and electrostatic (ES) cases. The maximum gradients are found in the source region, while the minimum gradients are found in the SOL region. In each region, the electrostatic gradients are roughly constant as $\hat{n}$ increases. In the electromagnetic cases, the source-region gradients increase with $\hat{n}$ while the SOL gradients decrease.}
    \label{fig:power-scan-minmax-Lp}
\end{figure}

Focusing on $\beta$, on the left side of  \cref{fig:power-scan-beta-Lp-profiles} we show the electron $\beta$ profiles for the entire $\hat{n}$ parameter scan, this time normalized to $\hat{n}$. In the electromagnetic cases (top left), we can see again that as $\hat{n}$ increases, the gradients steepen in the source region, and flatten in the SOL region. In the electrostatic cases (bottom left), there is little change in the profiles as $\hat{n}$ increases. Since the collisionality is the only parameter changing with $\hat{n}$ in the electrostatic cases, this indicates that collisions are not playing a major role in changing the dynamics (at least at the 10\% reduced collisionality that we use here). Thus we are effectively isolating changes in dynamics due to electromagnetic effects as we scan $\hat{n}$. Even if the profiles look relatively similar, the changes in the gradients in the electromagnetic cases are still significant. On the right side of \cref{fig:power-scan-beta-Lp-profiles} we compute the inverse pressure gradient scale length, $L_p^{-1} = -\nabla_\perp \ln \beta_e$. Since larger values of $L_p^{-1}$ indicate steeper gradients, we can again see that increasing $\hat{n}$ gives steeper gradients in the source region, but only in the electromagnetic cases. Plotting the maximum gradient values in \cref{fig:power-scan-minmax-Lp}, we see that the  gradients in the source region increase by about $60\%$ over the electromagnetic $\hat{n}$ scan, while there is no change in the electrostatic cases. In the SOL region the gradients decrease with $\hat{n}$ in the electromagnetic cases; after plotting the minimum values of $L_n^{-1}$ in \cref{fig:power-scan-minmax-Lp}, we see that the SOL gradients fall by about $50\%$ over the scan. In the $\hat{n}=10$ case, the ratio between the gradients in the source region and the SOL region is $\sim 6$, while in the electrostatic case the ratio is only $\sim 2$.
A decrease in pressure gradient with increasing $\beta$ is consistent with the results of \citet{halpern2013}, which showed (in a circular-flux-surface geometry with a limiter on the high-field side) that there is transition between resistive and ideal ballooning modes at some critical $\beta_e$ that leads to flattening due to increased transport. 

We should note that experimental SOL profiles on NSTX are much steeper, falling off to near zero within a few centimeters of the last-closed-flux-surface. There are many effects that we are not currently modeling that could reduce transport and make the profiles steeper, including using the magnetic geometry from the experiment with magnetic shear and an X-point. This is left to future work (with some preliminary results including magnetic shear shown in \cref{sec:heli-geo}), and so for now we do not expect agreement between our profiles and the experiment. Nonetheless, we can still investigate interesting physical aspects of the simulations and the influence of electromagnetic effects on the dynamics.

\subsection{Interchange instability and $E\times B$ shear stabilization}

All of these cases are unstable to the interchange mode due to (constant) unfavorable curvature in our helical magnetic geometry. The ideal interchange growth rate is $\gamma_\mathrm{int} = \sqrt{2}c_s/\sqrt{R L_p}$, and the modes are constant along the field line so that  $k_\parallel=0$. This is analogous to the Rayleigh-Taylor instability in fluid dynamics, with unfavorable magnetic curvature giving an effective gravity $g_\mathrm{eff}=2 c_s^2/R$. 
On open field lines that end on conducting plates, true $k_\parallel=0$ ideal interchange modes are not possible because this would imply $\Phi=$ const everywhere (since $\Phi=$ const on the plates). One way to restore interchange dynamics is to consider sheath effects, which allow jumps in the potential near the ends so that we can have $k_\parallel\sim0$ in the interior with a finite electric field. The interchange growth rate can be  reduced at low $k_\perp \rho$ due to sheath boundary conditions when the current to the sheath is large, but this does not change the stability threshold \citep{myra1997}; for a nice derivation of this effect, see
\citet{shi-thesis}.  In \cref{fig:interchange}, we show the effective ideal interchange growth rate, $\gamma_\mathrm{int,eff}=\max ( \sqrt{2}c_s/\sqrt{R L_p} )$, for each electromagnetic and electrostatic case, computed using the maximum value of $\gamma_\mathrm{int}$ in the source region in each case. This does not account for stabilization from sheath-connection or possible electromagnetic effects. The effective interchange growth rate increases by about 20\% with increasing $\hat{n}$ in the electromagnetic cases, and stays relatively constant in the electrostatic cases. The fact that the effective ideal growth rate increases with $\hat{n}$ suggests that there is some stabilizing effect due to electromagnetic effects that allows the gradients to steepen.

It is well known that the interchange mode can also be stabilized by shear in the velocity of plasma flows. A recent study by \citet{zhang2020} that uses a constant effective gravity (like our constant curvature in helical geometry) and a pedestal-like density profile that has radial variation in both the density and its gradient appears particularly relevant to our results. In that work it was found that short wavelength interchange modes are very efficiently stabilized by $E\times B$ shear if the shearing rate $\omega_{E\times B}=v_E'$ is comparable to (but not necessarily larger than) the interchange growth rate, with significant stabilization at $\omega_{E\times B}/\gamma_\mathrm{int}\sim 0.4$. Recent work by \citet{goldston2020} has suggested that this stabilization effect could have important implications for the trigger for pedestal formation in the L-H transition.

\begin{figure}[t!]
    \centering
    \includegraphics[width=.7\textwidth]{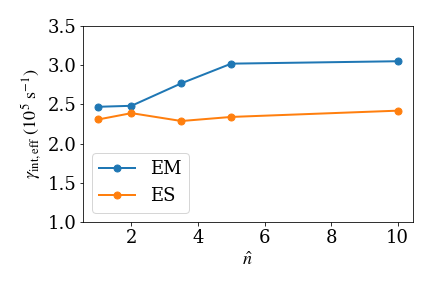}
    \caption[Effective ideal interchange growth rates for electromagnetic and electrostatic cases.]{Effective ideal interchange growth rates for the electromagnetic (EM) and electrostatic (ES) cases as a function of $\hat{n}$. The growth rate has been computed using the maximum value of $L_p^{-1}$ in the source region for each case. In the electromagnetic cases, the growth rate increases by about 20\% with increasing $\hat{n}$, while in the electrostatic cases the growth rate stays roughly constant with $\hat{n}$. These effective growth rates do not account for stabilization from sheath-connection or possible electromagnetic effects. }
    \label{fig:interchange}
    
    \includegraphics[width=\textwidth]{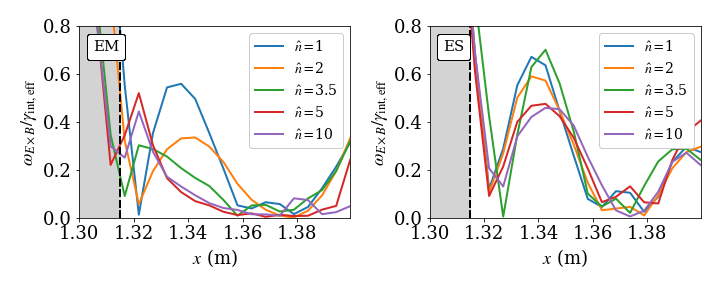}
    \caption[Ratios of the $E\times B$ shearing rate to the effective interchange growth rate.]{Ratio of the average $E\times B$ shearing rate (which varies radially), $\omega_{E\times B}$, to the effective interchange growth rate, $\gamma_\mathrm{int, eff}$. }
    \label{fig:exb-shear}
\end{figure}
In \cref{fig:exb-shear}, we show the ratio of the  average $E\times B$ shearing rate (which varies radially) to the growth rate, $\omega_{E\times B}/\gamma_\mathrm{int,max}$ for the electromagnetic and electrostatic cases. In the electrostatic cases (right), the ratio peaks in all cases near $x=1.34$ m, with the ratio decreasing somewhat with $\hat{n}$. In the electromagnetic cases, the peak in the ratio shifts radially inward in the higher $\hat{n}$ cases, so that the peak is just outside the source region near $x=1.32$ m in these cases. This is also the radial location where the gradients begin to steepen in the high-$\hat{n}$ electromagnetic cases, as shown in \cref{fig:power-scan-beta-Lp-profiles}. Thus it is plausible that elevated $E\times B$ shear just outside the source region is producing stabilization of the interchange mode. The gradients are then able to steepen until the interchange mode is destabilized again with a growth rate large enough to overcome the $E\times B$ shearing. Determining why the peak in the shearing rate moves radially inward in the electromagnetic cases is an important issue requiring further investigation.
Another possibility is that electromagnetic effects result in a change in the mode structure that allows steeper gradients. Since typically $\phi \sim T_e$, an increase in $T_e''$ that results from a steepening gradient could result in increased $E\times B$ shear. These two proposed mechanisms can also form a feedback loop, so it can be difficult to establish causality without identifying the initial trigger for the loop. Nonetheless, it is clear that electromagnetic effects are playing a key role since this behavior is not seen in the electrostatic cases. This mechanism has potential importance for pedestal formation and the L-H transition. 

\subsection{Destabilization of ballooning-type modes}

In the electromagnetic cases, ballooning-type modes with finite $k_\parallel$ can be destabilized as $\beta$ (or more precisely, the gradient of $\beta$) increases. In the core, the ideal ballooning stability parameter is typically defined as $\alpha = -q^2 R \nabla_\perp \beta = q^2 R \beta/L_p $, where $q$ is the safety factor, $R$ is the major radius, and $\beta = \beta_i + \beta_e$ is total plasma $\beta$. From the simplified ideal MHD ballooning mode equation in circular geometry \citep{coppi1977,connor1978, freidberg2014}, we have
\begin{equation}
    \pderiv{}{\theta}\left[ (1+\Lambda^2)\pderiv{X}{\theta}\right] + \alpha\left[ \hat{\omega}^2 \left(1 +\Lambda^2\right) + (\Lambda \sin\theta + \cos \theta)\right]X = 0, \label{ballooning}
\end{equation}
where $X=X(\theta)$ is the eigenfunction with $\theta$ is the ballooning angle, and $\Lambda = \hat{s}\theta - \alpha \sin\theta$ with $\hat{s}=(r/q)\dx{q}/\dx{r}$ the magnetic shear. From this one can obtain the complex frequency of the ballooning mode, $\hat{\omega} \equiv \omega/\gamma_\mathrm{int}$. In the core, ballooning is the result of unfavorable magnetic curvature on the outboard (low-field) side of the tokamak and favorable curvature on the inboard (high-field) side. This means the mode is most unstable on the outboard side, resulting in eigenmodes that peak (or balloon) on the outboard side. For circular flux surfaces, this variation in the curvature (and hence the ballooning drive) is given by the sinusoidal terms in \cref{ballooning}, with $\theta=0$ the outboard side where $\cos\theta$ is maximized. 

In our simple helical geometry, we neglect magnetic shear and Shafranov shift ($\Lambda =0$). We also have no favorable curvature, so that we take $\theta=0$ in the curvature term; our entire domain is effectively on the outboard side. Transforming the ballooning coordinate to be the length along the field line via $z=qR \theta$,
the result is the simple equation
\begin{equation}
    \pderiv{^2 X}{z^2} + \frac{\alpha}{q^2 R^2}(\hat{\omega}^2+1)X = 0.
\end{equation}
If we have constant curvature and no favorable curvature region, why should we have ballooning? The answer lies in the boundary condition along the field line. We have open field lines that end on conducting plates, resulting in line-tying. This means the footpoints of the field lines stay relatively fixed, while near the midplane the field lines are free to bend and bow with the plasma, as we saw in \cref{fig:bstream-y,fig:bstream-x}.  The result is the line-tied ballooning mode \citep{cowley1985,cowley1997,zhu2006}.

As a simple way to account for line-tying, we can constrain the eigenmode to vanish at the ends of our finite domain at $z=\pm L_z/2$. This means that the eigenmode is constrained to be a Fourier mode with wavenumber $k_\parallel = \ell\pi/L_z$, with $\ell$ some integer, so that we have  
\begin{equation}
    k_\parallel^2 = \frac{\alpha}{q^2 R^2}(\hat{\omega}^2+1)
\end{equation}
To find the critical value of $\alpha$ for instability, we take $\ell=1$ to get the lowest $k_\parallel$ eigenmode that has a zero crossing at the ends of the domain. This gives
\begin{equation}
    \frac{\alpha}{q^2 R^2}(\hat{\omega}^2+1) = \left(\frac{\pi}{L_z}\right)^2.
\end{equation}
Defining a new $\alpha$-like parameter for our helical SMT geometry,
\begin{equation}
 \alpha^\mathrm{SMT} \equiv \frac{L_z^2}{\pi^2}\frac{\alpha}{q^2 R^2}  = \frac{L_z^2}{\pi^2}\frac{\beta_e+\beta_i}{R L_p},
\end{equation}
the ideal ballooning instability growth rate is
\begin{equation}
    \gamma_\mathrm{bal} = \gamma_\mathrm{int}\sqrt{1 - \frac{1}{\alpha^\mathrm{SMT}}}. \label{gamma-bal}
\end{equation}
This gives an instability threshold of $\alpha^\mathrm{SMT}\gtrsim 1$. The growth rate is below the ideal interchange growth rate for all $\alpha^\mathrm{SMT}$, approaching $\gamma_\mathrm{int}$ from below for $\alpha^\mathrm{SMT} \gg 1$. \citet{halpern2013} showed a similar calculation for circular flux surfaces, and also showed that the threshold can be lowered due to non-ideal effects not captured in the simple derivation above. Additional sheath-related modifications could also be required, similar to the sheath-modified interchange mode \citep{myra1997}.

\begin{figure}[t!]
    \centering
    \includegraphics[width=.7\textwidth]{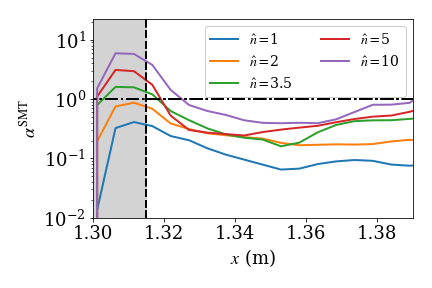}
    \caption[Ballooning stability parameter in electromagnetic cases.]{Ballooning stability parameter $\alpha^\mathrm{SMT} = (L_z/\pi)^2(\beta_e+\beta_i)/(R L_p)$ for each of the electromagnetic cases. The ballooning stability threshold $\alpha^\mathrm{SMT}=1$ is shown by the horizontal dash-dotted line. All but the $\hat{n}=1$ case are near or above the instability threshold in the source region, indicating that ballooning-type modes are destabilized.}
    \label{fig:power-scan-alpha}
\end{figure}

\begin{figure}[t!]
    \centering
    \includegraphics[width=\textwidth]{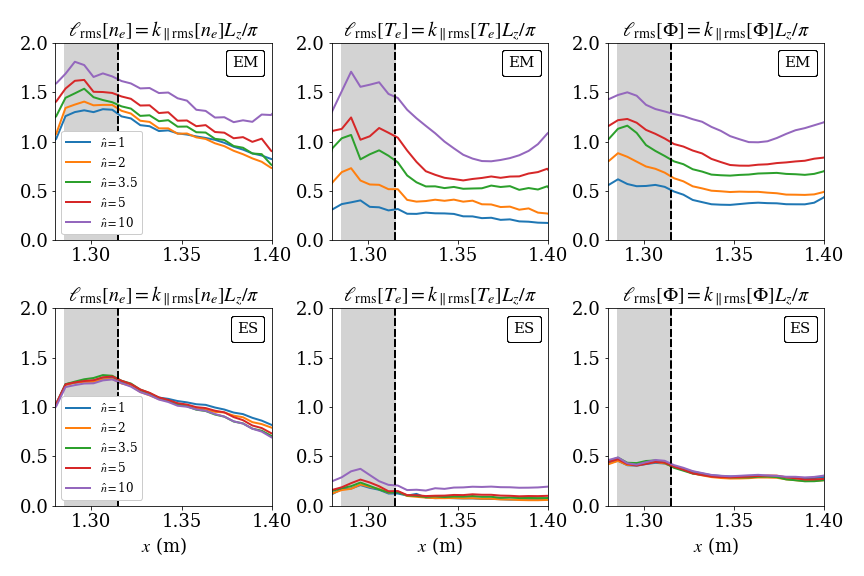}
    \caption[$k_{\parallel \mathrm{rms}}$ computed for electromagnetic and electrostatic cases.]{Radial profiles of $\ell_\mathrm{rms}=k_{\parallel \mathrm{rms}}L_z/\pi$, computed using fluctuations in the electron density (left), electron temperature (middle), and potential (right). In the electromagnetic cases (top row), $k_{\parallel \mathrm{rms}}$ peaks in the source region and increases with $\hat{n}$, consistent with ballooning-type modes becoming destabilized in the source region. In the electrostatic cases $k_{\parallel \mathrm{rms}}$ does not change with $\hat{n}$, and stays at or below the levels of the $\hat{n}=1$ electromagnetic case, with very low levels for the temperature and potential fluctuations indicating that interchange modes $(k_\parallel=0)$ are dominant.}
    \label{fig:power-scan-kpar-rms}
\end{figure}

In \cref{fig:power-scan-alpha} we plot radial profiles of the $\alpha^\mathrm{SMT}$ parameter for each of the electromagnetic cases. In the source region, all but the $\hat{n}=1$ case are near or above the $\alpha^\mathrm{SMT}\gtrsim 1$ threshold for ballooning instability, indicated in the plot with a dot-dashed line. This means that ballooning-type modes with finite $k_\parallel$ are destabilized in these cases.

We can also compute a measure of the root-mean-square (RMS) $k_\parallel$ in the fluctuations of field $f$ as
\begin{equation}
    \ell_\mathrm{rms}[f] = \frac{L_z}{\pi}k_{\parallel \mathrm{rms}}[f]= \frac{L_z}{\pi}\frac{1}{\tilde{f}_\mathrm{rms}}\left(\pderiv{\tilde{f}}{z}\right)_\mathrm{rms}
\end{equation}
In \cref{fig:power-scan-kpar-rms}, we compute radial profiles of this quantity for electron density, electron temperature, and potential fluctuations (\emph{i.e.}, $f=n_e,T_e,\Phi$). The top row shows the electromagnetic simulations and the bottom row shows the electrostatic ones. The trend is most noticeable in the temperature and potential fluctuations, with $k_{\parallel\mathrm{rms}}$ peaking in the source region and increasing with $\hat{n}$ in the electromagnetic cases, consistent with ballooning-type modes becoming destabilized in the source region. In the electrostatic cases $k_{\parallel\mathrm{rms}}$ stays at or below the levels of the $\hat{n}=1$ electromagnetic case for all $\hat{n}$, indicating that the transition to ballooning-type modes is a purely electromagnetic effect. While we expect finite $k_\parallel$ for the density fluctuations even in the electrostatic cases due to the parallel variations in the background density (and its gradient), the temperature and potential fluctuations show $k_\parallel \sim 0$ interchange modes are dominant in the electrostatic cases.

\subsection{Particle balance and transport}
The profiles in the SOL are set by a balance between the sources, cross-field (perpendicular) transport, and parallel transport, including parallel end losses to the walls. That is, in quasi-steady state, we have
\begin{equation}
    \nabla\cdot\vec{\Gamma} = \nabla_\perp\cdot \vec{\Gamma}_\perp + \nabla_\parallel \Gamma_\parallel = S,
\end{equation}
where $\vec{\Gamma}$ is the particle flux with perpendicular and parallel components $\vec{\Gamma}_\perp$ and $\Gamma_\parallel$, and $S$ is the particle source.
Since our numerical scheme conserves particles (and energy) both locally and globally, we are able to examine this particle balance and its consequences carefully. Recall that our radial boundary conditions are such that particles cannot leave through the side walls\footnote{In tokamak experiments there can be net particle and heat fluxes to the first wall, which can be concerning for large filaments and ELMs. Apart from large heat loads, this can also lead to main-chamber recycling that can degrade performance. This should be included in future models.}, so all losses are at the sheath entrances. Without cross-field transport upstream, the parallel fluxes to the endplates would have the same narrow footprint as the source in the simulations. Consequently, the widening of the footprint effectively gives the end result of the competition between upstream parallel and cross-field transport. In the following two sections we will examine the cross-field (perpendicular) and parallel particle transport.

\subsubsection{Cross-field (perpendicular) particle transport}


We compute the time- and $y$-averaged midplane profiles of cross-field (perpendicular, with respect to the background magnetic field) particle flux for the electromagnetic cases in  \cref{fig:power-scan-pflux}($a$), normalized in each case by $\hat{n}$. This is defined as
\begin{gather}
    \Gamma_{\perp e} = \langle \tilde{n}_e \tilde{v}_r \rangle + \langle \widetilde{n_e u_{\parallel e}} b_r \rangle
\end{gather}
where the first term is the contribution from the $E\times B$ drift, with $v_r = E_r/B = -(1/B)\pderivInline{\Phi}{y}$, and the second term is the flux due to magnetic flutter, with $b_r = (1/B)\pderivInline{A_\parallel}{y}$. The tilde indicates the fluctuation of a time-varying quantity, defined as $\tilde{A}=A - \bar{A}$ with $\bar{A}$ the time average of $A$. The brackets $\langle A \rangle$ denote an average in $y$ and time. The radial particle flux at the midplane scales linearly with source power, with very little change in the radial profile after scaling by $\hat{n}$. We would also see little difference if we directly compared the radial particle flux profiles between each electrostatic and electromagnetic case. 
Given the differences in the profiles and gradients we saw in \cref{fig:power-scan-profiles} and \cref{fig:power-scan-beta-Lp-profiles}, it is perhaps somewhat surprising that as $\hat{n}$ varies there are no differences in the profiles of radial particle flux at the midplane. In the core, where there is a clear scale separation in length scales between background and fluctuations, a linear flux-gradient parametrization of the transport in terms of an effective diffusivity $D_\perp$ and effective convective velocity $V_\perp$ can be justified, resulting in
\begin{equation}
    \Gamma_\perp = n V_\perp - D_\perp \nabla_\perp n.
\end{equation}
From this, one might expect that if gradients increase at constant flux then the diffusion coefficient must decrease (due to a mode transition, for example), and vice versa. In principle, the mode transition from interchange to ballooning that we observed in the previous section could result in a change in the diffusion coefficient.
However, in the edge/SOL we do not have the  scale separation required for this simple transport characterization, resulting in non-diffusive transport with large fluctuations and significant intermittency \citep{naulin2007}. 

\begin{figure}[t!]
    \centering
    \includegraphics[width=\textwidth]{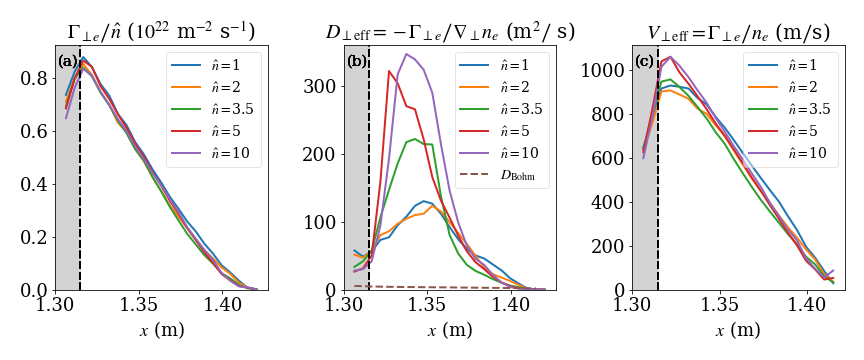}
    \caption[Electron radial particle fluxes, along with effective diffusion and convective velocity coefficients, at the midplane for the electromagnetic cases in the power scan.]{ Electron radial particle flux (normalized to $\hat{n}$) $(a)$, effective diffusion coefficient $(b)$, and effective convective velocity $(c)$ at the midplane for the electromagnetic cases. This includes transport contributions from $E\times B$ drifts and magnetic flutter. After normalizing the particle fluxes to $\hat{n}$ there is very little change in the radial profiles, so the particle flux scales linearly with $\hat{n}$. The effective diffusion coefficient increases by an order of magnitude from the source region to the SOL region. Also shown is the Bohm diffusion coefficient $D_\mathrm{Bohm}=\frac{1}{16}\frac{T_e}{e B}\sim 5$ m$^2/$s. The peak in the effective convective velocity just outside the source region is slightly higher in the higher $\hat{n}$ cases, which could indicate more convective blob transport due to sheath disconnection. Overall, the strong radial variation in these effective transport coefficients suggests that a simple linear flux-gradient relationship is not adequate.} 
    \label{fig:power-scan-pflux}
\end{figure}

\begin{figure}[t!]
    \centering
    \includegraphics[width=.8\textwidth]{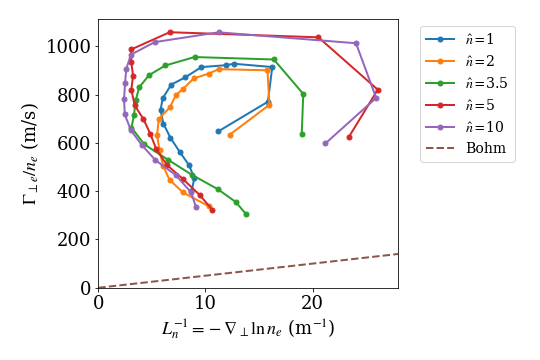}
    \caption[Flux-gradient relationships for the electromagnetic cases in the power scan.]{Flux-gradient relationships for the electromagnetic cases, computed using the time-averaged radial profiles of density and particle flux evaluated at each radial grid point. In each case, the starting point (around 600 on the $y$-axis) represents data taken at $x=1.3$ m, with $x$ increasing moving along each line until the endpoints (below 400 on the $y$-axis) representing data at $x=1.39$ m. The flux-gradient relationship is clearly nonlinear, indicating that the transport cannot be simply characterized in terms of an effective diffusion coefficient and a convective velocity. For reference, the effective flux from Bohm diffusion is also shown (brown dashed).}
    \label{fig:flux-gradient}
\end{figure}
In \cref{fig:power-scan-pflux}$(b)$ and $(c)$ we compute the effective flux-gradient parametrization parameters via
\begin{gather}
    D_{\perp \mathrm{eff}} = - \Gamma_{\perp e}/\nabla_\perp n_e \\
    V_{\perp \mathrm{eff}} =  \Gamma_{\perp e}/n_e.
\end{gather}
The large radial variation of these quantities suggests that the transport is inherently non-local, so that the transport is not determined by local background gradients but induced by propagating coherent structures \citep{xu2010}. We also plot $\Gamma_\perp/n$ versus $\nabla_\perp \ln n$ in \cref{fig:flux-gradient}, with the data taken from each radial point in the profiles. If we had diffusive transport so that the flux-gradient relationship was linear, one would expect that as the gradient increases the flux should also increase, and one could evaluate the coefficients $V_\perp$ and $D_\perp$ based on a linear fit. We see no such linear relationship, which is further indication that the transport is non-diffusive and non-local.

\begin{figure}[t!]
    \centering
    \includegraphics[width=\textwidth]{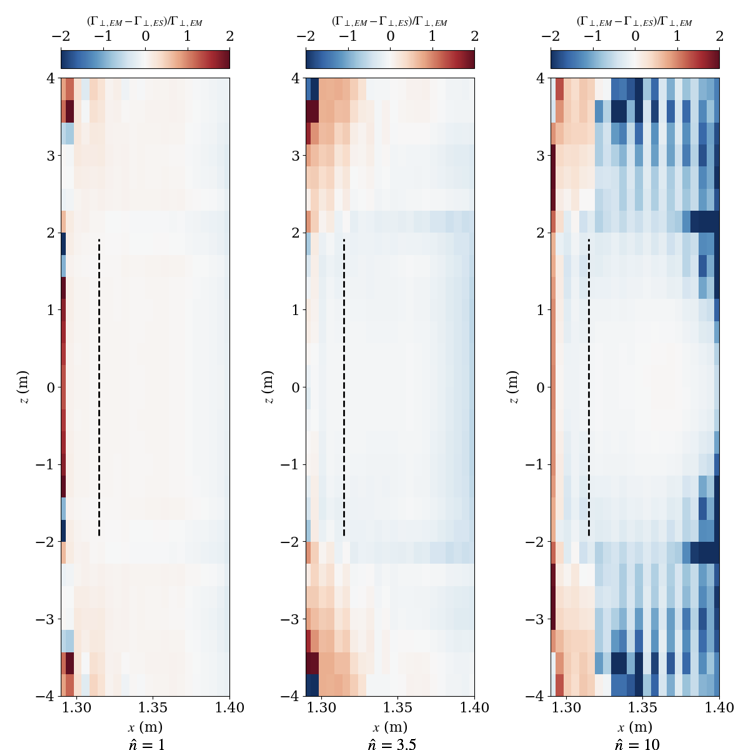}
    \caption[Difference in perpendicular particle transport between electromagnetic and electrostatic cases in the $x-z$ plane.]{Difference in perpendicular particle transport between corresponding electromagnetic and electrostatic cases in the $x-z$ plane for $\hat{n}=1,3.5,10$ cases, normalized to the flux from the electromagnetic case. Blue regions indicate $\Gamma_{\perp, ES}>\Gamma_{\perp, EM}$, while red regions indicate $\Gamma_{\perp ,ES}<\Gamma_{\perp, EM}$. Near the midplane $(z=0)$ the transport is roughly the same in the electromagnetic and electrostatic cases, consistent with the results of \cref{fig:power-scan-pflux}. However, off-midplane there is some reduction in the perpendicular transport in the SOL region in the higher $\hat{n}$ cases. There is also a small region near the endplates in the source region where the transport is larger in the electromagnetic cases. The black dotted line indicates the boundary of the source region.}
    \label{fig:flux-diff-xz}
\end{figure}

To better understand differences in perpendicular particle transport between the electromagnetic and electrostatic cases, in \cref{fig:flux-diff-xz} we compute the difference between the electromagnetic and electrostatic $\Gamma_{\perp e}$ in the $x-z$ plane, averaged over $y$ and time, normalized to the electromagnetic flux $\Gamma_{\perp, EM}$. Here, regions where the perpendicular particle flux is larger in the electrostatic case than in the corresponding electromagnetic case are indicated in blue ($\Gamma_{\perp ,ES}>\Gamma_{\perp, EM}$), while red regions indicate the opposite ($\Gamma_{\perp, ES}<\Gamma_{\perp, EM}$). Near the midplane ($z=0$) the transport is roughly the same between each corresponding electrostatic and electromagnetic case, consistent with the results of \cref{fig:power-scan-pflux}. Off-midplane there is some reduction in transport in the SOL region in the high $\hat{n}$ cases.

\begin{figure}[t!]
    \includegraphics[width=\textwidth]{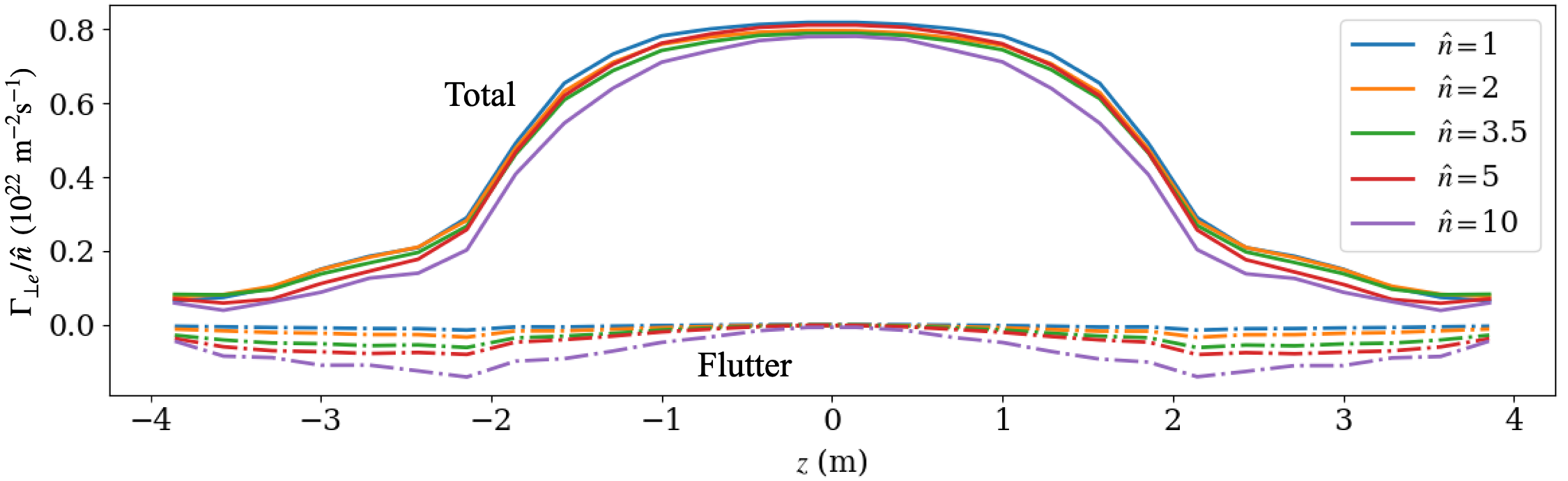}
    \caption[Radial electron particle fluxes for the electromagnetic cases in the power scan, as a function of distance along the field line.]{Radial electron particle fluxes for the electromagnetic cases normalized to $\hat{n}$, evaluated at $x=1.32$ m (just outside the source region) and plotted as a function of the distance along the field line, $z$. Magnetic flutter transport (dash-dotted) becomes stronger relative to the total transport as $\hat{n}$ increases, resulting in a slight reduction in cross-field transport off-midplane at higher $\hat{n}$.}
    \label{fig:power-scan-pflux-z-profiles}
    \includegraphics[width=\textwidth]{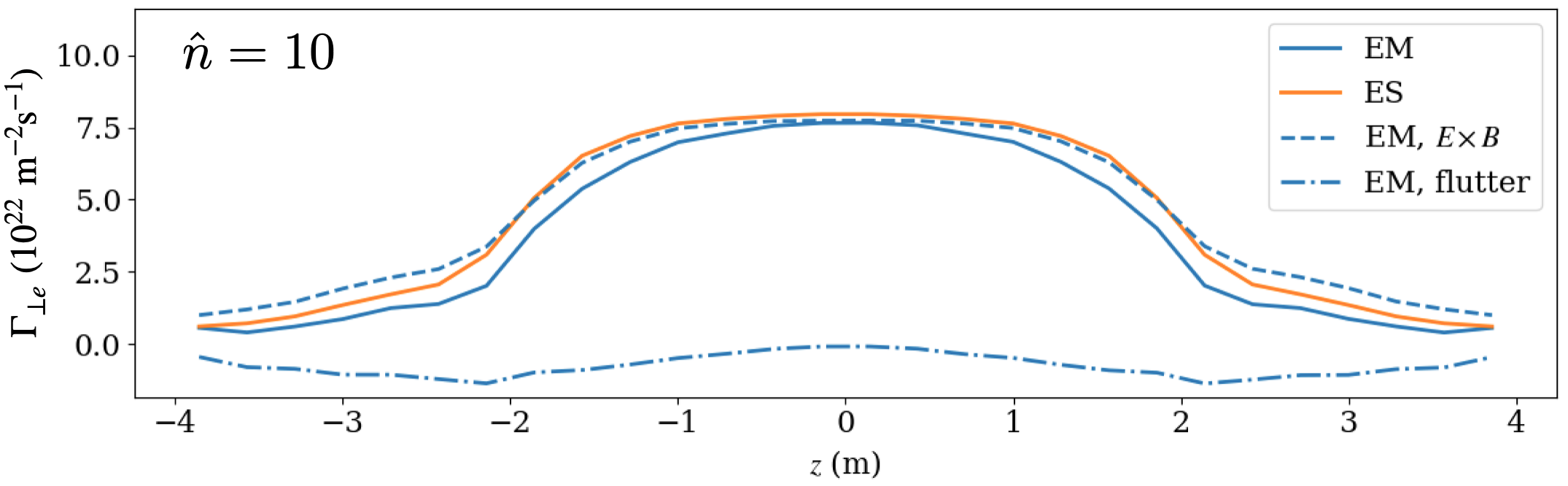}
    \caption[Electrostatic-electromagnetic comparison of radial electron particle fluxes as a function of distance along the field line.]{Electrostatic-electromagnetic comparison of radial electron particle fluxes for the $\hat{n}=10$ case, evaluated at $x=1.32$ m (just outside the source region) and plotted as a function of the distance along the field line, $z$. In the electromagnetic case (EM), we have broken the total transport into $E\times B$ (dashed) and flutter (dot-dashed) components. Radially-inward magnetic flutter transport peaking near $z=\pm 2$ reduces the net radially-outward transport off-midplane.}
    \label{fig:em-es-pflux-n10}
\end{figure}

To investigate this further, in \cref{fig:power-scan-pflux-z-profiles} we show the radial particle fluxes as a function of the distance along the field line, $z$, evaluated just outside the source region at $x=1.32$ m and normalized to $\hat{n}$. As $\hat{n}$ increases, the particle flux falls off more quickly along the field line. This is despite the fact that the $E\times B$ fluxes remain near the levels seen in the electrostatic cases, as shown in an electrostatic-electromagnetic comparison of the $\hat{n}=10$ case in \cref{fig:em-es-pflux-n10}. The differences can be attributed to magnetic flutter transport (dotted lines) along the perturbed field lines becoming stronger (relative to the total radial transport) with increasing $\hat{n}$; here, negative values indicate radially inward transport. 

\begin{figure}[t]
    \includegraphics[width=1\textwidth]{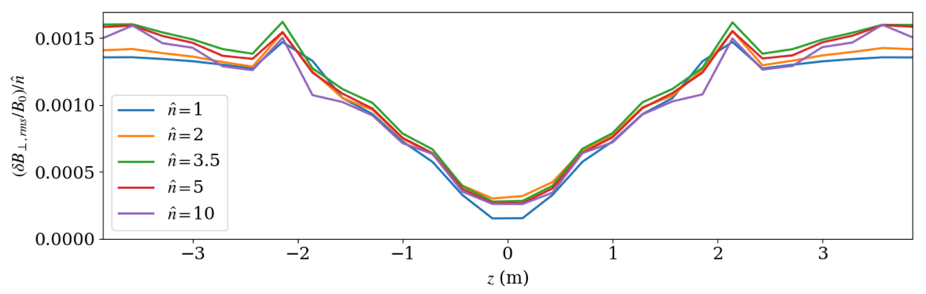}
    \caption[RMS fluctuation amplitude of magnetic perturbations as a function of distance along the field line.]{RMS fluctuation amplitude of magnetic perturbations, $\delta B_\perp/B_0$, normalized to $\hat{n}$, evaluated at $x=1.32$ m (just outside the source region) and plotted as a function of the distance along the field line, $z$. The fluctuation amplitude scales well with $\hat{n}$.}
    \label{fig:power-scan-deltaB-z}
    \includegraphics[width=1\textwidth]{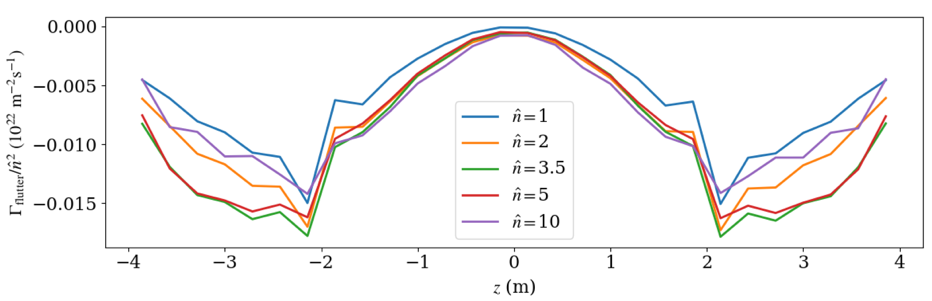}
    \caption[Particle flux due to magnetic flutter for the power scan.]{Particle flux due to magnetic flutter, evaluated at $x=1.32$ m (just outside the source region) and plotted as a function of the distance along the field line, $z$. Here we have normalized to $\hat{n}^2$, indicating that the flutter transport roughly scales with $\hat{n}^2$.}
    \label{fig:power-scan-flutter}
\end{figure}
Radial magnetic flutter transport is the result of parallel motion along radially perturbed field lines, and is given by
\begin{gather}
    \Gamma_\mathrm{flutter} =  \langle \widetilde{n_e u_{\parallel e}} b_r \rangle
\end{gather}
In our system, as a blob is transported radially outwards by the $E\times B$ drift, it can drag the field lines with it at higher $\beta$. However, the footpoints of the field lines are relatively fixed due to line-tying, so the field lines bow out radially at the midplane, as can be seen in \cref{fig:bstream-x}$b$. As particles travel from the midplane to the end plates along these bowed field lines, they are moving radially inward. This flutter transport cancels out some of the radially outward $E\times B$ transport at the midplane, resulting in a net reduction of radial transport off-midplane. To better understand the scaling of the flutter transport, we compute the RMS amplitude of the magnetic fluctuations along the field line at $x=1.32$ m in \cref{fig:power-scan-deltaB-z}. We see that the fluctuations scale well with $\hat{n}\sim \beta$, where we have normalized to $\hat{n}$ in the plot. Since the flutter transport is roughly proportional to $n_e \delta B_\perp$, and both $n_e$ and $\delta B_\perp$ scale linearly with $\hat{n}$, this means we might expect that the flutter transport scales roughly with $\hat{n}^2$. In \cref{fig:power-scan-flutter} we see that when we normalize the flutter profiles along the field line from \cref{fig:power-scan-pflux-z-profiles} by $\hat{n}^2$ instead of $\hat{n}$, the flutter transport is indeed roughly scaling with $\hat{n}^2$.

\subsubsection{Parallel particle transport: particle fluxes to the endplates}
\begin{figure}[t]
    \centering
    \includegraphics[width=1\textwidth]{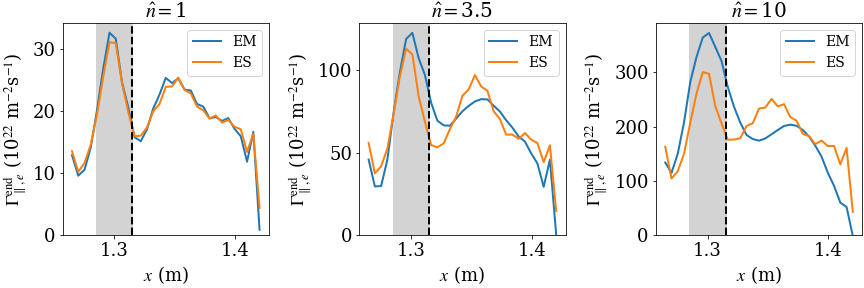}
    \caption[Parallel electron particle fluxes to the lower endplate.]{Parallel electron particle fluxes to the lower endplate, $\Gamma_{\parallel,e}^\mathrm{end}$, for $\hat{n}=\{1,3.5,10\}$. As $\hat{n}$ increases in the electromagnetic cases, reduced radial transport upstream due to magnetic flutter results in $\sim 10\%$ higher peak particle fluxes than in the corresponding electrostatic cases, peaked near $x=1.3$ m (the source peak).}
    \label{fig:power-scan-end-pfluxes}
\end{figure}

We examine the parallel electron  
particle fluxes to the lower endplate at $z=-L_z/2$ in \cref{fig:power-scan-end-pfluxes} (with the fluxes to the upper endplate nearly identical). This is defined as
\begin{gather}
    \Gamma^\mathrm{end}_{\parallel, e} = \Big\langle \int {\mathcal{J}f_{e\,h}}\dot{\vec{R}}_h\cdot \uv{b}\, \dx{x}\,\dx{y}\, \dx{^3\vec{v}} \Big|_{z=-L_z/2} \Big\rangle.
\end{gather}
Note that this counts only the net flux of high energy electrons that can overcome the sheath potential. When we integrate the resulting profiles in $x$, we obtain an integrated lower particle flux that is approximately half the integrated particle source rate (with the other half due to the upper particle flux), indicating that we have a steady state with the sources balanced by parallel end losses (recalling that there are no perpendicular losses to the radial boundaries here). 

As $\hat{n}$ increases in the electromagnetic cases, reduced radial transport upstream due to magnetic flutter results in $\sim 10\%$ higher peak particle fluxes than in the corresponding electrostatic cases, peaked near $x=1.3$ m (the source peak). There is virtually no change in the profiles in the electrostatic cases other than scaling with $\hat{n}$. 
In each case there is also a second, smaller peak in the SOL region, with this peak slightly lower in the electromagnetic cases than in the electrostatic ones. This second peak is likely due to end losses from blobs that escape the source region and propagate some finite distance into the SOL region. That the electromagnetic fluxes are higher in the source region and slightly lower in the remainder of the domain is consistent with less upstream cross-field transport in the electromagnetic cases.
Note that while the width of the flux profiles in the source region is certainly influenced by the width of the source, the shape of the source is identical for all cases. This means that the (relative) differences in the widths and heights of the profiles are physical. Nonetheless, since the absolute peak values and widths are sensitive to the source parameters, a comparison to experimental divertor fluxes is out of the scope of this work; this would likely require the inclusion of closed-field-line regions, since most of the sourcing of particles (and heat) is on closed field lines in tokamaks.

\subsection{Heat fluxes to the endplates}
A critical issue for future tokamak experiments and reactors is the heat exhaust problem, with large heat loads posing a risk to the survivability of the device walls. Thus it is important to develop high-fidelity modeling capability to be able to predict the heat loads and heat-flux widths on the divertor plates. While our present simulations do not have the realistic X-point geometry (including both closed- and open-field-line regions) or neutral particle dynamics required to produce experimentally-relevant heat flux predictions, we can still examine the heat flux profiles that result from our simulations. We can compute the total (ion plus electron) heat flux to the lower endplate at $z=-L_z/2$ via
\begin{equation}
    Q^\mathrm{end}_{\parallel} = \sum_s \Big\langle \int H_{s\,h} \mathcal{J}f_{s\,h}\dot{\vec{R}}_h\cdot \uv{b}\, \dx{x}\,\dx{y}\, \dx{^3\vec{v}} \Big|_{z=-L_z/2} \Big\rangle \label{end-hflux}.
\end{equation}
Here, we include the potential energy via the Hamiltonian to account for slowing of electrons as they climb the potential drop from the sheath entrance to the grounded wall. We plot the radial profiles of this quantity for the $\hat{n}=\{1,3.5,10\}$ electromagnetic and electrostatic cases in \cref{fig:power-scan-end-hfluxes}. Like in the previous section, the peak flux increases in the electromagnetic cases relative to the electrostatic cases as $\hat{n}$ increases, with a $\sim 20\%$ higher peak in the $\hat{n}=10$ case. Again, this is consistent with upstream cross-field transport being reduced by electromagnetic effects.

\begin{figure}[t]
    \centering
    \includegraphics[width=\textwidth]{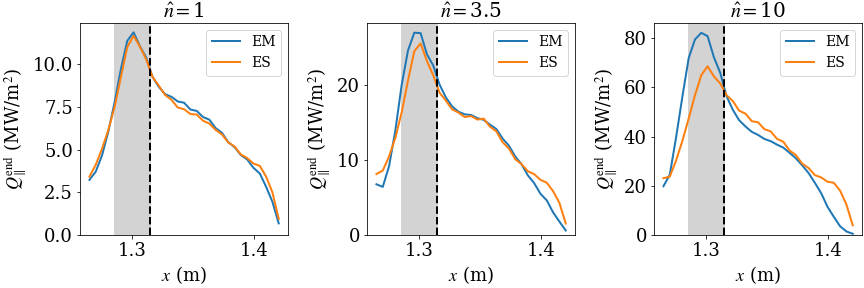}
    \caption[Heat fluxes to the lower endplate.]{Total (ion plus electron) heat fluxes to the lower endplate, $Q^\mathrm{end}_{\parallel}$, for $\hat{n}=\{1,3.5,10\}$. As $\hat{n}$ increases in the electromagnetic cases the peak heat flux increases relative to the electrostatic cases, with a $\sim 20\%$ higher peak in the $\hat{n}=10$ case.}
    \label{fig:power-scan-end-hfluxes}
\end{figure}

\subsection{Fluctuation statistics}

Experimental measurements have shown that the SOL is characterized by large, intermittent fluctuations. We compare fluctuation statistics between the electromagnetic and electrostatic $\hat{n}=10$ cases in \cref{fig:n10-stats}. Statistics of the electron density are shown on the top row, the {middle row} shows statistics of electrostatic potential fluctuations{ and the bottom row shows statistics of the radial electron particle flux.} All statistics are averaged over $y$ and $z$ near the midplane. The root-mean-square (RMS) density fluctuation level $n_{rms}/\bar{n}$ is at least 20\% throughout the domain in both cases, consistent with the large fluctuations observed in experiments. Despite the fact that the electromagnetic and electrostatic cases show the same level of particle transport at the midplane, the RMS  density fluctuations  are slightly larger in the electromagnetic case. Meanwhile the RMS relative potential fluctuations are slightly smaller in the electromagnetic case. Since intermittency is a key feature of SOL transport observed in experiments, we also measure the skewness and excess kurtosis of the fluctuations. Positive values of these higher-order statistics generally indicate more intermittency. Both cases show comparable levels of skewness and excess kurtosis of the density fluctuations. The potential fluctuations seem to be more intermittent in the electrostatic case, with higher skewness and kurtosis in much of the domain. On the bottom row, we see that the electromagnetic case has some larger particle flux fluctuations approaching the far edge of the domain. In both the electromagnetic and electrostatic cases the particle flux is also intermittent, perhaps slightly more so in the electrostatic case, as indicated by positive skewness and kurtosis in much of the domain.

\begin{figure}[t]
    \includegraphics[width=1\textwidth]{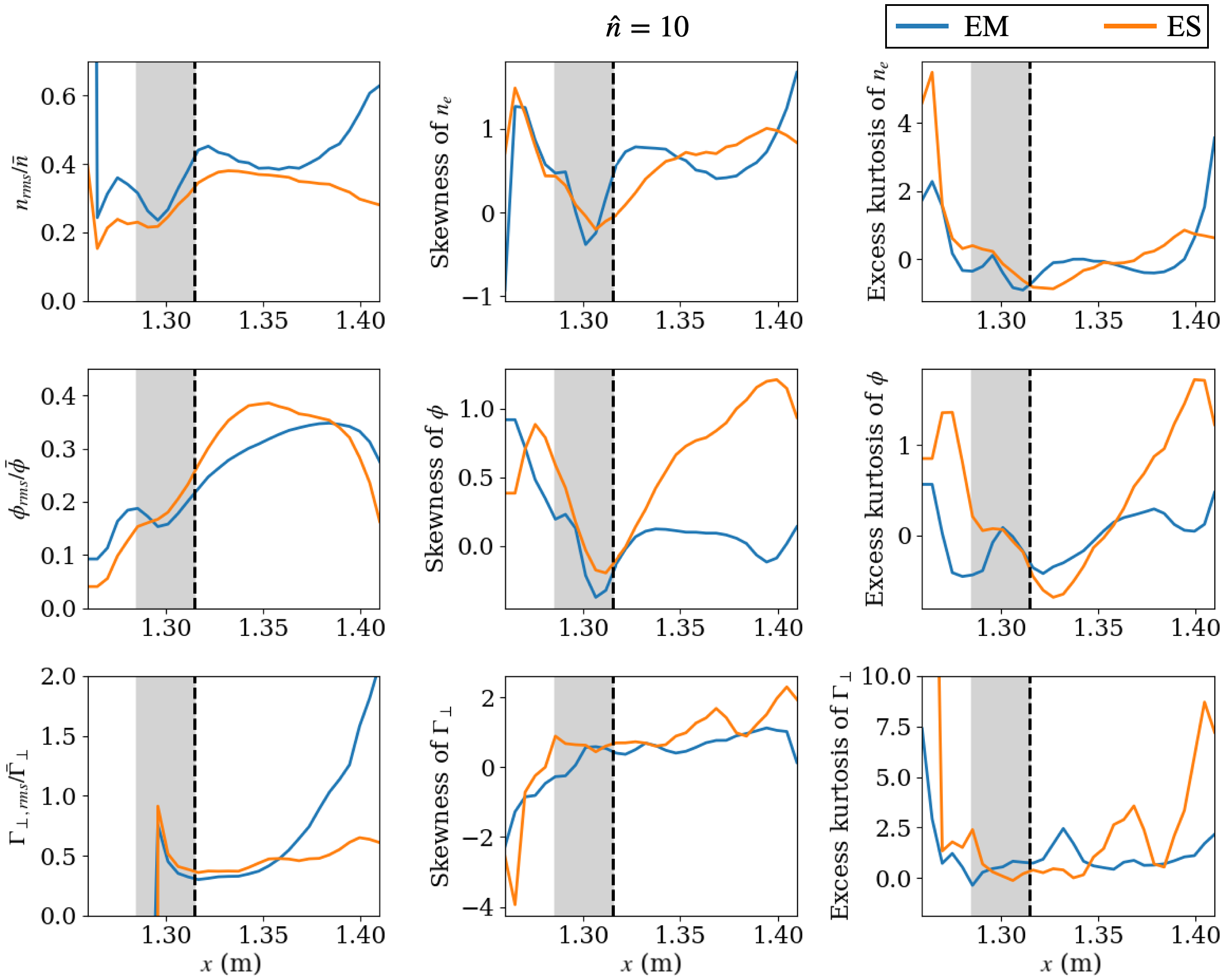}
    \caption{Electrostatic-electromagnetic comparison of fluctuation statistics of the electron density, the electrostatic potential, and the electron particle flux for the $\hat{n}=10$ case.}
    \label{fig:n10-stats}
\end{figure}

\begin{figure}[t]
    \centering
    \includegraphics[width=\textwidth]{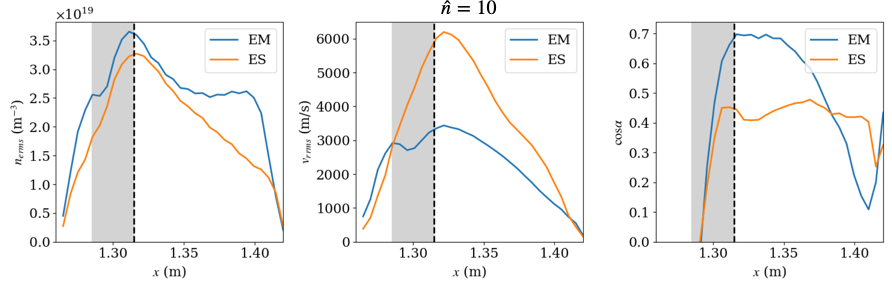}
    \caption{Components of the radial $E\times B$ particle flux, $\Gamma_{n,r} = n_{e,\text{rms}} v_{r,\text{rms}} \cos\alpha$, where $\cos \alpha$ measures the phase between the density and $E\times B$ fluctuations.}
    \label{fig:n10-exb-comps}
\end{figure}

It is perhaps somewhat counter-intuitive that even though the density fluctuations are slightly larger in the electromagnetic case, the resulting transport is the same. The radial $E\times B$ particle flux can also be written as $\Gamma_{\perp,E\times B} = n_{e,\text{rms}} v_{r,\text{rms}} \cos\alpha$, where $\cos\alpha \equiv \langle \tilde{n}_e \tilde{v}_r\rangle / (n_{e,\text{rms}} v_{r,\text{rms}})$ accounts for the phase between density and $E\times B$ velocity fluctuations. \cref{fig:n10-exb-comps} shows these three components of $\Gamma_{\perp,E\times B}$ for the $\hat{n}=10$ case. Despite the electromagnetic case having slightly larger density fluctuations and better correlation between the density and $E\times B$ fluctuations in most of the domain, the resulting particle flux is identical in both cases, this is offset by reduced $E\times B$ fluctuation amplitude in the electromagnetic case.

\section{Summary of results}

In this chapter we presented the first electromagnetic gyrokinetic simulations on open field lines. We showed that large magnetic fluctuations on the order $\delta B_\perp/B \sim 1\%$ can be handled in a stable and efficient manner. This is critical for enabling the study of electromagnetic effects in the edge and SOL, which are expected to be important for phenomena such as ELMs and the pedestal.

In \cref{sec:emgk-res} we showed a preliminary set of simulations, one electromagnetic and one electrostatic, and examined qualitative differences in the dynamics. In the electromagnetic case, we traced the perturbed magnetic field lines and found that they can be bent and stretched significantly by the plasma motion at high $\beta$. We found that blobs spin in the electrostatic case due to adiabatic electron dynamics. In the electromagnetic case the electron response is non-adiabatic and the blobs propagate ballistically radially outwards, suggesting electrical disconnection from the sheath. The dynamics observed here could be relevant for high $\beta$ blobs and ELMs, which involve high $\beta$ filament-like structures that carry significant uni-directional current. 

In \cref{sec:power-scans}, we performed a study of the effects of increasing $\beta$ on the SOL dynamics. At higher $\beta$, the influence of electromagnetic effects became stronger, resulting in steepening of pressure gradients near the source region and flattening of gradients in the remainder of the domain. The interplay between steepening pressure gradients in the source region and increased $E\times B$ shear just outside the source region could be relevant for pedestal formation and the L-H transition. We also observed a transition from interchange-like modes with $k_\parallel\sim 0$ to ballooning-like modes with finite $k_\parallel$ as pressure gradients $(\alpha^\mathrm{SMT})$ increased above the ballooning stability threshold in the source region. While cross-field perpendicular transport at the midplane was unaffected by increasing $\beta$, the transport was reduced off-midplane by magnetic flutter in the higher $\beta$ cases due to line bending. This resulted in the parallel particle and heat fluxes to the endplates being more peaked in the electromagnetic cases. 

One might note that at the nominal experimental source power $(\hat{n}=1)$, we observed electromagnetic effects to be mostly unimportant, and that we needed to scale up the source power ($\sim \beta$) by a factor of 3-10 to see electromagnetic effects impact the dynamics. While this is true in the simple setup that we have considered here, in a real experiment there are other effects that could make electromagnetic dynamics important at the experimental $\beta$ levels. These include steeper pressure gradients, stronger magnetic fields, longer connection lengths, and magnetic shear, all of which could push the system into a more electromagnetic regime at experimental $\beta$ levels. 

\begin{subappendices}
 
\section[Note on some results from Mandell \emph{et al.} (2020)]{Note on some results from \citet{mandell2020}}

In \citet{mandell2020}, we presented $\hat{n}=10$ results that showed a reduction in radial transport in the electromagnetic case compared to the electrostatic case  (see Fig. 10 of \citet{mandell2020}). After further analysis, we believe that this was a consequence of placing the source region too close to the inner-radial boundary of the simulation. This resulted in fast parallel losses in the cells at the boundary because the Dirichlet condition $\Phi=0$ on the walls meant that there could be no sheath potential to confine particles at the domain edge. In the electromagnetic cases, the issue was exacerbated by radially inward magnetic flutter transport near the boundary, resulting in even more losses from the boundary cells and consequently less perpendicular particle transport. This can be seen in \cref{fig:jpp-sims-flux}. After extending the domain 2 cm radially inward and redoing the simulations, we saw much less difference in particle transport levels between the electromagnetic and electrostatic cases, consistent with the results in the $\hat{n}=10$ cases in \cref{sec:power-scans}. 

\begin{figure}[t!]
    \centering
    \includegraphics[width=\textwidth]{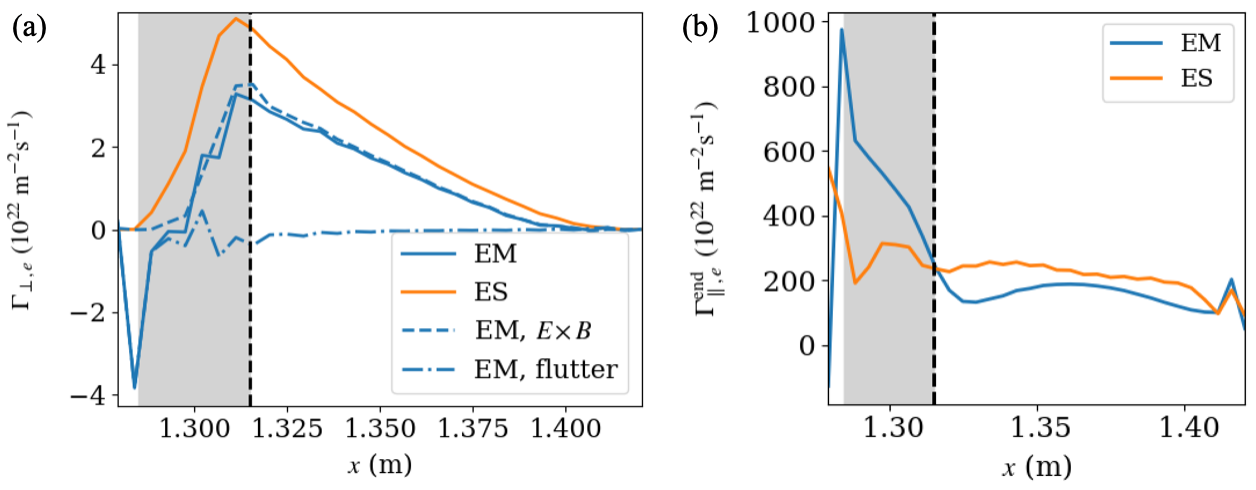}
    \caption[Perpendicular and parallel particle fluxes from Mandell \emph{et al.} (2020).]{Midplane radial electron particle flux $(a)$ and parallel electron particle flux to the endplate $(b)$ for electromagnetic and electrostatic $\hat{n}=10$ cases from \citet{mandell2020} (see Fig. 10 in that work). In these cases, the source region (gray) was too close to the radial boundary, resulting in fast parallel losses in the boundary cells. Parallel losses were even larger in the electromagnetic case due to radially inward magnetic flutter transport near the boundary. This resulted in less perpendicular particle transport at the midplane in the electromagnetic case. We have redone the simulations with the domain extended 2 cm radially inward and produced results consistent with the results in \cref{sec:power-scans}.}
    \label{fig:jpp-sims-flux}
\end{figure}

\end{subappendices}

\chapter{Generalizing the magnetic geometry: towards a more realistic tokamak scrape-off layer} \label{ch:geometry}

In this chapter we move towards more realistic tokamak SOL geometry by adopting a generalized field-aligned non-orthogonal coordinate system. 
Choosing field-aligned coordinates allows one to exploit the elongated nature of the turbulence, which is generally characterized by long wavelengths parallel to the background field and short perpendicular wavelengths ($k_\parallel \ll k_\perp$). 
In the local approach employed by several core gyrokinetic codes, the resulting domain is a thin flux tube extended along the field line. In the global approach, the domain remains field aligned, but extends radially to cover some or all of the minor radius of the device. In both approaches, the field-aligned coordinate can be coarse since $k_\parallel \ll k_\perp$. Thus when the (generalized) poloidal angle is chosen as the field-aligned coordinate, the resulting grid resolution is fine in the radial and toroidal (or binormal) directions to resolve short perpendicular wavelengths and coarse in the poloidal direction. 
In the toroidal direction, axisymmetry allows one to assume statistical periodicity so that only a fraction (wedge) of the full toroidal direction needs to be resolved, provided the toroidal domain extent is many turbulent correlation lengths wide. This approach has been used successfully by many local and global gyrokinetic codes, with the resulting computational savings comprising one of the main advantages of the field-aligned approach.\footnote{Alternatively, the toroidal angle can be used as the coarse field-aligned coordinate. However, in this case the full poloidal angle must be resolved with a fine grid (unlike the toroidal angle, statistical periodicity cannot be used in the poloidal angle to reduce the domain length). This is not as computationally efficient as using the poloidal angle as the field-aligned coordinate.} 

While a field-aligned coordinate system can be used both in the core and in the SOL, these coordinates are singular on the separatrix for diverted geometries due to the presence of the X-point \citep{stegmeir2016}. To deal with this issue, recent interest has focused on the flux-coordinate independent (FCI) approach, which abandons field- and flux-aligned coordinates in the poloidal plane but retains a field-line-following discretization of the parallel gradient operator \citep{hariri2013,hariri2014,stegmeir2016}. This allows a coarse toroidal grid to be used, but still requires fine perpendicular resolution covering the entire poloidal plane. Another recent approach uses a flux-aligned poloidal grid with controlled dealignment near the X-point \citep{McCorquodale2015,dorf2016continuum,dorf2020}. After breaking the toroidal direction into several blocks (wedges), a local field-aligned coordinate system is used in each block. Interpolation (similar to what is done in the FCI approach) is required to compute the parallel derivatives between blocks. 

Addressing the issue of the X-point in the \gke code is left as important future work. All of the approaches detailed above have the disadvantage that a fine grid is required on the entire poloidal plane. To avoid this, one might imagine a cross-separatrix simulation domain composed of a global field-aligned region in the core, a thin non-field-aligned region near the separatrix (perhaps using some version of FCI in this small region), and another field-aligned flux-tube-like region in the SOL. This way, we could keep the advantages of flux tubes for most of the domain, limiting the region needing to be resolved with a fine poloidal grid to the small area near the separatrix. Interpolation between these regions with different coordinate systems and different poloidal resolution would be required and may present challenges. 

This chapter takes a step towards these full geometry capabilities. In the first section we express the gyrokinetic system in general field-aligned coordinates. We then focus on how to formulate field-aligned coordinate systems for use in flux-tube-like domains in the SOL in \cref{sec:heli-geo,sec:solovev,sec:circ-geo}. 

\section{Gyrokinetics in a field-aligned coordinate system}

Here we will express the gyrokinetic system in a field-aligned coordinate system. The resulting equations contain various metric-related quantities since the coordinates are non-orthogonal. Many of these metric quantities were dropped in the simple helical geometry used in \cref{ch:nstx-results}.

\subsection{Preliminaries: general non-orthogonal curvilinear coordinates}
\label{sec:geovec}

Suppose we have a coordinate system in 3D space parametrized by coordinates $(x,y,z)$. If the coordinate system is orthogonal, 
a vector can be uniquely decomposed in terms of the coordinate basis vectors as
    $\vec{v} = v_x \uv{x} + v_y \uv{y} + v_z \uv{z}$,
where $v_x$ is the component of the vector in the direction of the $\uv{x}$ basis vector, and similarly for $y$ and $z$. However, if the coordinate system is non-orthogonal, there are \emph{two} natural sets of basis vectors, leading to the \emph{covariant} and \emph{contravariant} representation of a vector:
\begin{align}
    \vec{v} &= v_x \nabla x + v_y \nabla y + v_z \nabla z \qquad \text{(covariant)} \\
    \vec{v} &= v^x \vec{e}_x + v^y \vec{e}_y + v^z \vec{e}_z \qquad\quad\,  \text{(contravariant)}
\end{align}
In the covariant representation, the (contravariant) basis vectors are the gradient vectors, defined by
\begin{equation}
    \qquad\qquad \vec{e}^\alpha = \nabla \alpha, \qquad \alpha=(x,y,z).
\end{equation}
In the contravariant representation, the (covariant) basis vectors are the tangent vectors, defined by
\begin{equation}
    \qquad\qquad \vec{e}_\alpha = \pderiv{\vec{R}}{\alpha},\qquad \alpha=(x,y,z),
\end{equation}
where $\vec{R}$ is the position vector. Note that we will usually opt to use $\nabla\alpha$ in place of $\vec{e}^\alpha$ to denote the gradient basis vectors, but we will continue to use $\vec{e}_\alpha$ for the tangent basis vectors. 

These basis vectors are neither orthogonal nor unit vectors, which leads to the co- and contravariant \emph{metric coefficient} tensors, defined by
\begin{align}
    g_{\alpha\beta} &= \vec{e}_\alpha \cdot\vec{e}_\beta \qquad\ \ \,  \text{(covariant)} \\
    g^{\alpha\beta} &= \nabla\alpha \cdot\nabla\beta \qquad \text{(contravariant)}
\end{align}
These two tensors are inverses of each other, so that $(g_{\alpha \beta}) = (g^{\alpha\beta})^{-1}$. 
The two sets of basis vectors also obey the relationship
\begin{equation}
    \nabla\alpha \cdot \vec{e}_\beta = \delta^\alpha_\beta,
\end{equation}
where $\delta^\alpha_\beta$ is the Kronecker delta. It follows that the tangent basis vectors can be expressed in terms of the gradient basis vectors as
\begin{equation}
    \vec{e}_x = J\left( \nabla y\times\nabla z\right), \qquad \vec{e}_y = J\left(\nabla z\times\nabla x\right), \qquad \vec{e}_z = J\left(\nabla x\times\nabla y\right),
\end{equation}
where 
\begin{equation}
    J = \left[\left( \nabla y\times\nabla z\right)\cdot\nabla x\right]^{-1} = \left[\left( \nabla z\times\nabla x\right)\cdot\nabla y\right]^{-1} =\left[\left( \nabla x\times\nabla y\right)\cdot\nabla z\right]^{-1}
\end{equation}
is the Jacobian of the coordinate system written in terms of the gradient basis vectors.
Similarly, we can express the gradient basis vectors as
\begin{equation}
    \nabla x = \frac{1}{J}\left( \vec{e}_y\times\vec{e}_z\right), \qquad \nabla y = \frac{1}{J}\left(\vec{e}_z\times\vec{e}_x\right), \qquad \nabla z = \frac{1}{J}\left(\vec{e}_x\times\vec{e}_y\right),
\end{equation}
and we can also write the Jacobian in terms of the tangent basis vectors as
\begin{equation}
    J = \left( \vec{e}_y\times\vec{e}_z\right)\cdot\vec{e}_x = \left( \vec{e}_z\times\vec{e}_x\right)\cdot\vec{e}_y =\left( \vec{e}_x\times\vec{e}_y\right)\cdot\vec{e}_z.
\end{equation}
The Jacobian can also be written (up to a sign) via the determinants of the metric tensors,
\begin{equation}
    J = \det(g_{\alpha\beta})^{1/2} =\det(g^{\alpha\beta})^{-1/2}.
\end{equation}
Finally, note that the co- and contravariant components of a vector can be obtained from $v_z = \vec{v}\cdot \vec{e}_z$ and $v^z = \vec{v} \cdot \nabla z$, and similarly for $x$ and $y$.

We will also make use of the following vector calculus identities for the gradient, divergence, and curl:
\begin{gather}
    \nabla f = \pderiv{f}{x}\nabla x + \pderiv{f}{y} \nabla y + \pderiv{f}{z}\nabla z, \\
    \nabla \cdot \vec{F} = \frac{1}{J}\left[\pderiv{}{x}\left(J F^x\right) + \pderiv{}{y}\left(J F^y\right) + \pderiv{}{z}\left(J F^z\right) \right], \\
    \nabla \times \vec{F} = \frac{1}{J}\left[\left(\pderiv{F_z}{y}-\pderiv{F_y}{z}\right)\vec{e}_x+ \left(\pderiv{F_x}{z}-\pderiv{F_z}{x}\right)\vec{e}_y+\left(\pderiv{F_y}{x}-\pderiv{F_x}{y}\right)\vec{e}_z \right].
\end{gather}
Volume integrals can be expressed as
\begin{gather}
    \int f\, \dx{^3\vec{R}} = \int Jf\,\dx{x}\, \dx{y}\,\dx{z}, 
\end{gather}
and a surface integral over a constant $x$ surface is given by
\begin{gather}
    \int \vec{F}\cdot\dx{\vec{s}}_x = \int J \vec{F}\cdot\nabla x\, \dx{y}\,\dx{z} = \int J F^x\, \dx{y}\,\dx{z},
\end{gather}
and similarly for surface integrals over constant $y$ and $z$ surfaces.

For more details about non-orthogonal coordinate systems, see \citet{dhaeseleer1991}.

\subsection{Representation of the background field}
In order to take advantage of the fact that turbulent structures are much more elongated along the field line than perpendicular to it ($k_\parallel \ll k_\perp$), we adopt a field-aligned coordinate system, which we will denote by $(x,y,z)$. To do this, we write the background magnetic field in Clebsch-like form as
\begin{equation}
    \vec{B} = \mathcal{C}(x) \nabla x \times \nabla y.
\end{equation}
Here, $x$ and $y$ are coordinates perpendicular to the background field, with $x$ usually a radial-like coordinate, and $y$ a field-line-labeling coordinate. Importantly, $x$ and $y$ are constant on field lines since $\vec{B}\cdot\nabla x =\vec{B}\cdot\nabla y =0$. For now, $\mathcal{C}$ is an arbitrary function of $x$ (it cannot depend on $z$ because $\vec{B}$ must be divergence free, and it cannot depend on $y$ because we will assume axisymmetry). Since by construction the background field is perpendicular to the gradient basis vectors $\nabla x$ and $\nabla y$, the background field can then be written in contravariant form as
\begin{equation}
    \vec{B} = (\vec{B}\cdot\nabla z)\vec{e}_z = \frac{\mathcal{C}}{J}\vec{e}_z.
\end{equation}
Thus the magnetic field is in the direction of the tangent vector in the $z$ direction, $\vec{e}_z$, so $z$ is a field-aligned coordinate as desired. The Jacobian of the $(x,y,z)$ coordinate system is
\begin{equation}
    J = \left[(\nabla x\times \nabla y)\cdot \nabla z\right]^{-1}.
\end{equation}
Noting that the magnitude of the background field is given by
\begin{equation}
    B = \sqrt{\vec{B}\cdot\vec{B}} = \frac{\mathcal{C}}{J}\sqrt{g_{zz}},
\end{equation}
with $g_{zz} = \vec{e}_z \cdot\vec{e}_z = J^2 B^2/\mathcal{C}^2$, we can also write the background field as
\begin{equation}
    \vec{B} = \frac{B}{\sqrt{g_{zz}}}\vec{e}_z.
\end{equation}
Finally, we will assume that $\vec{B}$ and all other geometric quantities are axisymmetric, so that they have no $y$ dependence, \emph{i.e.} $\pderivInline{B}{y}=0$.


The definition of the coordinates $(x,y,z)$ that satisfy the above relations is relatively flexible, depending on desired properties of the coordinates. In \cref{sec:heli-geo,sec:solovev} we give specific definitions of the coordinates in different geometrical configurations.

\subsection{Gyrokinetic Poisson bracket in field-aligned coordinates}

Recall from \cref{gkpb} that the gyrokinetic Poisson bracket is defined as
\begin{equation}
     \{F,G\} = \frac{\vec{B}^*}{m B_\parallel^*} \cdot \left(\nabla F \frac{\partial G}{\partial v_\parallel} - \frac{\partial F}{\partial v_\parallel}\nabla G\right) - \frac{ \uv{b}}{q B_\parallel^*}\times \nabla F \cdot \nabla G,
\end{equation}
with $\vec{B}^* = \vec{B} + (mv_\parallel/q)\nabla\times\uv{b} +\nabla \times(A_\parallel \uv{b})$ and $B_\parallel^* = \uv{b}\cdot\vec{B}^*$. 
Using the identities from \cref{sec:geovec}, we can write $\vec{B}^*$ in contravariant form as
\begin{align}
    &\vec{B}^* = \frac{\mathcal{C}}{J}\vec{e}_z + \frac{mv_\parallel}{q} \frac{1}{J}\left[ -\pderiv{b_y}{z}\vec{e}_x + \left(\pderiv{b_x}{z}-\pderiv{b_z}{x}\right)\vec{e}_y + \pderiv{b_y}{x}\vec{e}_z\right] \notag \\
    &\  + \frac{1}{J}\left[\left( \pderiv{(A_\parallel b_z)}{y}-\pderiv{(A_\parallel b_y)}{z} \right)\vec{e}_x + \left( \pderiv{(A_\parallel b_x)}{z}-\pderiv{(A_\parallel b_z)}{x}\right)\vec{e}_y + \left(  \pderiv{(A_\parallel b_y)}{x}-\pderiv{(A_\parallel b_x)}{y}\right)\vec{e}_z   \right], 
\end{align}
so that the contravariant components of $\vec{B}^*$ are
\begin{align}
    B^{*x}&= \frac{1}{J}\left[-\frac{mv_\parallel}{q}\pderiv{b_y}{z}+\left( \pderiv{(A_\parallel b_z)}{y}-\pderiv{(A_\parallel b_y)}{z} \right)\right], \\
    B^{*y}&=\frac{1}{J}\left[\frac{mv_\parallel}{q}\left(\pderiv{b_x}{z}-\pderiv{b_z}{x}\right)+\left( \pderiv{(A_\parallel b_x)}{z}-\pderiv{(A_\parallel b_z)}{x}\right)\right],  \\
    B^{*z}&=\frac{1}{J}\left[\mathcal{C} + \frac{m v_\parallel}{q}\pderiv{b_y}{x} +\left(  \pderiv{(A_\parallel b_y)}{x}-\pderiv{(A_\parallel b_x)}{y}\right) \right].
\end{align}
Here, the covariant components of the unit vector $\uv{b}=\vec{B}/B$ are given by
\begin{equation}
    \qquad\qquad\qquad b_\alpha = \uv{b}\cdot\vec{e}_\alpha = \frac{g_{\alpha z}}{\sqrt{g_{zz}}},\qquad\qquad \alpha=(x,y,z).
\end{equation}
Then we can compute the bracket by using 
\begin{equation}
    \vec{B}^*\cdot\nabla F = B^{*x}\pderiv{F}{x} + B^{*y}\pderiv{F}{y} + B^{*z}\pderiv{F}{z}
\end{equation}
and
\begin{align}
    \uv{b}\times\nabla F\cdot\nabla G &= \frac{1}{J}\left(b_y\pderiv{F}{z}-b_z \pderiv{F}{y}\right)\pderiv{G}{x} +\frac{1}{J} \left(b_z\pderiv{F}{x}-b_x\pderiv{F}{z}\right)\pderiv{G}{y} \notag \\
    &\qquad +\frac{1}{J}\left(b_x\pderiv{F}{y}-b_y\pderiv{F}{x}\right)\pderiv{G}{z}.
\end{align}
Finally, the phase-space Jacobian $B_\parallel^*$ is given by
\begin{equation}
    B_\parallel^* = \uv{b}\cdot\vec{B}^*= B +\left(\frac{mv_\parallel}{q}+A_\parallel\right)\frac{1}{J}\left[-b_x\pderiv{b_y}{z} - b_y\left(\pderiv{b_z}{x}-\pderiv{b_x}{z}\right)+b_z\pderiv{b_y}{x}\right] \approx B.
\end{equation}

\subsection{Equations of motion in field-aligned coordinates}
Recall from \cref{eomR0,eomV0} that the gyrokinetic equations of motion are given by
\begin{gather}
    \dot{\vec{R}} = \{\vec{R},H\} = \frac{\vec{B^*}}{B_\parallel^*}v_\parallel + \frac{\uv{b}}{q B_\parallel^*}\times\left(\mu\nabla B + q \nabla \Phi\right), \\
    \dot{v}_\parallel = \{v_\parallel,H\}-\frac{q}{m}\pderiv{A_\parallel}{t} = -\frac{\vec{B^*}}{m B_\parallel^*}{\cdot}\left(\mu\nabla B + q \nabla \Phi\right)-\frac{q}{m}\pderiv{A_\parallel}{t}.
\end{gather}
We can write the velocity $\dot{\vec{R}}$ in contravariant form as  $\dot{\vec{R}}=\dot{x}\,\vec{e}_x + \dot{y}\,\vec{e}_y + \dot{z}\,\vec{e}_z$, with components
\begin{align}
    \dot{x} &= \{x,H\} = \dot{\vec{R}}\cdot\nabla x= \frac{B^{*x}}{ B_\parallel^*}v_\parallel+ \frac{\uv{b}}{q B_\parallel^*}\times\left(\mu\nabla B + q \nabla \Phi\right)\cdot \nabla x \notag \\
    &= \frac{1}{J B_\parallel^*}\left[-\frac{m v_\parallel^2}{qB} \pderiv{(B\, b_y)}{z} + \frac{mv_\parallel^2+\mu B}{qB}b_y\pderiv{B}{z}-b_z\pderiv{\Phi}{y}+b_y\pderiv{\Phi}{z} + v_\parallel \left( \pderiv{(A_\parallel b_z)}{y}-\pderiv{(A_\parallel b_y)}{z} \right) \right]\\
    \dot{y} &= \{y,H\} = \dot{\vec{R}}\cdot\nabla y= \frac{B^{*y}}{ B_\parallel^*}v_\parallel+ \frac{\uv{b}}{q B_\parallel^*}\times\left(\mu\nabla B + q \nabla \Phi\right)\cdot \nabla y \notag \\
    &= \frac{1}{J B_\parallel^*}\left[ -\frac{m v_\parallel^2}{qB}\left(\pderiv{(B\, b_z)}{x}- \pderiv{(B\, b_x)}{z}\right) + \frac{mv_\parallel^2+\mu B}{qB}\left(b_z \pderiv{B}{x}-b_x\pderiv{B}{z}\right) + b_z\pderiv{\Phi}{x} - b_x\pderiv{\Phi}{z} \right. \notag \\
    &\qquad\qquad\quad \left.+ v_\parallel \left( \pderiv{(A_\parallel b_x)}{z}-\pderiv{(A_\parallel b_z)}{x}\right) \right] \\
    \dot{z} &= \{z,H\} = \dot{\vec{R}}\cdot\nabla z= \frac{B^{*z}}{ B_\parallel^*}v_\parallel+ \frac{\uv{b}}{q B_\parallel^*}\times\left(\mu\nabla B + q \nabla \Phi\right)\cdot \nabla z \notag \\
    &= \frac{1}{JB_\parallel^*}\left[\mathcal{C}v_\parallel + \frac{m v_\parallel^2}{qB} \pderiv{(B\,b_y)}{x} - \frac{mv_\parallel^2+\mu B}{qB}b_y\pderiv{B}{x} + b_x\pderiv{\Phi}{y} - b_y\pderiv{\Phi}{x} + v_\parallel \left(  \pderiv{(A_\parallel b_y)}{x}-\pderiv{(A_\parallel b_x)}{y}\right)\right] \label{eomgeoz}
\end{align}
The parallel acceleration is given by
\begin{align}
    \dot{v}_\parallel &= -\frac{1}{J B_\parallel^*}\left[-\frac{mv_\parallel}{q}\pderiv{b_y}{z}+\left( \pderiv{(A_\parallel b_z)}{y}-\pderiv{(A_\parallel b_y)}{z} \right)\right]\left(\frac{\mu}{m}\pderiv{B}{x} + \frac{q}{m}\pderiv{\Phi}{x}\right) \notag \\
    &-\frac{1}{JB_\parallel^*}\left[\frac{mv_\parallel}{q}\left(\pderiv{b_x}{z}-\pderiv{b_z}{x}\right)+\left( \pderiv{(A_\parallel b_x)}{z}-\pderiv{(A_\parallel b_z)}{x}\right)\right] \frac{q}{m}\pderiv{\Phi}{y} \notag  \\
    &-\frac{1}{JB_\parallel^*}\left[\mathcal{C} + \frac{m v_\parallel}{q}\pderiv{b_y}{x} +\left(  \pderiv{(A_\parallel b_y)}{x}-\pderiv{(A_\parallel b_x)}{y}\right) \right]\left(\frac{\mu}{m}\pderiv{B}{z} + \frac{q}{m}\pderiv{\Phi}{z}\right) - \frac{q}{m}\pderiv{A_\parallel}{t}.
\end{align}

Here we have not neglected any terms due to smallness of parallel derivatives compared to perpendicular derivatives. While this is a common approximation made in local gyrokinetic codes, we note that dropping such terms can break Liouville's theorem, \cref{liouvillethm}, so that the gyrokinetic equation can no longer be written in conservative form (phase-space volume is no longer conserved exactly). In particular, Liouville's theorem in this case requires
\begin{equation}
    \nabla \cdot (\uv{b}\times\nabla H) = (\nabla \times \uv{b}) \cdot \nabla H
\end{equation}
so that corresponding terms cancel exactly. If parallel derivatives were dropped, this could become
\begin{equation}
    \nabla \cdot (\uv{b}\times\nabla H) \neq (\nabla \times \uv{b})_\perp \cdot \nabla H,
\end{equation}
which breaks Liouville's theorem. Since our energy-conserving discontinuous Galerkin scheme relies on the conservative form of the gyrokinetic equation, it is important for Liouville's theorem to be preserved.

\subsection{Field equations in field-aligned coordinates}
For the gyrokinetic Poisson equation, \cref{poisson1}, we must calculate
\begin{align}
    \nabla \cdot (\epsilon_\perp \nabla_\perp \Phi) &= \frac{1}{J}\left[\pderiv{}{x}\left(J \epsilon_\perp \nabla_\perp \Phi\cdot\nabla x\right)+\pderiv{}{y}\left(J \epsilon_\perp \nabla_\perp \Phi\cdot\nabla y\right)+\pderiv{}{z}\left(J \epsilon_\perp \nabla_\perp \Phi\cdot\nabla z\right)\right] \notag \\
    &\approx \frac{1}{J}\left[\pderiv{}{x}\left(J\epsilon_\perp\left(\pderiv{\Phi}{x}g^{xx}+\pderiv{\Phi}{y}g^{xy}\right)\right) +\pderiv{}{y}\left(J\epsilon_\perp\left(\pderiv{\Phi}{x}g^{xy}+\pderiv{\Phi}{y}g^{yy}\right)\right)\right].
\end{align}
Here, unlike above, we \emph{do} neglect the $\pderivInline{}{z}$ terms compared to the perpendicular derivative terms, so that the required Poisson solve remains two-dimensional. For  energetic consistency, a similar treatment would need to be made in the corresponding electrostatic field energy term (or in the second-order $E\times B$ energy term if it is kept in the Hamiltonian). Similarly, in Amp\`ere's law, \cref{ampere1}, we have
\begin{equation}
    \nabla_\perp^2 A_\parallel  =\frac{1}{J}\left[\pderiv{}{x}\left(J\left(\pderiv{A_\parallel}{x}g^{xx}+\pderiv{A_\parallel}{y}g^{xy}\right)\right) +\pderiv{}{y}\left(J\left(\pderiv{A_\parallel}{x}g^{xy}+\pderiv{A_\parallel}{y}g^{yy}\right)\right)\right].
\end{equation}

\subsection{Summary of geometric quantities}
The geometry quantities of interest, which appear in either the equations of motion or the field equations, are 
\begin{gather}
\texttt{bmag} = B\\
\texttt{cmag} = {\mathcal{C}} =  \frac{JB}{\sqrt{g_{zz}}} \\
\texttt{b\_x} = b_x = \frac{g_{xz}}{\sqrt{g_{zz}}} \\
\texttt{b\_y} = b_y = \frac{g_{yz}}{\sqrt{g_{zz}}}  \\
\texttt{b\_z} = b_z = \sqrt{g_{zz}}  \\
\texttt{gxx} = g^{xx}\\
\texttt{gxy} = g^{xy}\\
\texttt{gyy} = g^{yy} \\
\texttt{jacobPhase} = B_\parallel^* \approx B\\
\texttt{jacobGeo} = J = \sqrt{\det g_{ij}}.
\end{gather}
Here we have also included the variable names for these quantities in \gke.

\section{Helical SOL configuration including magnetic shear} \label{sec:heli-geo}
Thus far, our treatment of field-aligned geometry has been completely general, except for the assumption of axisymmetry. Now we will examine how a particular magnetic field geometry affects the choice of field-aligned coordinates and the resulting metric quantities that appear in the equations.

The helical field in an SMT is given in cylindrical  $(R,\varphi,Z)$ coordinates (where $\varphi$ is \textit{counter-clockwise} viewed from above) by
\begin{equation}
    \vec{B} = B_\varphi \uvg{\varphi} + B_v \uv{Z} = \frac{B_0 R_0}{R} \uvg{\varphi} + B_v \uv{Z}.
\end{equation}
Note we can also write this as
\begin{equation}
    \vec{B} = \nabla\Psi\times\nabla\varphi + B_0 R_0 \nabla \varphi,
\end{equation}
with $\Psi = R^2 B_v/2$ the vertical magnetic flux function (analogous to the poloidal flux in a tokamak). The field line pitch varies with radius, and can be expressed via \citep{perez2006}
\begin{equation}
    q(R) = \frac{H B_\varphi}{2\pi R B_v} = \frac{B_0 R_0 H}{2\pi R^2 B_v},
\end{equation}
where $q$ is analogous to the safety factor in a tokamak, and $H$ is the vertical height between the bottom and top end-plates where the field lines terminate. Similarly, we can define the magnetic shear as
\begin{equation}
    \hat{s} = \frac{R}{q}\deriv{q}{R}.
\end{equation}
Note that if we allow the vertical field to have some simple radial dependence via $B_v=B_v(R) = B_{v0}(R/x_0)^n$, the shear is
\begin{equation}
    \hat{s} = -2 - \frac{R}{B_v}\deriv{ B_v}{R} = -2 -n.
\end{equation}
When $B_v=$ const (as is the case for standard SMTs like the Helimak), we have $\hat{s}=-2$. The connection length is given by $L_c = H B/B_v$, which in general varies with radius. In \cref{fig:nstx-heli-Lc}$(a)$ we plot the connection length $L_c$ as a function of radius for several values of $\hat{s}$ with NSTX-like parameters. 

Note that here all quantities, including the field line pitch and the magnetic shear, are still constant \emph{along} the field lines. This is in contrast to a real tokamak SOL, where the pitch and shear can vary significantly along the field lines, especially near the X-point due to flux expansion. The field line pitch at the midplane is also nearly constant with radius in a real tokamak SOL, while here we have the pitch varying with radius. Nonetheless, this simple geometry will still allow us to study some of the effects of magnetic shear. Comparing to the actual connection length in the NSTX experiment shown in \cref{fig:nstx-heli-Lc}$(b)$ (with the figure adapted from Fig. 2 in \citet{boedo2014}), we see that the variation in $L_c$ with radius as $|\hat{s}|$ increases is approaching the variation of the connection length in the experiment. However, in the experiment the connection length varies even more than in the $\hat{s}=-10$ case over a shorter radial length, changing by a factor of almost four over a few centimeters. 

\begin{figure}[t]
    \centering
    \includegraphics[width=\textwidth]{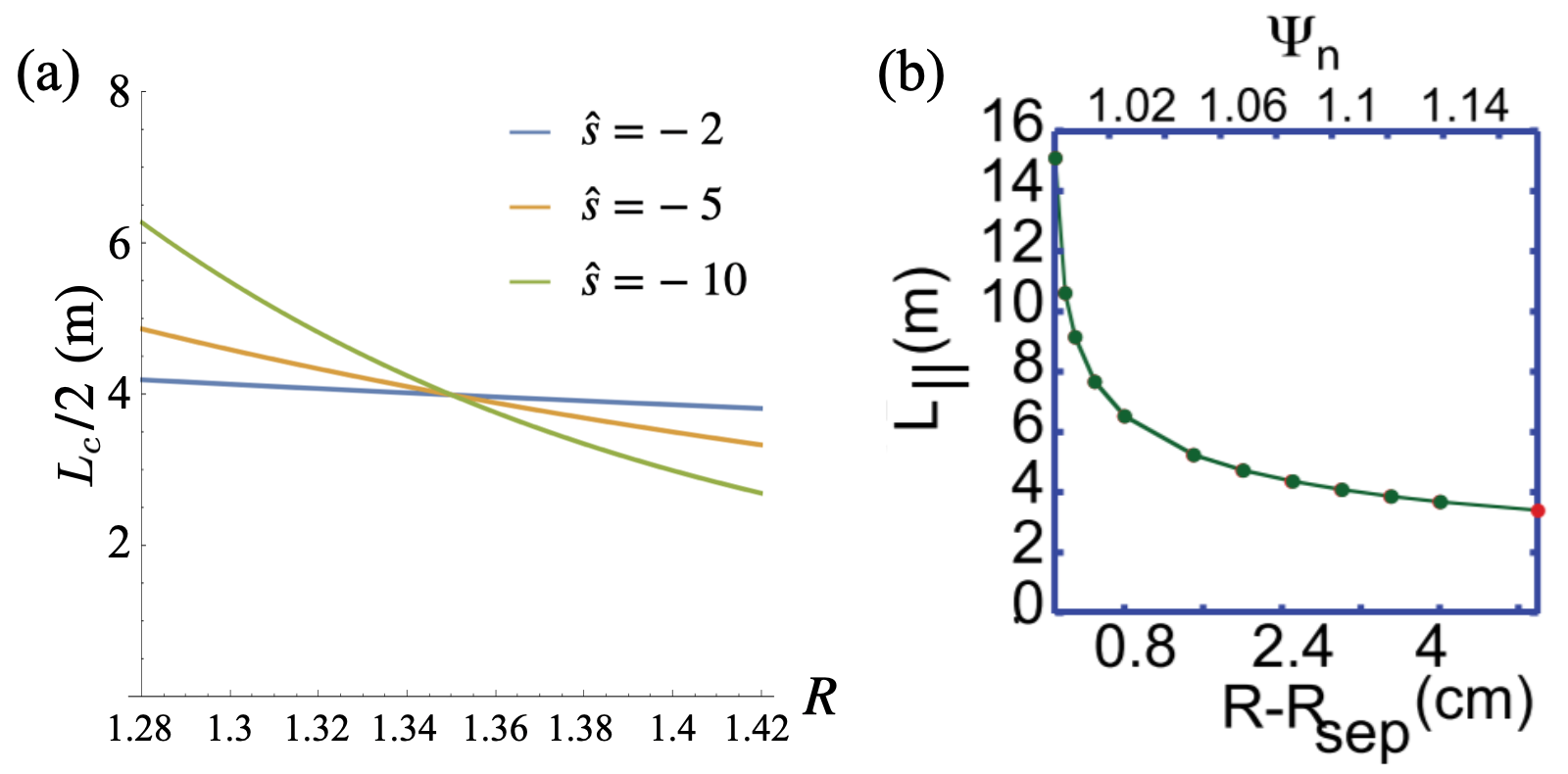}
    \caption[Connection length, $L_c=H B/B_v$, as a function of radius for several values of magnetic shear, $\hat{s}$, using NSTX-like parameters.]{$(a)$ Connection length, $L_c/2=H B/B_v/2$, as a function of radius for several values of magnetic shear, $\hat{s}$, using NSTX-like parameters: $R_0=0.85$ m, $x_0=R_0+a=1.35$ m, $B_0=0.5$ T, $H=2.4$ m. Here we choose $B_{v0}$ so that $L_c(x_0)=8$ m, which results in $B_{v0}\approx 0.1$ T. $(b)$ Connection length from the NSTX experiment, $L_\parallel$, defined as parallel length from midplane to lower divertor plate. Figure adapted from \citet{boedo2014}.}
    \label{fig:nstx-heli-Lc}
\end{figure}

Now we need field-line-following coordinates $(x,y,z)$, such that
\begin{equation}
    \vec{B} = \mathcal{C} \nabla x \times \nabla y.
\end{equation}
This can be achieved by choosing the coordinates to be\footnote{Alternatively, we could have chosen the $z$ coordinate to measure distance along the field line, with $z = Z/\sin\vartheta = Z B/B_v$, as we show in Appendix \ref{app:alt-heli-geo}. In this approach, the $z$ domain extent is given by the connection length; however, the connection length can vary with radius, so the $z$ domain extent would also need to vary with radius. Choosing the vertical height $Z$ for the $z$ coordinate does not have this issue (assuming the vertical height between the top and bottom end plates does not change with radius), so this is the approach that we take in the bulk of this chapter.}
\begin{gather}
    x = R, \qquad z = Z, \qquad y = x_0\left(\varphi - \frac{2\pi qZ}{H}\right)=x_0\left(\varphi - \frac{B_\varphi Z}{B_v R} \right). \label{heli-coord}
\end{gather}
The resulting gradient basis vectors are
\begin{gather}
    \nabla x = \uv{R}, \qquad
    \nabla y = -\frac{2\pi x_0 z}{H}\deriv{q}{x}\uv{R} + \frac{x_0}{x}\uvg{\varphi} - \frac{2\pi x_0 q}{H} \uv{Z} 
    , \qquad \nabla z = \uv{Z}, \label{gradvec}
\end{gather}
so that
\begin{equation}
    \nabla x \times \nabla y  = \frac{x_0}{x} \left(\frac{2\pi qx}{H}\uvg{\varphi} + \uv{Z}\right) =  \frac{x_0}{x} \left(\frac{B_\varphi}{B_v} \uvg{\varphi} + \uv{Z}\right). 
\end{equation}
Now taking $\mathcal{C} = B_v x/x_0$, we obtain the correct form of the field,
\begin{equation}
    \vec{B} = \frac{B_v x}{x_0} \nabla x \times \nabla y = B_\varphi\uvg{\varphi} + B_v\uv{Z}.
\end{equation}

\begin{figure}
    \centering
    \includegraphics[width=\textwidth]{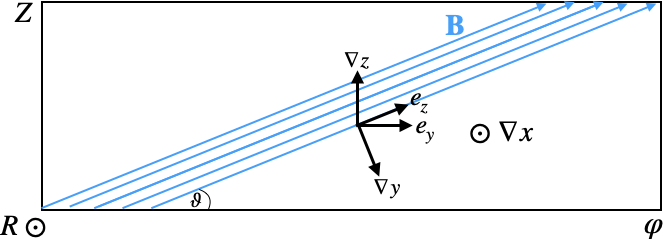}
    \caption[Diagram of field-aligned coordinate basis vectors in the $(\varphi,Z)$ plane for a helical magnetic geometry.]{Diagram of field-aligned coordinate basis vectors in the $(\varphi,Z)$ plane for a helical magnetic geometry. Note that vector magnitudes are not drawn to scale.}
    \label{fig:heli-geo}
\end{figure}

We can also calculate the tangent basis vectors
\begin{equation}
    \vec{e}_\alpha = \pderiv{\vec{X}}{\alpha} = \pderiv{R}{\alpha}\uv{R} + R\pderiv{\varphi}{\alpha}\uvg{\varphi} + \pderiv{Z}{\alpha}\uv{Z}, \qquad \alpha = (x,y,z),
\end{equation}
which gives
\begin{gather}
    \vec{e}_x = \uv{R} + \frac{2\pi x z}{H}\deriv{q}{x}\uvg{\varphi} = \uv{R} +\frac{B_\varphi z}{B_v x}\hat{s} \uvg{\varphi}, \qquad \vec{e}_y = \frac{x}{x_0}\uvg{\varphi},\qquad \vec{e}_z = \frac{2\pi q x}{H} \uvg{\varphi}+\uv{Z}=\frac{B_\varphi}{B_v}\uvg{\varphi}+\uv{Z}. \label{tanvec}
\end{gather}
A diagram of the field-aligned basis vectors in the $(\varphi,Z)$ plane is shown in \cref{fig:heli-geo}.

Finally, the mapping from the field-aligned coordinates $(x,y,z)$ to physical Cartesian coordinates $(X,Y,Z)$ is given by
\begin{gather}
    X = R \cos \varphi = x \cos\left(\frac{y}{x_0}+\frac{2\pi qz}{H}\right) = x\cos \left(\frac{y}{x_0}+\frac{B_\varphi z}{B_v x}  \right)  \label{c2pX}\\
     Y = R \sin \varphi = x \sin\left(\frac{y}{x_0}+\frac{2\pi qz}{H}\right) = x\sin \left(\frac{y}{x_0}+\frac{B_\varphi z}{B_v x}  \right)   \label{c2pY}\\
     Z = z. \label{c2pZ}
\end{gather} 
The geometry quantities of interest are 
\begin{gather}
\texttt{bmag} = B= B_v\sqrt{1+\frac{B_\varphi^2}{B_v^2}}\\
\texttt{cmag} = {\mathcal{C}} =  \frac{JB}{\sqrt{g_{zz}}}=\frac{B_v x}{x_0} \\
\texttt{b\_x} = b_x = \frac{g_{xz}}{\sqrt{g_{zz}}} =\frac{\hat{s} z B_0^2 R_0^2}{B_v B x^3}\\
\texttt{b\_y} = b_y = \frac{g_{yz}}{\sqrt{g_{zz}}}=\frac{B_0 R_0}{B x_0}  \\
\texttt{b\_z} = b_z = \sqrt{g_{zz}} = \frac{B}{B_v} \\
\texttt{gxx} = g^{xx} = 1 \\
\texttt{gxy} = g^{xy} = - \frac{\hat{s} z B_0 R_0 x_0}{B_v x^3}\\
\texttt{gyy} = g^{yy} = \frac{B_0^2 R_0^2 x_0^2}{B_v^2 x^6}\left(x^2 + \hat{s}^2 z^2\right) + \frac{x_0^2}{x^2} \\
\texttt{jacobPhase} = B_\parallel^* \approx B\\
\texttt{jacobGeo} = J = \sqrt{\det g_{ij}} = \frac{x}{x_0}.
\end{gather}
Note that in \gke, none of these expressions for the metric quantities are explicitly implemented; instead the mapping from computational to physical coordinates, \cref{c2pX,c2pY,c2pZ}, is supplied as an input and then the metric quantities are computed via automatic differentiation operations. This enables the flexibility to use more complicated mappings where analytical expressions for the metric quantities may not be available (see \emph{e.g.} \cref{sec:solovev}).

We can now compute how various terms in the equations of motion are affected by the geometry. For example, the components of the magnetic (curvature and $\nabla B$) drift are
\begin{gather}
    \vec{v}_d\cdot\nabla x = 0 \label{vdx}\\
    \vec{v}_d\cdot\nabla y = \frac{B}{J B_v}\left[- \frac{mv_\parallel^2+\mu B}{q B}\left(\frac{1}{x}+\frac{B_v^2}{B^2}(1+\hat{s})\right)+\frac{mv_\parallel^2}{qB} \frac{B_v^2}{B^2 x}(2+\hat{s})   \right] \label{vdy} \\
    \vec{v}_d\cdot\nabla z = \frac{B_\varphi}{B}\frac{m v_\parallel^2+\mu B}{qB}\left(\frac{1}{x} + \frac{B_v^2}{B^2}(1+\hat{s})\right) \label{vdz}.
\end{gather}
Comparing these terms to \cref{simplevd} that we used for the simplified helical geometry in \cref{ch:nstx-results}, we see that the first term in $\vec{v}_d\cdot\nabla y$ in \cref{vdy} is the same except for a factor of $B/(JB_v)$. We also have some new terms in $\vec{v}_d\cdot\nabla y$, which are small corrections when $B_v \ll B_\varphi\sim B$ as assumed in \cref{ch:nstx-results}. We also have a finite $\vec{v}_d\cdot\nabla z$, unlike in the simplified geometry treatment, though this term is small compared to $\vec{v}_d\cdot\nabla y$ when $B_v\ll B_\varphi\sim B$.

\subsection{Simulation results: dependence on magnetic shear in helical configuration}

In this section we present preliminary electrostatic gyrokinetic simulations in the helical configuration with magnetic shear described above. We perform a scan of the magnetic shear parameter, taking $\hat{s}=\{-2,-5,-10\}$, with the geometry becoming more sheared as $|\hat{s}|$ increases. We use NSTX-like geometry parameters: $R_0=0.85$ m, $x_0=R_0+a=1.35$ m, $B_0=0.5$ T, $H=2.4$ m, and we choose $B_{v0}=B_v(x_0)$ so that $L_c(x_0)=8$ m, resulting in $B_{v0}\approx 0.1$ T. The resulting connection length as a function of radius is shown in \cref{fig:nstx-heli-Lc}. The domain extents in the radial, binormal, and parallel directions, respectively, are $1.26 \leq x \leq 1.42$ m ($L_x \approx 56\rho_{\mathrm{s}0}$), $-0.485 \leq y \leq 0.485$ m ($L_y \approx 100\rho_{\mathrm{s}0} H/L_c(x_0)$), and $-H/2 \leq z \leq H/2$. All other parameters are the same as used in the base case ($\hat{n}=1$) from \cref{sec:power-scans}, including the source power $P_\mathrm{src} = P_\mathrm{SOL} L_y/(2\pi R_c)= 0.62$ MW and the source profile with a Gaussian peak at $x=1.3$ m.

\begin{figure}[t]
    \centering
    \includegraphics[width=\textwidth]{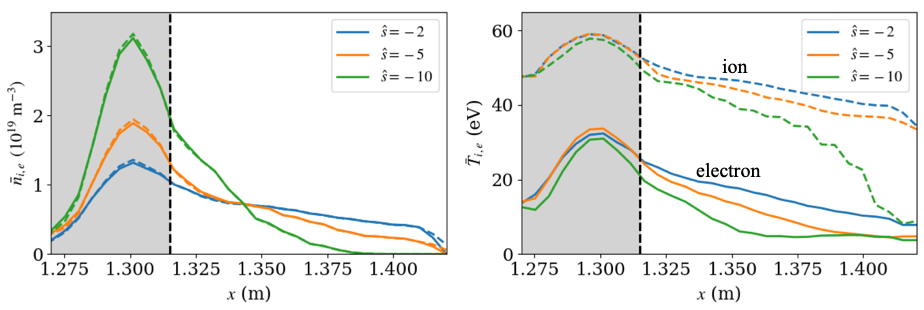}
    \caption[Midplane density and  temperature profiles for several values of magnetic shear, $\hat{s}$.]{Time-averaged midplane density and  temperature profiles for several values of magnetic shear, $\hat{s}$. The profiles steepen as $|\hat{s}|$ increases and the geometry becomes more sheared.}
    \label{fig:shat-profiles}
    \vspace{.5cm}
    \centering
    \includegraphics[width=\textwidth]{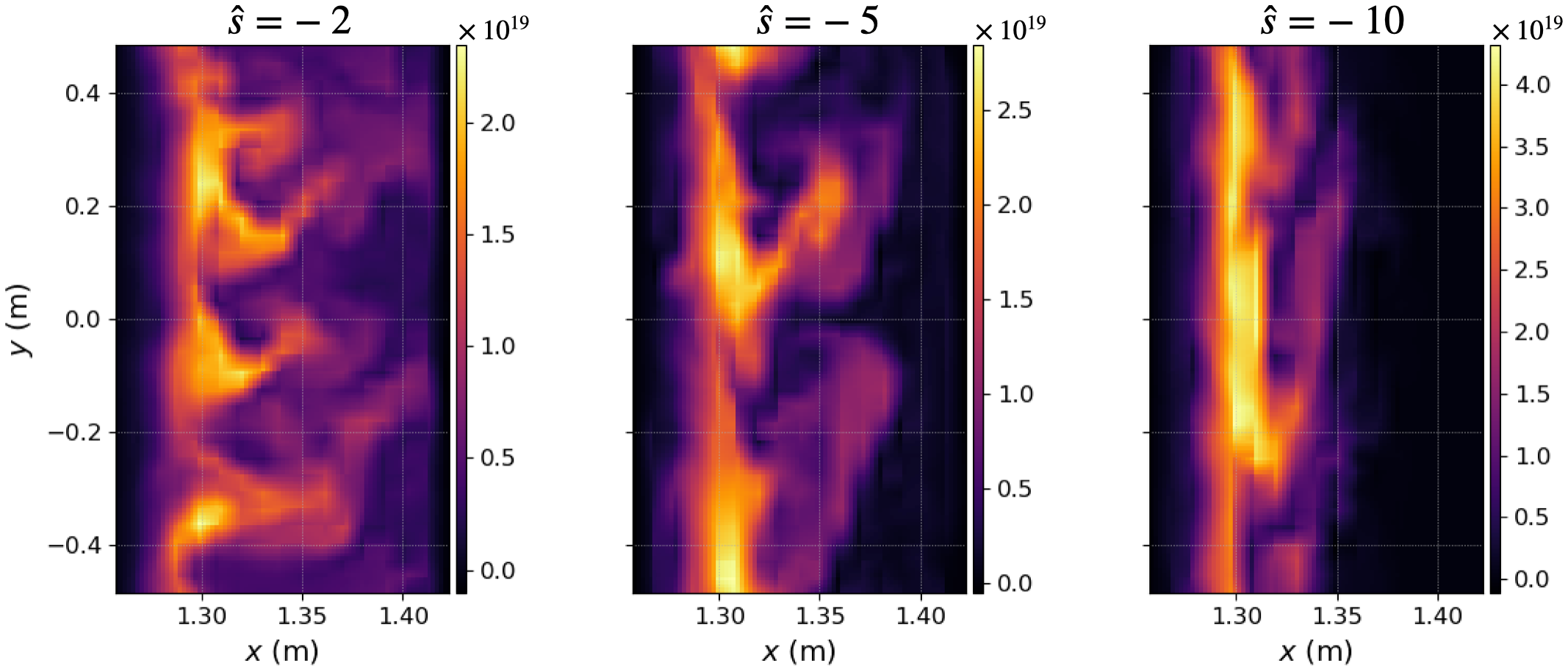}
    \caption[Snapshots of electron density at the midplane for $\hat{s}=-2,-,5,-10$.]{Snapshots of electron density at the midplane for $\hat{s}=-2,-,5,-10$.}
    \label{fig:shat-snapshots}
\end{figure}

In \cref{fig:shat-profiles} we show time-averaged density and temperature radial profiles for each case. The profiles for the $\hat{s}=-2$ case are similar to the profiles from the base case in \cref{sec:power-scans}, which used a simplified geometry neglecting magnetic shear and assumed a constant connection length. This is somewhat expected since \cref{fig:nstx-heli-Lc} shows that the connection length in the $\hat{s}=-2$ case varies little over the domain. As we move to more sheared geometries, the density profiles steepen, with the peak midplane density more than doubling between the $\hat{s}=-2$ and $\hat{s}=-10$ cases.  

Snapshots of the electron density at the midplane for each case are shown in \cref{fig:shat-snapshots}. While the $\hat{s}=-2$ case looks similar to the cases from \cref{ch:nstx-results} with blobs moving radially outwards, the $\hat{s}=-10$ case shows little evidence of radial transport. We measure the radial $E\times B$ particle flux near the midplane in \cref{fig:shat-exb-flux} and confirm that indeed particle transport is reduced as $|\hat{s}|$ increases. Recalling the blob dynamics discussion from \cref{sec:blob-dynamics}, magnetic shear can short-circuit the blob polarization by allowing currents to close through the thin sheared part of the blob, resulting in slower blobs. This seems to be consistent with the picture here, with blob transport getting weaker in more sheared cases.

\begin{figure}[t]
    \centering
    \includegraphics[width=.6\textwidth]{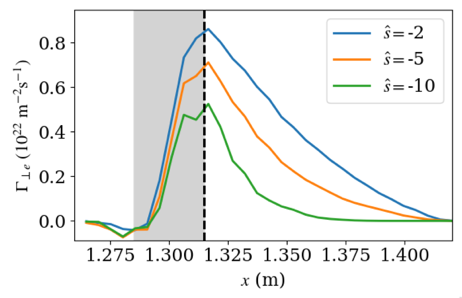}
    \caption[Radial electron $E\times B$ particle flux near the midplane for several values of $\hat{s}$.]{Time-averaged radial electron $E\times B$ particle flux near the midplane. The cross-field transport decreases as the geometry becomes more sheared at larger $|\hat{s}|$, resulting in the steeper profiles seen in \cref{fig:shat-profiles}.}
    \label{fig:shat-exb-flux}
    \vspace{.5cm}
    \centering
    \includegraphics[width=\textwidth]{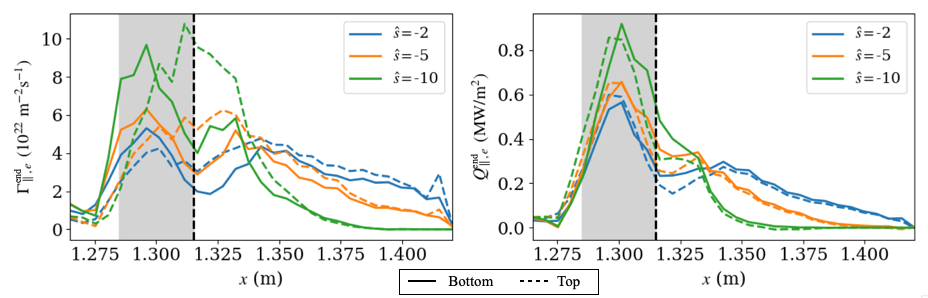}
    \caption[Electron particle and heat fluxes to the end plates for several values of $\hat{s}$.]{Time-averaged electron particle and heat fluxes to the bottom (solid) and top (dashed) end plates. The peak fluxes increase as $|\hat{s}|$ increases, consistent with less cross-field transport upstream.}
    \label{fig:shat-end-flux}
\end{figure}

\begin{figure}[t!]
    \centering
    \includegraphics[width=.8\textwidth]{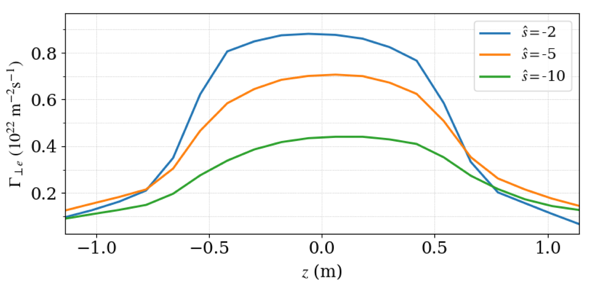}
    \caption[$\hat{s}$ scan: Electron radial $E\times B$ particle flux profiles along the field line.]{Time-averaged electron radial $E\times B$ particle flux profiles along the field line, taken just outside the source region at $x\sim 1.32$ m. There is a slight asymmetry in the profiles, with slightly more radial transport at $z>0$ than $z<0$ for the $\hat{s}=-5$ and $\hat{s}=-10$ cases. This is consistent with the differences in the particle fluxes to the bottom and top end plates shown in \cref{fig:shat-end-flux}.}
    \label{fig:shat-exb-flux-z}
    \vspace{.5cm}
    \centering
    \includegraphics[width=.8\textwidth]{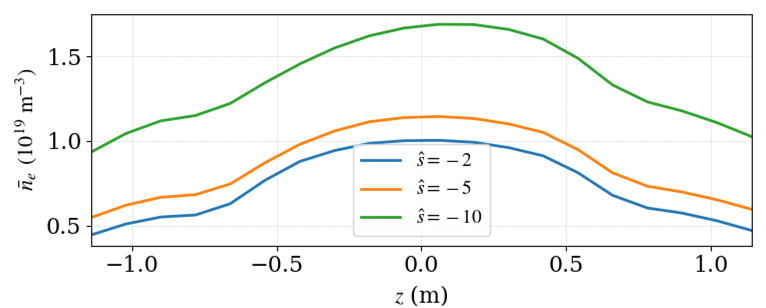}
    \caption[$\hat{s}$ scan: Electron density profiles along the field line.]{Time-averaged electron density profiles along the field line, taken just outside the source region at $x\sim 1.32$ m. There is a slight asymmetric shift in the profiles, with slightly higher density near the top of the domain ($z=H/2=1.4$ m) especially for the more sheared cases.}
    \label{fig:shat-ne-z}
\end{figure}

As a result of weaker cross-field transport, the peak particle and heat fluxes to the end plate increase in the more sheared cases, as shown in \cref{fig:shat-end-flux}. Note that we have separately shown the fluxes to the top and bottom end plates. There are some differences in the particle flux profiles between the ends, especially for the $\hat{s}=-10$ case, where there is a noticeable shift in the peak to higher $x$ between the bottom and top. When we examine how the radial particle flux (taken just outside the source region near $x=1.32$ m) varies along the field line in \cref{fig:shat-exb-flux-z}, we see asymmetry as well. There slightly more radial transport at $z>0$ than $z<0$ for the $\hat{s}=-5$ and $\hat{s}=-10$ cases. This is consistent with the radial shift in the end plate particle fluxes to higher $x$ from bottom to top.

The average electron density also shows asymmetry along the field line, as shown in \cref{fig:shat-ne-z}. Here we have again evaluated the profiles just outside the source region near $x=1.32$ m. There is a slight shift in the profiles to higher $z$, with the shift larger in the more sheared cases. One possible reason for the asymmetry is the presence of a vertical component of the $E\times B$ drift, 
\begin{equation}
    \vec{v}_E\cdot\nabla z = \frac{1}{J B_\parallel^*}\left(b_x \pderiv{\Phi}{y} - b_y \pderiv{\Phi}{x} \right) = \frac{x_0}{x B} \left(\frac{\hat{s} z B_0^2 R_0^2}{B_v B x^3}\pderiv{\Phi}{y} - \frac{B_0 R_0}{B x_0}\pderiv{\Phi}{x}\right).
\end{equation}
This term was not present in the earlier simplified geometry simulations from \cref{sec:emgk-res}. It has been suggested that this term is responsible for asymmetry between top and bottom profiles in the Helimak \citep{bernard2020}.

\section{Solov'ev model analytical equilibria in the SOL} \label{sec:solovev}

The Grad-Shafranov equation in cylindrical $(R,Z,\phi)$ coordinates,
\begin{equation}
-\mu_0 R J_\phi = -\nabla^* \Psi  = - R^2 \nabla \cdot \frac{1}{R^2}\nabla \Psi = \mu_0 R^2 p'(\Psi) + I I'(\Psi),
\end{equation}
relates the equilibrium, defined by the poloidal flux function $\Psi$, to the pressure and current profiles, $p(\Psi)$ and $I(\Psi)$ respectively. 
Given $\Psi(R,Z)$, one particular choice for the field-aligned coordinates are $(x,y,z) = (\Psi,-\alpha,\theta)$, where $x=\Psi$ is the poloidal flux, $z=\theta$ is a generalized poloidal angle, and $y=-\alpha$ is a field-line-labeling coordinate defined so that the Clebsch representation of the magnetic field is given by
\begin{equation}
    \vec{B} = \nabla \alpha \times \nabla \Psi = \nabla x \times \nabla y,
\end{equation}
with $\mathcal{C}=1$ for this choice of coordinates.
Note that for an axisymmetric system, the background magnetic field can also be expressed as
\begin{equation}
    \vec{B} = I(\Psi)\nabla\phi + \nabla\Psi\times\nabla\phi, \label{axiB}
\end{equation}
where $\phi$ is the toroidal angle and $I(\Psi) = R B_\phi$. Thus in practice, the functions $\Psi(R,Z)$ and $B_\phi(R,Z)$ are sufficient to determine the magnetic geometry. In this section we will take an analytical Solov'ev solution of the Grad-Shafranov equation for $\Psi(R,Z)$, and show how to compute the remaining $\theta$ and $\alpha$ coordinates. In principle, the procedure outlined could be used for an arbitrary $\Psi(R,Z)$ profile, including one from a numerical equilibrium file generated by \emph{e.g.} EFIT \citep{lao1985,lao1990}.

\begin{figure}[t!]
    \centering
    \includegraphics[width=.6\textwidth]{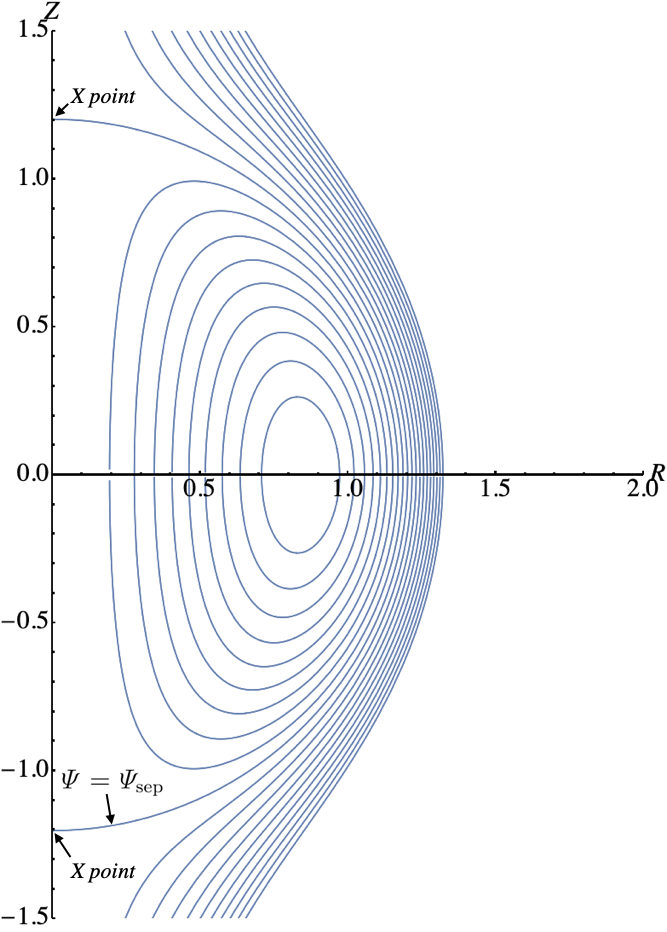}
    \caption[Flux surfaces for an analytical Solov'ev equilibrium with NSTX-like parameters.]{Flux surfaces (surfaces of constant $\Psi$) for the analytical Solov'ev equilibrium given by \cref{solovev}, with NSTX-like parameters: $R_0=0.85$, $B_0=0.55$, $\kappa_0=2$, and $\bar{q}=2$. For the open field lines, we will take the end-plates to be at $Z=\pm Z_\text{end}=\pm 1.5$.}
    \label{fig:solovev-nstx-psi-full}
\end{figure}

Taking a vacuum field (which is a good approximation in the SOL) so that $I=$ const and approximating $p\propto \Psi$, we obtain an analytical Solov'ev solution \citep{chance1978,jardin2010},
\begin{equation}
\Psi(R,Z) = \frac{B_0}{2R_0^2\kappa_0 \bar{q}}\left(R^2 Z^2 +\frac{\kappa_0^2}{4}(R^2-R_0^2)^2\right), \label{solovev}
\end{equation}
where $R_0$ is the major radius of the magnetic axis, $B_0$ is the toroidal field strength at the magnetic axis, $\kappa_0$ is the ellipticity at the  axis, and $\bar{q}$ is the safety factor at the axis. The resulting flux surfaces for NSTX-like parameters are shown in \cref{fig:solovev-nstx-psi-full}. The flux surfaces are up-down symmetric, and there are closed and open surfaces, separated by a separatrix given by the surface $\Psi=\Psi_\text{sep} = B_0 \kappa_0 R_0^2/(8 \bar{q})$. In this (unphysical) equilibrium, there are X-points where the separatrix intersects the $R=0$ axis. However, we will focus only on the open flux surfaces (sufficiently far away from the X-points), that is, those with $\Psi > \Psi_\text{sep}$. We will assume that the field lines on these open surfaces terminate on end-plates at $Z= \pm Z_\text{end}$. For a given surface with $\Psi>\Psi_\text{sep}$, we can then parametrize the surface with
\begin{equation}
    R(\Psi,Z) = \sqrt{R_0^2 + \frac{2 Z^2}{\kappa_0^2} + \frac{2}{\kappa_0^2}\sqrt{\frac{2 \kappa_0^3  \bar{q} R_0^2}{B_0}\Psi - \kappa_0^2 R_0^2 Z^2 + Z^4}}. \label{psi-param}
\end{equation}

We will now calculate the generalized poloidal angle, $\theta$. For this, first consider a magnetic surface coordinate system $(\Psi, \theta, \phi)$. The line element in these coordinates is given by
\begin{equation}
    \dx{\vec{\ell}} = \pderiv{\vec{R}}{\Psi}\dx{\Psi} + \pderiv{\vec{R}}{\theta}\dx{\theta} +\pderiv{\vec{R}}{\phi}\dx{\phi}. 
\end{equation}
On a magnetic surface we have $\dx{\Psi}=0$ and in a poloidal plane we have $\dx{\phi}=0$, so the line element on the surface in the poloidal plane is simply
\begin{equation}
    \dx{\vec{\ell}}_p = \pderiv{\vec{R}}{\theta}\dx{\theta} = {J}\nabla\phi\times\nabla\Psi\, \dx{\theta},
\end{equation}
with magnitude
\begin{equation}
    \dx{\ell_p} = |\dx{\vec{\ell}_p}|=|{J}\nabla\phi\times\nabla\Psi| \dx{\theta} = \frac{|{J}\nabla \Psi|}{R}\dx{\theta}.
\end{equation}
Now we can find $\theta$ by integrating
\begin{equation}
    \dx{\theta} = \frac{R}{|{J}\nabla\Psi|}\dx{\ell_p}
\end{equation}
along contours of $\Psi$,
\begin{equation}
    \theta =  \int_{\theta_0}^{\theta}\dx{\theta}'= \int_{\theta_0}^{\theta} \frac{R}{|{J}\nabla\Psi|}\dx{\ell_p}. \label{thetadef0}
\end{equation}
Using the flux surface parametrization from \cref{psi-param}, the differential length along contours of $\Psi$ is given by
\begin{equation}
    \dx{\ell_p} = \sqrt{1+\left(\pderiv{R(\Psi,Z)}{Z}\right)^2}\dx{Z}. 
\end{equation}

Instead of prescribing the form of the generalized poloidal angle and subsequently computing the Jacobian ${J}=(\nabla\Psi\times\nabla\theta\cdot\nabla\phi)^{-1}$, we can instead prescribe the Jacobian to give desired properties of the generalized poloidal angle \citep{jardin2010}. Here, we will choose 
\begin{equation}
    {J}=s(\Psi)\frac{R}{|\nabla\Psi|}, \label{jacob-phi}
\end{equation}
which gives an equal-arc-length poloidal angle in $(-\pi,\pi]$, with
\begin{equation}
    s(\Psi) = \frac{1}{2\pi}\oint \dx{\ell_p} = \frac{1}{\pi} \int_{-Z_\text{end}}^{Z_\text{end}} \sqrt{1+\left(\pderiv{R(\Psi,Z')}{Z'}\right)^2}\dx{Z'}
\end{equation}
a normalization factor. This means that on a particular flux surface, the arc length of each $\Delta \theta$ segment will be equal (see \cref{fig:solovev-nstx-psi-sol}). Inserting this Jacobian definition into \cref{thetadef0}, the poloidal angle is then given by
\begin{equation}
    \theta(R,Z) = \frac{1}{s(\Psi(R,Z))}\int_{-Z_\text{end}}^{Z} \sqrt{1+\left(\pderiv{R(\Psi,Z')}{Z'}\right)^2}\dx{Z}'.
\end{equation}
These integrals have no closed form in general and must be evaluated numerically.

Now we will define the third coordinate, $\alpha$, such that $\vec{B} = \nabla \alpha \times \nabla \Psi$. To do this, we take $\alpha$ to be of the form \citep{kruskal1958}
\begin{equation}
    \alpha = \phi - q(\Psi)\theta -\nu(\Psi,\theta,\phi),
\end{equation}
where $\nu$ is a to-be-determined function that is periodic in $\theta$ and $\phi$, and $q(\Psi)$ is the global safety factor, defined as the poloidal average of the local safety factor $\hat{q}(\Psi,\theta)$,
\begin{equation}
    q(\Psi) = \frac{1}{2\pi}\int_0^{2\pi} \hat{q}(\Psi,\theta)\, \dx{\theta},
\end{equation}
with
\begin{equation}
    \hat{q}(\Psi,\theta) = \frac{\vec{B}\cdot\nabla\phi}{\vec{B}\cdot\nabla\theta} = -\mathcal{J} \vec{B}\cdot\nabla\phi = -I(\Psi)\frac{\mathcal{J}}{R^2} = \frac{-I(\Psi)s(\Psi)}{R |\nabla\Psi|}.
\end{equation}
It is convenient to define a new toroidal angle $\zeta = \phi - \nu$, so that we have\footnote{Alternatively, we could have defined $\alpha$ in terms of the physical toroidal angle $\phi$, \emph{i.e.} $\zeta=\phi$, by modifying the poloidal coordinate to be $\theta' = \theta + \nu/q$. In either case, the result is straight field lines with slope $q$, be it in the $(\theta,\zeta)$ plane or the $(\theta',\phi)$ plane.}
\begin{equation}
    \vec{B}= \nabla \alpha \times \nabla \Psi = \nabla\Psi \times \nabla (q\theta-\zeta) .
\end{equation}
Here we can see that the field lines are straight lines with slope $q$ in the $(\theta,\zeta)$ plane, given by $\alpha = \zeta - q(\Psi)\theta = \text{const}$. 

\begin{figure}[t!]
    \centering
    \includegraphics[width=.6\textwidth]{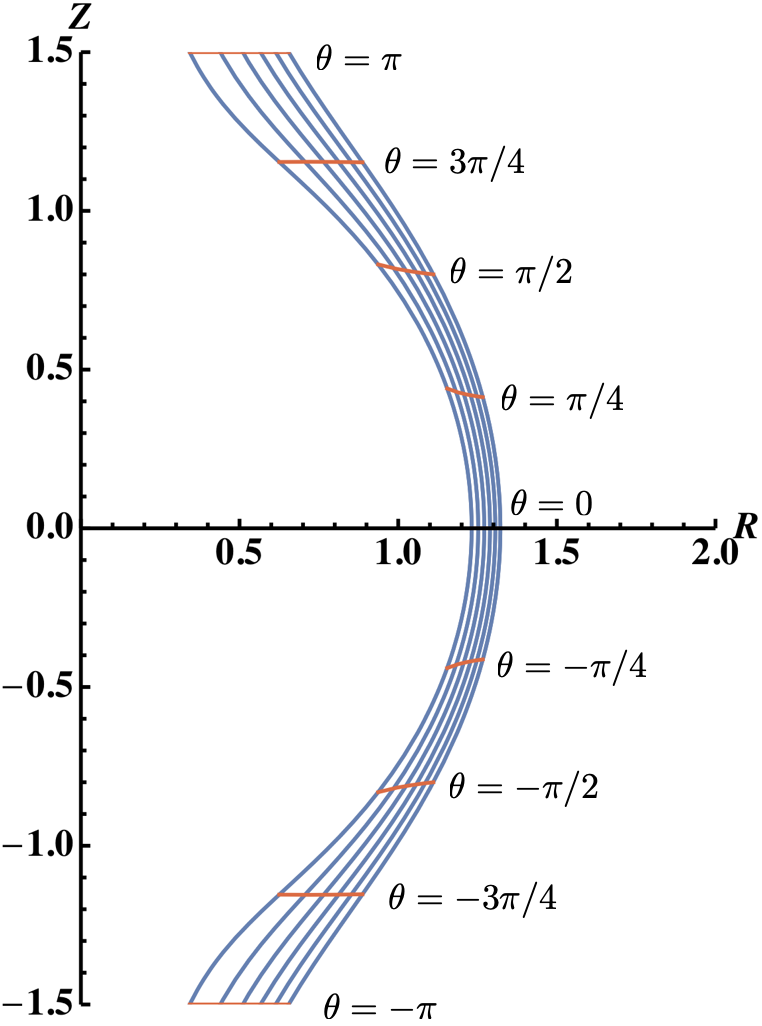}
    \caption[SOL flux surfaces for an analytical Solov'ev equilibrium with NSTX-like parameters.]{Flux surfaces with $\Psi_N = \Psi/\Psi_\text{sep}=(1.2,1.4,1.6,1.8,2.0)$ are shown in blue. Lines of constant $\theta$ are shown in red. Here, $\theta$ is defined to be an equal-arc-length poloidal angle. This means that on a particular flux surface, the arc length of each $\Delta \theta=\pi/4$ segment is equal.}
    \label{fig:solovev-nstx-psi-sol}
\end{figure}

To compute $\alpha$, we first note that
\begin{equation}
    \nabla \alpha = \pderiv{\alpha}{\theta}\nabla \theta + \pderiv{\alpha}{\Psi}\nabla\Psi +\pderiv{\alpha}{\phi}\nabla\phi,
\end{equation}
which gives 
\begin{equation}
    \vec{B} =\nabla\Psi\times  \nabla\alpha = \pderiv{\alpha}{\theta}\nabla\Psi\times \nabla\theta + \pderiv{\alpha}{\phi}\nabla\Psi\times\nabla\phi.
\end{equation}
Now notice
\begin{equation}
    \vec{B}\cdot\nabla\phi = \pderiv{\alpha}{\theta}\nabla\Psi\times\nabla\theta\cdot\nabla\phi = \frac{1}{{J}}\pderiv{\alpha}{\theta},
\end{equation}
and from \cref{axiB}, we also have 
\begin{equation}
    \vec{B}\cdot\nabla\phi = I(\Psi)|\nabla \phi|^2 = \frac{I(\Psi)}{R^2}.
\end{equation}
Thus we can integrate along $\theta$ (at constant $\Psi$) to find $\alpha$,
\begin{equation}
    \alpha = C(\Psi,\phi)+\int_0^\theta \pderiv{\alpha}{\theta'} \dx{\theta'} = \phi-I(\Psi)\int_0^\theta \frac{{J}}{R^2}\dx{\theta'} = \phi - R B_\phi \int \frac{1}{|\nabla \Psi|R}\dx{\ell_p},
\end{equation}
where we have taken the constant of integration to be $C(\Psi,\phi)=\alpha(\theta=0,\Psi,\phi)=\phi$. As we did for $\theta$, this integral can be computed by integrating along the contours of $\Psi$ parametrized by $Z$,
\begin{equation}
    \alpha(R,Z,\phi) = \phi - R B_\phi \int_{0}^{Z} \frac{1}{|\nabla \Psi|R(\Psi,Z')}\sqrt{1+\left(\pderiv{R(\Psi,Z')}{Z'}\right)^2}\dx{Z'}.
\end{equation}

\begin{figure}[t!]
    \centering
    \includegraphics[width=.8\textwidth]{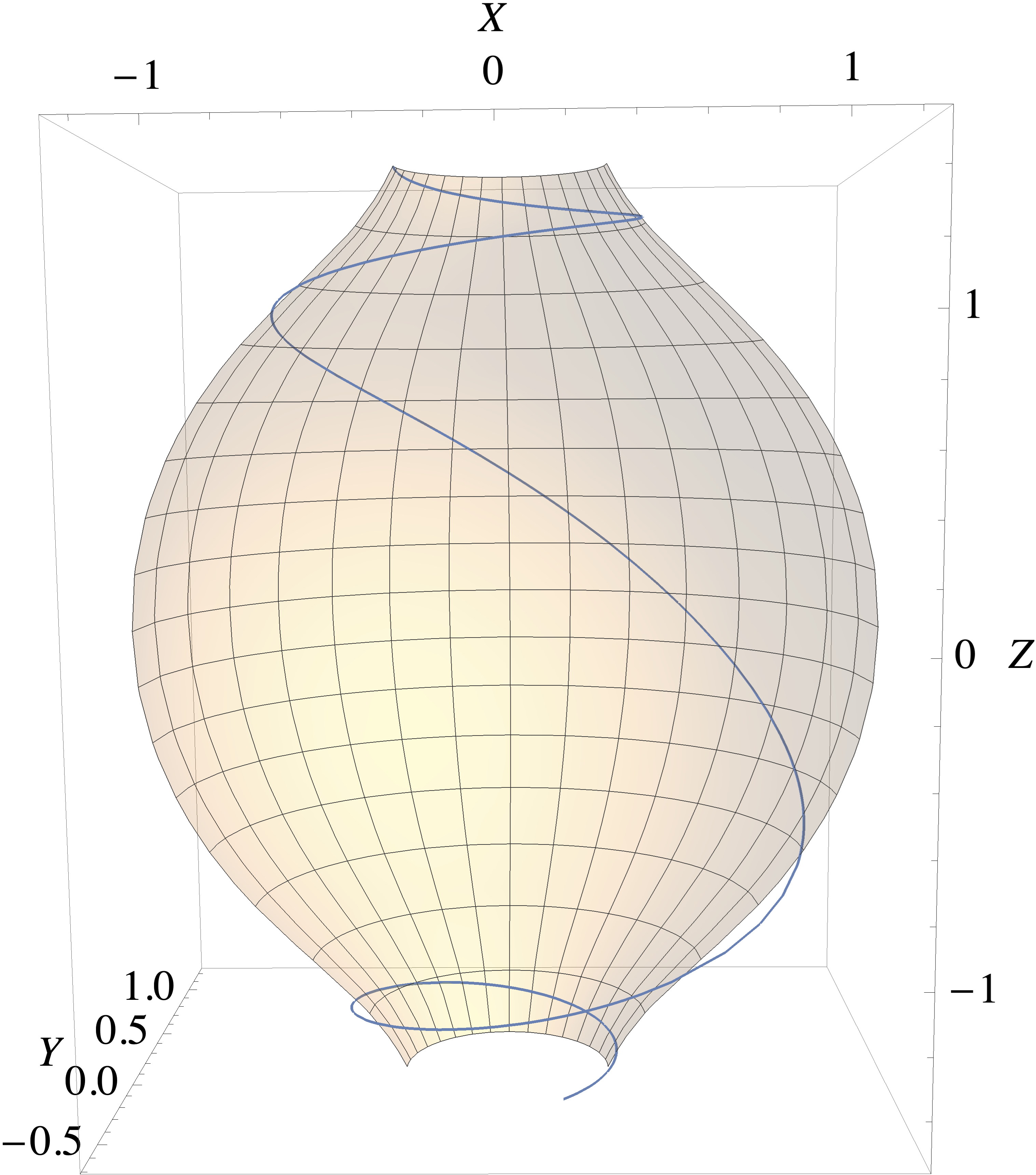}
    \caption[3D rendering of flux surface with field-aligned coordinate tracing.]{Selecting a line of constant $\alpha$ and constant $\Psi$ traces a field line (blue) from the bottom end-plate (at $\theta=-\pi$ and $Z=-Z_\text{end}=-1.5$) to the top end-plate (at $\theta=\pi$ and $Z=Z_\text{end}=1.5$). A rendering of part of the constant $\Psi_N=1.2$ surface is shown in the background.}
    \label{fig:solovev-nstx-line-trace}
\end{figure}

We now have expressions for $\Psi(\vec{R})$, $\alpha(\vec{R})$, and $\theta(\vec{R})$. We can thus define a field-aligned coordinate system with $(x,y,z) = (\Psi,-\alpha,\theta)$. (The difference of sign between $y$ and $\alpha$ is a matter of convention, and we follow \cite{beer1995field} here). The expressions for $\alpha$ and $\theta$ involve integrals that must be evaluated numerically in most cases. In \cref{fig:solovev-nstx-psi-sol}, we show lines of constant $\theta$ (and $\alpha$) in the poloidal plane for the open-field-line region of the NSTX-like equilibrium shown in \cref{fig:solovev-nstx-psi-full}, demonstrating the equal-arc-length poloidal angle. In \cref{fig:solovev-nstx-line-trace}, we show that a line of constant $\alpha$ and constant $\Psi$ traces a field line from the bottom end-plate to the top end-plate. Further, note that the connection length can be computed from
\begin{equation}
    L_c = \oint \sqrt{g_{zz}} \dx{\ell_p} = \int_{-Z_\text{end}}^{Z_\text{end}} \sqrt{g_{zz}(R,Z')}\sqrt{1+\left(\pderiv{R(\Psi,Z')}{Z'}\right)^2}\dx{Z}'.
\end{equation}
\cref{fig:solovev-Lc} shows the connection length as a function of the poloidal flux $\Psi$ for the NSTX-like equilibrium. 

\begin{figure}[t!]
    \centering
    \includegraphics[width=.7\textwidth]{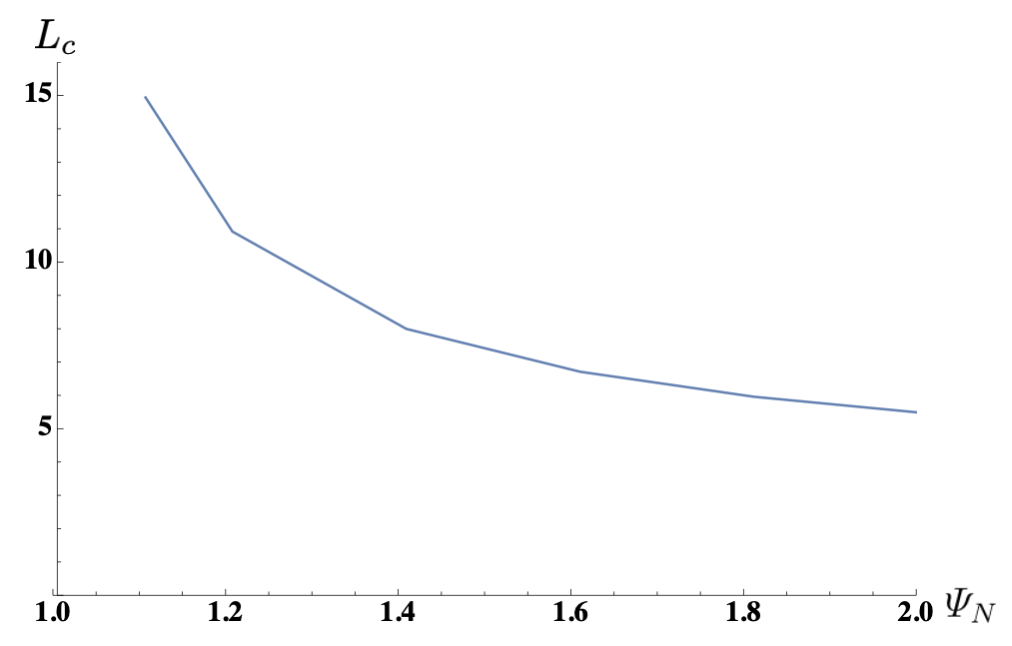}
    \caption{Connection length, $L_c$, as a function of normalized poloidal flux $\Psi_N=\Psi/\Psi_\text{sep}$ for the NSTX-like equilibrium.}
    \label{fig:solovev-Lc}
\end{figure}

We also need derivatives of the coordinates to compute metric quantities, which can also be computed numerically via automatic differentiation or finite differencing; alternatively, integral expressions for the derivatives can also be derived via the Leibniz integral rule. In \cref{fig:solovev-allGeo}, we show some of the resulting geometric quantities. There is significant flux expansion and magnetic shear for flux surfaces near the separatrix and X-points, which causes some of the metric quantities to diverge near the top-left and bottom-left corners of the $(\Psi,\theta)$ domain as $\Psi$ approaches $\Psi_\text{sep}$. Finally, note that the Jacobian of the $(\Psi,\alpha,\theta)$ coordinate system, 
\begin{equation}
    J = (\nabla \Psi \times \nabla \alpha \cdot \nabla \theta)^{-1} = s(\Psi)\frac{R}{|\nabla \Psi|}, \label{jacob-alpha}
\end{equation}
is equivalent to the Jacobian defined in \cref{jacob-phi} for the $(\Psi,\theta,\phi)$ magnetic surface coordinate system.

\begin{figure}[t!]
    \centering
    \includegraphics[width=\textwidth]{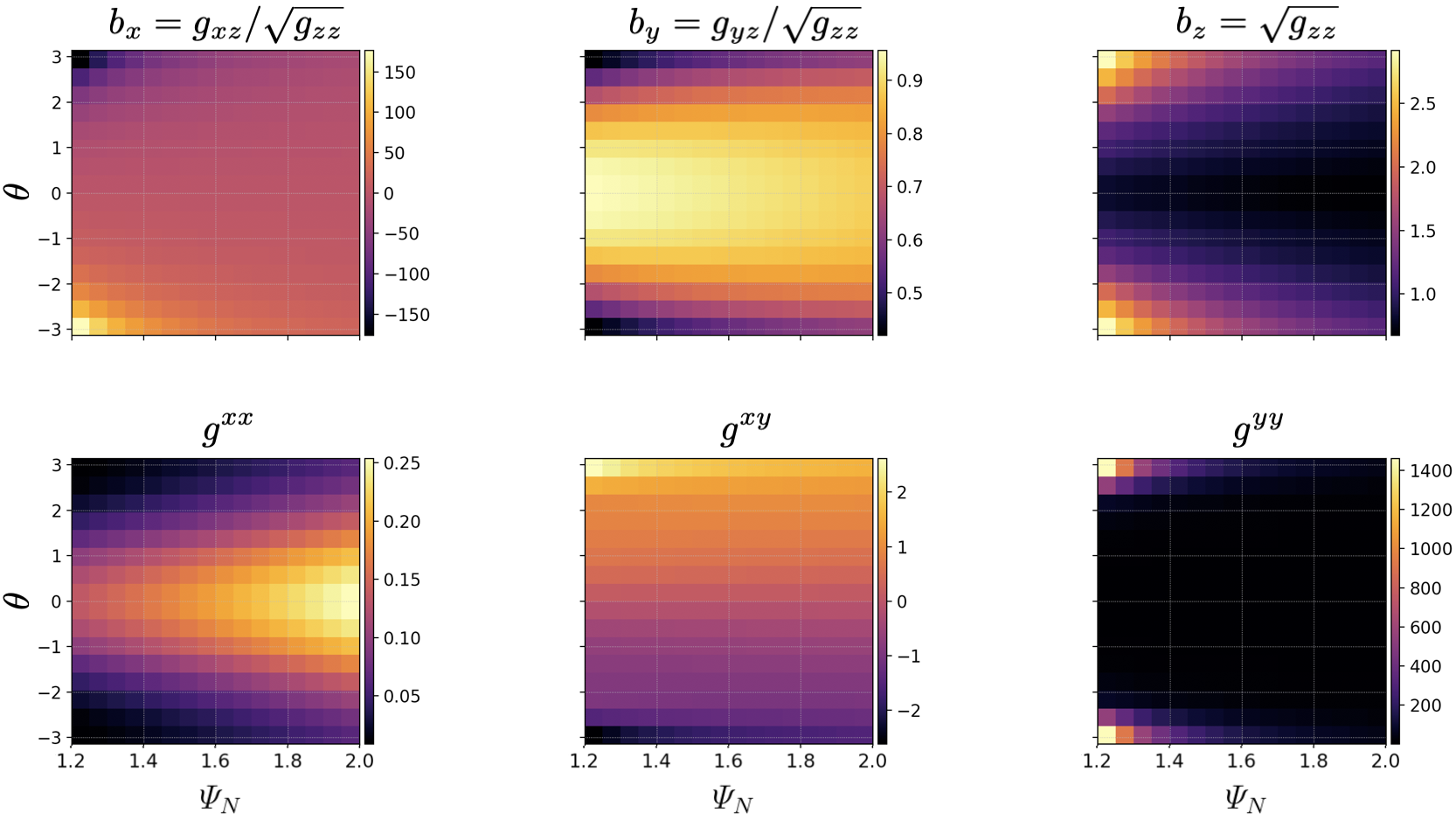}
    \caption[Geometric quantities resulting from the field-aligned coordinate system for the NSTX-like SOL equilibrium.]{Geometric quantities resulting from the field-aligned coordinate system as derived in this section for the NSTX-like equilibrium shown in \cref{fig:solovev-nstx-psi-full}.}
    \label{fig:solovev-allGeo}
\end{figure}

Linear and nonlinear simulations using the shaped SOL geometry presented in this section are left to future work. An important question is how close the simulation domain can approach the X-point before the significant magnetic shear and metric divergence near the X-point become numerically untenable.

\section{Analytical concentric circular equilibria} \label{sec:circ-geo}

In the tokamak core, many early results and inter-code benchmarks have used an \emph{ad hoc} analytical equilibrium model with circular concentric flux surfaces. This is commonly given by the popular $s-\alpha$ model with no Shafranov shift ($\alpha=0$).  This type of geometry has also been used to model circular SOL plasmas with a limiter at the inboard midplane, modeling an inner-wall-limited plasma \citep{ricci2013,halpern2013,francisquez2017global,zhu2017global}. Here we use a slightly modified circular equilibrium model that is more consistent than the $s-\alpha$ model in the large aspect ratio approximation \citep{lapillonne2009}.

As in the previous section, we start from a Solov'ev solution of the form of \cref{solovev}. We use a toroidal coordinate system $(r,\theta,\phi)$, where $(r,\theta)$ are the minor radius and poloidal angle coordinates in the $(R,Z)$ plane such that $R= R_0 + r\cos\theta$ and $Z=r\sin\theta$, and $\phi$ is the toroidal angle. Taking the large aspect ratio limit $R_0/a\gg 1$ and $\kappa_0=1$ for circular flux surfaces, we obtain
\begin{equation}
    \Psi = \frac{B_0}{2 \bar{q}}r^2. \label{psicirc}
\end{equation}
The resulting magnetic field is
\begin{equation}
\vec{B} = \nabla\phi\times\nabla\Psi + R B_\phi \nabla \phi = \frac{R_0 B_0}{R} \left[\vec{e}_\phi + \frac{r}{R_0 \bar{q}} \vec{e}_\theta\right].
\end{equation}
While \cref{psicirc} implies $\bar{q}=\text{const}$, we will generalize the equilibrium to allow $\bar{q}=\bar{q}(r)$, which is related to the true safety factor via \citep{lapillonne2009}
\begin{equation}
    q(r) = \frac{1}{2\pi}\int_0^{2\pi} \frac{{\bf B}\cdot\nabla\phi}{{\bf B}\cdot\nabla \theta}\ \dx{\theta} = \frac{\bar{q}}{2\pi}\int_0^{2\pi}\frac{\dx{\theta}}{1+\epsilon \cos \theta} = \frac{\bar{q}}{\sqrt{1-\epsilon^2}},
\end{equation}
where $\epsilon=r/R_0$ is the inverse aspect ratio. We instead define the poloidal flux in terms of its radial derivative, $\dx{\Psi}/\dx{r}= rB_0/\bar{q}$, giving
\begin{equation}
    \Psi = \int_0^r \frac{r' B_0}{\bar{q}(r')}\, \dx{r'}
\end{equation}
instead of \cref{psicirc}.

Here we will choose a straight-field-line poloidal angle $\chi$ defined such that field lines are straight in the $(\chi,\phi)$ plane with slope $q$. To do this we take $({\vec{B}\cdot\nabla \phi })/({ \vec{B}\cdot \nabla \chi })= q$, which leads to $\dx{\chi}/\dx{\theta} = (\vec{B}\cdot\nabla\phi)/(q\vec{B}\cdot \nabla \theta)$. 
Integrating over $\theta$ then gives
\begin{equation}
\chi(r,\theta) = \frac{1}{q}\int_0^\theta \frac{{\bf B}\cdot\nabla\phi}{{\bf B}\cdot\nabla\theta'}\ \dx{\theta'} = \frac{\bar{q}}{q}\int_0^\theta \frac{\dx{\theta'}}{1+\epsilon\cos\theta'} = 2 \arctan\left[\sqrt{\frac{1-\epsilon}{1+\epsilon}}\tan\left(\frac{\theta}{2}\right)\right].
\end{equation}

Now we can define the field-aligned coordinate system $(x,y,z)$ as
\begin{equation}
    x = r - x_0, \qquad y = \frac{r_0}{q_0}\left(q\chi - \phi\right) - y_0, \qquad z = \chi,
\end{equation}
where $r_0$ is the minor radius of a flux surface of interest, and $q_0=q(r_0)$. 

\begin{figure}
    \centering
    \includegraphics[width=.7\textwidth]{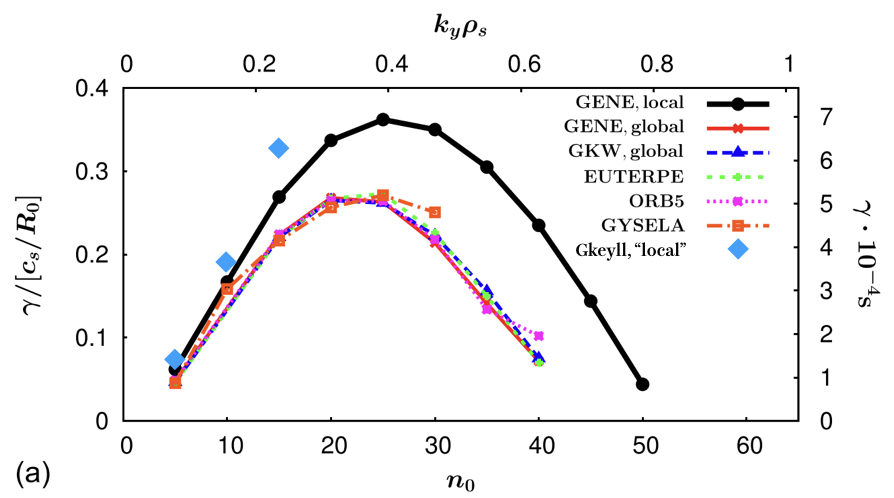}
    \caption[Linear growth rate results for the Cyclone base case, with \gke results compared to several other gyrokinetic codes.]{Linear growth rate results for the Cyclone base case, with \gke results (blue diamonds) compared to several other gyrokinetic codes. Note that since \gke does not include gyroaveraging we only expect accuracy in the small $k_y\rho_\mathrm{s}$ limit. Adapted from Figure 3a of \citep{gorler2016}.}
    \label{fig:cyclone-linear}
\end{figure}

\subsection{Cyclone base case linear benchmark}
Here we perform the now-standard Cyclone base case linear benchmark with \gke using the circular equilibrium described in the previous section. These calculations use the electrostatic approximation with adiabatic electrons, as in the original benchmark \citep{dimits2000comparisons}. Like the KBM calculations in \cref{sec:kbm}, these simulations use a single narrow cell in the radial direction, making them effectively ``local'', unlike a ``global'' calculation that covers some portion of the tokamak minor radius and accounts for the radial variation of background quantities.\footnote{While \gke could in principle be capable of performing global calculations, one must be careful with the initial conditions to avoid an equilibrium-scale $n=0$ mode that results from the neoclassical terms. This can obscure or alter the growth rate of the $n\neq 0$ mode of primary interest. The so-called ``canonical'' Maxwellian, which is formulated in terms of the canonical momentum so that it is an equilibrium of the full-$f$ gyrokinetic equation including the neoclassical terms, is used in some full-$f$ gyrokinetic codes to avoid exciting the $n=0$ mode \citep{angelino2006}. The implementation of a canonical Maxwellian initial condition is currently in progress, which will enable global instability calculations.} We use an extended domain along the field line with $-3\pi \leq \chi \leq 3\pi$ and Dirichlet boundary conditions $f(\chi = \pm 3\pi) = F_0$ so that no fluctuations are allowed at the domain ends.

\begin{figure}
    \centering
    \includegraphics[width=.6\textwidth]{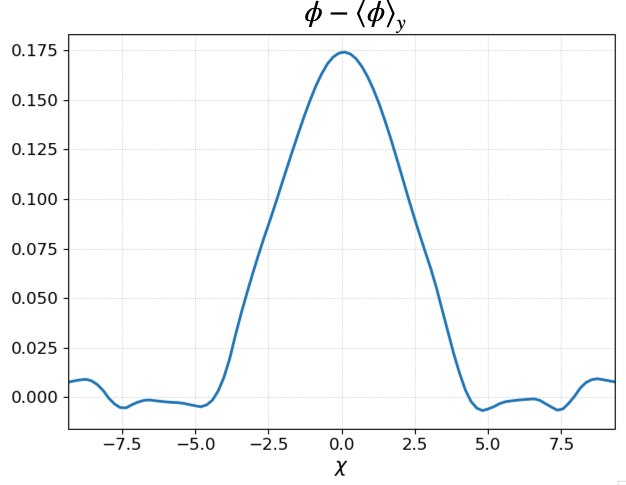}
    \caption{Linear eigenmode for $n=10$ case, showing ballooning mode structure along the field line.}
    \label{fig:cyclone-phi}
\end{figure}

We compare our results to the inter-code benchmark of \citet{gorler2016}; relevant parameters for the simulation setup are given in Tables I and II of that reference. In \cref{fig:cyclone-linear} we have reproduced Figure 3a of \citet{gorler2016} with the addition of some \gke results. We see good agreement for $n=5$, where $n = k_y r_0/q_0$ is the toroidal mode number. Since \gke does not include gyroaveraging, we only expect accuracy in the small $k_y \rho_\mathrm{s}$ limit. Consequently, we overpredict the growth rate as we move to higher mode numbers. We also show the linear eigenmode for the $n =10$ case in \cref{fig:cyclone-phi}, which has the characteristic peak at the outboard midplane $(\chi=0)$ of a ballooning mode. Here we have removed the $n=0$ component of the potential by subtracting its $y$-average.

Additional work will extend these benchmarks to the two-species electrostatic cases and the electromagnetic cases that are also examined in \citet{gorler2016}. We will also include global effects in future work.


\begin{subappendices}
\section{Alternative coordinate mappings for helical geometry} \label{app:alt-heli-geo}
One may note that the helical coordinate mapping we used in this chapter, \cref{heli-coord}, does not match the mapping proposed in \citet{shi2019} or \cref{ch:nstx-results}, given by
\begin{equation}
    R = x,\qquad\varphi = \frac{y \sin\vartheta+ z\cos\vartheta}{R_c} ,\qquad Z=z\sin\vartheta,
\end{equation}
where $\sin\vartheta=B_v/B=H/L_c$, with $\vartheta$ the field line pitch angle (as shown in \cref{fig:heli-geo}) so that
\begin{equation}
    x = R,\qquad z=\frac{Z}{\sin\vartheta}=\frac{L_c}{H} Z,\qquad y =\frac{R_c}{\sin\vartheta}\left(\varphi - \frac{Z}{R_c}\cot\vartheta\right) = R_c\frac{L_c}{H}\left(\varphi - \frac{Z}{R_c} \frac{B_\varphi}{B_v}\right).
\end{equation}
After computing the resulting gradient basis vectors, we have
\begin{equation}
    \nabla x\times\nabla y = \frac{L_c}{H}\left({\frac{B_\varphi}{B_v} \uvg{\varphi}+\frac{R_c}{x}\uv{Z}}\right).
\end{equation}
Taking $\mathcal{C}=H B_v/L_c$, the resulting background magnetic field in Clebsch form is
\begin{equation}
    \vec{B} = \mathcal{C}\nabla x \times \nabla y = B_\varphi \uvg{\varphi} + \frac{R_c}{x}B_v \uv{Z}.
\end{equation}
We see that this only gives the correct pitch of the magnetic field at $x=R_c$, so \emph{this is not a field-aligned coordinate system}.

We can fix this issue and make the coordinates field-aligned by changing the $\varphi$ mapping to 
\begin{equation}
    \varphi = \frac{y \sin\vartheta}{R_c} + \frac{z \cos\vartheta}{x} = \frac{y}{R_c}\frac{H}{L_c} + \frac{z}{x}\sqrt{1-\frac{H^2}{L_c^2}},
\end{equation}
so that 
\begin{equation}
    y = \frac{R_c}{ \sin\vartheta}\left(\varphi - \frac{Z}{R}\cot\vartheta\right) =  R_c\frac{L_c}{H}\left(\varphi - \frac{B_\varphi Z}{B_v R}\right),
\end{equation}
which is the same as the definition of $y$ from Eq. (\ref{heli-coord}) except for a factor of $1/\sin\vartheta$ (with $x_0=R_c$ here). 
This gives
\begin{equation}
    \nabla x\times \nabla y = \frac{R_c}{x}\left(\cos\vartheta \hat{\boldsymbol{\varphi}} + \sin\vartheta \uv{Z}\right) = \frac{R_c}{x}\frac{L_c}{H}\left(\frac{B_\varphi}{B_v}\uvg{\varphi} +  \uv{Z}\right).
\end{equation}
Taking $\mathcal{C}=B x/R_c$, we have
\begin{equation}
   {\bf B}= \mathcal{C}\nabla x\times\nabla y =B\left(\cos\vartheta \uvg{\varphi} + \sin\vartheta \uv{Z}\right) =B_\varphi \uvg{\varphi}+B_z\uv{Z},
\end{equation}
which is now the correct form of the magnetic field.

This mapping uses the distance along the field line as the field-aligned $z$ coordinate. For SMT configurations, where the connection length can vary radially, this choice could mean that the $z$ simulation domain extents must also vary radially. We could instead normalize the $z$ coordinate to the connection length, resulting in an equal-arc-length-like parallel coordinate, but this would result in the parallel coordinate becoming proportional to the vertical height $Z$. Thus we have used the vertical height $Z$ as the $z$ coordinate in \cref{sec:heli-geo}. 

\section{Self-consistently reproducing ``simplified'' helical geometry} \label{app:simple-heli-geo}

Here we would like to effectively reproduce the ``simplified'' helical geometry that we used to produce the results in \cref{sec:emgk-res}. As we will see, the actual geometry that matches the approximations we made in \cref{sec:simple-geo} is not helical, but purely toroidal; it essentially consists of rings of toroidal field stacked vertically on top of each other, as shown in \cref{fig:simple-heli}.
The mapping from cylindrical coordinates $(R, \varphi, Z)$ to field-aligned coordinates $(x,y,z)$ that gives this simplified geometry is
\begin{equation}
    x = R, \qquad y = Z,  \qquad z = R_c \varphi. \label{eq:simple-geo-map}
\end{equation}
The tangent unit vectors are 
\begin{equation}
    \vec{e}_x = \uv{R}, \qquad \vec{e}_y = \uv{Z}, \qquad \vec{e}_z = \frac{x}{R_c} \uvg{\varphi},
\end{equation}
and the gradient unit vectors are
\begin{equation}
    \nabla x = \uv{R}, \qquad \nabla y = \uv{Z},\qquad \nabla z = \frac{R_c}{x}\uvg{\varphi}.
\end{equation}
\begin{figure}[t]
    \centering
    \includegraphics[width=.5\textwidth]{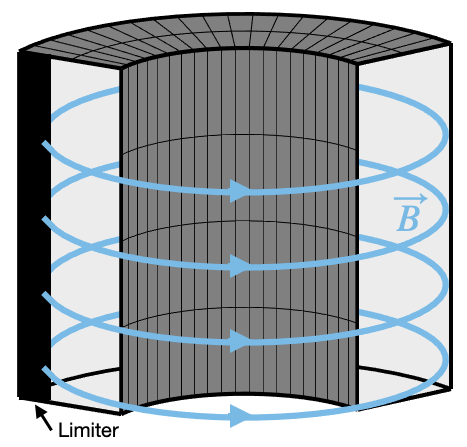}
    \caption[Geometry with purely toroidal magnetic field.]{Geometry with purely toroidal magnetic field. This is the geometry that matches the metric-related approximations we made for the results in \cref{sec:emgk-res}. In those simulations, sheath boundary conditions applied in the $z$ direction are equivalent to a limiter that extends vertically at some toroidal location (although the ``connection length'' along the field lines used in the simulations is not necessarily the circumference of the toroidal rings). Periodic boundary conditions were applied in the $y$ direction, which is the vertical direction here.}
    \label{fig:simple-heli}
\end{figure}
The magnetic field is then given by
\begin{equation}
    \vec{B} = \mathcal{C} \nabla x \times \nabla y = \frac{J B}{\sqrt{g_{zz}}} \uv{R} \times \uv{Z} = B \uvg{\varphi}
\end{equation}
with $B =B_\text{axis}(R_0/x)$, so that the field is indeed purely toroidal.  

Note that here, the coordinate along the field line, $z$, is defined to be proportional to the toroidal angle, $\varphi$. Whereas elsewhere in this chapter we used the poloidal angle as the field-aligned coordinate, this is not possible here where the field lines have no pitch. In fact, this is precisely the issue with using a field-aligned poloidal coordinate at the X-point, where the poloidal magnetic field vanishes and the field is purely toroidal. 

The remaining geometry quantities of interest are 
\begin{gather}
\texttt{cmag} = {\mathcal{C}} =  \frac{JB}{\sqrt{g_{zz}}}=B \\
\texttt{b\_x} = b_x = \frac{g_{xz}}{\sqrt{g_{zz}}} =0\\
\texttt{b\_y} = b_y = \frac{g_{yz}}{\sqrt{g_{zz}}}=0  \\
\texttt{b\_z} = b_z = \sqrt{g_{zz}} = \frac{x}{R_c} \\
\texttt{gxx} = g^{xx} = 1 \\
\texttt{gxy} = g^{xy} = 0\\
\texttt{gyy} = g^{yy} = 1 \\
\texttt{jacobPhase} = B_\parallel^* \approx B\\
\texttt{jacobGeo} = J = \sqrt{\det g_{ij}} = \frac{x}{R_c}.
\end{gather}
We can see that many of the metric quantities are trivial or vanish, just as we approximated in \cref{sec:simple-geo}. The magnetic (curvature and $\nabla B$) drifts are then
\begin{gather}
    \vec{v}_d \cdot \nabla x = 0 \\
    \vec{v}_d \cdot \nabla y 
    = -\frac{m v_\parallel^2 + \mu B}{qB} \frac{1}{x} \\
    \vec{v}_d \cdot \nabla z = 0,
\end{gather}
which matches \cref{simplevd}.

\end{subappendices}

\chapter{Positivity-preserving discontinuous Galerkin algorithm for hyperbolic conservation laws without \emph{post-hoc} diffusion}\label{ch:positivity}

Physically, the distribution function of particles is a non-negative scalar function, \emph{i.e.} $f(\vec{x},\vec{v},t) \geq 0$ throughout the phase space. However, there is no guarantee that a numerical scheme will preserve this property. The discontinuous Galerkin schemes that we described in \cref{ch:dg} are no exception, as they do not even ensure positivity of the cell average of the distribution function. In some cases small regions of negative $f$ do not impact the physics, but in other cases negative regions can lead to numerical instability. Thus we must develop a method to prevent negative regions in order for our simulations to be accurate and robust.

There is extensive literature on constructing positivity-preserving (or more generally, bound-preserving) DG schemes. For example, a widely used and now standard positivity-limiting scheme presented by \citet{zhang2011} works by limiting the amount of flux leaving a cell surface so that the cell average is not allowed to become negative. However, this limiter procedure by itself can produce unphysical steep slopes and higher moments. For this reason, a \emph{post-hoc} sub-cell diffusion step is applied whereby the slopes and higher moments in each cell are adjusted so that the solution remains positive at some control points. This leads to a robust positivity-preserving algorithm for many conservation laws, like the Euler equations and the ideal MHD equations. Generalizations of the Zhang-Shu scheme have also been made \citep{johnson2012}.

However, the Zhang-Shu (and related) algorithms cannot be used for evolution of kinetic equations in which the conservation properties are indirect (\emph{i.e.} when there is not a direct equation for the evolution of the energy). The reason for this is that \emph{post-hoc} sub-cell diffusion will change the energy and break energy conservation. One could try to readjust the energy after the diffusive step to maintain energy conservation, as done in \cite{shi-thesis}, but often such adjustments are not possible without disturbing the underlying physics.

In this chapter we develop a novel positivity-preserving DG scheme without \emph{post-hoc} diffusion. After showing a preliminary example of the positivity issue, we first define what we mean by positivity in the context of the discontinuous Galerkin representation of the cell. Next we construct the positivity-preserving scheme and show that conservation properties are maintained for Hamiltonian systems. Finally we show numerical results in several dimensionalities and equation systems.

\section{The positivity problem: 1D advection example}\label{pos:1dadvect}

Consider an one-dimensional advection equation,
\begin{equation}
    \pderiv{f}{t} + v \pderiv{f}{x} = 0.
\end{equation}
Taking $v$ to be constant in this simple example, we know the solution is
\begin{equation}
    f(x,t) = f(x-v t, 0).
\end{equation}
That is, the solution keeps its initial shape as it moves through the domain with constant velocity $v$. If the domain is periodic with length $L$, the pulse will return to its initial position at time $t=L/v$. 

\begin{figure}[t!]
    \centering
    \includegraphics[width=\textwidth]{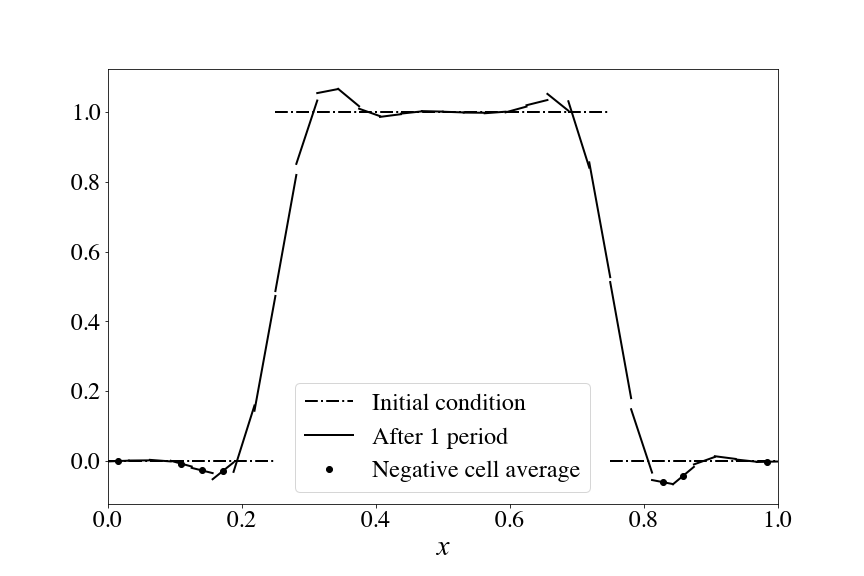}
    \caption{Advection of a square pulse for one period on a 1D periodic domain, overlaid with the initial condition.}
    \label{fig:no-pos-adv}
\end{figure}

We can easily discretize this system with a discontinuous Galerkin scheme. In \cref{fig:no-pos-adv}, we show the results of a piecewise-linear scheme with 32 cells on a one-dimensional periodic domain with length $L=1$. The initial condition is a square pulse centered at $x=0.5$ with width $w=0.5$, shown dash-dotted. After 1 period, the solution (solid lines) returns to its initial position, although the shape of the pulse has been distorted by numerical artifacts. Notably, we see 
unphysical negative overshoot regions at the bottom of the pulse, as well as positive overshoots at the top of the pulse. Cells with a negative cell average are marked with points showing the cell center. While in some applications small negative overshoots may be tolerable, often these unphysical negative regions can cause severe problems. Our goal in this chapter will be to devise a positivity-preserving scheme that eliminates these negative regions. 

\section{Defining positivity, in the weak sense} \label{sec:pos-definition}

The first challenge is to define what is meant by positivity in the context of the discontinuous Galerkin representation of the solution. In each cell, the solution is given by an expansion on some basis set. The goal of this section will be to develop a method to constrain the expansion coefficients to maintain positivity of the solution, in some sense. Let's first consider the simplest case: a piecewise-linear ($p=1$) representation in one dimension. Taking an orthonormal basis set $\psi = \{1/\sqrt{2}, \sqrt{3/2}x\}$ in a cell $x\in [-1,1]$, we have
\begin{equation}
    f_h = \sum_k f_k \psi_k = \frac{1}{\sqrt{2}} f_0 + \frac{\sqrt{3}}{\sqrt{2}}x f_1. 
\end{equation}
How can we constrain the coefficients $f_0$ and $f_1$ to ensure that the solution is positive? To start, we should at least ensure that the cell average is positive, so that $f_0\geq 0$. Should the solution be required to be positive on the whole cell domain $x=[-1,1]$, or can the solution be negative on some portion of the domain?

One possible way to answer these questions is to define positivity as \emph{weak equality to a positive-definite function} \citep{hakim2020a}. For example, we could consider a non-polynomial positive-definite exponential solution given by
\begin{equation}
    g_h = g_0 \exp(g_1 x).
\end{equation}
Weak equality of $f_h$ and $g_h$ in the $L_2$ sense, which we denote as $f_h \doteq g_h$, then requires that the projections of the two representations onto the basis be equivalent:
\begin{align}
    \int_{-1}^1 \frac{1}{\sqrt{2}} f_h\, \dx{x} = \int_{-1}^1 \frac{1}{\sqrt{2}} g_h\, \dx{x} \quad &\Rightarrow \quad f_0 = \frac{\sqrt{2} g_0 \sinh g_1}{g_1} \\
    \int_{-1}^1 \frac{\sqrt{3}}{\sqrt{2}} x f_h\, \dx{x} = \int_{-1}^1 \frac{\sqrt{3}}{\sqrt{2}} x g_h\, \dx{x} \quad &\Rightarrow \quad f_1 = \frac{\sqrt{6} g_0}{g_1^2}\left(g_1 \cosh g_1 -\sinh g_1\right).
\end{align}
Note that 
\begin{equation}
    \frac{f_1}{\sqrt{3} f_0} = \coth g_1 - \frac{1}{g_1} = L(g_1),
\end{equation}
where $L(g_1)$ is the well-known Langevin function in statistical mechanics. Notably, this function is bounded at $|L(g_1)|\leq1$ for all $g_1$. This means that in order for weak-equivalence of the solutions $f_h$ and $g_h$ to be possible, the coefficients of $f_h$ must satisfy
\begin{equation}
    \frac{|f_1|}{\sqrt{3}f_0}\leq 1. \label{posconstraintf1}
\end{equation}
Together with the constraint $f_0\geq0$, we now have positivity constraints for both coefficients of the piecewise-linear representation of the solution. 

We can now express $g_h(x)$ in terms of $f_0$ and $\bar{x}\equiv f_1/(\sqrt{3}f_0)$ via
\begin{equation}
    g_h(x) = \frac{f_0 g_1}{\sqrt{2}\sinh g_1}e^{g_1 x}, \label{gexp}
\end{equation}
where $g_1 = g_1(\bar{x}) = L^{-1}(\bar{x})$, and $L^{-1}$ is the inverse Langevin function.\footnote{Although the inverse of the Langevin function does not have a closed form, a number of Pad\'e approximations for the inverse have been developed, including \citep{cohen1991} 
\begin{equation}
    L^{-1}(x) \approx \frac{x(3-x^2)}{1-x^2}.\label{cohen}
\end{equation}
} \cref{fig:pos-diagram} shows an example of the weak-equivalent linear and exponential solutions with $f_0=0.4$ and $f_1=0.8$.

Note that the constraint from \cref{posconstraintf1} does not force the linear solution to be positive everywhere in the cell. In this case, the linear solution is negative for $x<-0.5$ but the exponential solution is still realizable. In fact, one can show that the constraint on $f_1$ is equivalent to requiring $f_h(x=\pm1/3)\geq 0$, so that as long as the linear solution remains positive at ``positivity control nodes'' $x=\pm 1/3$, the solution will be positive in the weak sense. The control nodes are also plotted in \cref{fig:pos-diagram}, and they are indeed both positive. 

\begin{figure}
    \centering
    \includegraphics[width=.7\textwidth]{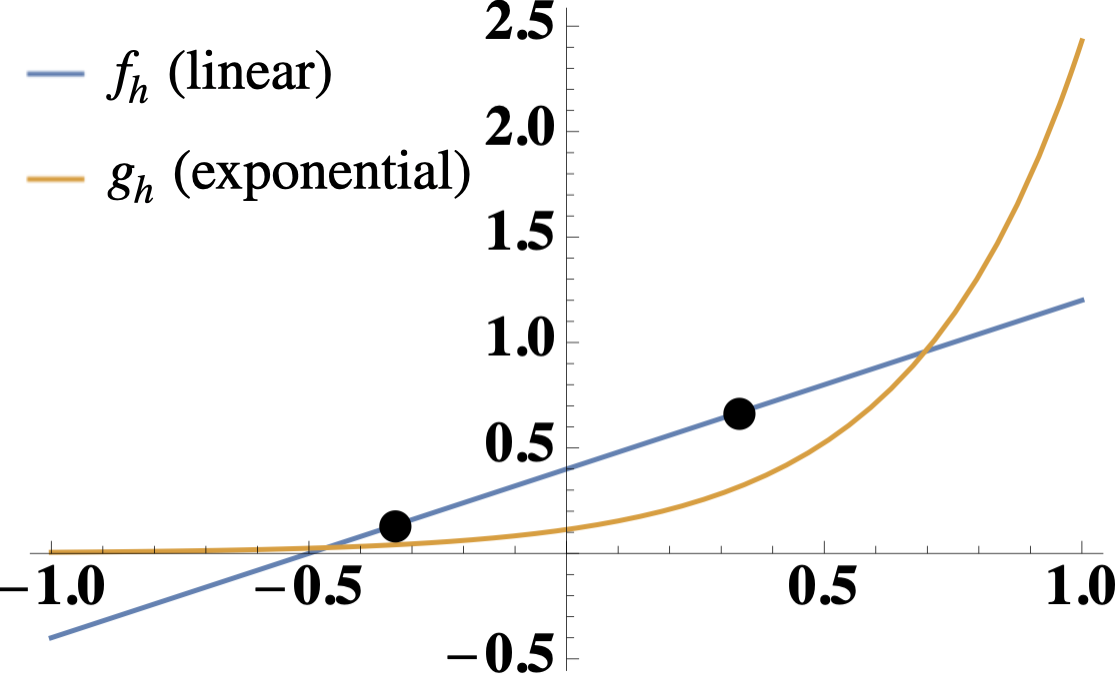}
    \caption[Weak-equivalent linear and exponential solutions.]{Weak-equivalent linear and exponential solutions, with $f_0=0.4$ and $f_1=0.8$. The exponential solution $g_h$ is given by \cref{gexp}. The ``positivity control nodes'' $f_h(x=\pm 1/3)$ are also shown.}
    \label{fig:pos-diagram}
\end{figure}

\subsection{Generalization to higher dimensionality}

It is not immediately clear how to tractably extend the procedure of the previous section to higher dimensionality. For example, to extend to a two-dimensional case, one might consider taking the exponential representation to be of the form $g_h = \exp(g_0 + g_1 x + g_2 y + g_3 x y)$. However, 2D integrals of this function involve error functions, making it difficult to evaluate positivity constraints based on weak equality. 

We will discuss a more rigorous procedure for generalizing the definition of the weak-equivalent positive-definite solution to higher dimensionality (and higher polynomial order) in Appendix \ref{sec:pos-gen}. For now, however, let's continue with the idea of ``positivity control nodes''. In one dimension, we saw above that if the piecewise-linear solution is positive at the control nodes at $x=\pm 1/3$, then the solution can be made weak-equivalent to an exponential solution. A sensible extension of this idea to higher dimension is to take tensor products of the control nodes, so that for example in 2D, we have control nodes at $(-1/3,-1/3),\ (-1/3,1/3),\ (1/3,-1/3)$, and $(1/3,1/3)$. Thus in the following section, we will consider the solution to be positive (in the weak sense) if the $N$-dimensional piecewise-linear solution is non-negative at all of the $2^N$ control nodes. 

\section{Constructing a positivity-preserving scheme without \emph{post-hoc} diffusion}

Now that we have a definition of positivity
, we next focus on how to construct a discontinuous Galerkin scheme that preserves positivity. In our scheme, we would like to avoid \emph{post-hoc}  sub-cell diffusion (rescaling slopes or higher moments of the solution in a cell if they become too extreme and violate positivity constraints after taking a timestep), which can break conservation laws involving higher-order moments (such as energy conservation in Hamiltonian systems like gyrokinetics). 

To begin, we will once again consider a generic hyperbolic conservation law of the form of \cref{pde2},
\begin{equation}
    \pderiv{f}{t} + \nabla\cdot \vec{F} = 0, \label{hyper-pos}
\end{equation}
with $\vec{F}(f)$ some arbitrary nonlinear flux. Recall from \cref{DGweakgen} that the DG discretization of this equation is given by multiplying by a test function $\psi$ and integrating by parts over a cell $\mathcal{K}_i$:
\begin{equation}
    \int_{\mathcal{K}_i}  \psi \pderiv{f_h}{t}\dx{\vec{x}}- \int_{\mathcal{K}_i} \vec{F}_h\cdot\nabla \psi\, \dx{\vec{x}}  + \oint_{\partial \mathcal{K}_i}\psi^- \vec{\hat{F}}\cdot\dx{\vec{s}} = 0. \label{dgweakpos}
\end{equation}

Let us first focus on the one-dimensional piecewise-linear ($p=1$) case. Mapping each cell $\mathcal{K}_i$ to the interval  $x\in[-1,1]$ using the transformation $x' = 2(x-x_i)/\Delta x$, with $x_i$ the cell center and $\Delta x$ the cell width, this gives (after dropping primes for simplicity)
\begin{equation}
   \frac{\Delta x}{2} \int_{-1}^1 \psi \pderiv{f_h}{t}\dx{x} - \int_{-1}^1 F_h \pderiv{\psi}{x}\dx{x}+ \psi(1)\hat{F}(1) - \psi(-1) \hat{F}(-1) =0.
\end{equation}
In our standard discretization scheme, we would substitute the one-dimensional $p=1$ orthonormal modal basis functions $\psi_j=\{1/\sqrt{2}, \sqrt{3/2}x\}$ for the test functions, which results in
\begin{gather}
    \pderiv{f_0}{t} = -\frac{\sqrt{2}}{\Delta x}\left[\hat{F}(1)-\hat{F}(-1)\right] \label{posf0dt} \\
    \pderiv{f_1}{t} = -\frac{\sqrt{6}}{\Delta x}\left[\hat{F}(1)+\hat{F}(-1)\right]+\frac{2\sqrt{3}}{\Delta x}F_0, \label{posf1dt}
\end{gather}
where we have also expanded $f_h$ and $F_h$ on the orthonormal modal basis.
In terms of these modal coefficients, we learned above that positivity requires $f_0\geq0$ and $|f_1|\leq\sqrt{3}f_0$. This is equivalent to ensuring that control node values at $x=\pm 1/3$ remain non-negative, $f_h(x=\pm1/3)\geq 0$.

We would like to find a way to limit the surface and volume terms to ensure that these constraints are not violated as the solution evolves. Existing positivity-preserving schemes attempt to limit the boundary fluxes so that the cell-average $f_0$ stays positive (the volume term for the cell-average always vanishes, so the cell-average is only affected by surface terms). However, it is not immediately clear how to account for an additional constraint $|f_1|\geq\sqrt{3}f_0$, since the cell-slope $f_1$ is also affected by the same fluxes; for this reason, existing schemes often rescale the cell-slope $f_1$ \emph{post-hoc}, which effectively gives sub-cell diffusion that can break higher-order conservation laws. 

It is more convenient to instead consider the evolution of the control nodes, $f_\pm \equiv f_h(x=\pm1/3)=f_0/\sqrt{2}\pm f_1/\sqrt{6}$. Taking the appropriate linear combinations of \cref{posf0dt,posf1dt}, we have
\begin{align}
   \pderiv{f_+}{t} &= -\frac{2}{\Delta x}\hat{F}(1) + \frac{\sqrt{2}}{\Delta x}F_0 \label{g+}\\
   \pderiv{f_-}{t} &= \frac{2}{\Delta x}\hat{F}(-1) - \frac{\sqrt{2}}{\Delta x}F_0. \label{g-}
\end{align}
Unlike in \cref{posf0dt,posf1dt}, each of the control nodes is only affected by a flux from one side of the cell: the left node $f_-$ is affected by the flux on the left boundary $\hat{F}(-1)$, and the right node $f_+$ is affected by the flux on the right boundary $\hat{F}(1)$. 

Neglecting the volume terms for now, this means that we can separately limit the left flux to maintain $f_-\geq 0$ and limit the right flux to maintain $f_+\geq 0$. However, note that the neighboring cells are also affected by these fluxes. To account for this, let us instead examine the evolution of two cells, $i$ and $i+1$, due to a flux $\hat{F}^{i+1/2}$ at their interface:
\begin{align}
   \pderiv{f_+^i}{t} &= -\frac{2}{\Delta x}\hat{F}^{i+1/2} + \frac{\sqrt{2}}{\Delta x}F_0^i \\
   \pderiv{f_-^{i+1}}{t} &= \frac{2}{\Delta x}\hat{F}^{i+1/2} - \frac{\sqrt{2}}{\Delta x}F_0^{i+1}.
\end{align}
We see that the flux $\hat{F}^{i+1/2}$ is simply exchanging information between $f_+^i$ and $f_-^{i+1}$, while neither $f_-^{i}$ nor $f_+^{i+1}$ is affected by this flux. This also means that only one of $f_+^i$ or $f_-^{i+1}$ is decreased by the flux, with the other increasing by the same amount. Upon adopting a forward Euler timestepping scheme (which can be built into a higher-order Runge-Kutta scheme), and dropping the volume terms for now, it is easy to see how to limit the flux $\hat{F}^{i+1/2}$ so that neither $f_+^i$ nor $f_-^{i+1}$ can become negative after a single timestep. This gives
\begin{align}
    f_+^i -\frac{2\Delta t}{\Delta x}\hat{F}^{i+1/2}&\geq0 \\
   f_-^{i+1} + \frac{2\Delta t}{\Delta x}\hat{F}^{i+1/2}&\geq0.
\end{align}
The limit on $\hat{F}^{i+1/2}$ to ensure that the flux does not make either $f_+^i$ or $f_-^{i+1}$ negative in a single step is then
\begin{align}
   -f_-^{i+1}\frac{\Delta x}{2\Delta t} \leq \hat{F}^{i+1/2} \leq f_+^{i}\frac{\Delta x}{2\Delta t}. \label{fluxlim1}
\end{align}
This is illustrated in the diagram in \cref{fig:pos-demo2}.

\begin{figure}[t]
    \centering
    \includegraphics[width=\textwidth]{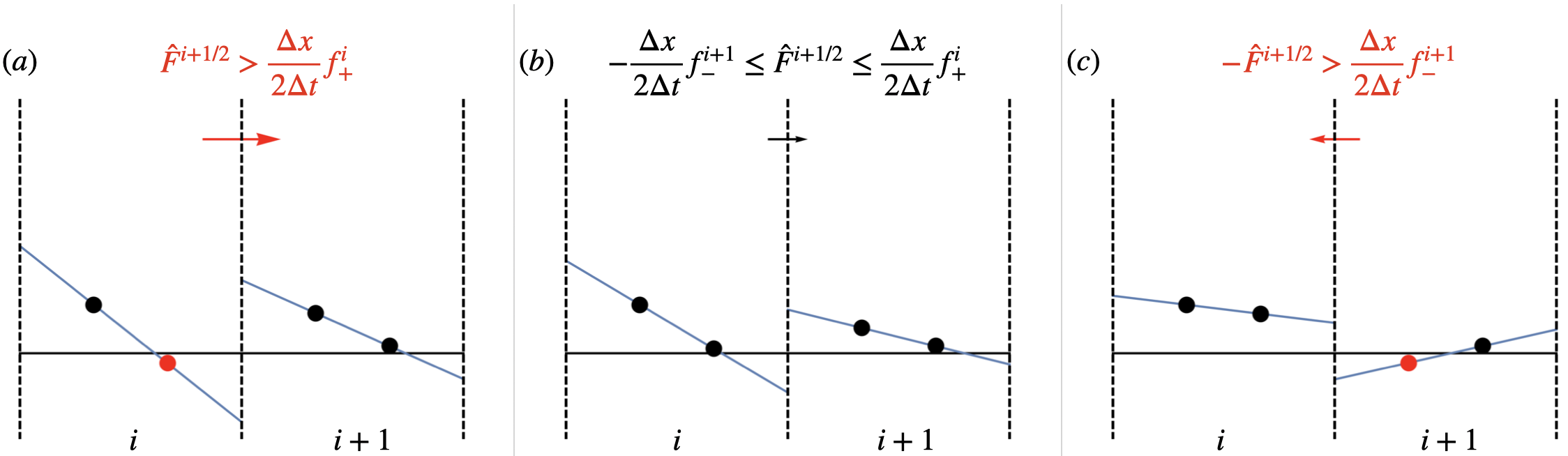}
    \caption[Illustration of limiter scheme for boundary fluxes.]{Illustration of limiter scheme for boundary fluxes based on \cref{fluxlim1}. The flux limit on the flux between cells $i$ and $i+1$, $\hat{F}^{i+1/2}$, is violated in the left $(a)$ and right $(c)$ panels, leading to negative control nodes (red points). In the left panel $(a)$, $\hat{F}^{i+1/2}$ is too large to the right and control node $f_+^i$ becomes negative (red). In the right panel $(c)$, $\hat{F}^{i+1/2}$ is too large to the left, leading to control node $f_-^{i+1}$ becoming negative (red). In the middle panel $(b)$, the flux is within the limits and the control nodes all remain positive (black points). Also note that the control points $f_-^i$ and $f_+^{i+1}$ remain fixed in all panels because they are unaffected by the flux (these points are affected by other fluxes, but not the one depicted here, $\hat{F}^{i+1/2}$).}
    \label{fig:pos-demo2}
\end{figure}

Once we have limited all fluxes to ensure that the surface terms cannot make $f_\pm$ negative in any cell in the domain, we can limit the volume terms. Considering \cref{g+,g-}, we can see that the volume terms exchange information between $f_-$ and $f_+$ within each cell. Thus one way to limit the volume terms is to simply scale all the volume terms in each cell by a common factor $0\leq\theta^i\leq1$ to ensure that neither $f_-$ or $f_+$ is made negative by the volume terms.

For cell $i$, the final forward-Euler update can be expressed as
\begin{align}
    f^i_+(t_n+\Delta t) &= f^i_+(t_n) - \frac{2\Delta t}{\Delta x}\hat{F}^{i+1/2} + \frac{\sqrt{2}\Delta t}{\Delta x}\theta^i F_0^i \\
    f^i_-(t_n+\Delta t) &= f^i_-(t_n) + \frac{2\Delta t}{\Delta x}\hat{F}^{i-1/2} - \frac{\sqrt{2}\Delta t}{\Delta x}\theta^i F_0^i,
\end{align}
with limits on the fluxes given by \cref{fluxlim1} and the volume scaling factor given by
\begin{equation}
\theta^i = \min\left(1,  \frac{f_+^i - \frac{2\Delta t}{\Delta x}\hat{F}^{i+1/2}}{\frac{-\sqrt{2}\Delta t}{\Delta x}F_0^i}, \frac{f_-^i + \frac{2\Delta t}{\Delta x}\hat{F}^{i-1/2}}{\frac{\sqrt{2}\Delta t}{\Delta x}F_0^i}\right).
\end{equation}
Here we have prioritized the surface terms over the volume terms in that we limit the surface terms first. This can, for example, allow a maximal flux out the left boundary to lower $f_-$ to zero. Then if the volume terms wanted to decrease $f_-$ further, the volume terms would be essentially turned off in this cell. In principle, one could instead prioritize the volume terms, allowing the maximum flow within the cell and then possibly turning off boundary fluxes. A comparison of these two approaches is left to future work.  

\subsection{Exponential surface extrapolation} \label{sec:exp}

While the scheme described above will rigorously preserve positivity of the solution, we can make an additional improvement involving how the boundary fluxes are computed. In a standard DG scheme with upwinded fluxes, the flux between cells $i$ and $i+1$ would be computed as
\begin{equation}
    \hat{F}^{i+1/2} =\begin{cases}
                        u^{i+1/2} f_h^i(x^{i+1/2})\qquad &u^{i+1/2}>0 \\
                        u^{i+1/2} f_h^{i+1}(x^{i+1/2})\qquad &u^{i+1/2}<0
                     \end{cases}
\end{equation}
with $f_h^i(x^{i+1/2})$ and $f_h^{i+1}(x^{i+1/2})$ computed using the piecewise-linear representation of the solution, and $u^{i+1/2}$ the advection velocity at the cell interface. After mapping to a unit cell on $x\in[-1,1]$, the boundary values would be given by
\begin{gather}
    f_h^i(1) = \frac{1}{\sqrt{2}}f_0+\frac{\sqrt{3}}{\sqrt{2}}f_1 \label{fRlin}\\
    f_h^{i+1}(-1) = \frac{1}{\sqrt{2}}f_0-\frac{\sqrt{3}}{\sqrt{2}}f_1.
\end{gather}

Let us consider an extreme case: a cell where the flux from the left boundary is zero, with advection velocity $u>0$ a constant. In this case, the modal coefficients of $f_h$ in the cell are given by (from \cref{posf0dt,posf1dt})
\begin{gather}
    \pderiv{f_0}{t} = -\frac{u\sqrt{2}}{\Delta x}f_h^i(1) \\
    \pderiv{f_1}{t} = -\frac{u\sqrt{6}}{\Delta x}f_h^i(1)+\frac{2u\sqrt{3}}{\Delta x}f_0.
\end{gather}
From above, we know that we need $|f_1|/(\sqrt{3}f_0)<1$ for the solution to remain positive and realizable. Thus, let us compute the evolution of $\bar{x}\equiv f_1/(\sqrt{3}f_0)$:
\begin{equation}
    \pderiv{\bar{x}}{t} = \frac{1}{\sqrt{3}f_0}\pderiv{f_1}{t} - \frac{\bar{x}}{f_0}\pderiv{f_0}{t}= \frac{u\sqrt{2}}{\Delta x}\left[2 - \sqrt{2}(1-\bar{x})\frac{f_h^i(1)}{f_0}\right].\label{evx}
\end{equation}
If we use the standard linear extrapolation for $f_h^i(1)$ from \cref{fRlin}, this gives
\begin{align}
    \pderiv{\bar{x}}{t} = \frac{u}{\Delta x}\left[2 - (1-\bar{x})(1+3\bar{x})\right]&=\frac{u}{\Delta x}\left[1-2\bar{x}+3\bar{x}^2\right] \label{xbardt} \\&> 0 \quad \text{for all $\bar{x}$}.\notag
\end{align}
This means that without any limiters, $\bar{x}$ always grows without bound in this extreme case, which would violate the realizability limit $\bar{x}<1$ in a finite time. Note also that any reduction or limit on the boundary value $f^i_h(1)$ only makes the issue worse, so that $\bar{x}$ increases more quickly and becomes unphysical sooner. In this case, the volume term is steepening the slope in the cell faster than the boundary flux can flatten it. In practice, the volume term limiter in our scheme would eventually prevent $|x|>1$. Nonetheless, perhaps it would help to \textit{enhance} the extrapolated boundary flux, in a way that $\pderivInline{\bar{x}}{t}\rightarrow 0$ as $\bar{x}\rightarrow 1$. This way, perhaps we wouldn't need to limit the volume terms as often or as much. 

\begin{figure}
    \centering
    \includegraphics[width=.7\textwidth]{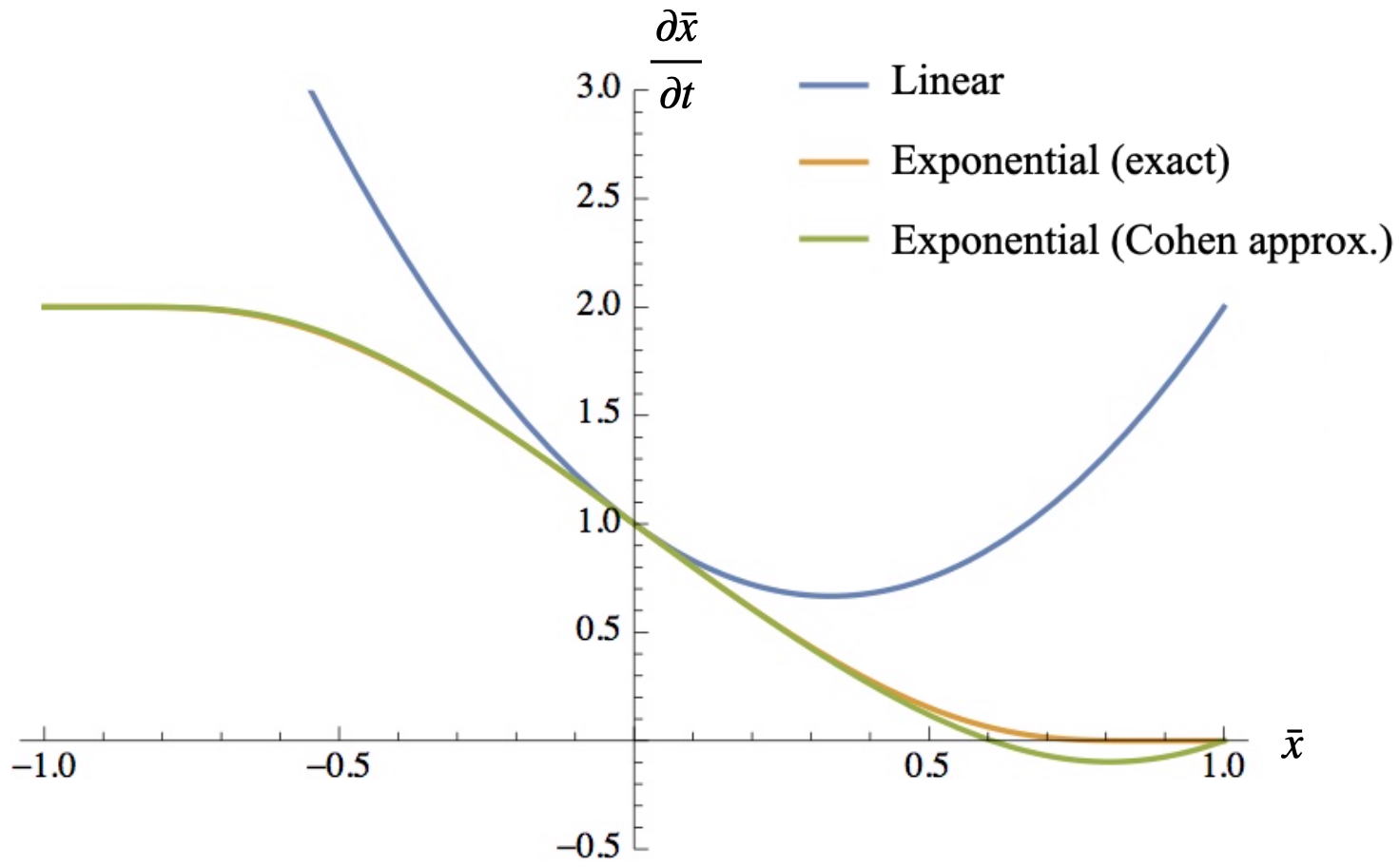}
    \caption[Change in slope as a function of slope for various flux extrapolations.]{$\pderivInline{\bar{x}}{t}$ vs. $\bar{x}$, given by \cref{xbardt}, for different choices of how to evaluate the surface value $f_h^i(1)$. A linear extrapolation, \cref{fRlin}, gives $\partial\bar{x}/\partial t>0$ for all $\bar{x}$, which drives the solution to be unrealizable in finite time (without limiters). Using an exponential extrapolation instead, given by \cref{gRexp}, gives $\partial\bar{x}/\partial t\rightarrow 0$ as $\bar{x}\rightarrow 1$, which means that $\bar{x}$ should stop increasing before becoming unphysical.}
    \label{fig:extrap}
\end{figure}

One way to \emph{enhance} the boundary value is to make use of the exponential reconstruction given by \cref{gexp}. Extrapolating $g_h$ to the right edge of the cell at $x=1$, we have
\begin{equation}
    g_h(1) = \frac{f_0 g_1}{\sqrt{2}\sinh g_1}e^{g_1}. \label{gRexp}
\end{equation}
In \cref{fig:extrap}
we plot ${\partial\bar{x}}/{\partial t}$ for three cases: linear extrapolation, \cref{fRlin}; exact exponential extrapolation, \cref{gRexp} with $g_1(\bar{x})=L^{-1}(\bar{x})$; and approximate exponential extrapolation, \cref{gRexp} with $g_1(\bar{x}) \approx \bar{x}(3-\bar{x}^2)/(1-\bar{x}^2)$, \emph{i.e.} using the Cohen approximation, \cref{cohen}, for the inverse Langevin function. We see that the both the exact and the approximate exponential extrapolation give $\partial\bar{x}/\partial t\rightarrow 0$ as $\bar{x}\rightarrow 1$, which means that $\bar{x}$ should stop increasing before becoming unphysical. The approximate exponential extrapolation gives a region where $\partial \bar{x}/\partial t<0$, which may in fact make the algorithm more robust because in this case the equilibrium value is $\bar{x}\approx0.6$. Meanwhile, the linear extrapolation gives $\pderivInline{\bar{x}}{t}>0$ for all $\bar{x}$ as we showed above.


\subsection{Extension to higher dimensionality with $p=1$}

To extend the scheme to higher dimensionality, we will again track the evolution of control nodes, which are given by tensor products of the 1D control nodes. For example, in 2D, we have four control nodes: $f_{--}=f_h(-1/3,-1/3)$, $f_{-+}=f_h(-1/3,1/3)$, $f_{+-}=f_h(1/3,-1/3)$, and $f_{++}=f_h(1/3,1/3)$. We will illustrate the scheme for the two-dimensional case, with extension to higher dimensions relatively straightforward.

In 2D, the DG weak form from \cref{dgweakpos} mapped to a cell in $x\in[-1,1],\ y\in[-1,1]$ is given by
\begin{align}
   &\int_{-1}^1 \dx{x} \int_{-1}^1 \dx{y}\, \psi \pderiv{f_h}{t}- \int_{-1}^1 \dx{x}\int_{-1}^1 \dx{y} \left[\frac{2}{\Delta x}F_{x\, h} \pderiv{\psi}{x} + \frac{2}{\Delta y}F_{y\, h} \pderiv{\psi}{y}\right] \notag \\
   &\quad+ \frac{2}{\Delta x}\int_{-1}^{1} \dx{y}
   \left[\psi(1,y)\hat{F}_x(1,y) - \psi(-1,y) \hat{F}_x(-1,y)\right] \\
    &\quad+ \frac{2}{\Delta y}\int_{-1}^{1} \dx{x} 
   \left[\psi(x,1)\hat{F}_y(x,1) - \psi(x,-1) \hat{F}_y(x,-1)\right]  =0.
\end{align}
We can then compute the evolution of the four control nodes as
\begin{align}
    \pderiv{f_{--}}{t} &= - \frac{1}{\Delta x}\left(F_{x 0} - \frac{1}{\sqrt{3}}F_{x 2}\right) - \frac{1}{\Delta y}\left(F_{y 0} -\frac{1}{\sqrt{3}}F_{y 1}\right) 
    +\frac{2}{\Delta x}\hat{F}_x(-1,-\frac{1}{3}) +\frac{2}{\Delta y}\hat{F}_y(-\frac{1}{3},-1) \\
    \pderiv{f_{-+}}{t} &= - \frac{1}{\Delta x}\left(F_{x 0} + \frac{1}{\sqrt{3}}F_{x 2}\right) + \frac{1}{\Delta y}\left(F_{y 0} -\frac{1}{\sqrt{3}}F_{y 1}\right)
    +\frac{2}{\Delta x}\hat{F}_x(-1,\frac{1}{3}) - \frac{2}{\Delta y}\hat{F}_y(-\frac{1}{3},1)\\
    \pderiv{f_{+-}}{t} &=  \frac{1}{\Delta x}\left(F_{x 0} - \frac{1}{\sqrt{3}}F_{x 2}\right) - \frac{1}{\Delta y}\left(F_{y 0} +\frac{1}{\sqrt{3}}F_{y 1}\right)
    -\frac{2}{\Delta x}\hat{F}_x(1,-\frac{1}{3}) + \frac{2}{\Delta y}\hat{F}_y(\frac{1}{3},-1) \\
    \pderiv{f_{++}}{t} &=  \frac{1}{\Delta x}\left(F_{x 0} + \frac{1}{\sqrt{3}}F_{x 2}\right) + \frac{1}{\Delta y}\left(F_{y 0} +\frac{1}{\sqrt{3}}F_{y 1}\right) 
    -\frac{2}{\Delta x}\hat{F}_x(1,\frac{1}{3}) - \frac{2}{\Delta y}\hat{F}_y(\frac{1}{3},1).
\end{align}
\begin{figure}[t]
    \centering
    \includegraphics[width=.6\textwidth]{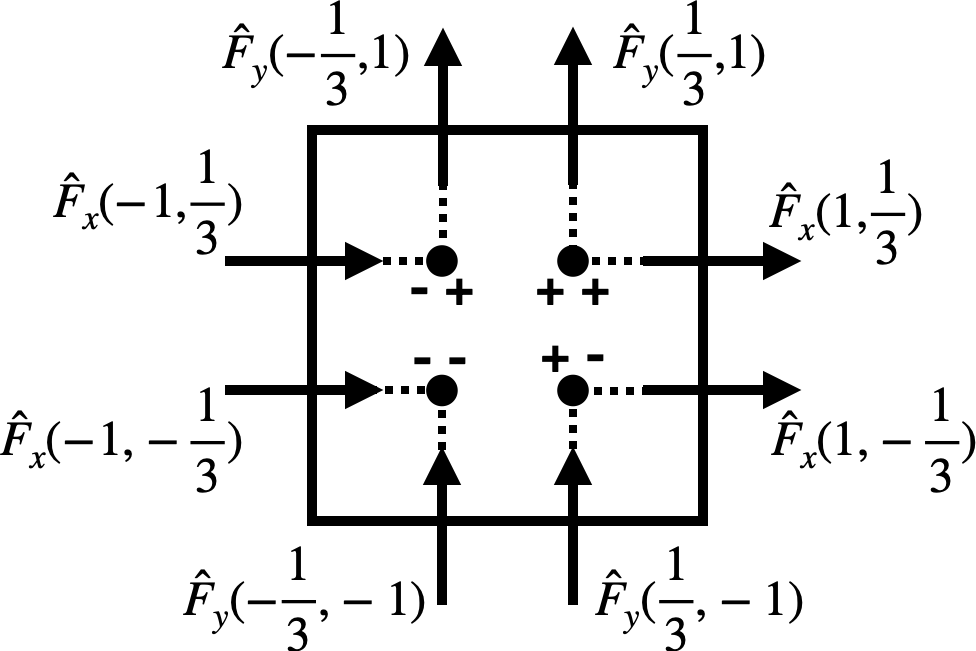}
    \caption[Diagram of fluxes affecting each interior control node in 2D.]{Diagram of fluxes affecting each interior control node in 2D. Each interior control node is affected by only the fluxes at the nearest surface control nodes (a dotted line between an interior node and a flux indicates that the flux affects that node).}
    \label{fig:pos-2d}
\end{figure}

Notably, each interior control node is affected only by the fluxes at the nearest surface control nodes, as shown in the diagram in \cref{fig:pos-2d}. Similar to above, we can limit each flux 
Notably, each interior control node is affected only by the fluxes at the nearest surface control nodes, as shown in the diagram in \cref{fig:pos-2d}. Similar to above, we can limit each flux so that the affected control nodes cannot become negative on a single forward-Euler timestep. Focusing on the $f_{++}$ control node, if the fluxes $\hat{F}_x(1,1/3)$ and $\hat{F}_y(1/3,1)$ are both directed out of the cell (as depicted in \cref{fig:pos-2d}), we need to make sure that the combined flux does not exceed the limit given by
\begin{equation}
    \frac{2\Delta t}{\Delta x}\hat{F}_x(1,\frac{1}{3}) + \frac{2\Delta t}{\Delta y}\hat{F}_y(\frac{1}{3},1) \leq f_{++}.
\end{equation}
We can separately limit the $x$ and $y$ fluxes by apportioning a fraction of $f_{++}$ that is allowed to removed in each direction, which we will denote as $\eta_x$ and $\eta_y$, such that $\eta_x + \eta_y = 1$. Now we can limit the fluxes as
\begin{align}
    \hat{F}_x(1, \frac{1}{3}) &\leq \eta_x f_{++}\frac{\Delta x}{2\Delta t} \\
    \hat{F}_y(\frac{1}{3},1) &\leq \eta_y f_{++}\frac{\Delta y}{2\Delta t}.
\end{align}
The definition of the $\eta_d$ need not be exact; one possible choice is to use the ratio given by the contribution to the CFL rate from each direction, $r_d$, divided by the total CFL rate $r$, so that $\eta_d = r_d/r$. 

As above, each flux should be limited by the control nodes on each side of the boundary. Thus the full limit on $\hat{F}_x(x^{i+1/2},1/3)$ at the boundary between cells $i$ and $i+1$ is given by
\begin{align}
   -\eta_x^{i+1} f_{-+}^{i+1}\frac{\Delta x}{2\Delta t} \leq \hat{F}_x(x^{i+1/2},\frac{1}{3}) \leq \eta_x^i f_{++}^{i}\frac{\Delta x}{2\Delta t}, \label{fluxlim2d}
\end{align}
where note the flux fraction $\eta_x$ is computed locally in each cell. Similar limiter expressions can be given for each flux depicted in \cref{fig:pos-2d}.

Further, the fluxes can be computed with exponential extrapolation as in \cref{sec:exp}. We can compute the exponential extrapolation in one direction at a time, avoiding the need for a multi-dimensional exponential expression. For example, to compute the exponential extrapolation for $\hat{F}_x(1, 1/3)$, we find the exponential $g(x)$ that is weak-equivalent to $f(x,1/3)$, and then evaluate the exponential expression at the surface. 

Once all the surface terms have been limited to ensure that no control point can become negative, the volume terms can again be limited by scaling all volume terms by a common factor $0 \leq \theta^i \leq 1$ in each cell $i$. Writing the forward-Euler update of each control node $c$ in cell $i$ generically as
\begin{equation}
    f_c^i(t_n + \Delta t) = f_c^i(t_n) + \Delta t S_c^i + \theta^i \Delta t V_c^i,
\end{equation}
with $S_c^i$ and $V_c^i$ the surface and volume terms, respectively, the volume scaling factor is given by
\begin{equation}
    \theta^i = \min_{c\, |\, V_c^i<0} \left(1, \frac{f_c^i + \Delta t S_c^i}{-\Delta t V_c^i}\right).
\end{equation}


\section{Conservation properties for Hamiltonian systems}
Although the positivity-preserving scheme presented above could be used to solve any kind of hyperbolic conservation law of the form of \cref{hyper-pos}, the primary targets of our scheme are Hamiltonian systems like gyrokinetics. In \cref{sec:dgenergy} we showed a DG scheme that conserves energy in Hamiltonian systems. Now let us apply the positivity-preserving limiters and consider how the conservation properties are modified, if at all.

Starting from \cref{eq:dis-weak-form}, the positivity-preserving evolution of the distribution function $f$ is given by
\begin{equation}
    \int_{\mathcal{K}_i} \psi \pderiv{(\mathcal{J}f_h)}{t}\dx{\vec{Z}}- \int_{\mathcal{K}_i} \theta_i \mathcal{J}f_h \dot{\vec{Z}}_h\cdot\pderiv{\psi}{\vec{Z}}\, \dx{\vec{Z}}  + \oint_{\partial \mathcal{K}_i}\psi^- \Lambda\left[\widehat{\jac f_h \dot{\vec{Z}}_h}\right]\cdot\dx{\vec{s}} = 0, \label{eq:dis-weak-form-lim}
\end{equation}
where $\theta_i$ represents the volume term scaling factor in cell $i$, and the notation $\Lambda[\hat{F}]$ represents limiters applied to surface fluxes. To check energy conservation, we first insert the discrete Hamiltonian $H_h$ for the test function $\psi$ and sum over cells, giving 
\begin{equation}
    \sum_i\int_{\mathcal{K}_i} H_h \pderiv{(\mathcal{J}f_h)}{t}\dx{\vec{Z}}= \sum_i \int_{\mathcal{K}_i} \theta_i \mathcal{J}f_h \dot{\vec{Z}}_h\cdot\pderiv{H_h}{\vec{Z}}\, \dx{\vec{Z}}  - \sum_i \oint_{\partial \mathcal{K}_i}H_h^- \Lambda\left[\widehat{\jac f_h \dot{\vec{Z}}_h}\right]\cdot\dx{\vec{s}} = 0.
\end{equation}
Even with the scaling factor $\theta_i$, the volume term vanishes as in the standard case because $\dot{\vec{Z}}_h\cdot \pderivInline{H_h}{\vec{Z}} = \{H_h, H_h\} = 0$. The surface term also vanishes just as in the standard case; the fluxes still exactly cancel at cell boundaries, even with the flux limiters. For  Hamiltonian systems written in canonical form, the remainder of the energy conservation proof from \cref{sec:dgenergy} is unchanged. However, additional complexities arise in our scheme for the symplectic formulation of the electromagnetic gyrokinetic system, which requires the inclusion of limiters in some of the field equations. Extension of the positivity-preserving scheme to EMGK is left to future work, and we discuss the difficulties briefly in Appendix \ref{emgk-pos}.

\section{Results} \label{sec:pos-results}
In this section we implement the positivity-preserving scheme in \gke and present some numerical results. We first study passive advection and then we turn to Hamiltonian systems: the incompressible Euler equations and electrostatic gyrokinetics.

\subsection{1D advection}
\begin{figure}[t!]
    \centering
    \includegraphics[width=.8\textwidth]{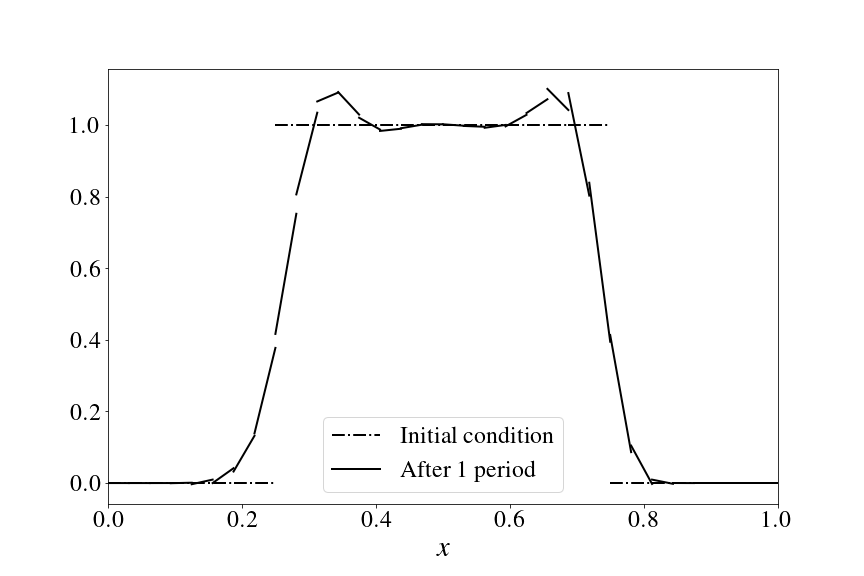}
    \caption[Positivity-preserving advection of a square pulse for one period on a 1D periodic domain.]{Positivity-preserving advection of a square pulse for one period on a 1D periodic domain with 32 cells, overlaid with the initial condition.}
    \label{fig:pos-adv-1d}
\end{figure}

Let us first return to the one-dimensional advection example from \cref{pos:1dadvect}. Again taking a square pulse on a periodic domain, \cref{fig:pos-adv-1d} shows the results of the new scheme. Compared to \cref{fig:no-pos-adv}, we now see no negative overshoots and no negative cell averages. In fact, not only do the cell averages remain positive, but the control nodes in each cell also remain positive, which ensures that slopes do not become unphysically large. Note however that points on cell boundaries can be negative, so long as the control nodes are positive, as can be seen at $x=5/32$. 

\subsection{2D advection}

\begin{figure}[t!]
    \centering
    \includegraphics[width=\textwidth]{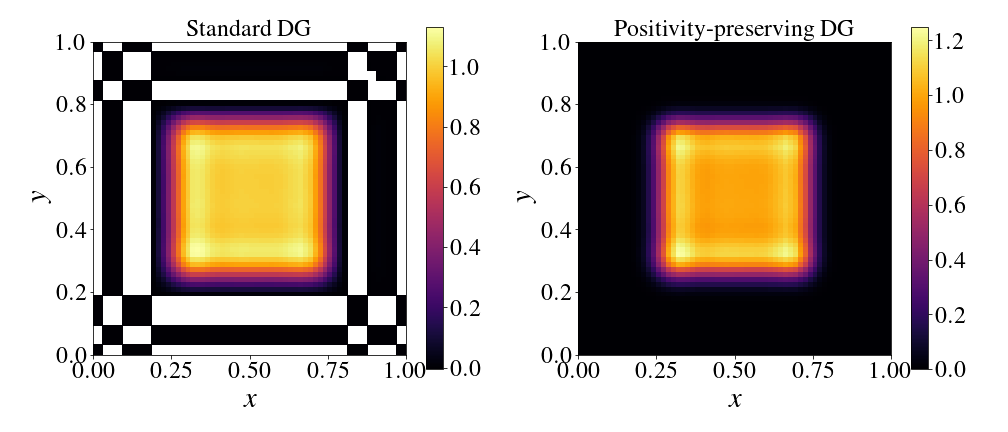}
    \caption[Positivity-preserving advection of a 2D square pulse for one period on a periodic domain.]{Two-dimensional square pulse after one period of diagonal advection on a periodic domain with $32\times32$ cells using standard DG (left) and positivity-preserving DG (right). In the standard case, cells with negative cell average are masked in white.
    }
    \label{fig:pos-adv-2d-diag}
\end{figure}

As a first two-dimensional test, we again consider uniform constant advection of a square pulse, given initially by
\begin{align}
   f(x,y,0) = \begin{cases}
   &1 \qquad |x-x_c| < 1/4\ \mathrm{and}\ |y-y_c| < 1/4 \\
    &0 \qquad \mathrm{otherwise}
    \end{cases}
\end{align}
with $x_c=y_c = 1/2$. We advect the solution diagonally, with velocity components $v_x=v_y=1$, through a periodic domain with $L_x = L_y = 1$. \cref{fig:pos-adv-2d-diag} shows a comparison of the results from the standard DG scheme and the positivity-preserving scheme. Both cases use piecewise-linear ($p=1$) basis functions with 32 cells in each direction. In the standard case, cells with negative cell average are masked in white. The positivity-preserving scheme successfully eliminates these negative regions.

\subsection{2D vortex waltz}

\begin{figure}[t!]
    \centering
    \includegraphics[width=\textwidth]{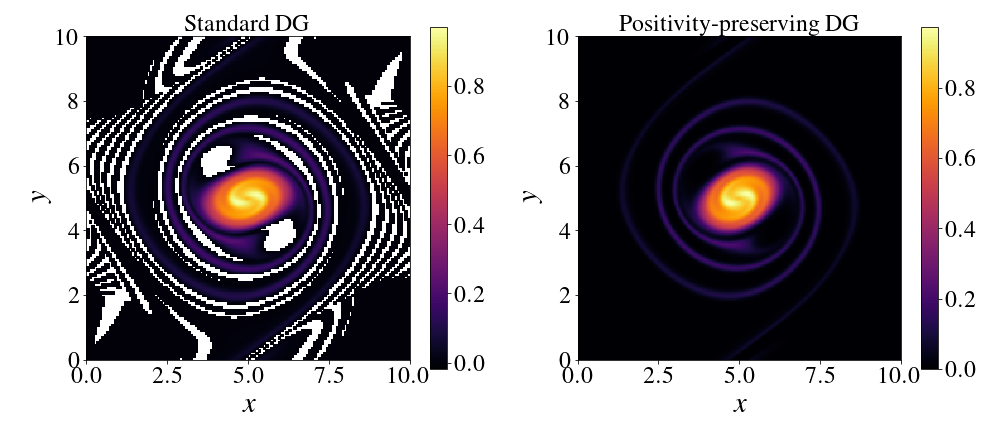}
    \caption[Vorticity for ``vortex waltz'' problem with positivity preservation.]{Vorticity at $t=100$ for ``vortex waltz'' problem on a periodic domain with $128\times128$ cells using standard DG (left) and positivity-preserving DG (right). In the standard case, cells with negative cell-average vorticity are masked in white.
    }
    \label{fig:vortex-waltz}
\end{figure}

A more stringent test of our positivity-preserving scheme and its conservation properties is given by the incompressible Euler system. As we saw in \cref{sec:incomp-euler}, this is a Hamiltonian system, with a conserved energy given by
\begin{equation}
    \mathcal{E} = \int |\nabla_\perp \psi|^2\dx{\vec{Z}}. \label{vortex-energy}
\end{equation}

In the ``vortex waltz'' problem \citep{nielsen1996}, we initialize two Gaussian vortices which merge as
they orbit around each other. The domain is doubly periodic of dimension $10\times 10$ length units. The initial vorticity given by
\begin{align}
  \varpi(x,y,0) = e^{-r_1^2/0.8} + e^{-r_2^2/0.8},
\end{align}
where
$r_i^2 = (x-x_i)^2 + (y-y_i)^2$ with $(x_1,y_1) = (3.5,5.0)$ and
$(x_2,y_2) = (6.5,5.0)$ the initial locations of the peaks. We discretize the system with piecewise-linear basis functions $(p=1)$ on a grid with $128 \times 128$ cells. We show a comparison of the vorticity at $t=100$ from the standard DG scheme and the positivity-preserving scheme in \cref{fig:vortex-waltz}, again masking cells in white that have negative cell-average vorticity.

\begin{figure}[t!]
    \centering
    \includegraphics[width=.7\textwidth]{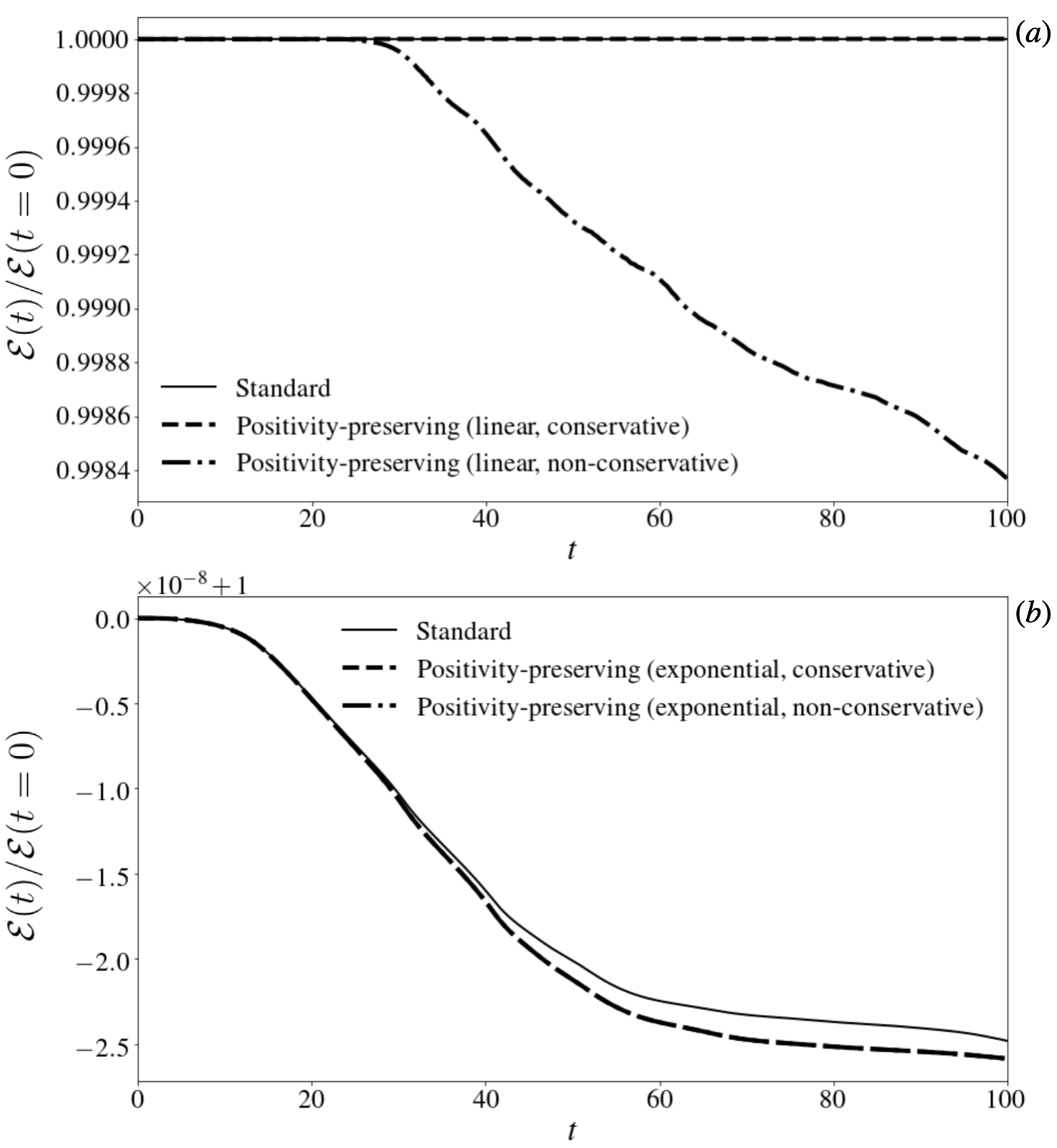}
    \caption[Energy conservation in vortex waltz test.]{Time traces of the total energy in the system $\mathcal{E}(t)$ normalized to the initial energy for three cases: standard DG, the full positivity-preserving DG scheme, and a non-conservative version of the positivity-preserving scheme which uses \emph{post-hoc} diffusion instead of volume term limiters. In $(a)$, a standard linear extrapolation is used to evaluate the surface fluxes, while in $(b)$ the exponential extrapolation from \cref{sec:exp} is used. Note the difference in scale between $(a)$ and $(b)$; the `standard' trace (solid) is the same in both plots. The negligible difference between the conservative and non-conservative traces in $(b)$ shows that using the exponential extrapolation reduces the need for volume limiters.}
    \label{fig:vortex-waltz-energy}
\end{figure}

To verify that the positivity-preserving scheme has not broken energy conservation, we show in \cref{fig:vortex-waltz-energy} time traces of the total energy, given by \cref{vortex-energy}, for three cases: standard DG, the full positivity-preserving scheme, and a non-conservative positivity scheme. In the non-conservative scheme, the surface term limiters are still applied as in the conservative scheme, but we do not apply the volume term limiters. This would keep cell averages positive but could allow unphysical slopes to develop, so we add \emph{post-hoc} rescaling of the slopes at the end of the timestep to maintain realizability, which breaks energy conservation. Indeed, the plot shows that the standard and conservative positivity-preserving schemes conserve the energy well, while the non-conservative scheme has energy errors.

\subsection{5D electrostatic gyrokinetics}

Our most challenging test of the positivity algorithm is its application to the 5D electrostatic gyrokinetic system. 
Without implementing the positivity algorithm, the standard DG discretization of the electrostatic gyrokinetic system results in regions of negative distribution function, leading to regions of negative density and negative temperature. This can lead to unphysical behavior in the collision operator\footnote{To improve robustness, we have altered the implementation of the collision operator in the standard version so that collisions are effectively turned off in cells with negative temperature, which avoids unphysical anti-diffusion in these cells. This improves robustness but does not completely eliminate positivity-related issues in the simulations.} and the sheath boundary conditions, resulting in numerical instabilities.

\begin{figure}[t]
    \centering
    \includegraphics[width=1.0\textwidth]{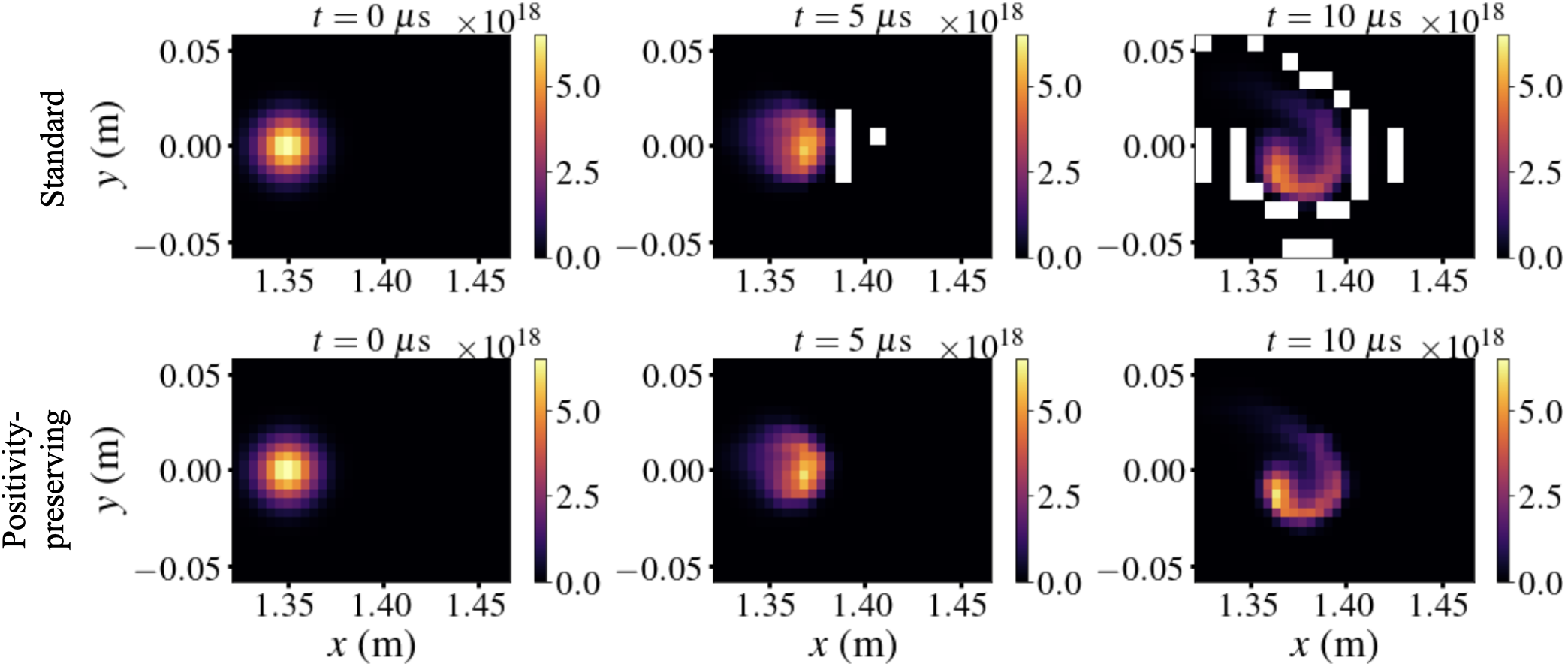}
    \caption[Positivity-preserving 5D collisionless electrostatic gyrokinetic seeded blob simulations.]{Timeseries of electron density at the midplane ($z=0$) for 5D collisionless electrostatic gyrokinetic seeded blob simulations with the standard DG scheme (top row) and the positivity-preserving scheme (bottom row). In the standard case, cells with negative cell-average density are masked in white.}
    \label{fig:nstx-blob-pos}
\end{figure}

As a first test of the positivity algorithm in the electrostatic gyrokinetic system, we perform a collisionless seeded blob test. We initialize a Gaussian blob in helical NSTX-like geometry with $L_z = 100$ m and sheath boundary conditions at $z=\pm L_z/2$. As the blob polarizes it begins to advect radially outwards and also spin due to the Boltzmann spinning effect \citep{angus2012}. With the standard DG discretization, this results in cells with negative cell-average density, as shown in the top row of \cref{fig:nstx-blob-pos}. With the positivity-preserving algorithm, these negative cells are eliminated. \cref{fig:nstx-blob-pos-energy} shows that energy conservation is not altered by the positivity algorithm, with energy still conserved in the system to $\sim \mathcal{O}(10^{-5})$.

\begin{figure}[h!]
    \centering
    \includegraphics[width=.5\textwidth]{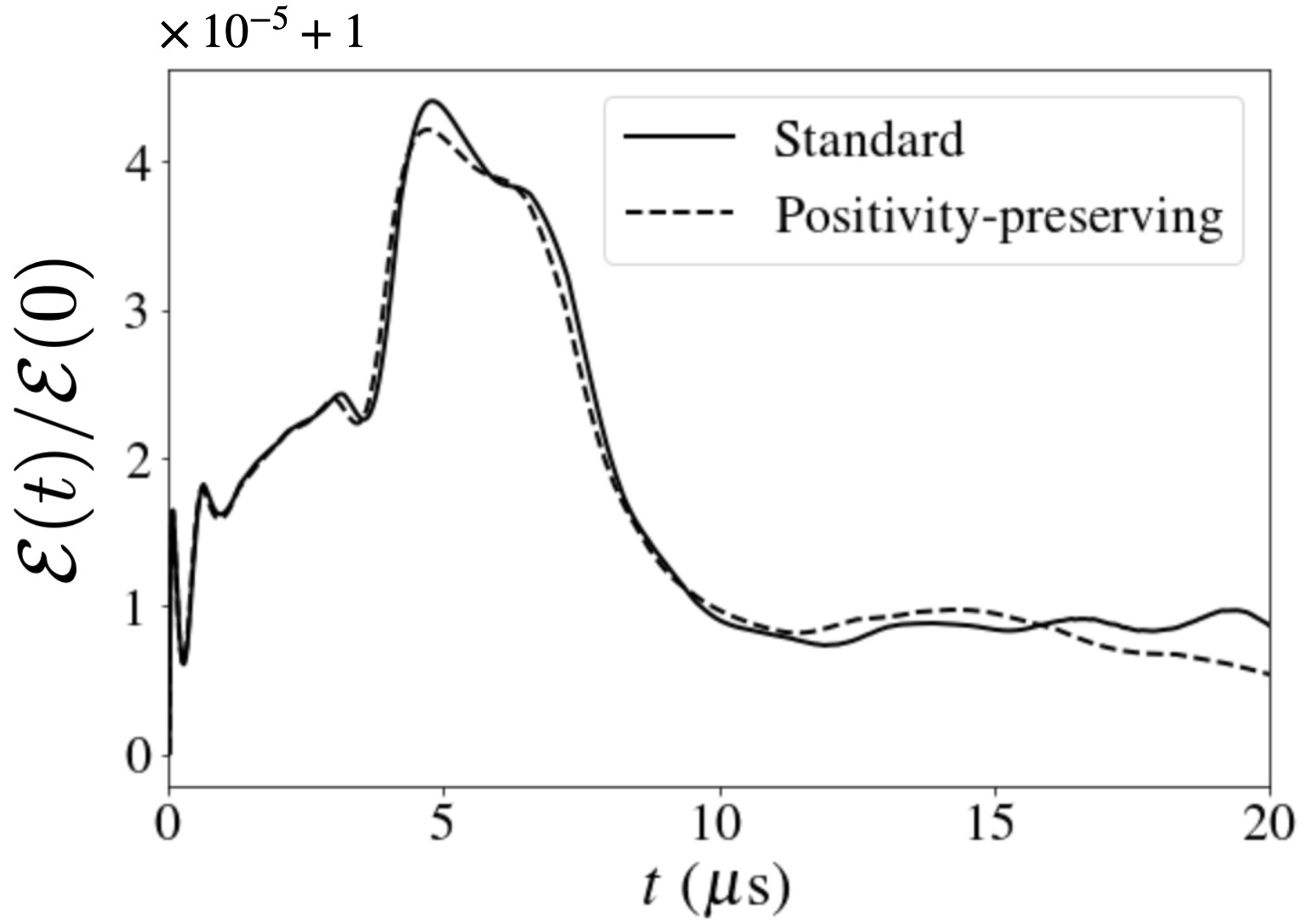}
    \caption[Energy conservation in the electrostatic seeded blob simulations.]{Energy conservation in the electrostatic seeded blob simulations with the standard algorithm and the positivity-preserving algorithm. The positivity algorithm does not adversely affect energy conservation, as desired.}
    \label{fig:nstx-blob-pos-energy}
\end{figure}


%
%

An energy-conserving positivity-preserving DG algorithm for the Dougherty collision operator has also been formulated and will be presented in future work. This will enable full electrostatic  simulations like the ones presented in \cref{sec:emgk-res} that maintain positivity of the distribution function. Implementing the positivity-preserving algorithm in the electromagnetic gyrokinetic system is somewhat more challenging, as we detail in the Appendix \ref{emgk-pos}.

\section{Summary}
In this chapter we have developed a discontinuous Galerkin scheme for maintaining positivity of the distribution function. The scheme has been carefully constructed to avoid \emph{post-hoc} diffusion so that conservation properties are preserved for Hamiltonian systems. The results in \cref{sec:pos-results} show that the scheme is successful in maintaining positivity for passive advection, incompressible Euler, and collisionless electrostatic gyrokinetic systems. Extension to include collisions and electromagnetic effects to the gyrokinetic system is left as future work.

While the simulations in the bulk of this thesis were able to run somewhat robustly with the standard DG algorithm (without any assurances of positivity of the distribution function), there were a number of simulations attempted as part of this thesis that failed due to positivity issues. For example, simulations failed when we tried to use a collision frequency that varied in space and time based on local plasma parameters, because negative local values of density and temperature resulted in an ill-defined collision frequency. We also expect the positivity problem to only get worse as we move to more realistic (and complex) simulation setups and geometries. Thus the work of this chapter is an important and necessary step towards robust, high-fidelity simulations.

\begin{subappendices}

\section{Difficulties in extending the positivity scheme to electromagnetic gyrokinetics} \label{emgk-pos}

Extending the positivity scheme to the electromagnetic gyrokinetics algorithm presented in \cref{sec:gkscheme} is challenging, in part because one does not have all the information needed to compute limiters when the limiters themselves are needed. To illustrate this, first imagine that we (magically) already know what all the limiters will be at the beginning of the timestep, such that no terms (surface or volume) can lead to a negative control node. Neglecting collisions, the DG discretization of the gyrokinetic equation might look something like
\begin{align}
    \int_{\mathcal{K}_i}&\psi \pderiv{(\mathcal{J}f_h)}{t}\dx{^3\vec{R}}\,\dx{^3\vec{v}}  
    \notag \\ &\quad 
    - \int_{\mathcal{K}_i} \theta_i^H \mathcal{J}f_h \dot{\vec{R}}_h\cdot\nabla \psi \,\dx{^3\vec{R}}\,\dx{^3\vec{v}} 
    - \int_{\mathcal{K}_i} \mathcal{J}f_h \left(\theta_i^H \dot{v}^H_{\parallel h}-\theta_i^A \frac{q}{m}\pderiv{A_{\parallel h}}{t}\right) \pderiv{\psi}{v_\parallel}\dx{^3\vec{R}}\,\dx{^3\vec{v}}
    \notag \\ &\quad
    + \oint_{\partial \mathcal{K}_i}\psi^- \Lambda\left[ \widehat{\mathcal{J}f_h}\dot{\vec{R}}_h\right]\cdot \dx{\vec{s}}_R\, \dx{^3\vec{v}}
    + \oint_{\partial \mathcal{K}_i}  \psi^-  \Lambda\left[\widehat{\mathcal{J}f_h} \left(\dot{v}^H_{\parallel h}-\frac{q}{m}\pderiv{A_{\parallel h}}{t}\right)\right]\dx{^3\vec{R}}\,\dx{s}_v 
    \notag \\ &\quad
    = 0. \label{DGgk-pos}
\end{align}
Here, $\theta_i^H$ and $\theta_i^A$ represent volume term scaling factors in cell $i$ for the Poisson bracket and inductive volume terms, respectively, and the notation $\Lambda[\hat{F}]$ represents limiters applied to surface fluxes.

We would like to ensure that this limiter scheme preserves energy conservation. While the volume term limiters $\theta^H$ on the Poisson bracket terms do not affect energy conservation, we have an additional volume term outside the bracket involving $\pderivInline{A_{\parallel h}}{t}$. We also have modifications to the surface terms. In order to maintain energy conservation, we must account for these limiters in the field equations. To see this, take $\psi= H_{s\,h}$ to compute
\begin{align}
    \sum_{s,i}&\int_{\mathcal{K}_i} H_{s\,h}\pderiv{(\mathcal{J}f_{s\,h})}{t}\,\dx{^3\vec{R}}\,\dx{^3\vec{v}}
     = \sum_{s,i}\int_{\mathcal{K}_i}\mathcal{J}f_{s\,h}\, \theta_i^H \left(\dot{\vec{R}}_h\vec{\cdot} \nabla H_{s\,h}+\dot{v}^H_{\parallel h}\pderiv{H_{s\,h}}{v_\parallel}\right)\dx{^3\vec{R}}\,\dx{^3\vec{v}} \notag\\ 
     &- \sum_{s,i}\int_{\mathcal{K}_i} \mathcal{J}f_{s\,h}\,\theta^A_i \frac{q_s}{m_s}\pderiv{A_{\parallel h}}{t} \pderiv{H_{s\,h}}{v_\parallel}\,\dx{^3\vec{R}}\,\dx{^3\vec{v}} 
     -  \sum_{s,i}\oint_{\partial \mathcal{K}_i}H_{s\,h}^- \Lambda\left[\widehat{\mathcal{J}f_h}\dot{\vec{R}}_h\right]\cdot \dx{\vec{s}}_R\, \dx{^3\vec{v}}
     \notag \\ &\quad 
     - \sum_{s,i} \oint_{\partial \mathcal{K}_i}  H_{s\,h}^-\Lambda\left[ \widehat{\mathcal{J}f_h} \left(\dot{v}^H_{\parallel h}-\frac{q}{m}\pderiv{A_{\parallel h}}{t}\right)\right]\dx{^3\vec{R}}\,\dx{s}_v. 
\end{align}
As noted above, despite the inclusion of the volume limiter $\theta^H$, the first volume term still vanishes exactly because $\dot{\vec{R}}_h\vec{\cdot} \nabla H_{h}+\dot{v}^H_{\parallel h}\pderivInline{H_{h}}{v_\parallel}=\{H_{h},H_{h}\}_h=0$. The surface terms also cancel exactly at cell interfaces, so we are left with
\begin{align}
    \sum_{s,i}\int_{\mathcal{K}_i}  H_{s\,h}\pderiv{(\mathcal{J}f_{s\,h})}{t}\,\dx{^3\vec{R}}\,\dx{^3\vec{v}}
    &= 
    -\sum_{s,i}\int_{\mathcal{K}_i} \mathcal{J}f_{s\,h}\,\theta^A_i \frac{q_s}{m_s}\pderiv{A_{\parallel h}}{t} \pderiv{H_{s\,h}}{v_\parallel} \,\dx{^3\vec{R}}\,\dx{^3\vec{v}} \notag \\
    &= 
    - \int_{\mathcal{T}^R}  \pderiv{A_{\parallel h}}{t} \tilde{J}_{\parallel h}\,\,\dx{^3\vec{R}}, \label{Hdfdt-lim}
\end{align}
where we will define a \emph{limited} parallel current, denoted by $\tilde{J}_{\parallel h}$, as
\begin{equation}
\tilde{J}_{\parallel h} = \sum_{s,j}\frac{q_s}{m_s} \int_{\mathcal{K}^v_j} \theta^A_j\pderiv{H_{s\,h}}{v_\parallel} \mathcal{J}f_{s\,h}\,\dx{^3\vec{v}}.
\end{equation}
To regain energetic consistency, this limited $\tilde{J}_{\parallel h}$ must be used in Amp\`ere's law, which becomes
\begin{equation}
\int_{\mathcal{K}^R_i} \nabla_\perp A_{\parallel h} \vec{\cdot} \nabla_\perp \varphi^{(i)}\,\dx{^3\vec{R}} - \oint_{\partial \mathcal{K}^R_i} \varphi^{(i)} \nabla_\perp A_{\parallel h}\cdot\dx{\vec{s}_R}  = \mu_0 \int_{\mathcal{K}^R_i} \varphi^{(i)}\ \tilde{J}_{\parallel h}\,\dx{^3\vec{R}}. \label{ampere-lim}
\end{equation}
Now we can insert $\varphi^{(i)} = (1/\mu_0)\pderivInline{A_{\parallel h}}{t}$ to compute
\begin{align}
    \pderiv{\mathcal{E}_{B\,h}}{t} &= \sum_{i} \int_{\mathcal{K}^R_i} \frac{1}{\mu_0}\nabla_\perp A_{\parallel h} \vec{\cdot} \nabla_\perp \pderiv{A_{\parallel h}}{t} \,\dx{^3\vec{R}} = \int_{\mathcal{T}^R} \pderiv{A_{\parallel h}}{t} \tilde{J}_{\parallel h}\,\dx{^3\vec{R}},
\end{align}
which now cancels the term leftover from \cref{Hdfdt-lim}. 
Now to derive the self-consistent Ohm's law, we take the time derivative of \cref{ampere-lim}, giving
\begin{align}
    &\int_{\mathcal{K}^R_i} \nabla_\perp \pderiv{A_{\parallel h}}{t} \vec{\cdot} \nabla_\perp \varphi^{(i)}\,\dx{^3\vec{R}} 
    - \oint_{\partial \mathcal{K}^R_i} \varphi^{(i)}\nabla_\perp \pderiv{A_{\parallel h}}{t} \cdot \dx{\vec{s}_R} 
    \notag \\ &\quad
    + \mu_0\sum_s q_s \int_{\mathcal{K}_i^R} \varphi^{(i)} \left[\sum_{j} \oint_{\partial\mathcal{K}^v_j}  \theta^A_j \bar{v}_\parallel^-\Lambda\left[\widehat{\mathcal{J} f_{s\,h}}\left({\dot{v}^H_{\parallel h}} -\frac{q_s}{m_s}\pderiv{A_{\parallel h}}{t}\right) \right]\dx{s_v}\right]\dx{^3\vec{R}} 
    \notag\\ &
    = \mu_0\sum_s q_s \int_{\mathcal{K}^R_i} \varphi^{(i)} \Bigg[\sum_j\int_{\mathcal{K}^v_j} \theta^A_j\bar{v}_\parallel \pderiv{(\mathcal{J}f_{s\,h})}{t}^\star \dx{^3\vec{v}} \Bigg]\dx{^3\vec{R}}, \label{ohm-lim}
\end{align}
where where we have assumed $p=1$, and
\begin{align}
    &\int_{\mathcal{K}_i}\psi \pderiv{(\mathcal{J}f_h)}{t}^\star\dx{^3\vec{R}}\,\dx{^3\vec{v}} = 
    \notag \\ &\quad
    \int_{\mathcal{K}_i} \theta^H_i \mathcal{J}f_h \left(\dot{\vec{R}}_h\cdot\nabla \psi +\dot{v}^H_{\parallel h} \pderiv{\psi}{v_\parallel} \right) \dx{^3\vec{R}}\,\dx{^3\vec{v}} 
    - \oint_{\partial \mathcal{K}_i}\psi^- \Lambda\left[\widehat{\mathcal{J}f_h}\dot{\vec{R}}_h\right]\cdot \dx{\vec{s}}_R\, \dx{^3\vec{v}}.
\end{align}
This limiter-modified Ohm's law presents a number of challenges. First, the surface limiters on the second line in \cref{ohm-lim} make the problem of solving for $\pderivInline{A_{\parallel h}}{t}$ a nonlinear one. These limiters act in phase-space, not real-space, which introduces additional degrees of freedom. And these are issues even when all the limiters are known at the time of the solve. In practice, there is the additional complication that all the limiters (the surface limiters and the volume limiters $\theta^A$) themselves depend on $\pderivInline{A_{\parallel h}}{t}$. Thus we essentially need to know $\pderivInline{A_{\parallel h}}{t}$ to evaluate limiters in Ohm's law in order to solve for $\pderivInline{A_{\parallel h}}{t}$. Inevitably, one will require an iteration scheme to solve this circular problem, although it is difficult to say whether such a scheme would converge quickly, if at all. 

\section{Generalization of positivity constraints to higher polynomial order and dimensionality} \label{sec:pos-gen}

In this section we consider a rigorous  procedure for tractably evaluating the positivity constraints in higher dimensionality and higher polynomial order. A key result of \cref{sec:pos-definition} was that in 1D, the most-extreme realizable solution has $f_1 = \pm \sqrt{3}f_0$. The corresponding exponential solution has $g_1 \rightarrow \pm \infty$, so that the exponential approaches a delta function at the cell boundary.

Another way to obtain this result is to project a delta function evaluated at the cell boundary (or more precisely, just inside the boundary) onto the piecewise-linear modal basis,
\begin{gather}
    \delta_0 = \int_{-1}^1 \frac{1}{\sqrt{2}}\delta(x-\pm 1)\, \dx{x} = \frac{1}{\sqrt{2}} \\
    \delta_1 = \int_{-1}^1 \frac{\sqrt{3}}{\sqrt{2}}x\delta(x-\pm 1)\, \dx{x} = \pm\frac{\sqrt{3}}{\sqrt{2}},
\end{gather}
so that the delta function on the boundary has a weak-equivalent piecewise-linear representation given by 
\begin{equation}
    \delta(x-\pm 1) \doteq  \frac{1}{\sqrt{2}}\delta_0 + \frac{\sqrt{3}}{\sqrt{2}}x \delta_1 = \frac{1}{2} \pm \frac{3}{2} x \equiv \delta^\pm(x).
\end{equation}
Indeed, $\delta_1 = \pm \sqrt{3} \delta_0$, so we have recovered the earlier result. Also note that the functions $\delta^\pm(x)$ have zeros at the positivity control nodes $x=\pm 1/3$. 

Now consider that we could use the functions $\delta^\pm(x)$ as basis functions and expand the solution as
\begin{equation}
    f_h = f_+ \delta^+ + f_- \delta^- = f_+\left(\frac{1}{2} + \frac{3}{2}x\right) + f_-\left(\frac{1}{2} - \frac{3}{2}x\right).
\end{equation}
This is effectively a nodal basis, with nodal values $f_\pm = f_h(x= \pm 1/3)$. Note that we can also obtain the coefficients $f_\pm$ via projection:
\begin{equation}
    f_\pm = \int_{-1}^1 \left(\frac{1}{2}\pm\frac{1}{2}x\right) f_h\,\dx{x} \equiv \int_{-1}^1 \delta_\pm(x) f_h\, \dx{x}.
\end{equation}
Thus we will denote the functions $\delta^\pm(x)$ as the positivity \emph{expansion} basis functions, and the functions $\delta_\pm(x)$ as the positivity \emph{projection} basis functions. Finally, as before, the positivity constraint for the $f_\pm$ coefficients is $f_\pm \geq 0$. 

Note that we have now slightly modified our positivity definition. Instead of using an exponential basis as the non-polynomial positive-definite basis set with which we require weak-equality, we will now use delta functions. Thus the positive-definite representation of the solution is
\begin{equation}
    g_h = g_+ \delta(x-1) + g_- \delta(x+1).
\end{equation}
Requiring this non-polynomial function to be positive on the entire cell domain gives the constraint $g_\pm\geq 0$.
Enforcing weak-equality with the piecewise-linear solution $f_h = f_+ \delta^+ + f_- \delta^-$ now simply gives $f_\pm = g_\pm$, since by construction the basis functions $\delta^\pm$ are weak-equivalent to the delta functions $\delta(x-\pm1)$. This once again gives that the positivity constraints on the coefficients of $f_h$ are $f_\pm\geq0$. Thus despite the change in positive-definite basis functions from exponentials to delta functions, the result is the same, which suggests some degree of equivalence between the two choices.

Now consider the one-dimensional piecewise-quadratic case. The orthonormal modal basis set on the cell $x\in[-1,1]$ is
\begin{equation}
    \psi = \left\{\frac{1}{\sqrt{2}},\ \frac{\sqrt{3}}{\sqrt{2}}x,\ \frac{3\sqrt{5}}{2\sqrt{2}}\left(x^2-\frac{1}{3}\right)\right\}.
\end{equation}
Once again, we can project delta functions just inside the cell boundaries, $\delta(x-\pm1)$, onto the basis, giving
\begin{gather}
    \delta^\pm = \frac{1}{2}\pm\frac{3}{2}x+\frac{15}{4}\left(x^2-\frac{1}{3}\right) \doteq \delta(x-\pm 1).
\end{gather}
We will again use these functions as expansion basis functions.
Given the extra degree of freedom in the piecewise-quadratic case, we need an additional basis function. We can obtain the final basis function by taking a delta function at the cell center, $\delta(x-0)$; we choose $x=0$ so that the basis is symmetric about the cell center. Projecting onto the piecewise-quadratic basis gives
\begin{equation}
    \delta^\circ = \frac{1}{2}-\frac{15}{8}\left(x^2-\frac{1}{3}\right) \doteq \delta(x-0).
\end{equation}
Thus, the piecewise-quadratic positivity expansion basis is given by $\{\delta^+,\delta^\circ,\delta^-\}$, allowing us to expand the solution as
\begin{equation}
    f_h = f_+\delta^+ + f_\circ \delta^\circ + f_- \delta^-.
\end{equation}
Unlike in the piecewise-linear case, the coefficients $\{f_+, f_\circ, f_-\}$ do not coincide with nodal values of $f_h$. Instead, they must be obtained by using projection basis functions, which can be shown to be
\begin{equation}
    \delta_+ = \frac{1}{2}x(x+1), \qquad \delta_\circ = 1-x^2, \qquad \delta_- = \frac{1}{2}x(x-1), \label{proj2}
\end{equation}
so that
\begin{align}
    f_{+} = \int_{-1}^1 \delta_{+} f_h\, \dx{x}, \qquad
    f_{\circ} = \int_{-1}^1 \delta_\circ f_h\, \dx{x}, \qquad
    f_{-} = \int_{-1}^1 \delta_{-} f_h\, \dx{x}.
\end{align}
Once again, the requirements $g_h = g_+ \delta(x-1) + g_\circ \delta(x) + g_- \delta(x+1) \geq 0$ (for all $x\in[-1,1]$) and $f_h\doteq g_h$ give the positivity constraints that the coefficients of $f_h = f_+ \delta^+ + f_\circ \delta^\circ + f_- \delta^-$ must be non-negative: $f_\pm,f_\circ \geq 0$. 

We have now successfully generalized the procedure for evaluating positivity constraints to $p=2$. We can further generalize to arbitrarily high order by making the following two observations. First, consider that we obtained the positivity expansion basis functions above by projecting delta functions centered at $x=-1,1$ for $p=1$ and $x=-1,0,1$ for $p=2$. We can recognize that these sets of points are Gauss-Lobatto nodes\footnote{The Gauss-Lobatto nodes always include the cell endpoints $x=\pm 1$ in the node set. Variants include the Legendre-Gauss-Lobatto nodes (commonly referred to as just the Gauss-Lobatto nodes), where the nodes of order $p$ are given by roots of the polynomial $P'_p(x)$, where $P(x)$ is a Legendre polynomial; and the Chebyshev-Gauss-Lobatto nodes, where the nodes are located at $x_j = \cos(\pi j/p)$ for $j=0,...,p$. For $p\leq 2$, these variants give identical nodes.} for $p=1$ and $p=2$, respectively. Second, note that the positivity projection functions are the Lagrange basis functions for the same set of Gauss-Lobatto nodes. For arbitrary $p$, the Lagrange basis functions for nodes $x_j$ are given by
\begin{equation}
    \ell_j(x) = \prod_{\begin{smallmatrix}0\le m\le p\\ m\neq j\end{smallmatrix}} \frac{x-x_m}{x_j-x_m}.
\end{equation}

Thus we now have a general procedure for generating positivity expansion and projection basis functions in one-dimension for arbitrary polynomial order. The steps are
\begin{enumerate}
    \item The positivity \emph{expansion} basis functions can be found by projecting delta functions centered at Gauss-Lobatto nodes $x_j$ onto the orthonormal modal basis $\vec{\psi}(x)$, 
    \begin{gather}
        \ell^j(x) = \vec{\psi}(x) \cdot \int_{-1}^1  \vec{\psi}(x)\  \delta(x-x_j)\,\dx{x}= \vec{\psi}(x)\cdot\vec{\psi}(x_j).
    \end{gather}
    We can then use this basis to expand $f_h$ as
    \begin{gather}
        f_h = \sum_{j=0}^p g_{j} \ell^j(x),
    \end{gather}
    where we will now use $g_j$ to denote the coefficients of the solution expanded on the positivity basis (to distinguish from modal coefficients $f_j$).
    \item The positivity \emph{projection} basis functions are the Lagrange basis functions for the Gauss-Lobatto nodes:
    \begin{equation}
        \ell_j(x) = \prod_{\begin{smallmatrix}0\le m\le p\\ m\neq j\end{smallmatrix}} \frac{x-x_m}{x_j-x_m}.
    \end{equation}
    Note that in general, the Lagrange basis functions for a particular set of nodes $x_j$ can be derived by computing the matrix
    \begin{equation}
        M_{jk} = \int_{-1}^1 \psi_j \ell^k\,\dx{x} = \psi_j(x_k),
    \end{equation}
    where the second equality assumes that the $\psi_j$ are orthonormal. Then the Lagrange basis functions are given by
    \begin{equation}
        \ell_j = \sum_k(M^{-1})_{jk} \psi_k.
    \end{equation}
    These projection basis functions can be used to find the coefficients
    \begin{equation}
        g_j = \int_{-1}^1 \ell_j(x) f_h\,\dx{x}, \label{pos-coeff}
    \end{equation}
    since 
    \begin{equation}
        \int_{-1}^1 \ell_i \ell^j\,\dx{x} = \delta_{ij}. \label{orthog}
    \end{equation}
    The $g_j$ are effectively the projection of $f_h$ onto a Gauss-Lobatto nodal (Lagrange) basis. Note however that this is not the same as directly evaluating $f_h$ at the Gauss-Lobatto nodes.
    \item The piecewise-polynomial solution $f_h$ is positive \emph{in the weak sense} if all coefficients $g_j$ from \cref{pos-coeff} are non-negative, so that $g_j\geq 0$ for all $j$. 
\end{enumerate}
The above procedure can then be further generalized to higher dimension by taking (possibly sparse) tensor products of the one-dimensional positivity expansion and projection basis functions. 
\end{subappendices}

\chapter{Summary and future work}\label{ch:conclusion}
\section{Summary}

The main advance of this thesis was the development of the first capabilities for simulating electromagnetic gyrokinetic turbulence on open magnetic field lines. This is an important step towards comprehensive electromagnetic gyrokinetic simulations of the coupled edge/SOL system. In the past, including electromagnetic effects in gyrokinetic codes has been challenging, as there are delicate issues such as the Amp\`ere cancellation problem that must be handled properly. In our continuum full-$f$ approach, we build on the successes of continuum $\delta f$ gyrokinetic codes in the core which have mostly avoided the cancellation problem. The inclusion of electromagnetic effects in gyrokinetic simulations that can handle the unique challenges of the boundary plasma (large fluctuations, open and closed field line regions, \emph{etc.}) is critical to the understanding of phenomena such as edge-localized modes and the pedestal, for which electromagnetic dynamics are expected to play a key role.

In \cref{ch:emgk} we gave a first-principles derivation of the electromagnetic gyrokinetic system, in the limit of interest for our present work. This derivation used phase-space-Lagrangian Lie perturbation methods to systematically derive a self-consistent, energy-conserving, and global gyrokinetic system, including electromagnetic perturbations. We used the weak-flow ordering, which simultaneously allows large perturbations $q\Phi/T\sim1$ at long wavelengths ($k_\perp \rho\sim\epsilon_V\ll1$) and small perturbations $q\Phi/T\sim\epsilon_V$ at short wavelengths ($k_\perp \rho\sim 1$), along with perturbations at intermediate scales. We also used the symplectic ($v_\parallel$) formulation of electromagnetic gyrokinetics, which results in the explicit presence of the inductive electric field in the gyrokinetic equation. After deriving the general formalism including finite-Larmor-radius (FLR) corrections, we consistently reduced the system to the long-wavelength limit by neglecting first- and second-order terms in the single particle Lagrangian to obtain the guiding-center Lagrangian, which contains no gyroaverages. Variational derivation of the field equations resulted in a self-consistent, energy-conserving system for electromagnetic gyrokinetics in the long-wavelength limit. We take this limit for simplicity of implementation in the \gke code, with extension of the implementation to include FLR terms left as future work. We summarized the system implemented in \gke in \cref{sec:gk-summary}.

We went to great lengths to ensure that the underlying system is self-consistent and conservative, so we also needed a robust numerical method with a discretization scheme that preserves these properties. This was the topic of \cref{ch:dg}. We have employed the discontinuous Galerkin method, a high-order numerical method that combines attractive features of finite-element and finite-volume methods. After discussing a discontinuous Galerkin scheme for general Hamiltonian systems that preserves energy by design (in the continuous-time limit), we applied the scheme to the electromagnetic gyrokinetic system. The scheme was then implemented in the gyrokinetic module of the \gke plasma simulation framework.
Linear benchmarks were shown to verify the implementation. The success of these benchmarks, especially for cases with high $\beta$ and small $k_\perp \rho$, indicated that the Amp\`ere cancellation problem is avoided. We confirmed this by deriving a semi-discrete Alfv\'en wave dispersion relation. As a result, we can handle electromagnetic fluctuations in a stable, robust, and efficient manner.

The success of the scheme led to the first published simulations of electromagnetic gyrokinetic turbulence on open field lines, detailed in \cref{ch:nstx-results}. As a rough model of the scrape-off layer in the National Spherical Torus Experiment (NSTX) experiment at PPPL, we took a simple helical configuration (like a simple magnetized torus, or SMT) with field lines wrapping helically around the torus and terminating on conducting plates at the top and bottom. This model system contains many of the necessary ingredients for SOL dynamics, including bad curvature and Debye sheath effects, which are handled via conducting-sheath boundary conditions.
Initial results showed that when electromagnetic effects are included, high $\beta$ blobs can bend and stretch the magnetic field lines as they move radially outwards in the SOL. Qualitative comparisons to a corresponding electrostatic simulation showed differences in blob dynamics, with non-adiabatic electron dynamics playing a key role in the electromagnetic case due to slowing of the parallel response. We then performed a study of the effects of increasing $\beta$ on the SOL dynamics. At higher $\beta$, the influence of electromagnetic effects became stronger, resulting in steepening of pressure gradients near the source region and flattening of gradients in the remainder of the domain. We observed a transition from interchange-like modes with $k_\parallel\sim 0$ to ballooning-like modes with finite $k_\parallel$ as pressure gradients $(\alpha^\mathrm{SMT})$ increased above the ballooning stability threshold in the source region. Radially inward magnetic flutter particle transport off midplane, resulting from parallel motion of electron along radially-bowed-out field lines, was observed to increase roughly as $\beta^2$. Meanwhile the $E\times B$ component of the radial particle transport only scaled linearly with $\beta$. This led to slightly reduced radial transport in the high $\beta$ electromagnetic cases, resulting in slightly higher peak particle and heat loads on the end plates compared to corresponding electrostatic cases. These results could have important implications for the transport of high $\beta$ blobs and ELM filaments. Further, the electromagnetic mechanism resulting in the steepening of gradients in the source region could have implications for pedestal formation and thus deserves more thorough study. Crucially, our electromagnetic simulations were not significantly more expensive than  corresponding electrostatic simulations, which should allow the routine inclusion of electromagnetic effects in future results. 

We worked on advancing to more realistic SOL geometry in \cref{ch:geometry}. We adopted a generalized field-aligned non-orthogonal coordinate system, and expressed the gyrokinetic system in these coordinates. While these coordinates break down at the separatrix in diverted geometries due to a singularity at the X-point, field-aligned coordinates could still be used for efficient discretization on either side of the separatrix, stitched to a non-aligned domain in the near vicinity of the separatrix. We then focused on how to formulate field-aligned coordinate systems for use in flux-tube-like domains in the SOL. We started with a helical configuration with magnetic shear, which generalizes the simple geometry used in \cref{ch:nstx-results}. Preliminary electrostatic simulations in this configuration showed that transport is reduced in more sheared configurations. We then formulated field-aligned coordinate systems based on an analytical Solov'ev model SOL equilibrium and an analytical concentric circular equilibrium. The latter is a common geometry used in the core region, especially for inter-code benchmarking. We presented results from a preliminary electrostatic ITG benchmark based on the Cyclone base case in circular core geometry, which compared well with results from other codes in the long-wavelength limit where our system is valid.

In \cref{ch:positivity} we tackled the problem of positivity in our discontinuous Galerkin scheme. Simulations can suffer from accuracy and robustness issues because the standard DG scheme does not guarantee that the distribution function will remain positive (even in the cell-average). We developed a novel scheme for both defining and preserving positivity in the DG discretization. Importantly, the scheme was designed without \emph{post-hoc} diffusion that is used in many existing positivity-preserving algorithms. This allows the scheme to preserve energy conservation while maintaining positivity, even in Hamiltonian systems like gyrokinetics where energy conservation relies on higher-order moments. We then implemented the scheme in \gke and performed a variety of numerical tests for advection, the 2D incompressible Euler system, and the 5D electrostatic gyrokinetic system. The success of the scheme in maintaining positivity and preserving energy conservation, even in 5D, is a significant advance. 

\section{Future work}

While the work of this thesis has advanced the modeling capabilities of the \gke code, there are a number of areas that remain in order to produce realistic results for direct comparison with existing experiments or prediction of future ones. The following list focuses on enhancements requiring further code and algorithm development, many of which are already in progress.

\begin{itemize}
    \item {\bf Closed-field-line boundary conditions:} The field-aligned geometry formulation presented in \cref{ch:geometry} can also be used for closed-field-line regions. What remains is the implementation of a boundary condition for closed-field-line regions. This work is currently underway, led by Mana Francisquez, using the twist-and-shift approach \citep{beer1995field}. This requires careful interpolation of mis-aligned sheared grids at the ends of the domain along the field line. Once ready, this will allow simulations in a limiter configuration containing both open and closed field line regions. This configuration has been used by several fluid and gyrofluid codes \citep{ribeiro2008,halpern2016, francisquez2017global}, and can be used to study SOL flows, the edge radial electric field, and resulting edge toroidal rotation. These are all relevant to pedestal formation and the L-H transition. \citet{parra2008, parra2010} have stressed the importance of third-order terms in the Hamiltonian that are required to accurately calculate toroidal rotation. While these subtleties must be investigated in detail, Parra \& Catto's results are for the low-flow, up-down symmetric, gyro-Bohm regime. In the edge, the gyro-Bohm scaling breaks down because eddy sizes are not much smaller than radial gradient scale lengths, and so including the third-order Hamiltonian terms may not be required for studying rotation mechanisms in the edge.
    
    \item {\bf Diverted geometry with X-point:} As we have discussed, the X-point is a significant challenge because field-aligned coordinates are singular on the separatrix. This has led to a number of new approaches, such as the flux-coordinate-independent (FCI) approach \citep{hariri2013,hariri2014,stegmeir2016}, which move away from the field-aligned approach. Implementation of FCI or a related approach near the X-point could allow simulation of diverted geometries. Ideally, one could still use conventional field-aligned domains in the core and in the SOL, and only use a non-aligned domain in the immediate vicinity of the separatrix. Stitching these domains together will require sophisticated interpolation and mapping schemes, especially if conservation laws are to be preserved.
    
    \item {\bf Neutral modeling:} Neutral interactions play a significant role in plasma-material interactions that dictate much of the SOL dynamics and evolution. As such, modeling neutrals is critical to producing experimentally-relevant results and predictions. Neutral modeling work is underway in \gke, led by Tess Bernard, leveraging the existing 6D Vlasov kinetic module to produce a kinetic Boltzmann neutral model. The main interaction mechanisms of electron-impact ionization, charge exchange, and radiative recombination are modeled.
    
    \item {\bf Gyroaveraging and higher order terms:} While we derived a long-wavelength limit of the gyrokinetic system, it is important to generalize to shorter wavelengths $k_\perp \rho \sim 1$ within the weak-flow ordering. This involves gyroaveraging operations in the gyrokinetic equation and the field equations. Gyroaveraging is relatively simple in the Fourier spectral representation of many core gyrokinetic codes, where simply multiplying by the Bessel function $J_0(k_\perp v_\perp/\Omega )$ gives gyroaveraging; however, in real space implementations, gyroaveraging requires integral operations that sample around the gyro-orbit. The finite-element implementation of \citet{maurer2020} is likely a good starting point for a gyroaveraging implementation in \gke. Additionally, the second-order $E\times B$ energy term in the Hamiltonian should be included, so that a time-evolving density can be used in the polarization term in the Poisson equation instead of the linearized polarization used in this work. This will be important for cases like pedestal formation where there is significant evolution of the density profile. With these additions, we could use the system given in Case 1 from \cref{variational}.
    
    \item {\bf More realistic/efficient collision operators:}
    We have used a model Dougherty collision operator in this work, and we have taken a constant-in-space and constant-in-time collision frequency. A time- and spatially-varying collision frequency has been implemented in \gke, but it suffers from robustness issues likely related to positivity. We have also used an artificially-reduced  collision frequency in this work to avoid severe timestep restrictions; this issue could be alleviated with an implicit or super-timestepping implementation of the collision operator. Further, a more realistic collision operator beyond the simple Dougherty model should be implemented. Preliminary work on a full nonlinear Fokker-Planck collision model in Rosenbluth potential form in \gke has been led by Petr Cagas.
    
    \item {\bf Porting \gke to GPUs:} Today, many of the world's fastest supercomputers all derive a majority of their computing power from graphics processors (GPUs). The rise of GPUs in scientific computing over the past decade has been driven in large part due to their supreme performance for machine learning applications. To fully leverage the power of these machines, the algorithms in \gke must be efficiently implemented on GPUs. Work has begun to port the compute-intensive \gke solver kernels to a CUDA implementation for use on NVIDIA GPUs, with significant progress made by myself, Ammar Hakim, Jimmy Juno, Mana Francisquez, and others on the \gke team as part of GPU hackathons hosted by Princeton.
    
    \item {\bf Extensions of the positivity algorithm, including collisions, electromagnetic gyrokinetics, and higher polynomial order:} 
    The positivity algorithms detailed in \cref{ch:positivity} are a significant step towards improving robustness of \gke simulations. Currently, these algorithms are only implemented for the electrostatic gyrokinetic system. An implementation including collisions has also been made by Mana Francisquez, but at the time of writing there is some issue in the implementation with energy conservation. Further extension to the electromagnetic gyrokinetic system will require additional work, due to the issues discussed in Appendix \ref{emgk-pos}. The algorithm is also formulated in a general way so that in principle it could be generalized to higher polynomial order. Finally, the algorithm could also be implemented into the Vlasov-Maxwell module in \gke.
    
\end{itemize}

\noindent
Along with the further development detailed above, there are a number of interesting and important physics problems that can leverage the electromagnetic gyrokinetic capabilities developed in this thesis. An immediate goal will be to investigate the importance of electromagnetic effects on SOL dynamics in realistic tokamak geometries at experimental parameters. Once the additional capability to simultaneously model open- and closed-field-lines has been developed, we will be able to study the dynamics of the coupled pedestal/SOL system. This is of critical importance for the development of a fusion pilot plant, and a major theme of the recent FESAC Long Range Planning report: ``a sustained burning plasma at high power density is required simultaneously with a solution to the power exhaust challenge: mitigating the extreme heat fluxes to materials surrounding the plasma'' \citep{carter2020}. 

Electromagnetic effects are expected to play a significant role in the pedestal region, in part due to large pressure gradients that push the plasma close to the ideal-MHD stability threshold \citep{snyder2011}. Further, while electrostatic turbulence is often suppressed in the pedestal region by $E\times B$ shear and other effects, the transport can remain above neoclassical levels due to the presence of electromagnetic instabilities such as microtearing modes \citep{hatch2016}. Since turbulence suppression in the pedestal region plays a key role in pedestal formation and sustenance in H-mode, understanding the impact of electromagnetic effects on pedestal transport is of critical importance to the success of current and future fusion devices such as ITER. Edge-localized modes (ELMs) can also play a key role in limiting the pedestal pressure gradient, and ELMs are strongly electromagnetic. Thus self-consistent study of pedestal dynamics requires modeling the coupled pedestal/SOL system, but previous efforts have relied on the electrostatic approximation to neglect electromagnetic perturbations \citep{idomura2009, abiteboul2013,churchill2017,ku2018fast}. By leveraging and extending the unique capabilities developed in this thesis to model full-$f$ electromagnetic gyrokinetic turbulence in the boundary region, we will be able to model the evolution of the coupled pedestal/SOL system in the presence of electromagnetic microturbulence. This will enable exciting and impactful research that will be valuable for understanding current experiments and ensuring the success of future fusion reactors.

\singlespacing

\cleardoublepage
\ifdefined\phantomsection
  \phantomsection  
\else
\fi
\addcontentsline{toc}{chapter}{Bibliography}
\bibliography{library}

\end{document}